\documentclass[lettersize,journal,twoside]{IEEEtran}
\IEEEoverridecommandlockouts
\normalsize

\usepackage{cite}
\usepackage{url}
\usepackage{hyperref}
\hypersetup{colorlinks=true, linkcolor=red, citecolor=red, urlcolor=blue} 
\usepackage{amssymb}
\usepackage{amsmath, amsfonts}
\usepackage{array}
\usepackage{amsthm}
\allowdisplaybreaks 
\usepackage{autobreak}
\usepackage{mathptmx}
\usepackage{mathtools, cuted}
\usepackage{algpseudocode}
\usepackage[ruled,vlined,linesnumbered]{algorithm2e}
\usepackage{graphicx}
\usepackage{bbm}
\usepackage{tablefootnote}
\usepackage{multirow}
\usepackage{stfloats}
\usepackage{tabu}
\usepackage{array,booktabs}
\newcolumntype{C}[1]{>{\centering\arraybackslash}p{#1}}
\usepackage{color}

\usepackage[table]{xcolor}
\usepackage{caption, subcaption}
\allowdisplaybreaks

\newtheorem{Theorem}{Theorem}
\newtheorem{Lemma}{Lemma}
\newtheorem{Corollary}{Corollary}

\newtheorem{Remark}{Remark}

\newcommand{\rs}{\!\!}
\newcommand{\bblue}{\textcolor{black}}

\newcommand{\Bquad}{\qquad\qquad\qquad\qquad\qquad\qquad\qquad\qquad\qquad}
\newcommand{\Mquad}{\qquad\qquad\qquad\qquad\qquad}
\newcommand{\Squad}{\qquad\qquad\qquad}

\title{Hierarchical Federated Learning in Wireless Networks: Pruning Tackles Bandwidth \\Scarcity and System Heterogeneity}

\author{
Md Ferdous Pervej, \IEEEmembership{Member, IEEE}, Richeng Jin, \IEEEmembership{Member, IEEE}, and Huaiyu Dai, \IEEEmembership{Fellow, IEEE}
\thanks{
This research was supported in part by the National Natural Science Foundation of China under Grants 62301487, in part by the Zhejiang Provincial Natural Science Foundation of China under Grant No. LQ23F010021 and LD21F010001, in part by the Ng Teng Fong Charitable Foundation in the form of ZJU-SUTD IDEA Grant under Grant No. 188170-11102,  and in part by the US National Science Foundation under grants CNS-1824518 and ECCS-2203214. (\textit{Corresponding author: Richeng Jin}.)}
\thanks{M. F. Pervej was with the Department of Electrical and Computer Engineering, North Carolina State University, Raleigh, NC 27695, USA, and now is with the Ming Hsieh Department of Electrical and Computer Engineering, University of Southern California, Los Angeles, CA 90089, USA (e-mail: pervej@usc.edu).}
\thanks{H. Dai is with the Department of Electrical and Computer Engineering, North Carolina State University, Raleigh, NC 27695, USA (e-mail: hdai@ncsu.edu).}
\thanks{R. Jin is with the College of Information Science and Electronic Engineering, Zhejiang University, the Zhejiang–Singapore Innovation and AI Joint Research Lab, and the Zhejiang Provincial Key Lab of Information Processing, Communication and Networking (IPCAN), Hangzhou, China, 310000  (e-mail: richengjin@zju.edu.cn).} 
\thanks{\copyright 2024 IEEE. Personal use of this material is permitted. Permission from IEEE must be obtained for all other uses, in any current or future media, including reprinting/republishing this material for advertising or promotional purposes, creating new collective works, for resale or redistribution to servers or lists, or reuse of any copyrighted component of this work in other works.}}

\begin{document}
\maketitle

\begin{abstract}
While a practical wireless network has many tiers where end users do not directly communicate with the central server, the users' devices have limited computation and battery powers, and the serving base station (BS) has a fixed bandwidth.
Owing to these practical constraints and system models, this paper leverages model pruning and proposes a pruning-enabled hierarchical federated learning (PHFL) in heterogeneous networks (HetNets).
We first derive an upper bound of the convergence rate that clearly demonstrates the impact of the model pruning and wireless communications between the clients and the associated BS.
Then we jointly optimize the model pruning ratio, central processing unit (CPU) frequency and transmission power of the clients in order to minimize the controllable terms of the convergence bound under strict delay and energy constraints.
However, since the original problem is not convex, we perform successive convex approximation (SCA) and jointly optimize the parameters for the relaxed convex problem.
Through extensive simulation, we validate the effectiveness of our proposed PHFL algorithm in terms of test accuracy, wall clock time, energy consumption and bandwidth requirement.
\end{abstract}


\section{Introduction}
\IEEEPARstart{F}{ederated} learning (FL) has garnered significant attention as a privacy-preserving distributed edge learning solution in wireless edge networks \cite{mcmahan2017communication,tran2019Fedearated,yang2020energy}.
Since the original FL follows the parameter server paradigm, many state-of-the-art works consider a single server with distributed clients as the general system model in order to study the analytical and empirical performance \cite{tran2019Fedearated,yang2020energy,richeng2022communication,chen2023exploring,zhang2023joint}.
Given that there are $\mathcal{U} \coloneqq \{u\}_{u=1}^{\mathrm{U}}$ clients, each with a local dataset of $\mathcal{D}_u \coloneqq \{\mathbf{x}_a, y_a\}_{a=1}^A$, where $\mathbf{x}_a$ and $y_a$ are the $a^{th}$ feature vector and the corresponding label, the central server wants to train a global machine learning (ML) model $\mathbf{w}$ by minimizing a weighted combination of the clients' local objective functions $f_u(\mathbf{w})$'s, as follows
\begin{align}
    &f(\mathbf{w}) \coloneqq \sum\nolimits_{u=1}^{\mathrm{U}} \alpha_u f_u(\mathbf{w}), \label{globalLoss} \\
    &f_{u}(\mathbf{w}) \coloneqq (1/|\mathcal{D}_u|) \sum\nolimits_{(\mathbf{x}_a, y_a) \in \mathcal{D}_u} \mathrm{l}(\mathbf{w}, \mathbf{x}_a, y_a),\label{ueLossFun}
\end{align}
where $\alpha_u$ is the corresponding weight, and $\mathrm{l}(\mathbf{w}, \mathbf{x}_a, y_a)$ denotes the loss function associated with the $a^{th}$ data sample.
However, general networks usually follow a hierarchical structure \cite{wang2022demystifying}, where the clients are connected to edge servers, the edge servers are connected to fog nodes/servers, and these fog nodes/servers are connected to the cloud server \cite{hosseinalipour2022multi}.
Naturally, some recent works \cite{xu2021adaptive, Liu2022jointUE, luo2020hfel, liu2020client, feng2022Mobility,abad2020hierarchical} have extended FL to accommodate this hierarchical network topology.

A client\footnote{The terms client and UE are interchangeably used when there is no ambiguity.} does not directly communicate with the central server in the hierarchical network topology. 
Instead, the clients usually perform multiple local rounds of model training before sending the updated models to the edge server. 
The edge server aggregates the received models and updates its edge model, and then broadcasts the updated model to the associated clients for local training.
The edge servers repeat this for multiple edge rounds and finally send the updated edge models to the upper-tier servers, which then undergo the same process before finally sending the updated models to the cloud/central server.
This is usually known as hierarchical federated learning (HFL) \cite{luo2020hfel}.
On the one hand, HFL acknowledges the practical wireless heterogeneous network (HetNet) architecture.
On the other hand, it avoids costly direct communication between the far-away cloud server and the capacity-limited clients \cite{abad2020hierarchical}.
Moreover, since local averaging improves learning accuracy \cite{wang2022demystifying}, the central server ends up with a better-trained model.

While HFL can alleviate communication bottlenecks for the cloud server, data and system heterogeneity amongst the clients still need to be addressed.
Since the clients are usually scattered in different locations and have various onboard sensors, the data collected/sensed by these clients are diverse, causing statistical data heterogeneity that the server cannot govern.
As such, we need to embrace it in our theoretical and empirical study.
Besides, the well-known system heterogeneity arises from the clients' diverse computation powers \cite{kairouz2021advances}.
Recently, some works have been proposed to deal with system heterogeneity. 
For example, FedProx \cite{MLSYS2020_38af8613}, anarchic federated averaging (AFA) \cite{yang2022anarchic} and federated normalized averaging algorithm (FedNova) \cite{wang2020tackling}, to name a few, considered different local rounds for different clients to address the system heterogeneity.
More specifically, FedProx adds a proximal term to the client's local objective function to handle heterogeneity. 
AFA and FedNova present different ways to aggregate clients trained models' weights at the server to tackle this system heterogeneity.
However, these algorithms still assume that the client has and trains the original ML model, i.e., neither the computation time for the client's local training nor the communication overhead for offloading the trained model to the server is considered in system design.

Model pruning has attracted research interest recently \cite{lin2020Dynamic, jiang2022model}. 
It makes the over-parameterized model sparser, which allows the less computationally capable clients to perform local training more efficiently without sacrificing much of the test accuracy.
Besides, since the trained model contains fewer non-zero entries, the communication overhead over the unreliable wireless link between the client and the associated base station (BS) also dramatically reduces.
However, pruning generally introduces errors that only partially vanish, causing the pruned model to converge only to a neighborhood of the optimal solution \cite{lin2020Dynamic}.
Besides, unlike the traditional FL, where the model averaging happens only at the central server, HFL has multiple hierarchical levels that may adopt their own aggregation strategy.
Therefore, model pruning at the local client level leads to additional errors in the available models at different levels, eventually contributing to the global model. 
As such, more in-depth study is in need to understand the full benefit of model pruning in hierarchical networks.

\subsection{Related Work}
\noindent
Some recent works studied HFL \cite{xu2021adaptive, Liu2022jointUE, luo2020hfel, liu2020client, feng2022Mobility, abad2020hierarchical} and model pruning-based traditional single server based FL \cite{jiang2022model, zhu2023fedlp, liu2021adaptive, Liu2022Joint, ren2022toward} separately.
In \cite{xu2021adaptive}, Xu \textit{et al.} proposed an adaptive HFL scheme, where they optimized edge aggregation intervals and bandwidth allocation to minimize a weighted combination of the model training delay and training loss.
Liu \textit{et al.} proposed network-assisted HFL in \cite{Liu2022jointUE}, where they optimized wireless resource allocation and user associations to minimize $1)$ learning latency for independent identically distributed (IID) data distribution and $2)$ weighted sum of the total data distance and learning latency for the non-IID data distribution. 
Similar to \cite{xu2021adaptive, Liu2022jointUE}, Luo \textit{et al.} jointly optimized the wireless network parameters in order to minimize the weighted combination of the total energy consumption and delay during the training process in \cite{luo2020hfel}. 
Besides, \cite{liu2020client} also proposed an HFL algorithm based on federated averaging (FedAvg) \cite{mcmahan2017communication}.
In \cite{feng2022Mobility}, Feng \textit{et al.} proposed a mobility-aware clustered FL algorithm owing to user mobility.
More specifically, assuming that all users had an equal probability of staying at a cluster, the authors derived an upper bound of the convergence rate to capture the impact of user mobility, data heterogeneity and network heterogeneity. 
Abad \textit{et al.} also optimized wireless resources to reduce the communication latency and facilitate HFL in \cite{abad2020hierarchical}.

On the model pruning side, Jiang \textit{et al.} considered two-stage distributed model pruning in \cite{jiang2022model} with traditional single server based FL setting without any wireless network aspects.  
In a similar setting, Zhu \textit{et al.} proposed a layer-wise pruning mechanism in \cite{zhu2023fedlp}. 
Liu \textit{et al.} optimized the pruning ratio and time allocation in \cite{liu2021adaptive} in order to maximize the convergence rate in a time division multiple access operated small BS (sBS).
The idea was extended to joint client selection, pruning ratio optimization and time allocation in \cite{Liu2022Joint}. 
Using a similar network model, Ren \textit{et al.} optimized pruning ratios and bandwidth allocations jointly to minimize a weighted combination of the FL training time and pruning error in \cite{ren2022toward}.
These works \cite{liu2021adaptive,Liu2022Joint,ren2022toward} decomposed the original problem into different sub-problems that they solved iteratively in an attempt to solve the original problem sub-optimally. 
Moreover, \cite{liu2021adaptive,Liu2022Joint,ren2022toward} considered a simple network system model with a single BS serving the distributed clients with the wireless links.

\subsection{Our Contributions}
\noindent
While the studies mentioned above shed some light on HFL and model pruning in the traditional single server based FL, the impact of pruning on HFL in resource-constrained wireless HetNet is yet to be explored.
On the one hand, the clients need to train the original model for a few local episodes to determine the neurons they shall prune, which adds additional time and energy costs.
On the other hand, pruning adds errors to the learning performance.
Therefore, it is necessary to theoretically and empirically study these errors from different levels in HFL.
Moreover, it is also crucial to justify how and when one should adopt model pruning in practical wireless HetNets.
Motivated by these, in this work, we present our pruning-enabled HFL (PHFL) framework with the following major contributions:
\begin{itemize}
    \item Considering a practical wireless HetNet, we propose a PHFL solution in which the clients perform local training on the initial models to determine the neurons to prune, perform extensive training on the pruned models, and offload the trained models under strict delay and energy constraints.
    \item We theoretically analyze how pruning introduces errors in different levels under resource constraints in wireless HetNets by deriving a convergence bound that captures the impact of the wireless links between the clients and server and the pruning ratios. More specifically, the proposed solution converges to the neighborhood of a stationary point of traditional HFL with a convergence rate of $\mathcal{O} \big(1 / \sqrt{\mathrm{U} T} \big) + \mathcal{O} (\beta^2 D^2 \delta^{\mathrm{th}})$, where $\mathrm{U}$ is the total number of clients, $T$ is the total local iterations, $\beta$ quantifies smoothness of the loss function, $D^2$ is an upper bound of the $L_2$ norm of the model weights, and $0 < \delta^{\mathrm{th}} < 1$ is the maximum allowable pruning ratio.
    \item Then, we formulate an optimization problem to maximize the convergence rate by jointly configuring wireless resources and system parameters. To tackle the non-convexity of the original problem, we use a successive convex approximation (SCA) algorithm to solve the relaxed convex problem efficiently.
    \item Finally, using extensive simulation on two popular datasets and three popular ML models, we show the effectiveness of our proposed solution in terms of test accuracy, training time, energy consumption and bandwidth requirement.
\end{itemize}

The rest of the paper is organized as follows:
Section \ref{sec_sysModel} introduces our system model.
Detailed theoretical analysis is performed in Section \ref{section_convBound}, followed by our joint problem formulation and solution in Section \ref{section_jointProbSol}.
Based on our extensive simulation, we discuss the results in Section \ref{section_sim_results}.
Finally, Section \ref{section_conclusion} concludes the paper.
Moreover, Table \ref{tableOfNotations} summarizes the important notations used in the paper.

\begin{table}[!t]
\caption{Important Notations}
\fontsize{8}{8}\selectfont
\centering
\begin{tabular}{|C{1.9cm} |C {6.1cm}|}
\hline 
\textbf{Notation} & \textbf{Description} \\ \hline
$u$; $b$; $l$  & $u^{\mathrm{th}}$ user; $b^{\mathrm{th}}$ sBS; $l^{\mathrm{th}}$ mBS  \\ \hline 
$\mathcal{B}_l$; $k$ & sBS set under the $l^\mathrm{th}$ mBS; $k^{\mathrm{th}}$ sBS in $\mathcal{B}_l$ \\ \hline
$\mathcal{V}_{k,l}$; $j$ & VC set of sBS $k$ under the $l^\mathrm{th}$ mBS; $j^{\mathrm{th}}$ VC of sBS $k$ \\ \hline
$\mathcal{U}_{j,k,l}$; $i$ & Client set of the $j^{\mathrm{th}}$ VC of $l^{\mathrm{th}}$ sBS under the $l^\mathrm{th}$ mBS; $i^{\mathrm{th}}$ client in $\mathcal{U}_{j,k,l}$ \\ \hline
$z$; $\mathcal{Z}$ & $z^{\mathrm{th}}$ pRB; pRB set \\ \hline
$\mathrm{P}_i^t$ & Client $i$'s transmission power during $t$  \\ \hline
$\mathbf{w}$; $\mathbf{m}$; $\tilde{\mathbf{w}}$ & Original model; binary mask; pruned model \\ \hline
$\mathbf{w}_i$; $\mathbf{w}_j$; $\mathbf{w}_k$; $\mathbf{w}_l$ & Local model of the client, VC, sBS, and mBS \\ \hline
$f_i(\cdot)$; $f_j(\cdot)$; $f_k(\cdot)$; $f_l(\cdot)$; $f(\cdot)$ & Loss function of the client, VC, sBS, mBS, and central server, respectively \\ \hline
$\nabla f_i (\cdot)$; $g(\cdot)$; $\eta$ & True gradient; stochastic gradient; learning rate \\ \hline
$d$; $d_p$ & Total \& pruned parameters of the ML model \\ \hline
$\delta_i^t$; $\delta^{\mathrm{th}}$ & Pruning ratio of client $i$ during $t$; max pruning ratio \\ \hline
$\alpha_i$; $\alpha_j$; $\alpha_k$; $\alpha_l$ & Weight of $i^{th}$ client, $j^{\mathrm{th}}$ VC, $k^{\mathrm{th}}$ sBS, and $l^{\mathrm{th}}$ mBS \\ \hline 
$\rho$ & Number of SGD rounds on $\mathbf{w}$ to get winning ticket \\ \hline
$\kappa_0$; $\kappa_1$; $\kappa_2$; $\kappa_3$ & Number of local, VC, sBS, and mBS rounds\\ \hline
$\mathbf{1}_i^t$; $p_i^t$ & Binary indicator function to define if sBS receives $i^{\mathrm{th}}$ client's trained model; probability that $\mathbf{1}_i^t=1$ \\ \hline 
$\beta$ & Smoothness of the loss functions \\ \hline
$\sigma^2$ & Bounded variance of the gradients \\ \hline
$\epsilon_{\cdot}^2$ & Bounded divergence of the loss functions of two inter-connected tiers\\ \hline
$G^2$; $D^2$ & Upper bound of the $L_2$-norm of stochastic gradients and model weights, respectively \\ \hline 
$\mathrm{f}_i^t$; $\mathrm{f}_i^{\mathrm{max}}$ & CPU clock cycle of $i$ during $t$; max CPU cycle of i \\ \hline
$b$, $n$ & Batch size; number of mini-batch\\ \hline
$c_i$; $\mathrm{D}_i$ & Required number of CPU cycle of $i$ to process $1$-bit data; each data sample size in bits \\ \hline
$\mathrm{FPP}$ & Floating point precision \\ \hline
$\mathrm{t}_i^{\mathrm{cp_d}}$; $\mathrm{e}_i^{\mathrm{cp_d}}$ & Time and energy overheads to get the lottery ticket \\ \hline
$\mathrm{t}_i^{\mathrm{cp_s}}$; $\mathrm{e}_i^{\mathrm{cp_s}}$ & Time and energy overheads
to compute $\kappa_0$ local SGD rounds with the pruned model \\ \hline 
$\mathrm{t}_i^{\mathrm{up}}$; $\mathrm{e}_i^{\mathrm{up}}$ & Time and energy overheads for offloading client $i$'s trained model \\ \hline
$\mathrm{t}_i^{\mathrm{tot}}$; $\mathrm{e}_i^{\mathrm{tot}}$ & Client $i$'s total time and energy overheads to finish one VC round \\ \hline
$\mathrm{t^{th}}$; $\mathrm{e}_i^{\mathrm{th}}$ & Time and energy budgets to finish one VC round \\ \hline
\end{tabular}
\label{tableOfNotations}
\end{table}

\section{System Model}
\label{sec_sysModel}
\subsection{Wireless Network Model}
\noindent
We consider a generic heterogeneous network (HetNet) consisting of some UEs, sBSs and macro BSs (mBSs), as shown in Fig. \ref{sysMod}.
Denote the UE, sBS and mBS sets by $\mathcal{U} \coloneqq \{u\}_{u=1}^{\mathrm{U}}$, $\mathcal{B} \coloneqq \{b\}_{b=1}^{\mathrm{B}}$ and $\mathcal{L} \coloneqq \{l\}_{l=1}^{ L }$, respectively.
Each UE and sBS are connected to one sBS and mBS, respectively.
The mBSs are connected to the central server. 
While the UEs communicate over wireless links with the sBS, the connections between the sBS and mBS and between the mBS and central server are wired.
Moreover, due to UEs' system heterogeneity, we consider that the sBS groups UEs with similar computation and battery powers into a virtual cluster (VC).

\begin{figure}[!t]
    \centering
    \includegraphics[width=0.4\textwidth]{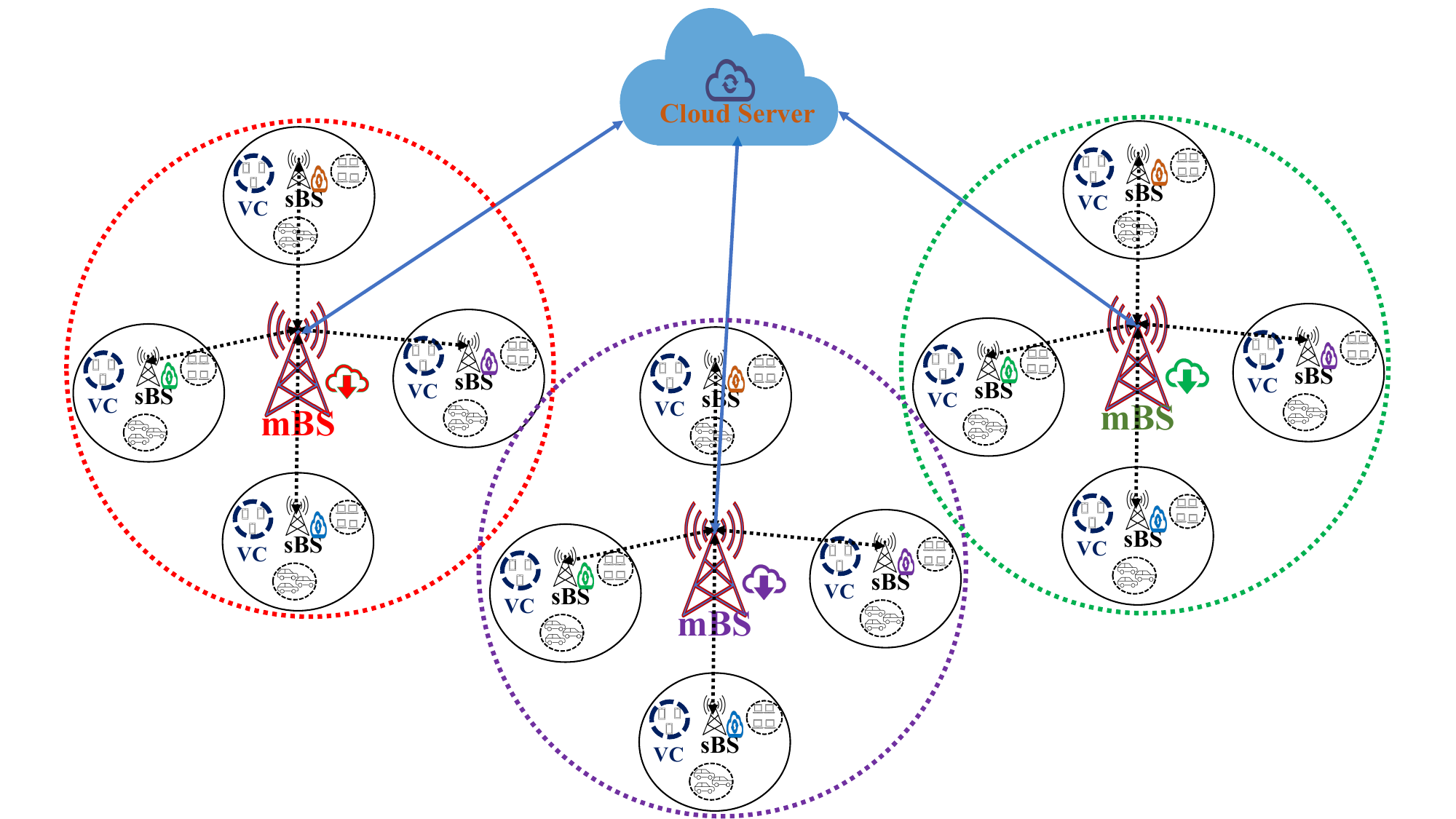}
    \caption{Pruning-enabled hierarchical FL system model}
    \label{sysMod}
\end{figure}

We can benefit from the VC since it enables one additional aggregation tier. 
Besides, thanks to the recent progress of the proximity-services in practical networks \cite{3GPP_TS_22_554}, one can also select a cluster head and leverage device-to-device communication to receive and distribute the models to the UEs in the same VC. 
However, in this work, we assume that the sBS creates the VC and also manages it.
We use the notation $\mathcal{B}_l \coloneqq \{k\}_{k=1}^{B_{l}}$, $\mathcal{V}_{k,l} \coloneqq \{j\}_{j=1}^{V_{k,l}}$ and $\mathcal{U}_{j,k,l} \coloneqq \{i\}_{i=1}^{U_{j,k,l}}$ to represent the sBS set associated to mBS $l$, VC set of sBS $k \in \mathcal{B}_l$ and UE set of VC $j \in \mathcal{V}_{k,l}$, respectively.
Moreover, denote the UEs associated with sBS $k$, mBS $l$ by $\mathcal{U}_{k,l} = \bigcup_{j=1}^{V_{k,l}}$ $\mathcal{U}_{j,k,l}$ and $\mathcal{U}_{l} = \bigcup_{k=1}^{B_{l}}\mathcal{U}_{k,l}$, respectively.
Finally, $\mathcal{U} = \bigcup_{l=1}^{L} \mathcal{U}_l$.
Note that we consider that these associations are known and provided by the network administrator.
The network has a fixed bandwidth divided into orthogonal physical resource blocks (pRBs) for performing the FL task. 
Denote the pRB set by $\mathcal{Z}\coloneqq \{z\}_{z=1}^Z$.
Each mBS reuses the same pRB set, i.e., the frequency reuse factor is $1$.
Besides, each mBS allocates dedicated pRBs to its associated sBSs.
Each sBS further uses dedicated pRBs to communicate with the associated UEs. 
As such, there is no intra-tier interference in the system.

To this end, denote the distance between UE $i$ and the sBS $k$ by $d_{i,k}$.
The wireless fading channel between the UE and sBS follows Rayleigh distribution\footnote{Rayleigh distribution is widely used for its simplicity. Other distributions are not precluded.}, and is denoted by $h_{i,k}^t$.
The transmission power of UE is denoted by $\mathrm{P}_{i}^t$.
As such, the uplink signal-to-noise-plus-interference ratio (SINR) is expressed as 
\begin{equation}
\label{ue_VC_sinr}
\begin{aligned}
    \gamma_{i,k}^t &= \mathrm{P}_{i}^t h_{i,k}^t d_{i,k}^{-\alpha} / (\omega \zeta^2 + I_{i,k}^t ),
\end{aligned}
\end{equation}
where $\alpha$ is the path loss exponent, $\zeta^2$ is the variance of the circularly symmetric zero-mean Gaussian distributed random noise and $\omega$ is the pRB size.
Moreover, $I_{i,k}^t = \sum_{l=1}^{L} \sum_{k'=1, k \neq k'}^{B_l} \sum_{j'=1}^{V_{k,l}} \sum_{i' \in \mathcal{U}_{j',k',l}} \mathrm{P}_{i'}^t h_{i',k'}^t d_{i',k'}^{-\alpha}$ is the inter-cell interference.
Thus, the data rate 
\begin{equation}
\begin{aligned}
    r_{i}^t & = \omega \log_2 \big[ 1 + \gamma_{i,k}^t \big]. 
\end{aligned}
\end{equation}

\subsection{Pruning-Enabled Hierarchical Federated Learning (PHFL)}
\label{subSecPHFL}
\noindent
In this work, we consider that each client uses mini-batch stochastic gradient descent (SGD) to minimize (\ref{ueLossFun}) over $n$ mini-batches since the gradient descent over the entire dataset is time-consuming.
Denote the stochastic gradient $g(\mathbf{w})$ such that $\mathbb{E}_{\xi_{n} \sim \mathcal{D}_i} [ g (\mathbf{w})] \coloneqq \nabla f_i(\mathbf{w})$, where $\xi_n$ is a randomly sampled batch from dataset $\mathcal{D}_i$.
However, computation with the original model $\mathbf{w} \in \mathbb{R}^d$ is still costly and may significantly extend the training time of a computationally limited client. 
To alleviate this, the UE trains a pruned model by removing some of the weights of the original model $\mathbf{w}$ \cite{lin2020Dynamic}.

Denote a binary mask by $\mathbf{m}_i \in \mathbb{R}^d$ and the pruned model by $\tilde{\mathbf{w}}_i \coloneqq \mathbf{w}_i \odot \mathbf{m}_i$, where $\odot$ means element-wise multiplication.
Note that training the pruned model $\tilde{\mathbf{w}}_i$ is computationally less expensive as it has fewer parameters than the original model. 
It is worth pointing out that the UE utilizes the state-of-the-art lottery ticket hypothesis \cite{franklelottery} to find the winning ticket $\tilde{\mathbf{w}}_i$ and the corresponding mask with the following key steps.
Denote the number of parameters required to be pruned by $d_p$. 
The UE performs $\rho$ local iterations on the original model $\mathbf{w}_i$ as
\begin{equation}
\label{lotTicket}
    \mathbf{w}_i^{\rho} = \mathbf{w}_i - \eta \sum\nolimits_{\bar{\rho}=1}^{\rho} g (\mathbf{w}_i^{\bar{\rho}}),
\end{equation}
where $\eta$ is the step size.
The UE then prunes $d_p$ entries of $\mathbf{w}_i^\rho \in \mathbb{R}^d$ with the smallest magnitudes\footnote{The time complexity to sort the $d$ parameters and then prune the $d_p$ smallest ones depends on the sorting technique.
Many sorting algorithms have logarithmic time complexity, which can be computed quickly in modern graphical processing units. 
Following the common practice in literature \cite{liu2021adaptive,Liu2022Joint,ren2022toward}, the overhead for pruning is ignored in this work. Moreover, the proposed method can be readily extended to incorporate the time overhead for pruning.} and generates a binary mask $\mathbf{m}_i \in \{0,1\}^d$.
To that end, the client obtains the winning ticket $\tilde{\mathbf{w}}_i$ by retaining the original weights of the corresponding nonzero entries of the mask $\mathbf{m}_i$ from the original initial model $\mathbf{w}_i$ \cite{franklelottery}. 
Note that other pruning techniques can also be adopted. 
Moreover, we denote the pruning ratio by \cite{Liu2022Joint} 
\begin{equation}
    \delta_{i} \coloneqq d_p/d.
\end{equation}
Given the pruned model $\tilde{\mathbf{w}}_i^{t,0} \coloneqq \mathbf{w}_i^t \odot \mathbf{m}_i^t$, the loss function of the UE is rewritten as
\begin{equation}
\label{ueLossFunPrune}
    f_{i}(\tilde{\mathbf{w}}_i^{t,0}) \coloneqq [1/ |\mathcal{D}_i|] \sum\nolimits_{(\mathbf{x}_a, y_a) \in \mathcal{D}_i} f(\tilde{\mathbf{w}}_i^{t,0}, \mathbf{x}_a, y_a),
\end{equation}
Each UE updates its pruned model as 
\begin{equation}
\label{localRound}
    \tilde{\mathbf{w}}_i^{t+1} 
    \coloneqq \tilde{\mathbf{w}}_i^{t,0} - \eta g (\tilde{\mathbf{w}}_i^{t,0}) \odot \mathbf{m}_i^t.
\end{equation}

As such, we denote the loss functions of the $j^{th}$ VC, $k^{th}$ sBS, $l^{th}$ mBS and central server as $f_j (\mathbf{w}_j) \coloneqq \sum\nolimits_{i=1}^{U_{j,k,l}}  \alpha_{i} f_i (\tilde{\mathbf{w}}_i)$, 
$f_{k} (\mathbf{w}_k)  \coloneqq  \sum\nolimits_{j=1}^{V_{k,l}}  \alpha_j f_j (\mathbf{w}_j)$,
$f_l (\mathbf{w}_l) \coloneqq  \sum\nolimits_{k=1}^{B_l}  \alpha_k f_k (\mathbf{w}_k)$, 
and $ f (\mathbf{w}) \coloneqq \sum\nolimits_{l=1}^{ L } \alpha_l f_l (\mathbf{w}_l)$, respectively. 
Note that for simplicity, we consider identical weights, i.e., $\alpha_i=1/U_{j,k,l}$, $\alpha_j=1/V_{k,l}$, $\alpha_k=1/B_l$ and $\alpha_l=1/L$, which can be easily adjusted for other weighting strategies. 
Besides, since aggregation happens at different times and different levels, we need to capture the time indices explicitly. 
Let each UE perform $\kappa_0$ local iterations before sending the updated model to the associated VC.
It is worth noting that the winning ticket and the corresponding binary mask are only obtained before these $\kappa_0$ local rounds begin. 
Besides, since training the original model is costly, it is reasonable to consider $\rho \ll \kappa_0$\footnote{In our simulation, we considered $\rho <\kappa_0$ and observed that even $\rho=1$ with $\kappa_0 \geq 5$ performed well.}.
Moreover, although the sBS-mBS and mBS-central server links are wired, communication and computation at these nodes incur additional burdens.
As such, we assume that each VC, sBS and mBS perform $\kappa_1$, $\kappa_2$ and $\kappa_3$ rounds, respectively, before sending the trained model to the respective upper layers.
Denote the indices of the current global round, mBS round, sBS round, VC round and UE's local round by $m$, $t_3$, $t_2$, $t_1$ and $t_0$, respectively.
Besides, similar to \cite{feng2022Mobility,xu2021adaptive}, let $t \coloneqq 
[\{ (m\kappa_3 + t_3) \kappa_2 + t_2 \}\kappa_1 + t_1 ]\kappa_0 + t_0$ denote the index of local update iterations.

If $t ~\textsl{mod}~ \kappa_0 = 0$, the UE receives the latest available model of its associated VC, i.e.,
\begin{equation}
\label{UE_Update_with_VC_Model}
    \mathbf{w}_i^{ \Bar{t}_0 } \gets \mathbf{w}_{j}^{\Bar{t}_0},
\end{equation}
where $\Bar{t}_0 = [\{ (m\kappa_3 + t_3) \kappa_2 + t_2 \}\kappa_1 + t_1 ]\kappa_0$.
The UE then computes the pruned model $\tilde{\mathbf{w}}_i^{\Bar{t}_0,0}$ and the binary mask $\mathbf{m}_i^{\Bar{t}_0}$. 
It then performs $\kappa_0$ local SGD rounds as
\begin{equation}
\label{UE_Update}
    \tilde{\mathbf{w}}_i^{\Bar{t}_0 + \kappa_0} = \tilde{\mathbf{w}}_i^{\Bar{t}_0,0} - \eta \sum\nolimits_{t_0=1}^{\kappa_0} g \big(\tilde{\mathbf{w}}_i^{\Bar{t}_0 + t_0, 0}\big) \odot \mathbf{m}_i^{\Bar{t}_0}. 
\end{equation}
Each VC $j$ performs $t_1=0, \dots, \kappa_1 - 1$ local rounds. 
When $(t_1+1) ~\textsl{mod}~ \kappa_1 = 0$, the VC's model gets updated by the latest available sBS model, i.e.,
\begin{equation}
\label{VC_Update_with_sBS_Model}
    \mathbf{w}_{j}^{\Bar{t}_1} \gets \mathbf{w}_{k}^{\Bar{t}_1},
\end{equation}
where $\Bar{t}_1 = \{ (m\kappa_3 + t_3) \kappa_2 + t_2 \}\kappa_1\kappa_0$.
Besides, between two VC rounds, the local model of the VC is updated as
\begin{align}
    &\mathbf{w}_{j}^{\Bar{t}_1 + (t_1+1) \kappa_0} = \sum\nolimits_{i=1}^{U_{j,k,l}} \alpha_i \tilde{\mathbf{w}}_{i}^{\Bar{t}_1 + (t_1+1) \kappa_0} 
    = \tilde{\mathbf{w}}_j^{\Bar{t}_1 + t_1 \kappa_0,0}  - \nonumber\\
    &\eta \rs \sum\nolimits_{i=1}^{U_{j,k,l}} \rs \big[\pmb{1}_i^{\Bar{t}_1 + t_1\kappa_0} \rs / \! p_i^{\Bar{t}_1 + t_1\kappa_0} \big] \alpha_i \sum\nolimits_{t_0=1}^{\kappa_0} \! g \big(\tilde{\mathbf{w}}_i^{\Bar{t}_1 + t_1\kappa_0 + t_0,0}\big) \odot \mathbf{m}_i^{\Bar{t}_1 + t_1 \kappa_0}, \label{VC_Update}
\end{align}
where $\tilde{\mathbf{w}}_j^{\Bar{t}_1 + t_1 \kappa_0,0} \coloneqq \sum_{i=1}^{U_{j,k,l}} \alpha_i \tilde{\mathbf{w}}_{i}^{\Bar{t}_1 + t_1\kappa_0,0}$ and $\pmb{1}_i^{\Bar{t}_1 + t_1 \kappa_0}$ is a binary indicator function that indicates whether the sBS receives the trained model of $i$ during the VC aggregation round $\Bar{t}_1 + (t_1+1) \kappa_0$ or not, and is defined as follows:
\begin{equation}
\label{suc_indicator}
\begin{aligned}
  \rs \pmb{1}_i^{\Bar{t}_1 + t_1 \kappa_0}  & \coloneqq 
         \begin{cases}
            1, \rs  & \text{with probability $p_i^{\Bar{t}_1 + t_1 \kappa_0}$}, \\
             0, \rs & \text{otherwise},
        \end{cases},
\end{aligned}
\end{equation}
where $p_i^{\Bar{t}_1 + t_1 \kappa_0}\rs$ is the probability of receiving the trained model over the wireless link and is calculated in the subsequent section (\textit{c.f.} (\ref{probSuc})).
Note that since the sBS has to receive the gradient over the wireless link, we use the binary indicator function $\pmb{1}_i^{\Bar{t}_1 + t_1 \kappa_0}\rs$ in (\ref{VC_Update}) as a common practice  \cite{pervej2023resource,feng2022Mobility}.

The sBS performs $t_2 \rs = \rs 0, \dots, \rs\kappa_2\rs - \! 1 \rs$ local rounds before updating its model.
When $(t_2 + 1) ~\textsl{mod} ~\kappa_2=0$, the sBS updates its local model with the latest available model at its associated mBS, i.e.,
\begin{equation}
\label{sBS_Update_with_mBS_Model}
    \mathbf{w}_{k}^{\Bar{t}_2} \gets \mathbf{w}_{l}^{\Bar{t}_2},
\end{equation}
where $\Bar{t}_2 = (m\kappa_3 + t_3)\kappa_2\kappa_1\kappa_0$.
In each sBS round $t_2$, the sBS updates its model as 
\begin{align}
    &\mathbf{w}_{k}^{\Bar{t}_2 + (t_2+1) \kappa_1\kappa_0} = \sum\nolimits_{j=1}^{V_{k,l}} \alpha_j \mathbf{w}_{j}^{\Bar{t}_2 + (t_2+1) \kappa_1\kappa_0}
    = \tilde{\mathbf{w}}_k^{\Bar{t}_2 + t_2 \kappa_1\kappa_0,0} - \nonumber\\
    &\eta \sum_{j=1}^{V_{k,l}} \alpha_j \sum_{t_1=0}^{\kappa_1-1} \sum_{i=1}^{U_{j,k,l}} \alpha_i \frac{\pmb{1}_i^{ \Bar{t}_2 + \Tilde{t}_2 } } {p_i^{\Bar{t}_2 + \Tilde{t}_2} } \sum_{t_0=1}^{\kappa_0} g \big( \tilde{\mathbf{w}}_i^{ \Bar{t}_2 + \Tilde{t}_2 + t_0, 0} \big) \odot \mathbf{m}_i^{\Bar{t}_2 + \Tilde{t}_2}, \label{sBS_Update}
\end{align} 
where $\tilde{\mathbf{w}}_k^{\Bar{t}_2 + t_2 \kappa_1\kappa_0,0} \coloneqq \sum_{j=1}^{V_{k,l}} \alpha_j \tilde{\mathbf{w}}_{j,k,l}^{\Bar{t}_2 + t_2 \kappa_1\kappa_0,0}$ and $\Tilde{t}_2 = (t_2 \kappa_1 + t_1)\kappa_0$.
Similarly, the mBS performs $t_3=0,\dots,\kappa_3-1$ local rounds before updating its local model with the latest available global model when $(t_3+1) ~ \textsl{mod}~ \kappa_3=0$, i.e.,
\begin{equation}
\label{mBS_Update_with_global_Model}
    \mathbf{w}_l^{m\kappa_3 \kappa_2\kappa_1\kappa_0} \gets \mathbf{w}^{m\kappa_3\kappa_2\kappa_1\kappa_0}.
\end{equation}
Moreover, between two mBS rounds, we can write 
\begin{align}
    &\mathbf{w}_l^{(m\kappa_3 + (t_3+1)) \kappa_2\kappa_1\kappa_0} 
    \rs = \tilde{\mathbf{w}}_l^{(m\kappa_3 + t_3) \kappa_2\kappa_1\kappa_0,0} \rs - \nonumber\\
    &\quad \eta \sum_{k=1}^{B_l} \rs \alpha_k \rs \sum_{t_2=0}^{\kappa_2-1} \sum_{j=1}^{V_{k,l}} \rs \alpha_j \sum_{t_1=0}^{\kappa_1-1} \sum_{i=1}^{U_{j,k,l}} \rs \alpha_i \frac{\pmb{1}_i^{\Bar{t}_3}} {p_i^{\Bar{t}_3}} \rs \sum_{t_0=1}^{\kappa_0} \rs g \big( \tilde{\mathbf{w}}_i^{\Bar{t}_3 + t_0,0} \big) \odot \mathbf{m}_i^{\Bar{t}_3}, \label{mBS_Update}
\end{align}    
where $\tilde{\mathbf{w}}_l^{(m\kappa_3 + t_3) \kappa_2\kappa_1\kappa_0,0} \coloneqq \sum_{k=1}^{B_l} \alpha_k \tilde{\mathbf{w}}_k^{(m\kappa_3 + t_3) \kappa_2\kappa_1\kappa_0,0}$ and $\Bar{t}_3 = [((m\kappa_3 + t_3)\kappa_2 + t_2)\kappa_1 + t_1]\kappa_0$.
Finally, the central server performs global aggregation by collecting the updated models from all mBSs as follows:
\begin{align}
    &\mathbf{w}^{(m + 1) \prod_{z=0}^3 \kappa_{z}} 
    = \tilde{\mathbf{w}}^{ m \prod_{z=0}^3 \kappa_z,0} \rs - \eta \sum\nolimits_{l=1}^{ L } \rs \alpha_l  \sum\nolimits_{t_3 = 0}^{\kappa_3-1} \sum\nolimits_{k=1}^{B_{l}} \rs \alpha_k \nonumber\\
    &\qquad \sum_{t_2 = 0}^{\kappa_2-1} \sum_{j=1}^{V_{k,l}} \rs \alpha_j \rs \sum_{t_1=0}^{\kappa_1-1} \sum_{i=1}^{U_{j,k,l}} \rs \alpha_i \frac{\pmb{1}_i^{\Bar{t}_3 }} {p_i^{\Bar{t}_3}} \sum_{t_0=1}^{\kappa_0} g \big( \tilde{\mathbf{w}}_i^{\Bar{t}_3 + t_0,0} \big) \odot \mathbf{m}_i^{\Bar{t}_3}, \label{global_Update}
\end{align} 
where $\tilde{\mathbf{w}}^{m\prod_{z=0}^3 \kappa_z,0} \coloneqq \sum_{l=1}^{ L } \alpha_l \tilde{\mathbf{w}}_l^{m \prod_{z=0}^3 \kappa_z,0}$.

The proposed PHFL process is summarized in Algorithm \ref{HFLalgo}.
\begin{algorithm}[t!]
\fontsize{8}{8}\selectfont
\SetAlgoLined 
\DontPrintSemicolon
\KwIn{Total global round $M$, initial model $\mathbf{w}^0$}
Synchronize all edge devices with the initial model $\mathbf{w}^0$ \;
\For{All global rounds $m=0$ to $M-1$}{
    \For{All mBS rounds $t_3 = 0, 1, \dots, \kappa_3-1$}{
        \For{All sBS rounds $t_2=0, 1, \dots, \kappa_2 - 1 $}{
            \For{All VC rounds $t_1=0,1,\dots, \kappa_1 - 1$}{
                \For {$i \in \mathcal{U}_{j,k,l}$ in parallel} {
                    UE receives the latest available model from the associated VC \;
                    Compute binary mask and get the winning ticket using lottery ticket hypothesis \;
                    \For{All local rounds $t_0=1,2, \dots, \kappa_0$}{
                        $t \gets [\{ (m\kappa_3 + t_3) \kappa_2 + t_2 \}\kappa_1 + t_1 ]\kappa_0 + t_0$ \;
                        UE updates local model using (\ref{localRound}) \;}
                    }
                sBS updates VC model using (\ref{VC_Update_with_sBS_Model}) and (\ref{VC_Update}) \; 
            }
            sBS update local cell model using (\ref{sBS_Update_with_mBS_Model}) and (\ref{sBS_Update}) \;
        }
        mBS updates local cell model using (\ref{mBS_Update_with_global_Model}) and (\ref{mBS_Update}) \;
    }
    Central server updates global model using (\ref{global_Update}) \;
}
\KwOut{Global ML model $\mathbf{w}^{M-1}$}
\caption{Pruning-Enabled Hierarchical FL}
\label{HFLalgo}
\end{algorithm}

\section{PHFL: Convergence Analysis}
\label{section_convBound}

\subsection{Assumptions}
\label{convergenceAssumptions}
\noindent
We make the following standard assumptions \cite{feng2022Mobility, wang2022demystifying,jiang2022model,Liu2022Joint,stich2018sparsified}:
\begin{enumerate}
    \item The loss functions are lower-bounded, i.e., $f(\mathbf{w}) \geq f_{\text{inf}}$.
    \item The loss functions are $\beta$-smooth, i.e., $\Vert \nabla f_i(\mathbf{w}) - \nabla f_i(\mathbf{w}') \Vert \leq \beta \Vert \mathbf{w} - \mathbf{w}' \Vert$.
    \item Mini-batch gradients are unbiased $\mathbb{E}_{\xi \sim \mathcal{D}_i} [g (\tilde{\mathbf{w}}_i)] = \nabla f_i(\tilde{\mathbf{w}}_i)$. Besides, the variance of the gradients is bounded, i.e., $\Vert g (\tilde{\mathbf{w}}_i) - \nabla f_i(\tilde{\mathbf{w}}_i) \Vert^2 \leq \sigma^2$.
    \item The divergence of the local, VC, sBS, mBS and global loss functions are bounded for all $i$, $j$, $k$, $l$ and $\mathbf{w}$, i.e.,
    \begin{align*} 
        &\sum\nolimits_{i=1}^{U_{j,k,l}} \alpha_i \Vert \nabla f_i (\mathbf{w}) - \nabla f_j (\mathbf{w}) \Vert^2 \leq \epsilon_{\mathrm{vc}}^2,\\
        & \sum\nolimits_{j=1}^{V_{k,l}} \alpha_j \Vert \nabla f_j (\mathbf{w}) - \nabla f_k (\mathbf{w}) \Vert^2 \leq \epsilon_{\mathrm{sbs}}^2,\\
        & \sum\nolimits_{k=1}^{B_l} \alpha_k \Vert \nabla f_k (\mathbf{w}) - \nabla f_l (\mathbf{w}) \Vert^2 \leq \epsilon_{\mathrm{mbs}}^2, \\
        & \sum\nolimits_{l=1}^{L} \alpha_l \Vert \nabla f_l (\mathbf{w}) - \nabla f (\mathbf{w}) \Vert^2 \leq \epsilon^2. 
    \end{align*}

    \item The stochastic gradients are independent of each other in different iterations.
    \item The stochastic gradients are bounded, i.e., $\mathbb{E} \Vert g(\mathbf{w}_i) \Vert^2 \leq  G^2$.
    \item The model weights are bounded, i.e., $\mathbb{E} \Vert \mathbf{w}_i \Vert^2 \leq D^2$.    
    \item The pruning ratio $\delta_i^t \in [0, \delta^{\mathrm{th}}]$, in which $0 < \delta^{\mathrm{th}} < 1$ and $\delta^{\mathrm{th}}$ is the maximum allowable pruning ratio, follows
    \begin{equation}
    \label{pruneRatio}
        \delta_i^t \geq \big\Vert \mathbf{w}_i^t - \tilde{\mathbf{w}}_i^{t,0} \big\Vert^2 / \Vert \mathbf{w}_i^t \Vert^2. 
    \end{equation}
\end{enumerate}

Since the updated global, mBS, sBS and VC models are not available in each local iteration $t$, similar to standard practice \cite{feng2022Mobility,wang2022demystifying,xu2021adaptive}, we assume the virtual copies of these models, denoted by $\bar{\mathbf{w}}^t$, $\bar{\mathbf{w}}_{l}^t$, $\bar{\mathbf{w}}_k^t$ and $\bar{\mathbf{w}}_j^t$, respectively, are available. 
Besides, we assume that the bounded divergence assumptions amongst the above loss functions also hold for these virtual models.
Moreover, analogous to our previous notations, we express $\tilde{\bar{\mathbf{w}}}^{t,0} \coloneqq \sum_{l=1}^{L} \alpha_l \sum_{k=1}^{B_l} \alpha_k \sum_{j=1}^{V_{k,l}} \alpha_j \sum_{i=1}^{U_{j,k,l}} \alpha_i \tilde{\mathbf{w}}_i^{t,0} = \sum_{u=1}^{\mathrm{U}} \alpha_u \tilde{\mathbf{w}}_u^{t,0}$ and $\tilde{\bar{\mathbf{w}}}^{0} \coloneqq \tilde{\bar{\mathbf{w}}}^{0,0}$.

\subsection{Convergence Analysis}
\noindent
Similar to existing literature \cite{feng2022Mobility,wang2022demystifying,xu2021adaptive,jiang2022model}, we consider the average global gradient norm as the indicator of the proposed PHFL algorithm's convergence.
As such, in the following, we seek an $\theta_{\mathrm{PHFL}}$-suboptimal solution such that $\frac{1}{T}\sum_{t=0}^{T-1} \Vert \sum_{u=1}^{\mathrm{U}} \alpha_u \nabla f_u (\tilde{\bar{\mathbf{w}}}^{t,0}) \odot \mathbf{m}_u^t \Vert^2 \leq \theta_{\mathrm{PHFL}}$ and $\theta_{\mathrm{PHFL}} \geq 0$. 
Particularly, we start with Theorem \ref{theorem_1} that requires bounding the differences amongst the models in different hierarchical levels.
These differences are first calculated in Lemma \ref{Lemma2} to Lemma \ref{Lemma5} and then plugged into Theorem \ref{theorem_1} to get the $\theta_{\mathrm{PHFL}}$-suboptimal bound in Corollary \ref{Theorem1_Convergence}.

\begin{Theorem}
\label{theorem_1}
When the assumptions in Section \ref{convergenceAssumptions} hold and $\eta \leq 1/\beta$, we have  
\begin{align}
\label{theorem_1_eqn}
    &\rs \theta_\mathrm{{PHFL}} \leq
    \mathcal{O} \bigg( \rs \frac{f(\tilde{\bar{\mathbf{w}}}^0) - f_{\mathrm{inf}} } {\eta T} \rs \bigg)  + \mathcal{O} \bigg( \rs \frac{\beta \eta \sigma^2}{\mathrm{U}} \! \bigg) + \underbrace{\mathcal{O} \big(\! \delta^{\mathrm{th}} \beta^2 D^2 \big)}_{\mathrm{pruning ~error}} + \nonumber\\
    & \underbrace{\mathcal{O} \big( \beta \eta G^2 \cdot \varphi_\mathrm{w,0}(\pmb{\delta},\pmb{\mathrm{f}}, \pmb{\mathrm{P}}) \big)}_{\mathrm{wireless ~ factor}} + \mathcal{O} \big( \beta^2 \big[\mathrm{L}_1 + \mathrm{L}_2 + \mathrm{L}_3 + \mathrm{L}_4 \big] \big),
\end{align}
where $\pmb{\delta}=\{ \delta_1^t, \dots, \delta_U^t \}_{t=0}^{T-1}$, $\pmb{\mathrm{f}}=\{\mathrm{f}_1^t, \dots, \mathrm{f}_U^t\}_{t=0}^{T-1}$, $\pmb{\mathrm{P}}=\{P_1^t,\dots, P_U^t\}_{t=0}^{T-1}$ and $\mathrm{f}_i^t$ is the $i^{\mathrm{th}}$ client's central processing unit (CPU) frequency in the wireless factor.
Besides, the terms $\mathrm{L}_1$, $\mathrm{L}_2$, $\mathrm{L}_3$ and $\mathrm{L}_4$ are  
\begin{align}
    &\varphi_\mathrm{w,0}(\pmb{\delta},\pmb{\mathrm{f}}, \pmb{\mathrm{P}}) = \frac{1}{T} \sum_{t=0}^{T-1} \sum_{l=1}^{L} \rs \alpha_{l}^2 \sum_{k=1}^{B_l} \rs \alpha_k^2 \sum_{j=1}^{V_{k,l}} \rs \alpha_j^2 \sum_{i=1}^{U_{j,k,l}} \rs \alpha_i^2 \bigg[\frac{1}{p_i^t} - 1 \bigg], \\
    &\mathrm{L}_1 = \frac{1}{T} \sum_{t=0}^{T-1} \sum_{l=1}^{L} \alpha_l \sum_{k=1}^{B_l} \alpha_k \sum_{j=1}^{V_{k,l}} \alpha_j \sum_{i=1}^{U_{j,k,l}} \alpha_i\mathbb{E} \Vert \bar{\mathbf{w}}_j^t - \tilde{\mathbf{w}}_i^{t} \Vert^2, \label{HL1}\\
    &\mathrm{L}_2 = [1/T] \sum\nolimits_{t=0}^{T-1} \sum\nolimits_{l=1}^{L} \rs \alpha_l \sum\nolimits_{k=1}^{B_l} \rs \alpha_k \sum\nolimits_{j=1}^{V_{k,l}} \rs \alpha_j \mathbb{E} \Vert \bar{\mathbf{w}}_k^t - \bar{\mathbf{w}}_j^t \Vert^2, \rs \rs \label{HL2}\\
    &\mathrm{L}_3 = [1/T] \sum\nolimits_{t=0}^{T-1} \sum\nolimits_{l=1}^{L} \alpha_l \sum\nolimits_{k=1}^{B_l} \alpha_k \mathbb{E} \Vert \bar{\mathbf{w}}_{l}^t - \bar{\mathbf{w}}_k^t \Vert^2, \label{HL3}\\
    &\mathrm{L}_4 = [1/T] \sum\nolimits_{t=0}^{T-1} \sum\nolimits_{l=1}^{L} \alpha_l \mathbb{E} \Vert \bar{\mathbf{w}}^t - \bar{\mathbf{w}}_{l}^t \Vert^2.\label{HL4}
\end{align}
\end{Theorem}

The proof of Theorem \ref{theorem_1} and the subsequent Lemmas are left in the supplementary materials.

\begin{Remark}
    The first term in (\ref{theorem_1_eqn}) is what we get for centralized learning, while the second term arises from the randomness of the mini-batch gradients \cite{bottou2018optimization}. 
    The third term appears from model pruning.
    Besides, the fourth term arises from the wireless links among the sBS and UEs. 
    It is worth noting that $\varphi_\mathrm{w,0} (\pmb{\delta},\pmb{\mathrm{f}}, \pmb{\mathrm{P}}) = 0$ when all $p_i^t$'s are $1$'s.
    Finally, the last term is due to the difference among the VC-local, sBS-VC, sBs-mBS and mBS-global model parameters, respectively, which are derived in the following. 
\end{Remark}

\begin{Remark}
    When the system has no pruning, i.e., all UEs use the original models, all $\delta_i^t=0$. Besides, under the perfect communication among the sBS and UEs, we have $p_i^t=1$. In such cases, the $\theta_{\mathrm{PHFL}}$-suboptimal bound boils down to 
    \begin{subequations}
    \label{Lemma1_No_Prune_Wireless}
    \begin{align}
        \theta_\mathrm{{PHFL}} 
        &\leq \mathcal{O} \big( (f(\bar{\mathbf{w}}^0) - f_{\mathrm{inf}})/ [\eta T]  \big) + \mathcal{O} \big( \beta \eta \sigma^2 / \mathrm{U} \big) + \nonumber\\
        &\Squad\qquad \mathcal{O} \big( \beta^2 [ \mathrm{L}_1 + \mathrm{L}_2 + \mathrm{L}_3 + \mathrm{L}_4 ]\big). \tag{\ref{Lemma1_No_Prune_Wireless}}
    \end{align}
    \end{subequations}
    Besides, the last term in (\ref{Lemma1_No_Prune_Wireless}) appears from the four hierarchical levels. 
    When $\mathrm{U}=1$ and there are no levels, i.e., $\mathrm{L}_1=\mathrm{L}_2=\mathrm{L}_3=\mathrm{L}_4=0$, the convergence bound is exactly the same as the original SGD with non-convex loss function \cite{wang2022demystifying}.
\end{Remark}

To that end, we calculate the divergence among the local, VC, sBS, mBS and global model parameters, and derive the corresponding pruning errors in each level in what follows.

\begin{Lemma}
\label{Lemma2}
When $\eta \leq 1/[2 \sqrt{10} \kappa_0 \beta]$, the average difference between the VC and local model parameters, i.e., the $\mathrm{L}_1$ term of (\ref{theorem_1_eqn}), is upper bounded as 
\begin{align}
\label{lem2}
    &\rs\rs\rs \frac{\beta^2}{T} \rs \sum_{t=0}^T \sum_{l=1}^{L} \rs \alpha_l \rs \sum_{k=1}^{B_l} \rs \alpha_k \rs \sum_{j=1}^{V_{k,l}} \rs \alpha_j \rs \sum_{i=1}^{U_{j,k,l}} \rs \alpha_i \mathbb{E} \left\Vert \bar{\mathbf{w}}_j^t - \tilde{\mathbf{w}}_i^t \right\Vert^2 \rs 
    \leq \mathcal{O} \big( \kappa_0 \eta^2 \beta^2 \sigma^2 \big) + \nonumber\\
    &\rs \mathcal{O} \big( \kappa_0^2 \eta^2 \beta^2 \epsilon_{\mathrm{vc}}^2 \big) + \mathcal{O} \big(\delta^{\mathrm{th}}  \beta^2 D^2 \big) + \mathcal{O} \big( \kappa_0 \eta^2 \beta^2 G^2 \cdot \varphi_{\mathrm{w, L}_1} \big), \rs
\end{align}
where $\varphi_{\mathrm{w, L}_1} \rs \! = \! \frac{1}{T} \sum_{l=1}^{L} \rs \alpha_l \! \sum_{k=1}^{B_l} \rs \alpha_k \sum_{j=1}^{V_{k,l}} \rs \alpha_j \sum_{i=1}^{U_{j,k,l}} \rs \alpha_i \sum_{t=0}^{T-1} (1/p_i^t - 1)$.
\end{Lemma}

\begin{Remark}
    In (\ref{lem2}), the first term comes from the statistical data heterogeneity, while the second term arises from the divergence between the local and VC loss functions. 
    The third term emanates from model pruning.
    Finally, the fourth term stems from the wireless links among the UEs and sBS.
\end{Remark}

\begin{Lemma}
\label{Lemma3}
When $\eta \leq 1/[2 \sqrt{10} \kappa_0 \kappa_1 \beta]$, the difference between the sBS model parameters and VC model parameters, i.e., the $\mathrm{L}_2$ term of (\ref{theorem_1_eqn}), is upper bounded as 
\begin{align}
\label{lem3}
    &\rs \rs \rs \rs \frac{\beta^2}{T} \sum_{t=0}^{T-1} \sum_{l=1}^{L} \rs \alpha_l \rs \sum_{k=1}^{B_l} \rs \alpha_k \rs \sum_{j=1}^{V_{k,l}} \rs \alpha_j \mathbb{E} \left\Vert \bar{\mathbf{w}}_k^t - \bar{\mathbf{w}}_j^t \right\Vert^2 \rs
    \leq \mathcal{O} \big(\beta^4 \kappa_0^4 \kappa_1^2 \eta^4 \epsilon_{\mathrm{vc}}^2 \big) + \nonumber\\
    &\mathcal{O} \big(\kappa_0^2 \kappa_1^2 \eta^2 \beta^2 \epsilon_{\mathrm{sbs}}^2 \big) + \mathcal{O} \big( \kappa_0 \kappa_1 \eta^2 \sigma^2 \beta^2 \big) + \mathcal{O } \big( \delta^{th} \beta^2 D^2 \big) + \nonumber\\
    &\mathcal{O} \big( \kappa_0^3 \kappa_1^2 \beta^4 \eta^4 G^2 \varphi_{\mathrm{w, L}_1} \big) + \mathcal{O} \big( \kappa_0 \kappa_1 \beta^2 \eta^2  \varphi_{\mathrm{w, L}_2} \big),
\end{align}
where $\varphi_{\mathrm{w, L}_2} \rs = \rs \frac{1}{T} \sum_{t=0}^{T-1} \sum_{l=1}^{L} \rs \alpha_l \! \sum_{k=1}^{B_l} \rs \alpha_k \! \sum_{j=1}^{V_{k,l}} \rs \alpha_j \sum_{i=1}^{U_{j,k,l}} \rs \alpha_i^2 ( 1/p_i^t - 1 )$.
\end{Lemma}

\begin{Remark}
\label{rem_of_Lem2}
The first term in (\ref{lem3}) appears from the divergence of the loss functions of the clients and VC, while the second term stems from the divergence between the loss function of the VC and sBS. 
The rest of the terms are due to the statistical data heterogeneity, model pruning and wireless links, respectively. 
\end{Remark}

\begin{Lemma}
\label{Lemma4}
When $\eta \leq 1/[2\sqrt{14} \kappa_0 \kappa_1 \kappa_2 \beta]$, the average difference between the sBS and mBS model parameters, i.e., the $\mathrm{L}_3$ term of (\ref{theorem_1_eqn}), is upper bounded as
\begin{align}
\label{Lem4}
    &[\beta^2/T] \sum\nolimits_{t=0}^{T-1} \sum\nolimits_{l=1}^{L} \alpha_l \sum\nolimits_{k=1}^{B_l} \alpha_k \mathbb{E} \left\Vert \bar{\mathbf{w}}_l^t - \bar{\mathbf{w}}_k^t \right\Vert^2 \nonumber\\
    &\leq \mathcal{O} \big( \kappa_0^3 \kappa_1^2 \kappa_2^2 \eta^4 \beta^4 \epsilon_{\mathrm{vc}}^2 \big) + \mathcal{O} \big(\kappa_0^4 \kappa_1^4 \kappa_2^2 \eta^4 \beta^4 \epsilon_{\mathrm{sbs}}^2 \big) + \nonumber\\
    &\mathcal{O} \big( \kappa_0^2 \kappa_1^2 \kappa_2^2 \eta^2 \beta^4 \epsilon_{\mathrm{mbs}}^2 \big) + \mathcal{O} \big( \kappa_0\kappa_1\kappa_2 \eta^2 \beta^2 \sigma^2 \big) + \nonumber\\
    &\mathcal{O} \big( \delta^{\mathrm{th}} \beta^2 D^2 \big) + \mathcal{O} \big( \kappa_0^3 \kappa_1^2 \kappa_2^2 \eta^4 \beta^4 G^2 \cdot \varphi_{\mathrm{w, L}_1} \big) + \nonumber\\
    &\mathcal{O} \big( \kappa_0^2 \kappa_1^2 \kappa_2^2 \beta^4 \eta^4 \cdot \varphi_{\mathrm{w, L}_2} ) + \mathcal{O} \big( \kappa_0\kappa_1\kappa_2 \eta^2 \beta^2 G^2 \cdot \varphi_{\mathrm{w,L}_3} \big).
\end{align}
where $\varphi_{\mathrm{w,L}_3} \rs \!=\! \frac{1}{T} \sum_{l=1}^{L} \rs \alpha_l \sum_{k=1}^{B_l} \rs \alpha_k \sum_{j=1}^{V_{k,l}} \rs \alpha_j^2 \sum_{i=1}^{U_{j,k,l}} \rs \alpha_i^2 \sum_{t=0}^{T-1} \!(1\!/\!p_i^t - 1)$. 
\end{Lemma}



\begin{Lemma}
\label{Lemma5}
When $\eta \leq 1/[6\sqrt{2} \kappa_0 \kappa_1 \kappa_2 \kappa_3 \beta]$, the average difference between the global and the mBS models, i.e., the $\mathrm{L}_4$ term, is bounded as follows:
\begin{align}
\label{Lem5}
    &\rs \rs \rs [\beta^2/T] \! \sum\nolimits_{t=0}^{T-1} \rs \sum\nolimits_{l=1}^{L} \rs \alpha_l \mathbb{E} \left\Vert \bar{\mathbf{w}}^t - \bar{\mathbf{w}}_l^t \right\Vert^2 \rs
    \leq \mathcal{O} \big(\rs \kappa_0^4 \kappa_1^2 \kappa_2^2 \kappa_3^2 \eta^4 \beta^4 \epsilon_{\mathrm{vc}}^2 \! \big) + \nonumber\\
    &\mathcal{O} \big(\kappa_0^4 \kappa_1^4 \kappa_2^2 \kappa_3^2 \eta^4 \beta^4 \epsilon_{\mathrm{sbs}}^2 \big) + \mathcal{O} \big( \kappa_0^4 \kappa_1^4 \kappa_2^4 \kappa_3^2 \eta^4 \beta^6 \epsilon_{\mathrm{mbs}}^2 \big) + \nonumber\\
    & \rs \mathcal{O} \!(\! \kappa_0^2 \kappa_1^2 \kappa_2^2 \kappa_3^2 \beta^4 \eta^2 \epsilon^2 \! ) \! + \! \mathcal{O} (\kappa_0\kappa_1\kappa_2\kappa_3 \beta^2 \eta^2 \sigma^2 ) \! + \! \mathcal{O} \!\big( \! \delta^{\mathrm{th}} \beta^2 D^2 \! \big) \! + \! \nonumber\\
    &\mathcal{O} \big( \kappa_0^3 \kappa_1^2 \kappa_2^2 \kappa_3^2 \eta^4 \beta^4 G^2  \varphi_{\mathrm{w, L}_1} \big) \! + \! \mathcal{O} \big( \kappa_0^3 \kappa_1^3 \kappa_2^2 \kappa_3^2 \beta^4 \eta^4 \varphi_{\mathrm{w, L}_2} \big) + \nonumber\\
    &\rs\mathcal{O} \!(\! \kappa_0^3 \kappa_1^3 \kappa_2^3 \kappa_3^2 \eta^4 \beta^4 G^2 \varphi_{\mathrm{w,L}_3} \!)\! + \! \mathcal{O} \!(\! \kappa_0\kappa_1\kappa_2\kappa_3 \beta^2 \eta^2 G^2 \varphi_{\mathrm{w,L}_4} \!),\rs 
\end{align}
where $\varphi_{\mathrm{w,L}_4} \rs = \frac{1}{T} \sum_{l=1}^{L} \rs \alpha_l \! \sum_{k=1}^{B_l} \rs \alpha_k^2 \!\sum_{j=1}^{V_{k,l}} \rs \alpha_j^2 \! \sum_{i=1}^{U_{j,k,l}} \rs \alpha_i^2 \! \sum_{t = 0}^{T-1} \left(1\!/\!p_i^t - 1 \right)$.
\end{Lemma}


Note that we have similar observations for (\ref{Lem4}) and (\ref{Lem5}) as in Remark \ref{rem_of_Lem2}. 
Now, using the above Lemmas, we find the final convergence rate in Corollary \ref{Theorem1_Convergence}.  

\begin{Corollary}
\label{Theorem1_Convergence}
When $\eta \leq 1/[6\sqrt{2} \kappa_0 \kappa_1 \kappa_2 \kappa_3 \beta]$, the $\theta_\mathrm{{PHFL}}$ bound of Theorem \ref{theorem_1} boils down to
\begin{align}
\label{Theorem1}
    \theta_\mathrm{{PHFL}} 
    &\leq \mathcal{O} \big( [f(\tilde{\bar{\mathbf{w}}}^0) - f_{\mathrm{inf}}]/ [\eta T] \big)  + \mathcal{O} ( \beta \eta \sigma^2 / \mathrm{U} ) + \nonumber\\
    &\mathcal{O} \big( \kappa_0^2 \eta^2 \beta^2 \epsilon_{\mathrm{vc}}^2 \big) + \mathcal{O} \big(\kappa_0^2 \kappa_1^2 \eta^2 \beta^2 \epsilon_{\mathrm{sbs}}^2 \big) + \nonumber\\
    &\mathcal{O} \big( \kappa_0^2 \kappa_1^2 \kappa_2^2 \eta^2 \beta^4 \epsilon_{\mathrm{mbs}}^2 \big) + \mathcal{O} \big(\kappa_0^2 \kappa_1^2 \kappa_2^2 \kappa_3^2 \beta^4 \eta^2 \epsilon^2 \big) + \nonumber\\
    &\underbrace{\mathcal{O} \big(\delta^{\mathrm{th}} \beta^2 D^2 \big)}_{\mathrm{pruning ~error}} + \underbrace{\mathcal{O} \big( \beta \eta G^2 \cdot \varphi_\mathrm{w,0}(\pmb{\delta},\pmb{\mathrm{f}}, \pmb{\mathrm{P}}) \big) }_{\mathrm{wireless ~ factor}}. 
\end{align} 
\end{Corollary}

\begin{Remark}
\label{ConvergenceRemark}
    In (\ref{Theorem1}), the third, fourth, fifth and sixth terms appear from the divergence between client-VC, VC-sBS, sBS-mBS and mBS-global loss functions, respectively.
\end{Remark}
\begin{Remark}
    In typical HFL with no model pruning, i.e., $\delta_u^t=0$, $\forall u \in \mathcal{U}$, $\mathcal{O} \big(\delta^{\mathrm{th}} \beta^2 D^2 \big)=0$.
    Besides, when the wireless links are ignored, the last term in (\ref{Theorem1}) becomes zero.
    In such a special case, Corollary \ref{Theorem1_Convergence} boils down to
    \begin{align}
    \theta_\mathrm{{PHFL}} 
    &\leq \mathcal{O} \big( [f(\tilde{\bar{\mathbf{w}}}^0) - f_{\mathrm{inf}}]/ [\eta T] \big)  + \mathcal{O} ( \beta \eta \sigma^2/\mathrm{U} ) + \nonumber\\
    &\mathcal{O} \big( \kappa_0^2 \eta^2 \beta^2 \epsilon_{\mathrm{vc}}^2 \big) + \mathcal{O} \big(\kappa_0^2 \kappa_1^2 \eta^2 \beta^2 \epsilon_{\mathrm{sbs}}^2 \big) + \nonumber\\
    &\mathcal{O} \big( \kappa_0^2 \kappa_1^2 \kappa_2^2 \eta^2 \beta^4 \epsilon_{\mathrm{mbs}}^2 \big) + \mathcal{O} \big(\kappa_0^2 \kappa_1^2 \kappa_2^2 \kappa_3^2 \beta^4 \eta^2 \epsilon^2 \big).
    \end{align}
\end{Remark}

\begin{Remark}
    When $\eta = \sqrt{\mathrm{U}/T}$, we have $T \geq 1/[72 \mathrm{U} \kappa_0^2 \kappa_1^2 \kappa_2^2 \kappa_3^2 \beta^2]$.
    With a sufficiently large $T$, when the trained model reception success probability is $1$ for all users in all time steps, we have $\theta_{\mathrm{PHFL}} \approx \mathcal{O} \big(1/\sqrt{\mathrm{U} T} \big) + \mathcal{O} \big(\delta^{\mathrm{th}} \beta^2 D^2 \big)$, where the second term comes from the pruning error. 
    Therefore, the proposed PHFL solution converges to the neighborhood of a stationary point of traditional HFL.
\end{Remark}

\section{Joint Problem Formulation and Solution}
\label{section_jointProbSol}

\noindent
Similar to existing literature\cite{yang2020energy, richeng2022communication, hosseinalipour2022multi, pervej2023resource}, we ignore the downlink delay in this paper since the sBS can utilize the higher spectrum and transmission power to broadcast the updated model.
Moreover, since the sBS-mBS and mBS-cloud server links are wired, we ignore the transmission delays for these links\footnote{The transmissions over the wired sBS-mBS and mBS-cloud server links happen in the backhaul, and the corresponding delays are quite small.
In order to calculate these delays, one should also consider the overall network loads, which are beyond the scope of this paper.}. 
Furthermore, since the sBS, mBS and the cloud server usually have high computation power, we also ignore the model aggregation and processing delays\footnote{The addition of $d$ parameters and then taking the average have a time complexity of $\mathcal{O}(d+1)$. 
With highly capable CPUs at the sBS, mBS, and central server, the corresponding time delays for parameter aggregation are usually small and therefore ignored in the literature \cite{xu2021adaptive,Liu2022jointUE,Liu2022Joint}.}.
Therefore, at the beginning of each VC round, i.e., $t \ni (m\prod_{z=0}^{\kappa_z} + t_0) ~ \textsl{mod} ~ \kappa_0 = 0$, we first calculate the required computation time for finding the lottery ticket as
\begin{equation}
\label{localCompLottery}
\begin{aligned}
    \mathrm{t}_{i}^{\mathrm{cp_d}} = \rho \times \left(b n c_{i} \mathrm{D}_{i}/ \mathrm{f}_{i}^t \right),
\end{aligned}    
\end{equation}
where $b$ is the batch size, $n$ is the number of batches, $c_{i}$ is the CPU cycles to process $1$-bit data, $\mathrm{D}_{i}$ is UE $u_{i}$'s each data sample's size in bits and $\mathrm{f}_{i}^t$ is the CPU frequency.
Upon finding the pruned model, each client performs $\kappa_0$ local iterations, which require the following computation time \cite{Liu2022Joint} 
\begin{equation}
\label{localComp_Sparse}
\begin{aligned}
    \mathrm{t}_{i}^{\mathrm{cp_s}} = \kappa_0 \times \left( b n (1 - \delta_{i}^t) c_{i} \mathrm{D}_{i} / \mathrm{f}_{i}^t\right).
\end{aligned}    
\end{equation}
To that end, the UE only offloads the non-zero weights along with the binary mask to the sBS.
As such, we calculate the uplink payload size of UE $i$ as follows\footnote{Note that one may send the non-pruned weights and the corresponding indices, which are unknown until the original initial model is trained for $\rho$ iterations. We consider an upper bound for the uplink payload, which will be used during the joint parameters optimization phase.}: 
\begin{equation}
\label{uplinkPayload}
\begin{aligned}
    s_{i} \leq d \big[1 - \delta_{i}^t\big] \left(\mathrm{FPP} + 1\right) + d,  
\end{aligned}
\end{equation}
where $\mathrm{FPP}$ is the floating point precision.
Note that, in (\ref{uplinkPayload}), we need $1$ bit to represent the sign of the entry.
Therefore, we calculate the uplink payload offloading delay as follows:
\begin{equation}
\label{offlodingDelay}
\begin{aligned}
    \mathrm{t}_{i}^{\mathrm{up}} \leq d\left[ \left(1 - \delta_{i}^t\right) \left(\mathrm{FPP} + 1 \right) + 1 \right] / r_{i}^t.
\end{aligned}
\end{equation}
As such, UE $i$'s total duration for local computing and trained model offloading is 
\begin{equation}
\label{ueTotalTime}
    \mathrm{t}_{i}^{\mathrm{tot}} \leq \mathrm{t}_{i}^{\mathrm{cp_d}} + \mathrm{t}_{i}^{\mathrm{cp_s}} + \mathrm{t}_{i}^{\mathrm{up}}.
\end{equation}

We now calculate the energy consumption for performing the model training, followed by the required energy for offloading the trained models. 
First, let us calculate the energy consumption to get the lottery ticket as 
\begin{equation}
\label{energyConLot}
\begin{aligned}
    \mathrm{e}_{i}^{\mathrm{cp_d}} = \rho \times  0.5\xi b n c_i \mathrm{D}_{i} (\mathrm{f}_{i}^t)^2,
\end{aligned}    
\end{equation}
where $0.5\xi$ is the effective capacitance of UE's CPU chip.
Similarly, we calculate the energy consumption to train $\kappa_0$ local iterations using the pruned model as 
\begin{equation}
\label{energyConSparse}
\begin{aligned}
    \mathrm{e}_{i}^{\mathrm{cp_s}} = \kappa_0 \times 0.5\xi b n ( 1 - \delta_{i}^t ) c_i \mathrm{D}_{i} (\mathrm{f}_{i}^t)^2.
\end{aligned}    
\end{equation}
Moreover, we calculate the uplink payload offloading energy consumption as follows: 
\begin{equation}
\label{energyConOffload}
\begin{aligned}
    \mathrm{e}_{i}^{\mathrm{up}} \leq d\left[ \left(1 - \delta_{i}^t\right) \left(\mathrm{FPP} + 1 \right) + 1 \right] P_{i}^t / r_{i}.  
\end{aligned}
\end{equation}
Therefore, the total energy consumption is calculated as  
\begin{equation}
\label{ueTotalEnergyCons}
    \mathrm{e}_{i}^{\mathrm{tot}} \leq \mathrm{e}_{i}^{\mathrm{cp_d}} + \mathrm{e}_{i}^{\mathrm{cp_s}} + \mathrm{e}_{i}^{\mathrm{up}}.
\end{equation}

\subsection{Problem Formulation}
\noindent
Denote the duration between VC aggregation $\mathrm{t}^{\mathrm{th}}$.
Then, we calculate the probability of successful reception of UE's trained model as follows:
\begin{align}
\label{probSuc}
    p_{i}^t 
    &= \mathrm{Pr} \big\{ \mathrm{t}_i^{\mathrm{tot}} \leq \mathrm{t}^{\mathrm{th}} \big\} 
    \! = \! \mathrm{Pr} \big\{\! s_i \! \leq \! r_i^t \big[ \mathrm{t}^{\mathrm{th}} \rs - \mathrm{t}_i^{\mathrm{cp_d}} \rs - \mathrm{t}_i^{\mathrm{cp_s}} \big] \rs \big\} 
    \nonumber\\
    &\rs = \mathrm{Pr} \big\{h_{i,k}^t \geq \big[ (2^{\chi_i^t} -1)(\omega \zeta^2 + I_{i,k}^t)/(\mathrm{P}_i^t d_{i,k}^{-\alpha}) \big] \big\} \nonumber\\
    &\overset{(a)}{=} \exp \big[- (2^{\chi_i^t} -1)(\omega \zeta^2 + I_{i,k}^t) / (\mathrm{P}_i^t d_{i,k}^{-\alpha}) \big],
\end{align}
where $\chi_i^t = \frac{d \mathrm{f}_i^t \left[ \left(1 - \delta_{i}^t\right) \left(\mathrm{FPP} + 1 \right) + 1 \right] } {\omega \left[ \mathrm{f}_i^t \mathrm{t^{th}} - b n c_i \mathrm{D}_i \left(\rho + \kappa_0 (1 - \delta_i^t) \right) \right] }$ and $(a)$ follows from the Rayleigh fading channels between the UE and the sBS.

Notice that the pruning ratio $\delta_i^t$, CPU frequency $\mathrm{f}_i^t$, transmission power $\mathrm{P}_i^t$ and the probability of successful model reception $p_i^t$ are intertwined.
More specifically, $p_i^t$ depends on $\delta_i^t$, $\mathrm{f}_i^t$ and $\mathrm{P}_i^t$, given that the other parameters remain fixed.
As such, we aim to optimize these parameters jointly by considering the controllable terms in our convergence bound in Corollary \ref{Theorem1_Convergence}.
Therefore, we focus on each VC round, i.e., the local iteration round $t$ at which $t \textsl{ mod } \kappa_0=0$.
Specifically, we focus on minimizing the error terms due to pruning and wireless links, which are given by
\begin{subequations}
\label{objFunc_Original}
\begin{align}
    \mathcal{O} \big(\delta^{\mathrm{th}} \beta^2 D^2 \big) + \mathcal{O} \big( \beta \eta G^2 \cdot \varphi_\mathrm{w,0}(\pmb{\delta},\pmb{\mathrm{f}}, \pmb{\mathrm{P}}) \big) . \tag{\ref{objFunc_Original}} 
\end{align}
\end{subequations}

\begin{Remark}
In the above expression, the first term appears from the pruning error ${\frac{12 \beta^2}{T} \sum_{t=0}^{T-1} \sum_{l=1}^{L}} \cdot \alpha_l \sum_{k=1}^{B_l} \alpha_k \sum_{j=1}^{V_{k,l}} \alpha_j \sum_{i=1}^{U_{j,k,l}} \alpha_i \big( 1 + 2\big\{\alpha_i \big[1 + \alpha_j \big(1 + \alpha_k \big\{ 1 + \alpha_l \big\} \big) \big] \big\} \big) \delta_i^t \Vert \mathbf{w}_i^t \Vert^2 \leq \mathcal{O} \big(\delta^{\mathrm{th}} \beta^2 D^2 \big)$, while the second term comes from the wireless factor $\frac{2 \beta \eta}{T}\sum_{l=1}^{L} \alpha_l^2 \cdot \sum_{k=1}^{B_l} \alpha_k^2 \sum_{j=1}^{V_{k,l}} \alpha_j^2 \sum_{i=1}^{U_{j,k,l}} \alpha_i^2 \sum_{t=0}^{T-1} \left(1/p_i^t - 1 \right) \mathbb{E} \big\Vert \tilde{g}\big(\tilde{\mathbf{w}}_i^{t,0}\big) \big\Vert^2 
\leq \mathcal{O} \big( \frac{\beta \eta G^2}{T} \sum_{l=1}^{L} \alpha_l^2 \sum_{k=1}^{B_l} \alpha_k^2 \sum_{j=1}^{V_{k,l}} \alpha_j^2 \sum_{i=1}^{U_{j,k,l}} \alpha_i^2 \sum_{t=0}^{T-1} \left(1/p_i^t - 1 \right) \big)$.
\end{Remark}

Based on the above observations, we consider a weighted combination of these two terms as our objective function to minimize the bound in (\ref{objFunc_Original}).
Using (\ref{probSuc}) in the wireless factor, we, therefore, consider the following objective function.
\begin{subequations}
\label{objFunc_Orig_1}
\begin{align}
    &\rs \rs \varphi^t(\pmb{\delta}^t, \pmb{\mathrm{f}}^t, \pmb{\mathrm{P}}^t) \rs = 
     \phi_1 \sum\nolimits_{l=1}^{L} \rs \alpha_l \sum\nolimits_{k=1}^{B_l} \rs \alpha_k \sum\nolimits_{j=1}^{V_{k,l}} \rs \alpha_j \sum\nolimits_{i=1}^{U_{j,k,l}}  \alpha_i \delta_i^{t} + \rs \tag{\ref{objFunc_Orig_1}} \\
     &\rs \rs \phi_2 \sum_{l=1}^{L} \rs \alpha_l \sum_{k=1}^{B_l} \rs \alpha_k \rs \sum_{j=1}^{V_{k,l}} \rs \alpha_j \rs \sum_{i=1}^{U_{j,k,l}} \rs \alpha_i \bigg[ \exp \bigg( \frac{(2^{\chi_i^t} - 1)(\omega \zeta^2 + I_{i,k}^t)} {\mathrm{P}_i^t d_{i,k}^{-\alpha}} \bigg) - 1 \bigg], \nonumber  
\end{align}
\end{subequations}
where $\phi_1$ and $\phi_2$ are two weights to strike the balance between the terms.
Note that the wireless factor is multiplied by the learning rate and gradient in (\ref{objFunc_Original}).
Typically, the learning rate is small. 
Besides, the gradient becomes smaller as the training progresses. 
As such, the wireless factor term is relatively small when $p_i^t > 0$ for all UEs and VC aggregation rounds. 
The model weights are non-negative. 
Furthermore, a larger pruning ratio $\delta_i^t$ can dramatically reduce the computation and offloading time, making the wireless factor $0$. 
However, as a higher pruning ratio means more model parameters are pruned, we wish to avoid making the $\delta_i^t$'s large to reduce the pruning-induced errors. 
The above facts suggest we put more weight on the pruning error term to penalize more for the $\delta_i^t$'s. 
As such, we consider $\phi_1 \gg \phi_2$.
However, in our resource-constrained setting, a small $\delta_i^t$ can prolong the training and offloading time, leading $p_i^t$ to be $0$, i.e., the sBS will not receive the local trained model. 
Therefore, although $\phi_2$ is small, we keep the wireless factor to ensure $p_i^t$ is never $0$.

Therefore, we pose the following optimization problem to configure the parameters jointly.
\begin{subequations}
\label{originalOptProb}
\begin{align}
    \underset{\pmb{\delta}^t, \pmb{\mathrm{f}}^t, \pmb{\mathrm{P}}^t} {\text{minimize }} &\quad  \varphi^t(\pmb{\delta}^t, \pmb{\mathrm{f}}^t, \pmb{\mathrm{P}}^t ) , \tag{\ref{originalOptProb}} \\
    \text{   s.t.} \quad \label{cons1} (C1) & \quad \mathrm{t}_i^{\mathrm{tot}} \leq \mathrm{t}^{\mathrm{th}}, \Squad ~ \forall i, \\
    \label{cons2} (C2) & \quad \mathrm{e}_i^{\mathrm{tot}} \leq \mathrm{e}_i^{\mathrm{th}}, \Squad \forall i, \\
    \label{cons3} (C3) & \quad 0 \leq \mathrm{f}_i^t \leq \mathrm{f}_i^{\mathrm{max}}, \qquad \qquad~ \forall i, \\
    \label{cons4} (C4) & \quad 0 \leq \mathrm{P}_i^t \leq P_{i}^{\mathrm{max}}, \qquad\qquad \forall i,\\
    \label{cons5} (C5) & \quad 0 \leq \delta_{i}^t \leq \delta^{\mathrm{th}}, \qquad\qquad ~~ \forall i,
\end{align}
\end{subequations}
where constraint $(C1)$ ensures that the completion of one VC round is within the required deadline.
Constraint $(C2)$ controls the energy expense to be within the allowable budget.
Besides, $(C3)$ and $(C4)$ restrict us from choosing the CPU frequency and transmission power within the UE's minimum and maximum CPU cycles and transmission power, respectively.
Finally, constraint $(C5)$ ensures the pruning ratio to be within a tolerable limit $\delta^{\mathrm{th}}$.

\begin{Remark}
We assume that clients' system configurations remain unchanged over time, while the channel state information (CSI) is dynamic and known at the sBS. 
The clients share their system configurations with their associated sBS.
The sBSs share their respective users' system configurations and CSI with the central server.
As such, problem (\ref{originalOptProb}) is solved centrally, and the optimized parameters are broadcasted to the clients.
Besides, problem (\ref{originalOptProb}) is non-convex with the multiplications and divisions of the optimization variables in the second term. 
Moreover, constraints (C$1$) and (C$2$) are not convex.
Therefore, it is not desirable to minimize this original problem directly.
In the following, we transform the problem into an approximate convex problem that can be solved efficiently.
\end{Remark}

\subsection{Problem Transformation}
\noindent
Let us define $A(\delta_i^t,\mathrm{f}_i, P_i) \coloneqq \exp \big[ (2^{\chi_i^t} - 1)(\omega \zeta^2 + I_{i,k}^t) /(\mathrm{P}_i^t d_{i,k}^{-\alpha}) \big]$. 
Given an initial feasible point set ($\delta_{i}^{t,q}$, $\mathrm{f}_{i}^{t,q}$, $\mathrm{P}_{i}^{t,q}$),
we perform a linear approximation of this non-convex expression as follows:
\begin{align}
\label{offloadProbApprox}
   \rs\rs A(\delta_i^t,\mathrm{f}_i, P_i) \rs
   &\approx \rs A(\delta_{i}^{t,q},\mathrm{f}_{i}^{t,q}, \mathrm{P}_{i}^{t,q}) \rs + \rs \nabla_{\delta_i^t} \rs \big[ A(\delta_{i}^{t,q},\mathrm{f}_{i}^{t,q}, \mathrm{P}_{i}^{t,q})\big] (\delta_i^t - \delta_{i}^{t,q}) \nonumber\\
   &\rs\rs + \nabla_{\mathrm{f}_i^t} \big[ A (\delta_{i}^{t,q},\mathrm{f}_{i}^{t,q}, \mathrm{P}_{i}^{t,q})\big] (\mathrm{f}_i^t - \mathrm{f}_{i}^{t,q}) + \nonumber\\
   &\rs \rs \nabla_{\mathrm{P}_i^t} \big[ A(\delta_{i}^{t,q},\mathrm{f}_{i}^{t,q}, \mathrm{P}_{i}^{t,q})\big] (\mathrm{P}_i^t - \mathrm{P}_{i}^{t,q}) 
   \coloneqq \tilde{A}(\delta_i^t,\mathrm{f}_i^t, \mathrm{P}_i^t), \rs
\end{align}
where $A(\delta_{i}^{t,q},\mathrm{f}_{i}^{t,q}, \mathrm{P}_{i}^{t,q}) = \exp \big[(2^{\chi_i^{t,q}} -1)(\omega \zeta^2 + \tilde{I}_{i,k}^t) /(\mathrm{P}_{i}^{t,q} d_{i,k}^{-\alpha}) \big]$, $\chi_i^{t,q} = \frac{d \mathrm{f}_{i}^{t,q} \left[ \left(1 - \delta_{i}^t\right) \left(\mathrm{FPP} + 1 \right) + 1 \right]} {\omega \left[\mathrm{f}_{i}^{t,q} \mathrm{t^{th}} - b n c_i \mathrm{D}_i \big(\rho + \kappa_0 (1 - \delta_{i}^{t,q}) \big) \right]}$ and $\tilde{I}_{i,k}^t=\sum_{l=1}^{L} \sum_{k'=1, k \neq k'}^{K} \sum_{j'=1}^{J_{k',l}} \sum_{u_{i'} \in \mathcal{U}_{j',k',l}} \mathrm{P}_{i'}^{t,q} h_{i',k'}  d_{i',k'}^{-\alpha}$.
Moreover, $\nabla_{\delta_i^t} \big[ A(\delta_{i}^{t,q},\mathrm{f}_{i}^{t,q}, \mathrm{P}_{i}^{t,q})\big]$, \newline  $\nabla_{\mathrm{f}_i^t} \big[ A(\delta_{i}^{t,q},\mathrm{f}_{i}^{t,q}, \mathrm{P}_{i}^{t,q})\big]$ and $\nabla_{\mathrm{P}_i^t} \big[ A(\delta_{i}^{t,q},\mathrm{f}_{i}^{t,q}, \mathrm{P}_{i}^{t,q})\big]$ are calculated in (\ref{nablaDelta}), (\ref{nablafreq}) and (\ref{nablaPower}), respectively.

\begin{figure*}[!t]
\begin{align}
    &\nabla_{\delta_i^t} \big[ A(\delta_{i}^{t,q},\mathrm{f}_{i}^{t,q}, \mathrm{P}_{i}^{t,q})\big] = \{\ln(2) 2^{\chi_i^{t,q}} d \mathrm{f}_{i}^{t,q} A(\delta_{i}^{t,q},\mathrm{f}_{i}^{t,q}, \mathrm{P}_{i}^{t,q}) (\omega \zeta^2 + \tilde{I}_{i,k}^t) [b n c_i D_i (\rho (\mathrm{FPP} + 1) - \kappa_0 ) - \nonumber\\
    &\Squad \qquad \mathrm{f}_{i}^{t,q} \mathrm{t^{th}} (\mathrm{FPP} + 1)    ] \}/ \{\omega \mathrm{P}_{i}^{t,q} d_{i,k}^{-\alpha} \times [\mathrm{f}_{i}^{t,q} \mathrm{t^{th}} - b n c_i D_i (\rho + \kappa_0(1 - \delta_{i}^{t,q}))]^2\}. \label{nablaDelta}\\
    &\nabla_{\mathrm{f}_i} \big[ A(\delta_{i}^{t,q},\mathrm{f}_{i}^{t,q}, \mathrm{P}_{i}^{t,q})\big] = \{-\ln(2) 2^{\chi_i^{t,q}} b n c_i d D_i A (\delta_{i}^{t,q},\mathrm{f}_{i}^{t,q}, \mathrm{P}_{i}^{t,q}) (\omega \zeta^2 + \tilde{I}_{i,k}^t) [(1-\delta_{i}^{t,q}) (\mathrm{FPP} + 1) + 1] \times \nonumber\\
    &\Squad \qquad (\rho + \kappa_0 [1-\delta_{i}^{t,q}])\} / \{\omega \mathrm{P}_{i}^{t,q} d_{i,k}^{-\alpha} \times [\mathrm{f}_{i}^{t,q} \mathrm{t^{th}} - b n c_i D_i (\rho + \kappa_0 (1 - \delta_{i}^{t,q}))]^2\}.  \label{nablafreq} \hfill\\
    &\nabla_{P_i} \big[ A(\delta_{i}^{t,q},\mathrm{f}_{i}^{t,q}, \mathrm{P}_{i}^{t,q})\big] = - A(\delta_{i}^{t,q},\mathrm{f}_{i}^{t,q}, \mathrm{P}_{i}^{t,q}) (2^{\chi_i^{t,q}}-1) (\omega \zeta^2 + \tilde{I}_{i,k}^t)/[(P_{i}^{t,q})^2 d_{i,k}^{-\alpha}]. \label{nablaPower}
\end{align}
\begin{equation}
\label{offlodingDelay_Approx}
\begin{aligned}
    \mathrm{t}_{i}^{\mathrm{up}} 
    &\approx \frac{d \big(2 + \mathrm{FPP} - (1 + \mathrm{FPP}) \delta_i^t \big)}{\omega\log_2 (1 + \mathrm{P}_{i}^{t,q} h_{i,k} d_{i,k}^{-\alpha} /[\omega \zeta^2 + \tilde{I}_{i,k}^t] ) } + \frac{-\ln(2) d h_{i,k} d_{i,k}^{-\alpha}[(1 - \delta_{i}^{t,q}) (\mathrm{FPP} + 1) +1] \times (\mathrm{P}_i^t - \mathrm{P}_{i}^{t,q})} {\omega \{\ln (1 + [\mathrm{P}_{i}^{t,q} h_{i,k} d_{i,k}^{-\alpha}] / [\omega \zeta^2 + \tilde{I}_{i,k}^t] )\}^2 (\omega \zeta^2 + \tilde{I}_{i,k}^t + \mathrm{P}_{i}^{t,q} h_{i,k} d_{i,k}^{-\alpha}) } \coloneqq \mathrm{\tilde{t}}_{i}^{\mathrm{up}}. 
\end{aligned}
\end{equation}
\begin{align}
\label{energyConOffload_Approx}
    \mathrm{e}_{i}^{\mathrm{up}} 
    &\approx \{ d \mathrm{P}_{i}^{t,q} [ \left(\mathrm{FPP} + 2 \right) - (\mathrm{FPP}+1) \delta_{i}^t ]\} / \{ \omega\log_2 (1 + \mathrm{P}_{i}^{t,q} h_{i,k} d_{i,k}^{-\alpha}/(\omega \zeta^2 + \tilde{I}_{i,k}^t ) ) \} + \qquad \quad \nonumber\\
    &\Squad \qquad \frac{ d[(1-\delta_{i}^{t,q}) (\mathrm{FPP}+1) + 1]  \bigg[\rs \log_2 \rs \bigg(\rs \rs 1 + \frac{\mathrm{P}_{i}^{t,q} h_{i,k} d_{i,k}^{-\alpha} } {\omega \zeta^2 + \tilde{I}_{i,k}^t } \rs \bigg)\rs - \frac{\mathrm{P}_{i}^{t,q} h_{i,k} d_{i,k}^{-\alpha}} {\ln(2)(\omega \zeta^2 + \tilde{I}_{i,k}^t + \mathrm{P}_{i}^{t,q} h_{i,k} d_{i,k}^{-\alpha}) } \rs \bigg]} {\omega \{\log_2 (1 + \mathrm{P}_{i}^{t,q} h_{i,k} d_{i,k}^{-\alpha} / (\omega \zeta^2 + \tilde{I}_{i,k}^t) )\}^2 } (\mathrm{P}_i^t - \mathrm{P}_{i}^{t,q}) = \mathrm{\tilde{e}}_{i}^{\mathrm{up}}. \rs\rs\rs
\end{align}
\vspace{0.01in}
\hrule{}
\end{figure*}

As such, we approximate (\ref{objFunc_Orig_1}) as follows:
\begin{align}
\label{objFunc_Approx}
    \rs \rs \tilde{\varphi}^t(\pmb{\delta}^t, \pmb{\mathrm{f}}^t, \pmb{\mathrm{P}}^t) 
    &= \sum\nolimits_{l=1}^{L} \rs \alpha_l \sum\nolimits_{k=1}^{B_l} \rs \alpha_k \sum\nolimits_{j=1}^{V_{k,l}} \rs \alpha_j \sum\nolimits_{i=1}^{U_{j,k,l}} \rs \alpha_i  \big(  \phi_1 \delta_i^{t} + \nonumber\\
    &\Squad\qquad \phi_2 \big[ \tilde{A}(\delta_i^t,\mathrm{f}_i^t, \mathrm{P}_i^t) - 1 \big] \big), 
\end{align}
where $\tilde{A}(\delta_i^t,\mathrm{f}_i^t, \mathrm{P}_i^t)$ is calculated in (\ref{offloadProbApprox}). 

We now focus on the non-convex constraints.
First, let us approximate the local pruned model computation time as 
\begin{equation}
\label{localComp_Sparse_Approx}
\begin{aligned}
   \rs\rs\rs \mathrm{t}_{i}^{\mathrm{cp_s}} 
    \rs & \approx \rs [\kappa_0 b n c_{i} \mathrm{D}_{i} \!/\! \mathrm{f}_{i}^{t,q}] (\! 1 - \delta_i^t  - (1 - \delta_{i}^{t,q}) (\mathrm{f}_i^t - \mathrm{f}_{i}^{t,q}) / \mathrm{f}_{i}^{t,q} ) \rs = \rs \mathrm{\tilde{t}^{cp_s}}_i\rs.\rs\rs\rs\rs\rs
\end{aligned}    
\end{equation}
Then, we approximate the non-convex uplink model offloading delay as shown in (\ref{offlodingDelay_Approx}).
Using a similar treatment, we write
\begin{align}
\label{energyConSparse_Approx}
    \mathrm{e}_{i}^{\mathrm{cp_s}} 
    &\approx \kappa_0 \xi b n c_i \mathrm{D}_{i} \mathrm{f}_{i}^{t,q} \big[(\delta_{i}^{t,q} - 0.5) \mathrm{f}_{i}^{t,q} - 0.5 \mathrm{f}_{i}^{t,q} \delta_i^t + \nonumber\\
    &\Mquad \qquad (1 - \delta_{i}^{t,q})\mathrm{f}_i^t \big] \coloneqq \mathrm{\tilde{e}}_{i}^{\mathrm{cp_s}}.
\end{align}
Similarly, we approximate the energy consumption for model offloading as shown in (\ref{energyConOffload_Approx}).

Therefore, we pose the following transformed problem 
\begin{subequations}
\label{optimTransformed}
\begin{align}
    \underset{\pmb{\delta}^t, \pmb{\mathrm{f}}^t, \pmb{\mathrm{P}}^t} {\text{minimize }} &\quad  \tilde{\varphi}^t  (\pmb{\delta}^t, \pmb{\mathrm{f}}^t, \pmb{\mathrm{P}}^t) , \tag{\ref{optimTransformed}} \\
    \text{   s.t.} \quad \label{trans_cons1} (\tilde{C}1) & \quad \mathrm{\tilde{t}}_{i}^{\mathrm{cp_d}} + \mathrm{\tilde{t}}_{i}^{\mathrm{cp_s}} + \mathrm{\tilde{t}}_{i}^{\mathrm{up}} \leq \mathrm{t}^{\mathrm{th}}, \\
    \label{trans_cons2} (\tilde{C}2) & \quad \mathrm{e}_{i}^{\mathrm{cp_d}} + \mathrm{\tilde{e}}_{i}^{\mathrm{cp_s}} + \mathrm{\tilde{e}}_{i}^{\mathrm{up}} \leq \mathrm{e}_i^{\mathrm{th}}, \\
    & \quad (\ref{cons3}), (\ref{cons4}), (\ref{cons5}),
\end{align}
\end{subequations}
where the constraints are taken for the same reasons as in the original problem. 
Besides, $\mathrm{\tilde{t}}_{i}^{\mathrm{cp_d}} = 2\rho b n c_{i} \mathrm{D}_{i} / \mathrm{f}_{i}^{t,q} - \rho b n c_{i} \mathrm{D}_{i}\mathrm{f}_i^t / (\mathrm{f}_{i}^{t,q})^{2}$.

Note that problem (\ref{optimTransformed}) is now convex and can be solved iteratively using existing solvers such as CVX \cite{diamond2016cvxpy}.
The key steps of our iterative solution are summarized in Algorithm \ref{SCA_Sol_Algo}.
Moreover, as (\ref{optimTransformed}) has $3\mathrm{U}$ decision variables and $5\mathrm{U}$ constraints, the time complexity of running Algorithm \ref{SCA_Sol_Algo} for $Q$ iterations is $\mathcal{O} \left(Q \times [(3\mathrm{U})^3 \times 5\mathrm{U} ] \right)$ \cite{che2014joint}.
While Algorithm \ref{SCA_Sol_Algo} yields a suboptimal solution and converges to a local stationary solution set of the original problem (\ref{originalOptProb}), SCA-based solutions are well-known for fast convergence \cite{sun2019optimal}.
Moreover, our extensive empirical study in the sequel suggests that the proposed PHFL solution with Algorithm \ref{SCA_Sol_Algo} delivers nearly identical performance to the upper bounded performance.

\begin{algorithm}[t!]
\fontsize{8}{8}\selectfont
\SetAlgoLined 
\DontPrintSemicolon
\KwIn{Initial feasible set $(\pmb{\delta}^{t,0}, \pmb{\mathrm{f}}^{t,0}, \pmb{\mathrm{P}}^{t,0})$, maximum iteration $Q$, precision level $\epsilon^{\mathrm{prec}}$; set $q=0$}
\nl{\textbf{Repeat}:} \;
\Indp {
    Solve (\ref{optimTransformed}) using $(\pmb{\delta}^{t,q}, \pmb{\mathrm{f}}^{t,q}, \pmb{\mathrm{P}}^{t,q})$ to get the optimized $(\pmb{\delta}^{t}, \pmb{\mathrm{f}}^{t}, \pmb{\mathrm{P}}^{t})$\;
    $q \gets q+1$ ;
    $\pmb{\delta}^{t,q} \gets \pmb{\delta}^{t}$ ;
    $\pmb{\mathrm{f}}^{t,q} \gets \pmb{\mathrm{f}}^{t}$ ;
    $ \pmb{\mathrm{P}}^{t,q} \gets \pmb{\mathrm{P}}^{t}$ \;
    }
\Indm \textbf{Until} converge with $\epsilon^{\mathrm{prec}}$ precision or $q=Q$ \;
\KwOut{Optimal $(\pmb{\delta}^t, \pmb{\mathrm{f}}^t, \pmb{\mathrm{P}}^t)$}
\caption{Iterative Joint Pruning Ratio, CPU Frequency and Transmission Power Selection Process}
\label{SCA_Sol_Algo}
\end{algorithm}

\section{Simulation Results and Discussions}
\label{section_sim_results}
\subsection{Simulation Setting}
\label{simu_setting}
\noindent 
For the performance evaluation, we consider $L=2$, $B=4$ and $\mathrm{U}=48$.
We let each sBS maintain $2$ VCs, where each VC has $6$ UEs.
In other words, we have $U_{j,k,l}=6, \forall j,k$ and $l$, $V_{k,l}=2, \forall k$ and $l$, and $B_l=2, \forall l$.  
We assume $\omega=1$ megahertz (MHz).
We randomly generate maximum transmission power $P^{\mathrm{max}}$, energy budget for each VC aggregation round $\mathrm{e^{th}}$, CPU frequency $\mathrm{f^{max}}$ and required CPU cycle to process per-bit data $c$, respectively, from $[23, 30]$ dBm, $[10, 13]$ Joules, $[1.8, 2.8]$ gigahertz (GHz) and $[20, 25]$ for these two VCs.
Therefore, all UEs in a VC have the above randomly generated system configurations\footnote{Our approach can easily be extended where all clients can have random $\mathrm{f}_i^{\mathrm{max}}$'s, $\mathrm{e}_i^{\mathrm{th}}$'s and $P_i^{\mathrm{max}}$'s. 
Our approach is practical since these parameters depend on the clients' manufacturers and their specific models.}.
Moreover, as described earlier in Section \ref{sec_sysModel}, our proposed PHFL has $4$ tiers, namely ($1$) UE-VC, ($2$) VC-sBS, ($3$) sBS-mBS, and ($4$) mBS-central server.

For our ML task, we use image classification with the popular CIFAR-$10$ and CIFAR-$100$ datasets \cite{krizhevsky2009learning} for performance evaluation.
We use symmetric Dirichlet distribution $\mathrm{Dir}(\bar{\alpha})$ with concentration parameter $\bar{\alpha}$ for the non-IID data distribution as commonly used in literature \cite{richeng2022communication, pervej2023resource}.
Besides, we use $1)$ convolutional neural network (CNN), $2)$ residual network (ResNet)-$18$ \cite{he2016deep} and $3)$ ResNet-$34$ \cite{he2016deep}.
The CNN model has the following architecture:  $\tt{Conv2d(3,128)}$, $\tt{MaxPool2d}$, $\tt{Conv2d(128, 64)}$, $\tt{MaxPool2d}$, $\tt{Linear(256, 256)}$, $\tt{Linear(256, \#Labels)}$, whereas the ResNets have a similar architecture as in the original paper \cite{he2016deep}. 
Moreover, the total number of trainable parameters depends on various  configurations, such as the input/output shapes, kernel sizes, strides, etc. 
In our implementation, the original CNN, ResNet-$18$ and ResNet-$34$ models, respectively, have $151,882$; $6,992,138$ and $12,614,794$ trainable parameters on CIFAR-$10$, and $175,012$; $7,038,308$ and $12,660,964$ trainable parameters on CIFAR-$100$.
Besides, with $\mathrm{FPP}=32$, we have a wireless payload of about $5.01$ megabits (Mbs), $230.7$ Mbs and $416.3$ Mbs for CIFAR-$10$, and $5.8$ Mbs, $232.3$ Mbs and $417.8$ Mbs for CIFAR-$100$ datasets for the respective three original models.

\begin{figure*}
\begin{subfigure}{0.33\textwidth}
    \centering
    \includegraphics[trim=18 5 40 25, clip, width=\textwidth, height=0.16\textheight]{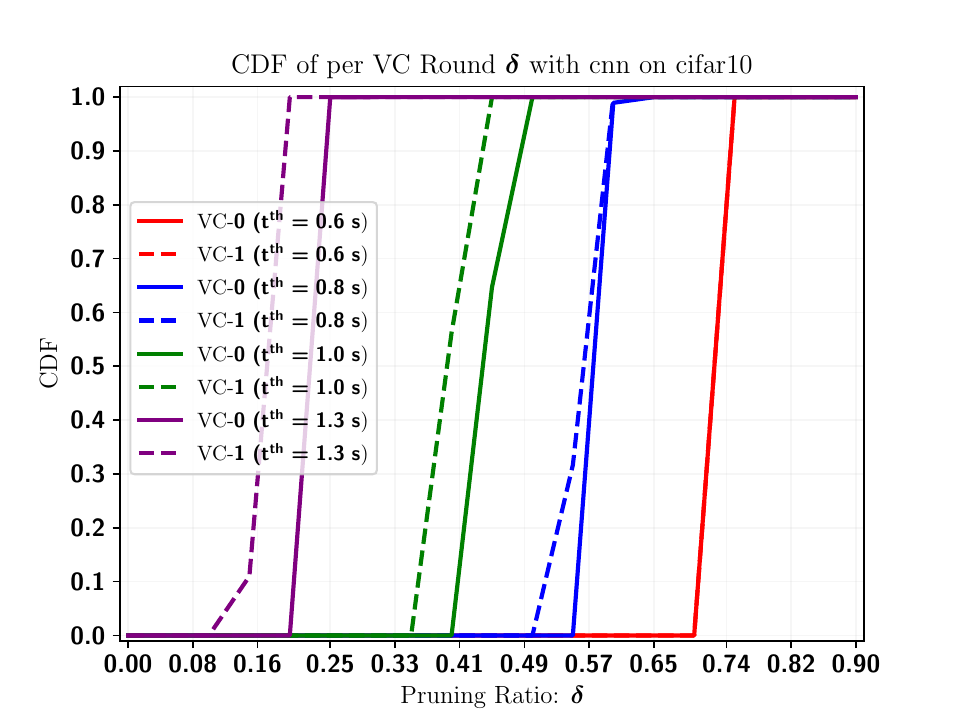}
    \caption{CDF of $\delta_u^t$'s with CNN}
    \label{cdfDelta5222cifar10cnn}
\end{subfigure} 
\begin{subfigure}{0.33\textwidth}
    \centering
    \includegraphics[trim=18 5 40 25, clip, width=\textwidth, height=0.16\textheight]{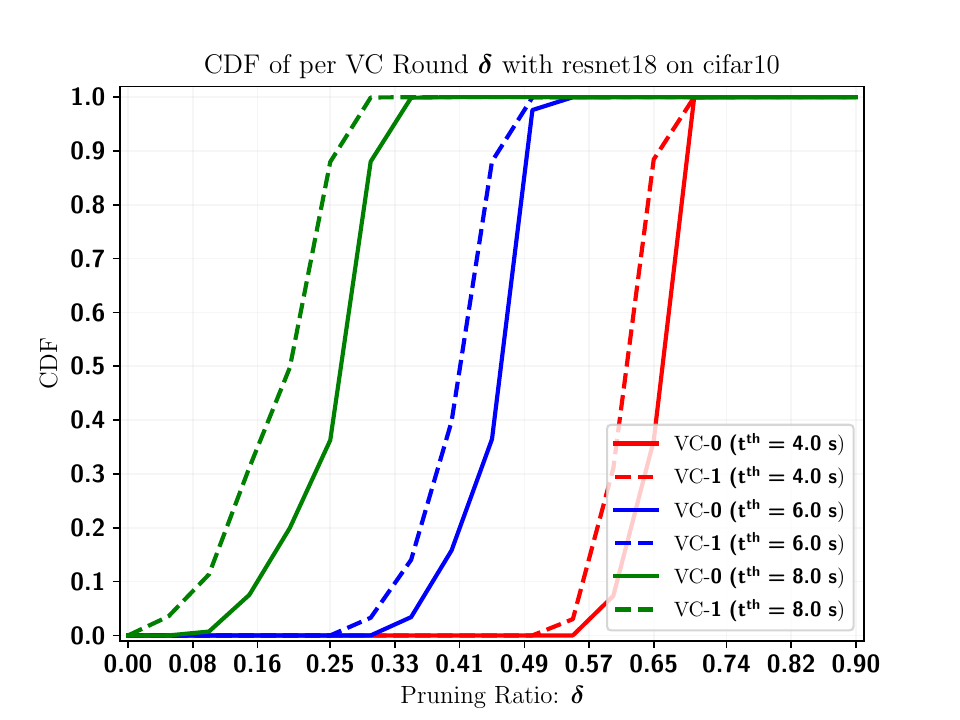}
    \caption{CDF of $\delta_u^t$'s with ResNet-$18$}
    \label{cdfDelta5222cifar10Resnet18}
\end{subfigure} 
\begin{subfigure}{0.325\textwidth}
    \centering
    \includegraphics[trim=18 5 40 25, clip, width=\textwidth, height=0.16\textheight]{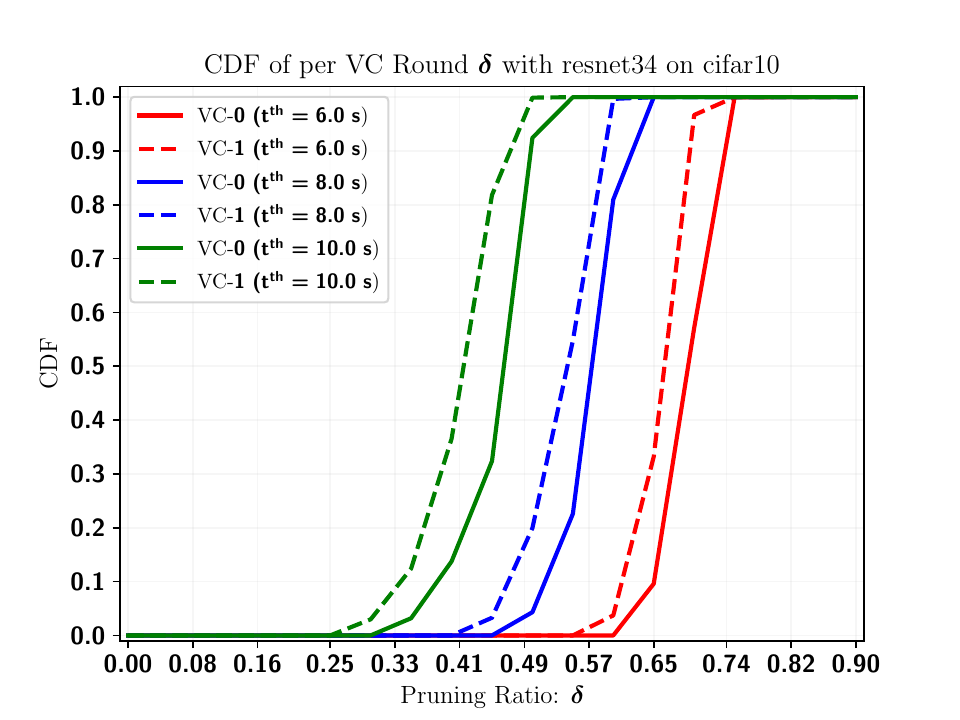}
    \caption{CDF of $\delta_u^t$'s with ResNet-$34$}
    \label{cdfDelta5222cifar10Resnet34}
\end{subfigure}
\caption{CDF of clients' pruning ratios in different VCs for different $\mathrm{t^{th}}$ with different ML models}
\label{pruningRatioCDFs}
\end{figure*}

\subsection{Performance Study}
\noindent 
First, we investigate the pruning ratios $\delta_i^t$'s in different VCs. 
When the system configurations remain the same, the pruning ratio depends on the deadline threshold $\mathrm{t^{th}}$.
More specifically, a larger deadline allows the client to prune fewer model parameters, given that the energy constraint is satisfied.
Intuitively, less pruning leads to a bulky model that takes longer training time.
The CNN model is shallower compared to the ResNets.
More specifically, the original non-pruned ResNet-$18$ and ResNet-$34$ models have about $46$ times and $83$ times the trainable parameters of the CNN model, respectively, on CIFAR-$10$.
Therefore, the clients require a larger $\mathrm{t^{th}}$ to perform their local training and trained model offloading as the trainable parameters increase.

Intuitively, given a fixed $\mathrm{t^{th}}$, the clients need to prune more model parameters for a bulky model in order to meet the deadline and energy constraints. 
Our simulation results also show that this general intuition holds in determining the $\delta_i^t$'s, as shown in Fig. \ref{pruningRatioCDFs}, which show the cumulative distribution function (CDF) of the $\delta_i^t$'s in different VCs. 
It is worth noting that the pruning ratios $\delta_i^t$'s in each VC aggregation rounds are not deterministic due to the randomness of the wireless channels. 
We know the optimal variables once we solve the optimization problem in (\ref{optimTransformed}), which depends on the realizations of the wireless channels.
Then, for a given VC $j$, we generate the plot by calculating $ \frac{\sum_{l=1}^{L} \sum_{k=1}^{B_l} \sum_{i=1}^{U_{j,k,l}} \mathbf{1} \left( \delta_i^{t} \leq \delta \right)} { \sum_{l=1}^{L} \sum_{k=1}^{B_l} \sum_{i=1}^{U_{j,k,l}} i }$, where $\mathbf{1} \left( \mathrm{\delta}_i^t \leq \delta \right)$ is an indicator function that takes value $1$ if $\delta_i^t \leq \delta$ and $0$ otherwise.
With the CNN model, about $50\%$ clients have a $\delta_i^t$ less than $0.23$, $0.43$, $0.58$ and $0.72$ in VC-$0$ in all cells, for $1.3$s, $1$s, $0.8$s and $0.6$s deadline thresholds, respectively, in Fig. \ref{cdfDelta5222cifar10cnn}.
Note that we use $\delta^{\mathrm{th}}=0.9$, i.e., the clients can prune up to $90\%$ of the neurons.
Moreover, we consider $\mathrm{t^{th}}=4$s and $\mathrm{t^{th}}=6$s, to make the problem feasible for all clients for the ResNet-$18$ and ResNet-$34$ models, respectively.
Furthermore, from Fig. \ref{cdfDelta5222cifar10cnn} - Fig. \ref{cdfDelta5222cifar10Resnet34}, it is quite clear that the UEs in VC-$1$ have to prune slightly lesser model parameters than the UEs in VC-$0$, even though the maximum CPU frequency $\mathrm{f^{max}}$ of the UEs in VC-$0$ is $2.58$ GHz, which is about $6.22\%$ higher than the UEs in VC-$1$.
However, due to the wireless payloads in the offloading phase, the transmission powers of the clients can also influence the $\delta_i^t$'s.
In our setting, the UEs' maximum transmission powers are $0.35$ Watt and $0.95$ Watt, respectively, in VC-$0$ and VC-$1$.
As such, with a similar wireless channel, the UEs in VC-$1$ can offload much faster than the UEs in VC-$0$. 
The above observations, thus, point out that trained model offloading time $\mathrm{t}_i^{\mathrm{up}}$ dominates the total time to finish one VC round $\mathrm{t}_i^{\mathrm{tot}}$.

\begin{figure*}
\begin{subfigure}{0.33\textwidth}
    \centering
    \includegraphics[trim=18 5 40 20, clip, width=\textwidth, height=0.16\textheight]{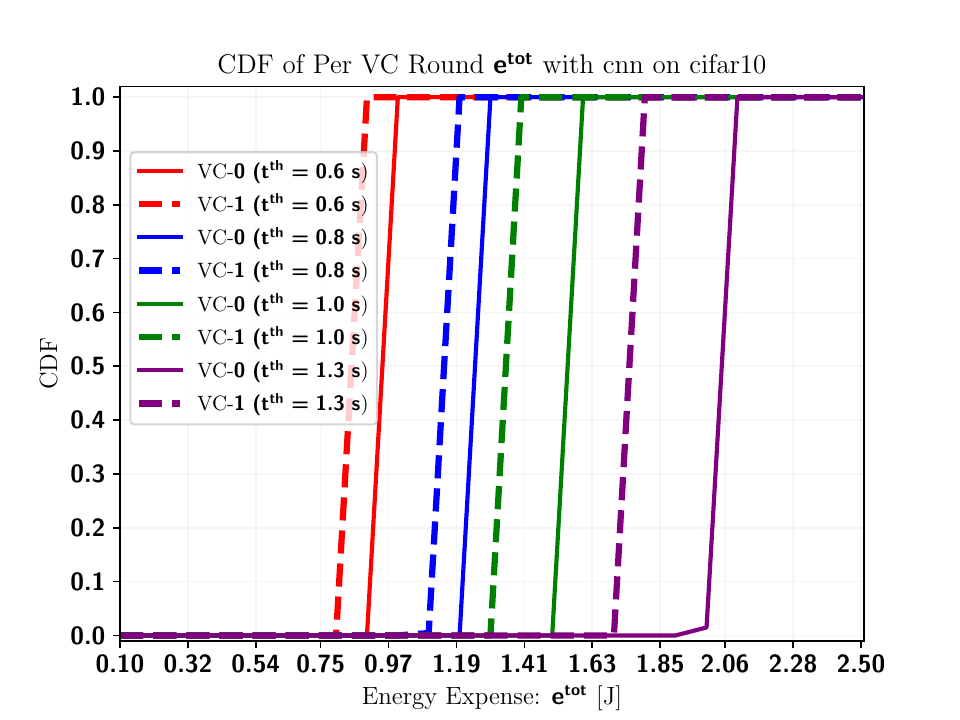}
    \subcaption{CDF of $\mathrm{e}_i^{\mathrm{tot}}$'s with CNN}
    \label{cdfEtot5222cifar10cnn}
\end{subfigure} 
\begin{subfigure}{0.33\textwidth}
    \centering
    \includegraphics[trim=18 5 40 20, clip, width=\textwidth, height=0.16\textheight]{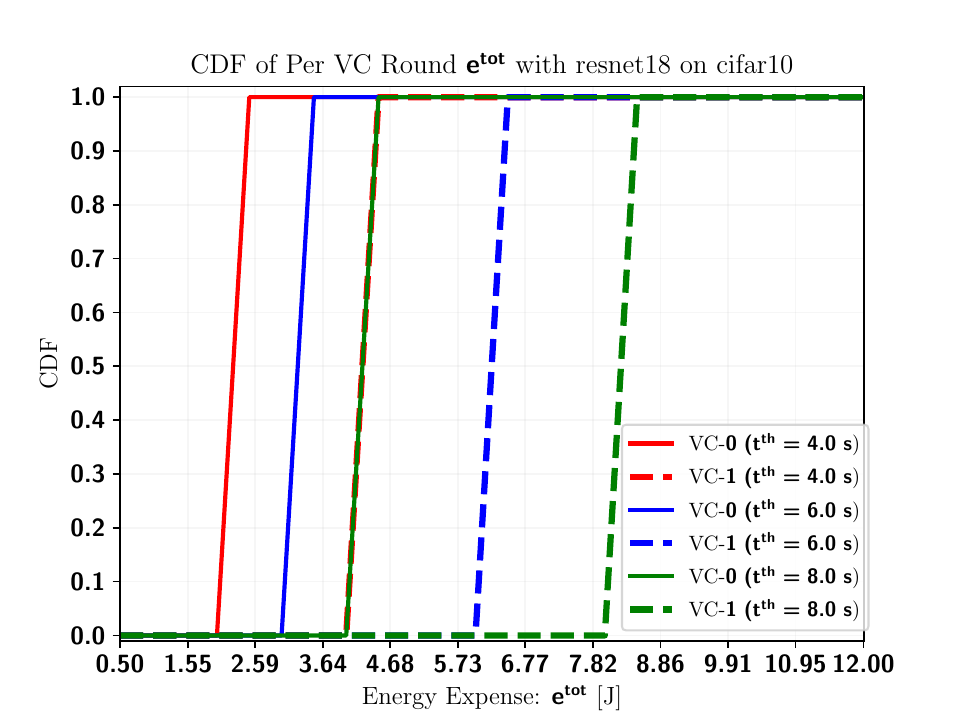}
    \subcaption{CDF of $\mathrm{e}_i^{\mathrm{tot}}$'s with ResNet-$18$}
    \label{cdfEtot5222cifar10Resnet18}
\end{subfigure} 
\begin{subfigure}{0.325\textwidth}
    \centering
    \includegraphics[trim=18 5 40 20, clip, width=\textwidth, height=0.16\textheight]{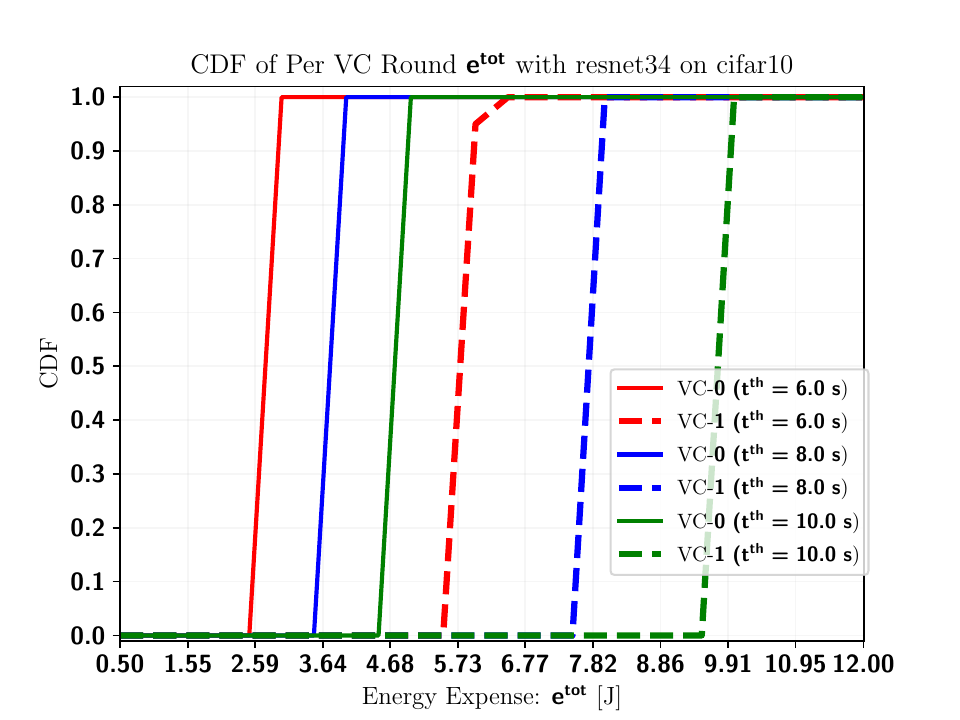}
    \subcaption{CDF of $\mathrm{e}_i^{\mathrm{tot}}$'s with ResNet-$34$}
    \label{cdfEtot5222cifar10Resnet34}
\end{subfigure}
\caption{CDF of clients' $\mathrm{e}_i^{\mathrm{tot}}$'s in different VCs for different $\mathrm{t^{th}}$ with different ML models}
\label{cdfEtot5222cifar10}
\end{figure*}

When it comes to energy expense, from (\ref{energyConSparse}) and (\ref{energyConOffload}), it is quite obvious that a high $\delta_i^t$ shall lead to less energy consumption for both the training and offloading.
However, the total energy expense of the clients boils down to the dominating factor between the required energy for computation and trained model offloading due to the interplay between the wireless and the learning parameters. 
Particularly, with $\omega=1$ MHz, the clients can offload the CNN model fast, leading the computational energy {consumption} to be the dominating factor.
The ResNets, on the other hand, have huge wireless overheads, leading the offloading time and energy be the dominating factors. 
This is also observed in our results in Fig. \ref{cdfEtot5222cifar10}, which is the CDF plot of the energy expense $\mathrm{e}_i^{\mathrm{tot}}$, calculated in (\ref{ueTotalEnergyCons}), of the clients in each VC and is generated following a similar strategy as in Fig. \ref{pruningRatioCDFs}. 
When the CNN model is used, the total energy cost of the clients in VC-$0$ is larger, even though they prune more parameters, compared to the clients in VC-$1$ since larger $\mathrm{f}_i^{\mathrm{max}}$'s of the clients in VC-$0$ leads to a higher computational energy cost. 
This, however, changes for the ResNets since the wireless communication burden dominates the computation burden. 
The clients in VC-$1$ can use their higher $P_i^\mathrm{max}$'s for the offloading time reduction when they determine the pruning ratios $\delta_i^t$'s.
As such, the total energy expenses of the clients in VC-$1$ are much larger than the ones in VC-$0$.  
Our simulation results in Fig. \ref{cdfEtot5222cifar10cnn} and Figs. \ref{cdfEtot5222cifar10Resnet18}-\ref{cdfEtot5222cifar10Resnet34} also reveal the same trends.

\begin{figure*}
\begin{subfigure}{0.33\textwidth}
    \centering
    \includegraphics[trim=10 5 35 25, clip, width=\textwidth, height=0.16\textheight]{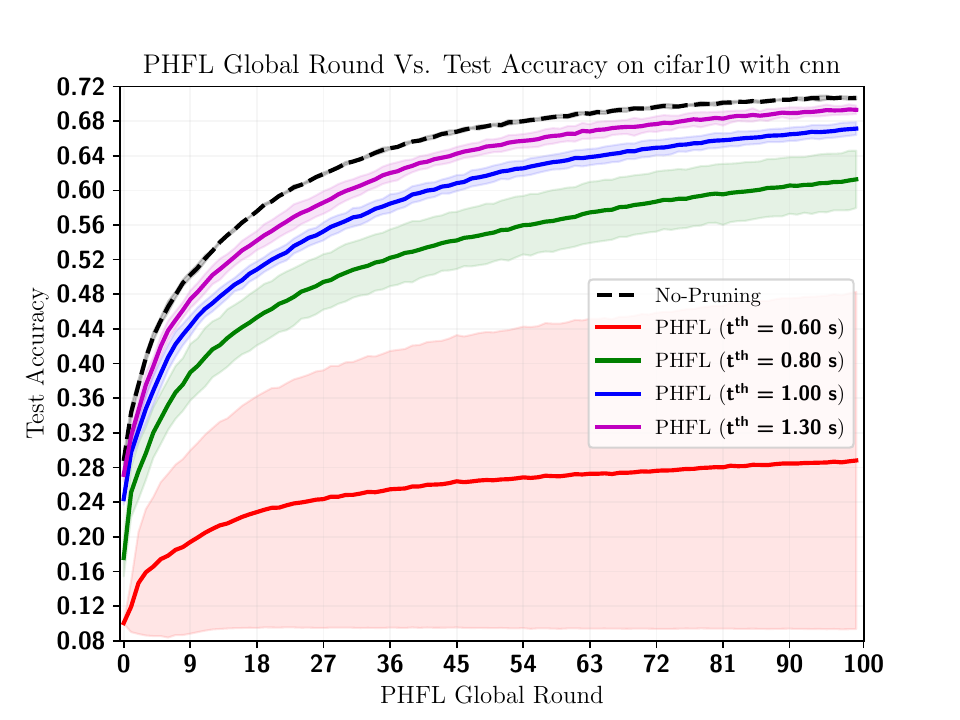}
    \subcaption{Test accuracy with CNN}
    \label{testAccVsDeadline5222cifar10cnn}
\end{subfigure} 
\begin{subfigure}{0.33\textwidth}
    \centering
    \includegraphics[trim=10 5 35 25, clip, width=\textwidth, height=0.16\textheight]{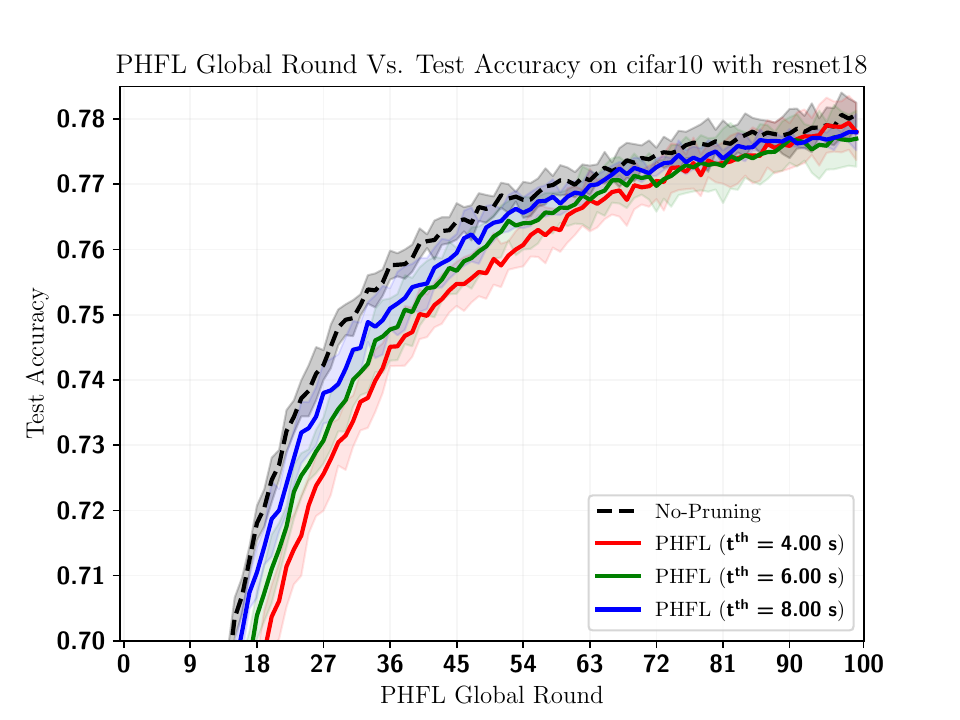}
    \subcaption{Test accuracy with ResNet-$18$}
    \label{testAccVsDeadline5222cifar10Resnet18}
\end{subfigure}
\begin{subfigure}{0.325\textwidth}
    \centering
    \includegraphics[trim=10 5 35 25, clip, width=\textwidth, height=0.16\textheight]{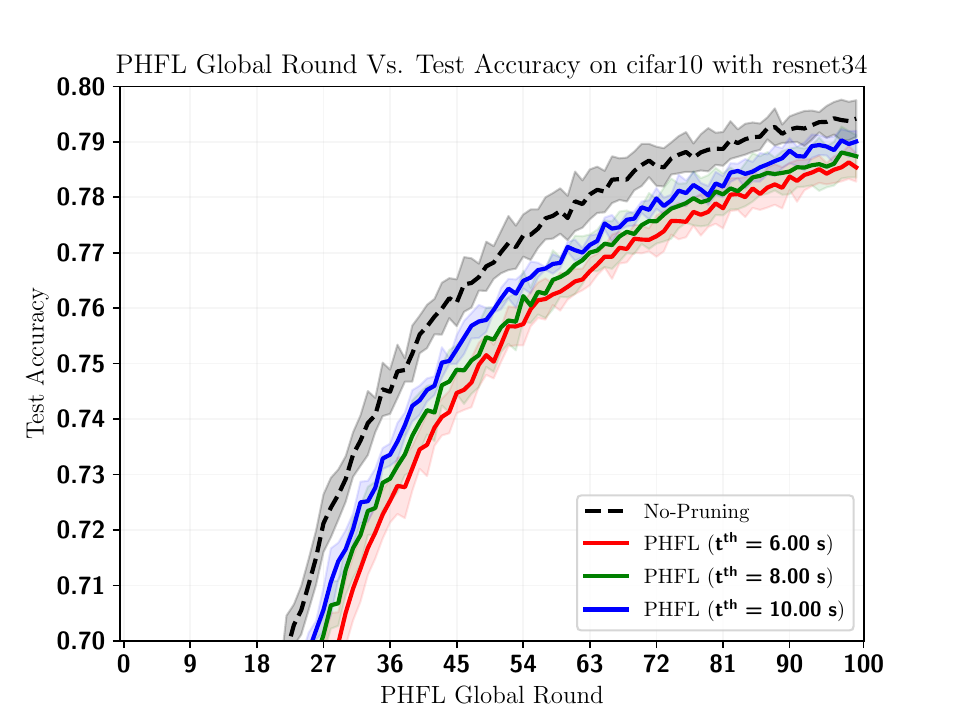}
    \subcaption{Test accuracy with ResNet-$34$}
    \label{testAccVsDeadline5222cifar10Resnet34}
\end{subfigure}
\begin{subfigure}{0.33\textwidth}
    \centering
    \includegraphics[trim=10 5 32 25, clip, width=\textwidth, height=0.16\textheight]{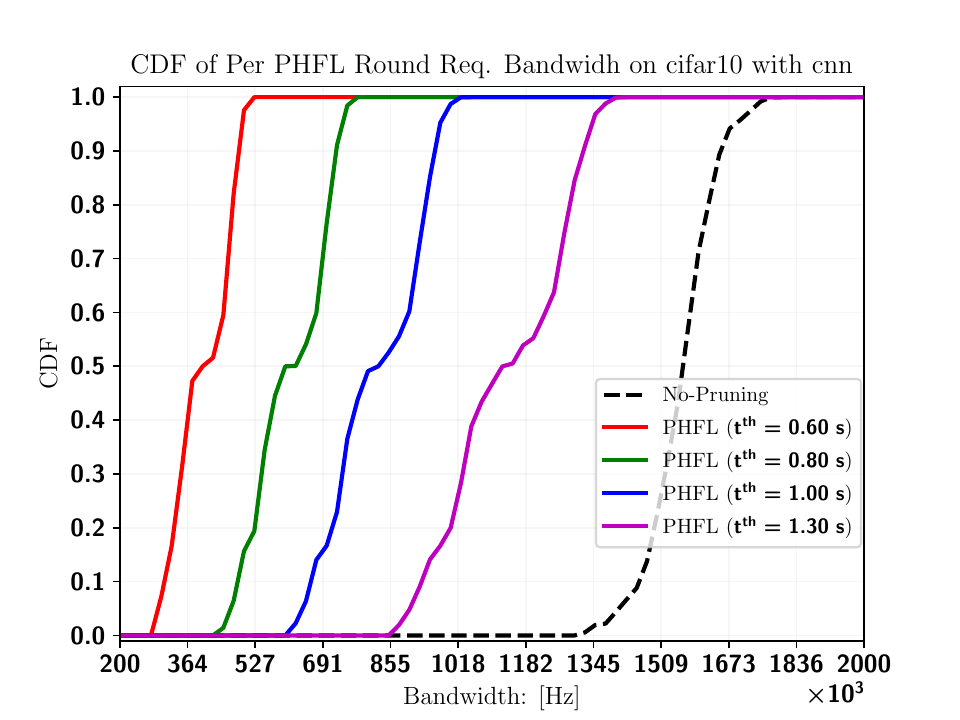}
    \subcaption{Required bandwidth with CNN}
    \label{bwVsDeadline5222cifar10cnn}
\end{subfigure} 
\begin{subfigure}{0.33\textwidth}
    \centering
    \includegraphics[trim=10 5 35 25, clip, width=\textwidth, height=0.16\textheight]{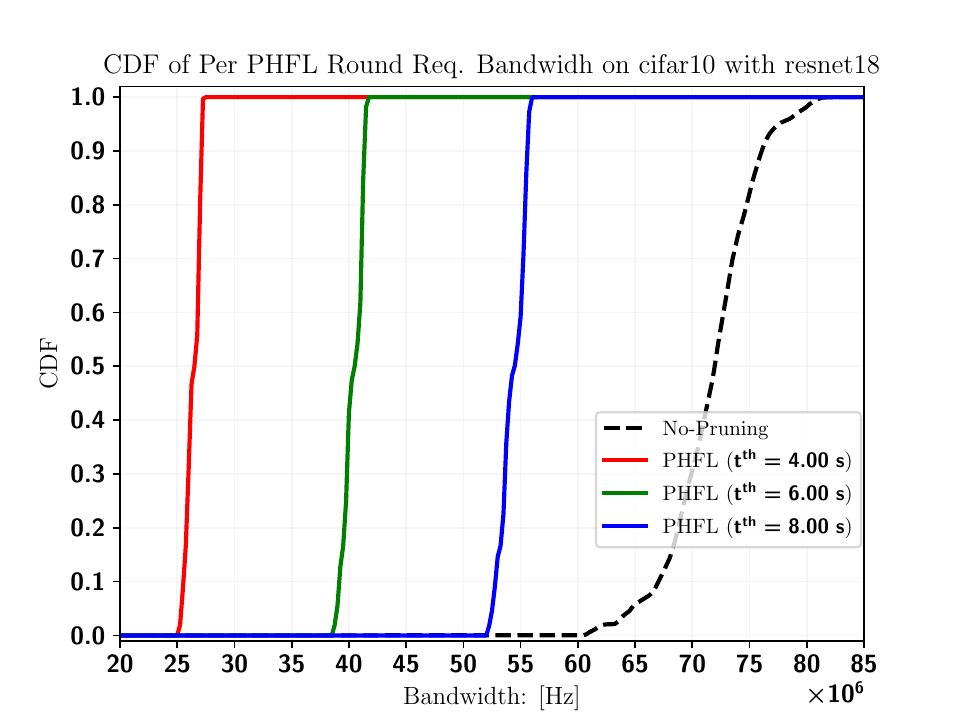}
    \subcaption{Required bandwidth with ResNet-$18$}
    \label{bwVsDeadline5222cifar10Resnet18}
\end{subfigure}
\begin{subfigure}{0.325\textwidth}
    \centering
    \includegraphics[trim=10 5 32 25, clip, width=\textwidth, height=0.16\textheight]{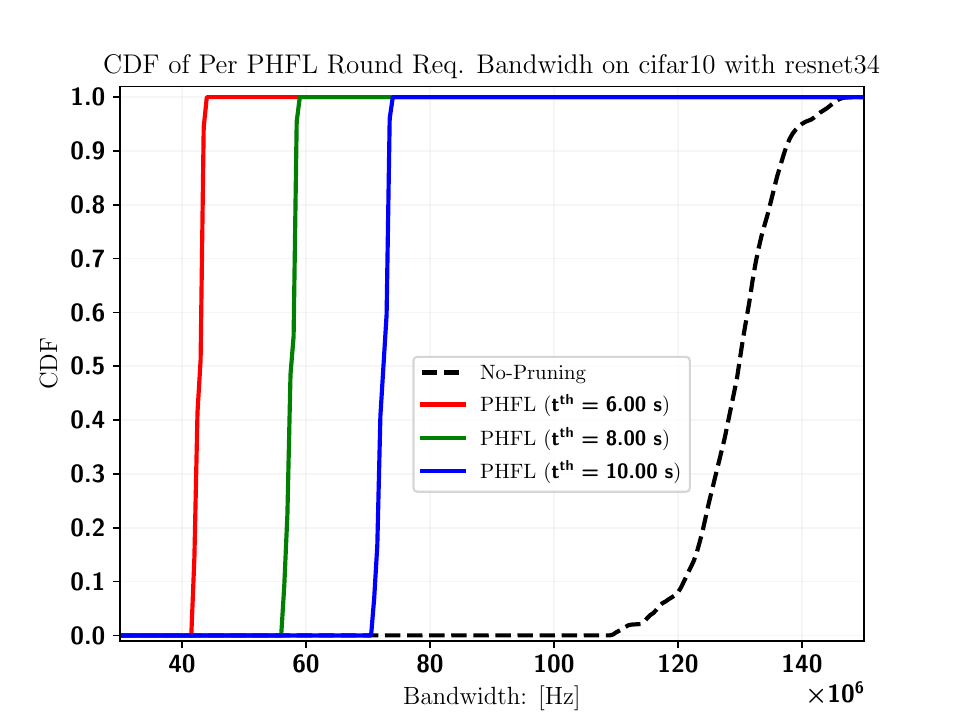}
    \subcaption{Required bandwidth with ResNet-$34$}
    \label{bwVsDeadline5222cifar10Resnet34}
\end{subfigure}
\caption{\bblue{Trade-offs between test accuracies and required bandwidth for different $\mathrm{t^{th}}$'s with different ML models}}
\label{accBWTradeoffsVsDeadline5222cifar10}
\end{figure*}

Now, we observe the impact of $\delta_i^t$'s on the test accuracies and required bandwidth for trained model offloading.
Intuitively, if the model is shallow, pruning further makes it shallower. 
Therefore, the test performance can exacerbate if $\delta_i^t$'s increases for a shallow model.
On the other hand, for a bulky model, pruning may have a less severe effect.
Specifically, under the deadline and energy constraints, pruning may eventually help because pruning a few neurons leads to a shallower but still reasonably well-constructed model that can be trained more efficiently.
Moreover, our convergence bound in (\ref{Theorem1}) clearly shows that the increasing $\delta_i^t$'s decreases the convergence speed. 
However, the wireless payload is directly related to $\delta_i^t$'s as shown in (\ref{uplinkPayload}).
Particularly, the wireless payload is an increasing function of the $\delta_i^t$'s.
As such, increasing the deadline threshold $\mathrm{t^{th}}$ should decrease the $\delta_i^t$'s but significantly increase the wireless payload size.

Our simulation results also reveal \bblue{the above} trends in Fig. \ref{accBWTradeoffsVsDeadline5222cifar10}.
\bblue{Note that, in Fig. \ref{accBWTradeoffsVsDeadline5222cifar10} and the subsequent figures, the (solid/dashed) lines are the average of $4$ independent simulation trials using the configurations mentioned in Section \ref{simu_setting}, while the shaded strips show the corresponding standard deviations.}
Particularly, we observe that the CNN model is largely affected by a small $\mathrm{t^{th}}$ because that leads to a large $\delta_i^t$, which eventually prunes more neurons of the already shallow model.
On the other hand, the bulky ResNets exhibit small performance degradation when $\mathrm{t^{th}}$ decreases.
Moreover, compared to the original non-pruned counterparts, the performance difference is small if $\mathrm{t^{th}}$ is selected appropriately, as shown in Fig. \ref{testAccVsDeadline5222cifar10cnn} to Fig. \ref{testAccVsDeadline5222cifar10Resnet34}.
However, even a slight increase in $d_p$ shall reduce the wireless payload size, which can significantly save a large portion of the bandwidth, as observed in Fig. \ref{bwVsDeadline5222cifar10cnn} to Fig. \ref{bwVsDeadline5222cifar10Resnet34}.
For example, if the CNN model is used, with $\mathrm{t^{th}}=1.3$s, the performance degradation on test accuracy is about \bblue{$1.92\%$} after $100$ PHFL round, while the per PHFL round bandwidth saving for at least $70\%$ of the clients is about \bblue{$20.84\%$}.
Similarly, if $\mathrm{t^{th}}=6$s, the test accuracy degradation is about \bblue{$0.47\%$} and \bblue{$1.11\%$} after $100$ global rounds, while the per PHFL round bandwidth saving for at least $70\%$ of the clients are about \bblue{$33.12\%$} and \bblue{$81.26\%$}, for ResNet-$18$ and ResNet-$34$, respectively.

\subsection{Baseline Comparisons}
\noindent
We now focus on performance comparisons. 
First, we consider the existing HFL \cite{Liu2022jointUE, xu2021adaptive} algorithm that does not consider model pruning or any energy constraints.
Besides, \cite{Liu2022jointUE, xu2021adaptive} only considered two levels, UE-BS and BS-cloud/server.
For a fair comparison, we adapt HFL into three levels $1$) UE-sBS, $2$) sBS-mBS and $3$) mBS-cloud/server.
Furthermore, we enforce the energy and deadline constraints in each UE-sBS aggregation round and name this baseline HFL with constraints (HFL-WC).
Moreover, since \cite{Liu2022jointUE, xu2021adaptive} did not have any VC and we have $\kappa_1$ VC aggregation rounds before the sBS aggregation round, we have adapted the deadline accordingly for HFL-WC to make our comparison fair.
Furthermore, we consider a random PHFL (R-PHFL) scheme, which has the same system model as ours and allows model pruning. 
In R-PHFL, the pruning ratios, $\delta_i^t$'s, are randomly selected between $0$ and $\delta^{\mathrm{th}}$ to satisfy constraint (\ref{cons5}). 
Moreover, in both HFL-WC and R-PHFL, a common $\kappa_0$ that satisfies both deadline and energy constraints for all clients in all VCs leads to poor test accuracy.
As such, we determine the local iterations of the UEs within a VC by selecting the maximum possible number of iterations that all clients within that VC can perform without violating the delay and energy constraints.
For our proposed PHFL, we choose $\kappa_0=5$ and $\kappa_1=2$.
Moreover, we let $\kappa_2=\kappa_3=2$ for both HFL-WC and PHFL. 
We also consider centralized SGD to show the performance gap of PHFL with the ideal case where all training data samples are available centrally.

\begin{figure*}[!t]
\begin{subfigure}{0.33\textwidth}
    \centering
    \includegraphics[trim=10 5 35 25, clip, width=\textwidth, height=0.16\textheight]{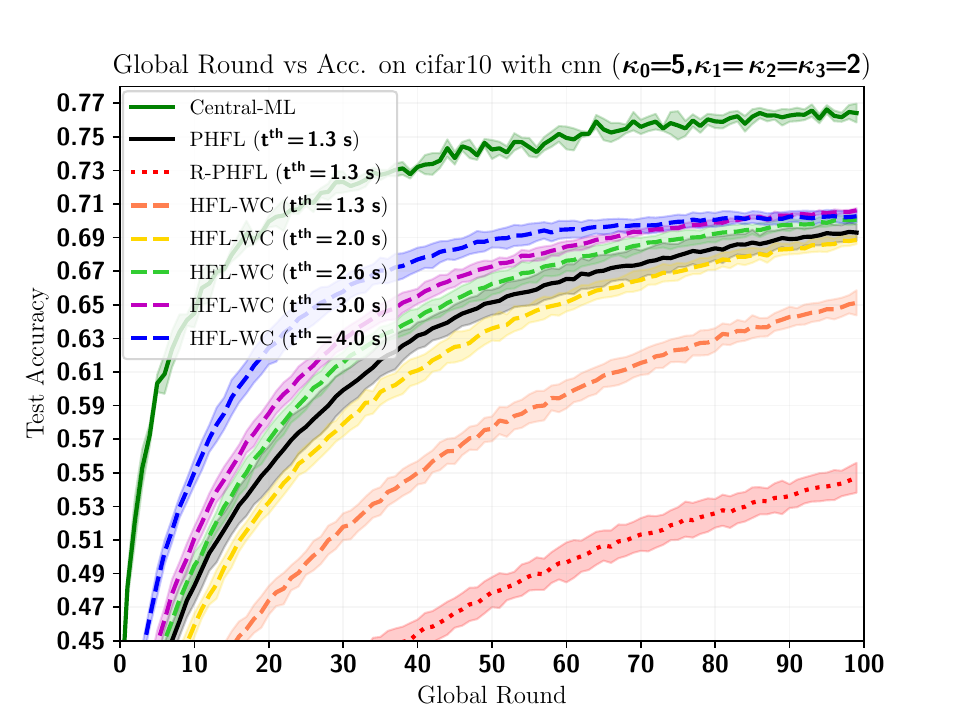}
    \subcaption{Test acc. (CNN on CIFAR-10)}
    \label{baselineAccVsDeadline5222cifar10cnn}
\end{subfigure} 
\begin{subfigure}{0.33\textwidth}
    \centering
    \includegraphics[trim=10 5 25 25, clip, width=\textwidth, height=0.16\textheight]{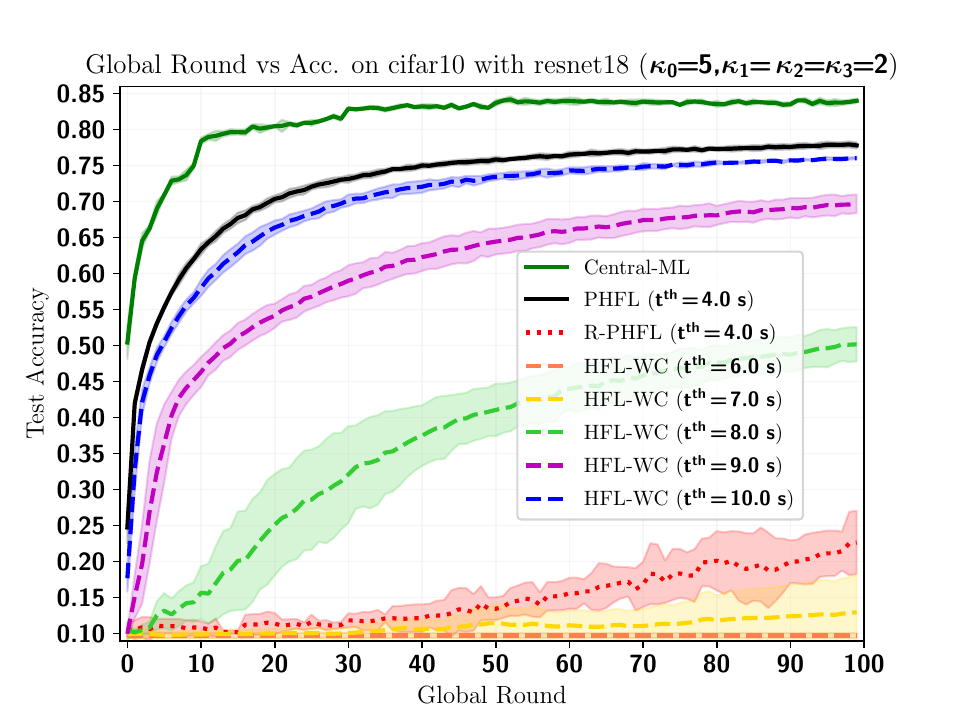}
    \subcaption{Test acc. (ResNet-$18$ on CIFAR-10)}
    \label{baselineAccVsDeadline5222cifar10Resnet18}
\end{subfigure}
\begin{subfigure}{0.325\textwidth}
    \centering
    \includegraphics[trim=10 5 25 25, clip, width=\textwidth, height=0.16\textheight]{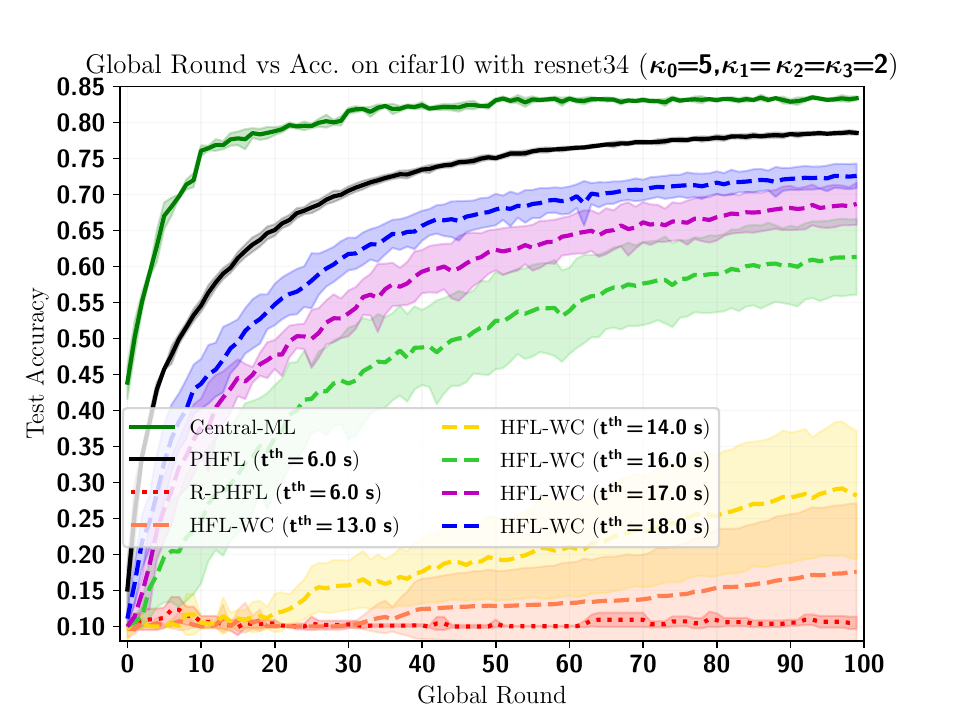}
    \subcaption{Test acc. (ResNet-$34$ on CIFAR-10)}
    \label{baselineAccVsDeadline5222cifar10Resnet34}
\end{subfigure}
\begin{subfigure}{0.33\textwidth}
    \centering
    \includegraphics[trim=10 5 35 25, clip, width=\textwidth, height=0.16\textheight]{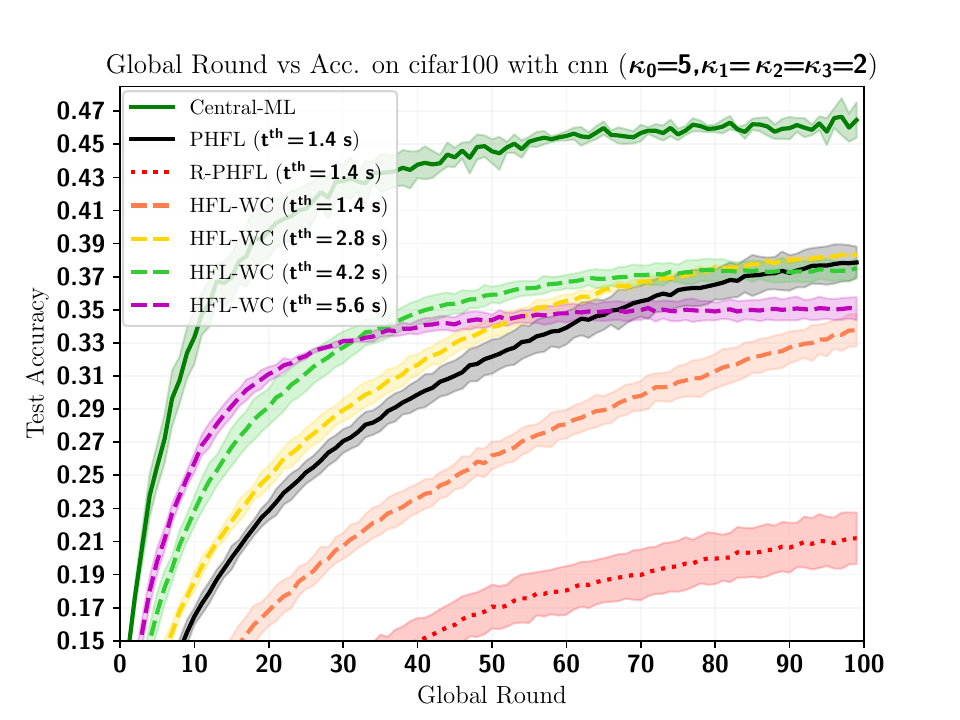}
    \subcaption{Test acc. (CNN on CIFAR-$100$)}
    \label{baselineAccVsDeadline5222cifar100cnn}
\end{subfigure} 
\begin{subfigure}{0.33\textwidth}
    \centering
    \includegraphics[trim=10 5 25 25, clip, width=\textwidth, height=0.16\textheight]{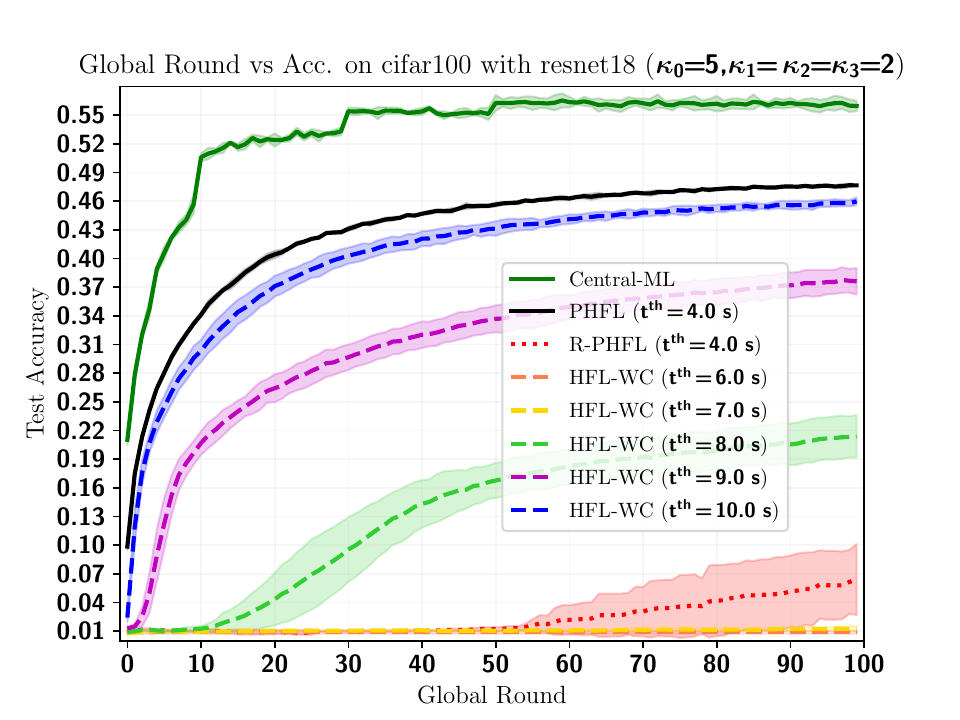}
    \subcaption{Test acc. (ResNet-$18$ on CIFAR-$100$)}
    \label{baselineAccVsDeadline5222cifar100Resnet18}
\end{subfigure}
\begin{subfigure}{0.325\textwidth}
    \centering
    \includegraphics[trim=10 5 25 25, clip, width=\textwidth, height=0.16\textheight]{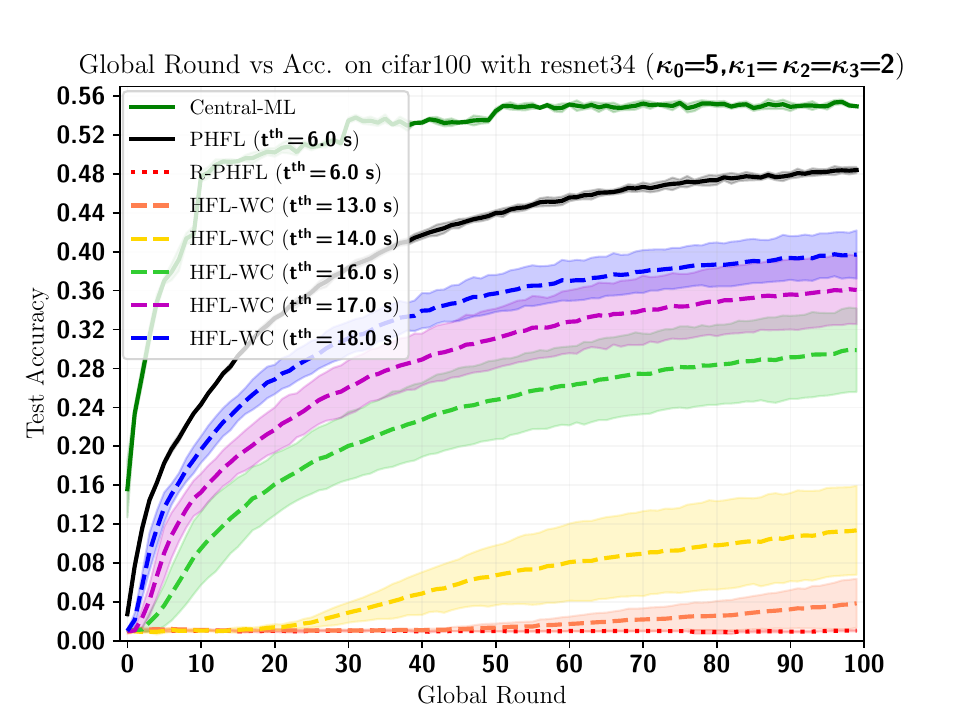}
    \subcaption{Test acc. (ResNet-$34$ on CIFAR-$100$)}
    \label{baselineAccVsDeadline5222cifar100Resnet34}
\end{subfigure}
\caption{\bblue{Test accuracies for different $\mathrm{t^{th}}$'s with different ML models on different datasets}}
\label{baselineAccVsDeadline10222cifar10}
\end{figure*}

From our above discussion, it is expected that pruning will likely not have the edge over the HFL-WC with a shallow model. 
Moreover, one may not need pruning for a shallow model in the first place. 
However, for a bulky model, due to a large number of training parameters and a huge wireless payload, pruning can be a necessity under extreme resource constraints.
Furthermore, it is crucial to jointly optimize $\delta_i^t$'s, $\mathrm{f}_i^t$'s and $\mathrm{P}_i^t$'s in order to increase the test accuracy.
The simulation results in Fig. \ref{baselineAccVsDeadline10222cifar10} also validate these claims.
We observe that when the $\mathrm{t^{th}}$ increases, HFL-WC's performance improves with the CNN model.
Moreover, when HFL-WC has $\kappa_1$ times the deadline of PHFL, the performance is comparable, as shown in Fig. \ref{baselineAccVsDeadline5222cifar10cnn} and Fig. \ref{baselineAccVsDeadline5222cifar100cnn}.
More specifically, the maximum performance degradation of the proposed PHFL algorithm is about \bblue{$1.88\%$} on CIFAR-$10$ when HFL-WC has $2.31$ times the deadline of PHFL.
However, for the ResNets model, HFL-WC requires a significantly longer deadline threshold to make the problem feasible. 
Particularly, with ResNet-$18$, $\mathrm{t^{th}} \leq 6$s does not allow the UEs to perform even a single local iteration, leading to the same initial model weights and, thus, the same test accuracy in HFL-WC.
Moreover, when the deadline threshold is $\kappa_1$ times the $\mathrm{t^{th}}$ of the PHFL, HFL-WC's test accuracy significantly lags, as shown in Fig. \ref{baselineAccVsDeadline5222cifar10Resnet18} and Fig. \ref{baselineAccVsDeadline5222cifar100Resnet18}.
Particularly, after $T=100$ rounds, our proposed solution with $\mathrm{t^{th}}=4$s provides about \bblue{$55.06\%$, $11.62\%$ and $2.33\%$}, and about \bblue{$122.8 \%$, $26.7 \%$ and $3.62\%$} better test accuracy on CIFAR-$10$ and on CIFAR-$100$, respectively, than the HFL-WC with $\mathrm{t^{th}}=8$s, $\mathrm{t^{th}}=9$s and $\mathrm{t^{th}}=10$s.
For the ResNet-$34$ model, the clients require \bblue{$\mathrm{t^{th}} \geq 13$s} for performing some local training in HFL-WC, whereas our proposed solution can achieve significantly better performance with only $\mathrm{t^{th}}=6$s. 
For example, our proposed solution with $\mathrm{t^{th}}=6$s yields about \bblue{$28.09\%$, $14.34\%$ and $8.22\%$}, and about \bblue{$61.97 \%$, $34.19 \%$ and $21.86 \%$} better test accuracy  than the HFL-WC with $\mathrm{t^{th}}=16$s, $\mathrm{t^{th}}=17$s and $\mathrm{t^{th}}=18$s on CIFAR-$10$ and CIFAR-$100$, respectively.
Moreover, our proposed solution provides about \bblue{$243.47\%$} and \bblue{$643.29\%$} and about \bblue{$648.36\%$} and \bblue{$4542.45 \%$} better test accuracy than R-PHFL on CIFAR-$10$ and CIFAR-$100$ with the ResNet-$18$ and ResNet-$34$ models, respectively.
The gap with the ideal centralized ML is also expected since FL suffers from data and system heterogeneity.

\begin{figure}[t!]
    \centering
    \includegraphics[trim=10 5 25 25, clip, width=0.4\textwidth, height=0.19\textheight]{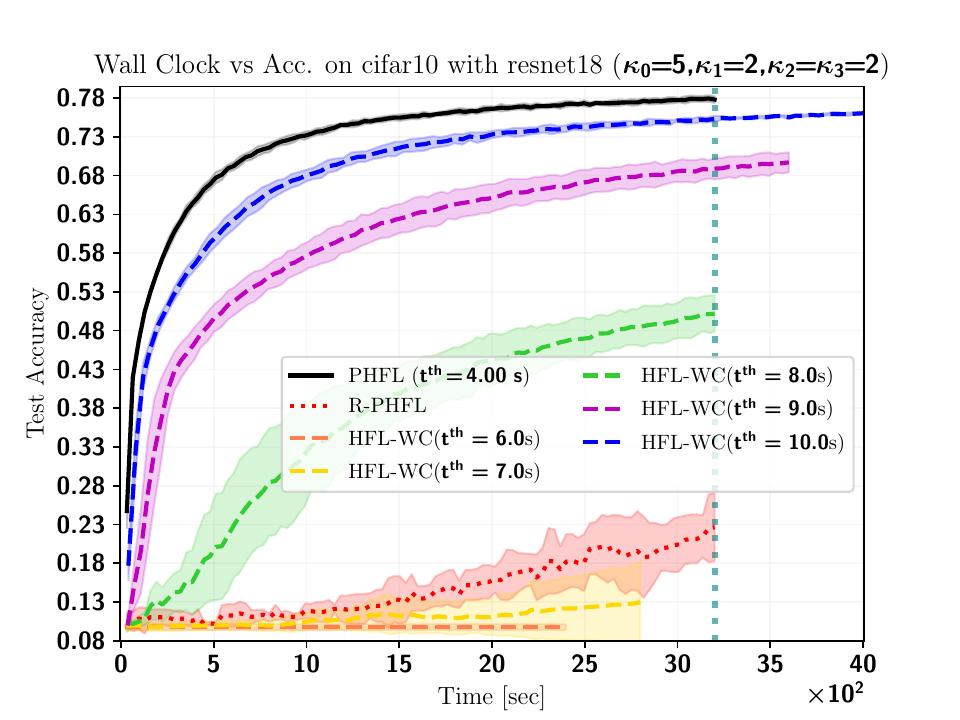}
    \caption{\bblue{Wall clock vs test accuracy on CIFAR-$10$ with ResNet-$18$}}
    \label{wallClockVsAccResnet18_5222}
\end{figure}

We also examine the performance of the proposed method in terms of wall clock time. 
Since HFL-WC does not have the VC tier, the wall clock time to run $M$ global rounds for HFL-WC with a deadline smaller than  $\kappa_1 \times \mathrm{t^{th}}$ will be lower than the proposed PHFL's wall clock time to run the same number of global rounds.
However, that does not necessarily guarantee a higher test accuracy than our PHFL since training and offloading the original bulky model may take a long time, allowing only a few local SGD rounds of the clients.
Besides, any deadline greater than $\kappa_1 \times \mathrm{t^{th}}$ for the HFL-WC will require a longer wall clock time than our proposed PHFL solution.
Our simulation results in Fig. \ref{wallClockVsAccResnet18_5222} clearly shows these trends.
We observe that when HFL-WC has a deadline $\kappa_1 \times \mathrm{t^{th}}$, i.e., $2 \times 4$s = $8$s, the test accuracies are about \bblue{$42.52\%$} and \bblue{$76.24\%$}, respectively, for HFL-WC and PHFL when the wall clock reaches $1800$ seconds. 
Even with $\mathrm{t^{th}} > 4\kappa_1$ seconds, the HFL-WC algorithm performs worse than our proposed PHFL solution.

From the above results and discussion, it is quite clear that R-PHFL yields poor test accuracy due to the random selection of $\delta_i^t$'s.
Besides, HFL-WC cannot deliver reasonable performance when the model has a large number of training parameters. 
Furthermore, our proposed PHFL's performance is comparable to the non-pruned HFL-WC performance with the shallow model. 
As such, in the following, we only consider the upper bound (UB) of the HFL baseline, which does not consider the constraints.
Moreover, to show how pruning degrades test accuracy, we also consider the UB, called HFL-VC-UB, by inheriting the same system model described in Section \ref{sec_sysModel}, but without the model pruning and the constraints.
In other words, we inherit the same underlying four levels, UE-VC, VC-sBS, sBS-mBS and mBS-cloud/server aggregation policy with the original non-pruned model.

\begin{figure*}[!t]
\begin{subfigure}{0.33\textwidth}
    \centering
    \includegraphics[trim=10 5 35 25, clip, width=\textwidth, height=0.16\textheight]{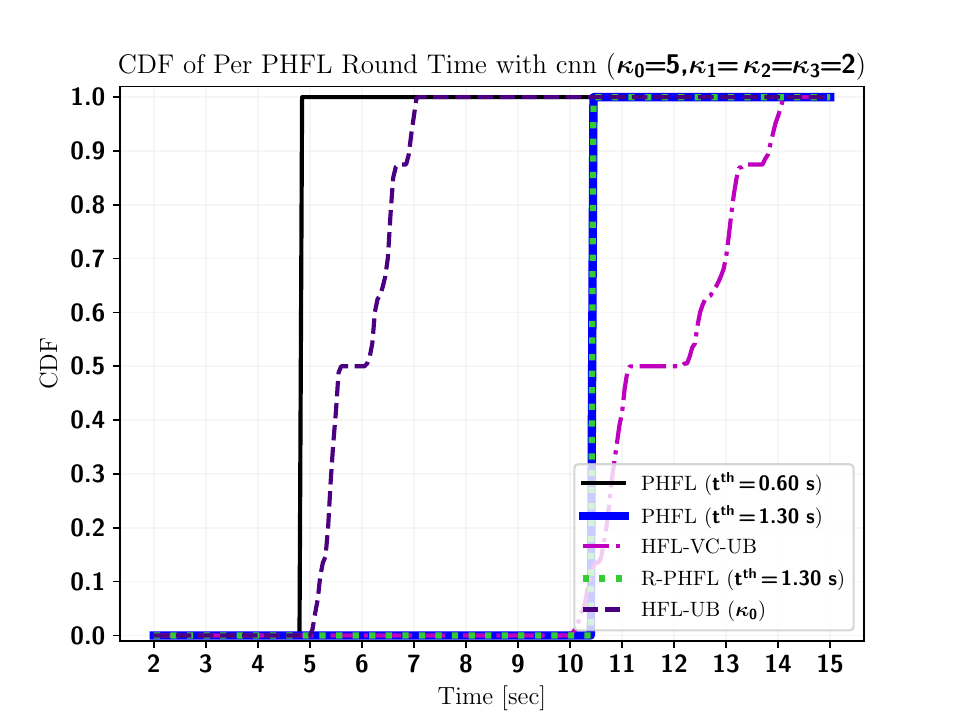}
    \subcaption{Per round time with CNN}
    \label{cdfPHFLRoundTime_CNN}
\end{subfigure} 
\begin{subfigure}{0.33\textwidth}
    \centering
    \includegraphics[trim=10 5 35 25, clip, width=\textwidth, height=0.16\textheight]{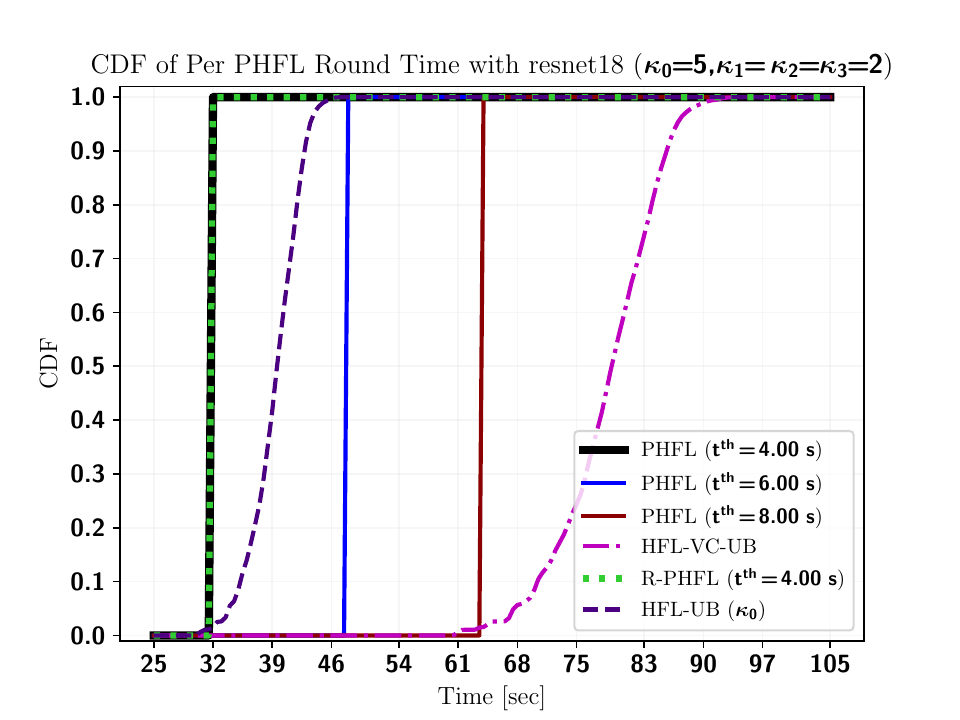}
    \subcaption{Per round time with ResNet-$18$}
    \label{cdfPHFLRoundTime_Resnet18}
\end{subfigure} 
\begin{subfigure}{0.325\textwidth}
    \centering
    \includegraphics[trim=10 5 25 25, clip, width=\textwidth, height=0.16\textheight]{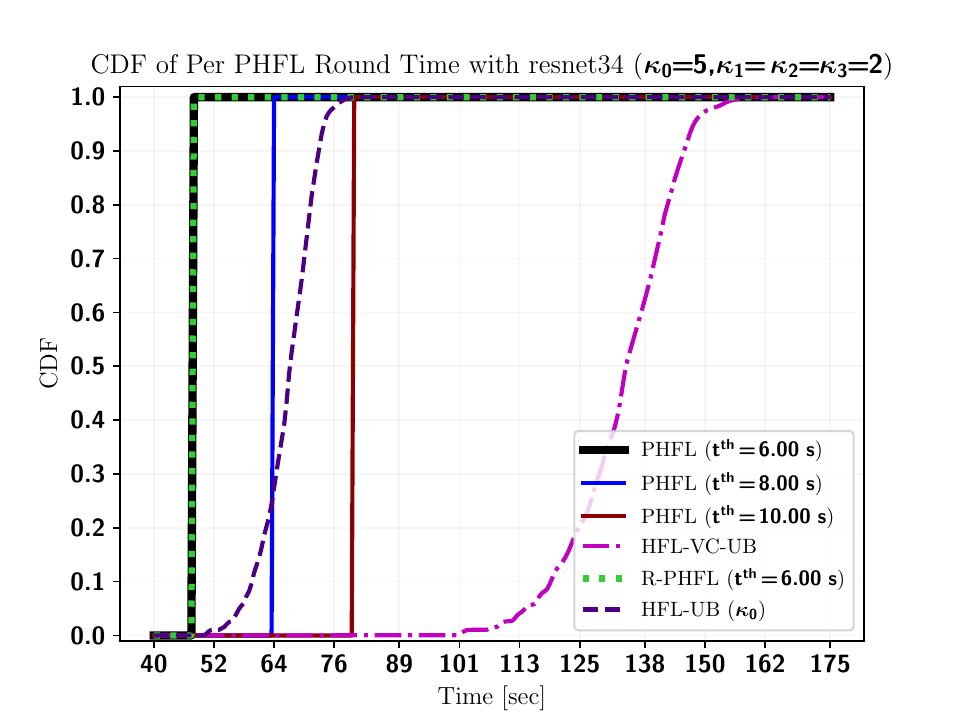}
    \subcaption{Per round time with ResNet-$34$}
    \label{cdfPHFLRoundTime_Resnet34}
\end{subfigure}
\begin{subfigure}{0.33\textwidth}
    \centering
    \includegraphics[trim=10 5 25 25, clip, width=\textwidth, height=0.16\textheight]{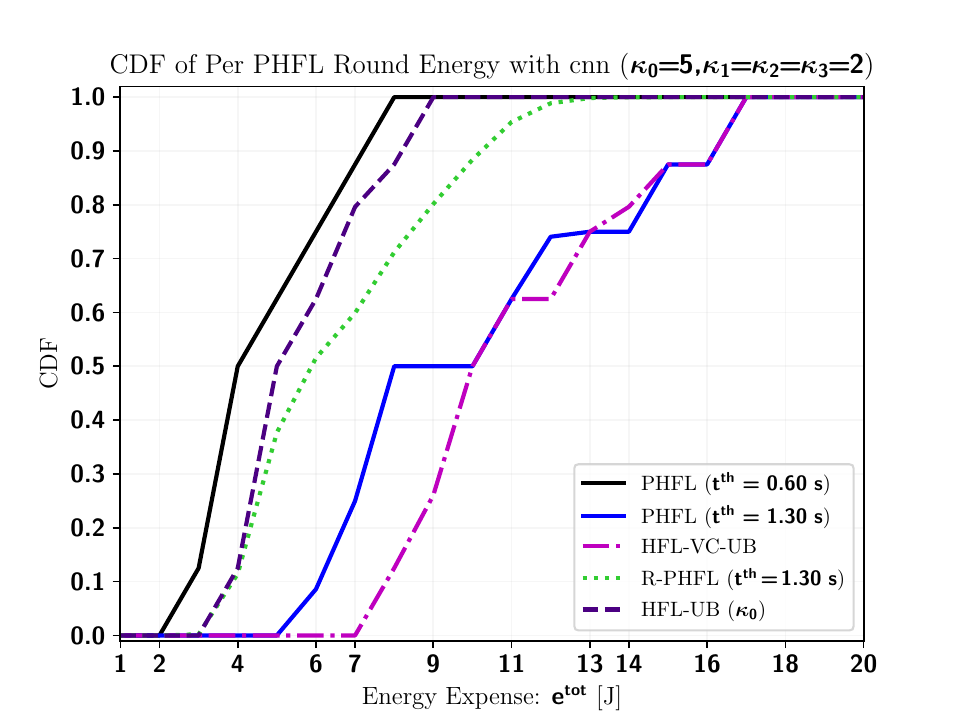}
    \subcaption{Per round $\mathrm{e^{tot}}$ with CNN}
    \label{cdfPHFLRoundEnergy_CNN}
\end{subfigure} 
\begin{subfigure}{0.33\textwidth}
    \centering
    \includegraphics[trim=10 5 25 25, clip, width=\textwidth, height=0.16\textheight]{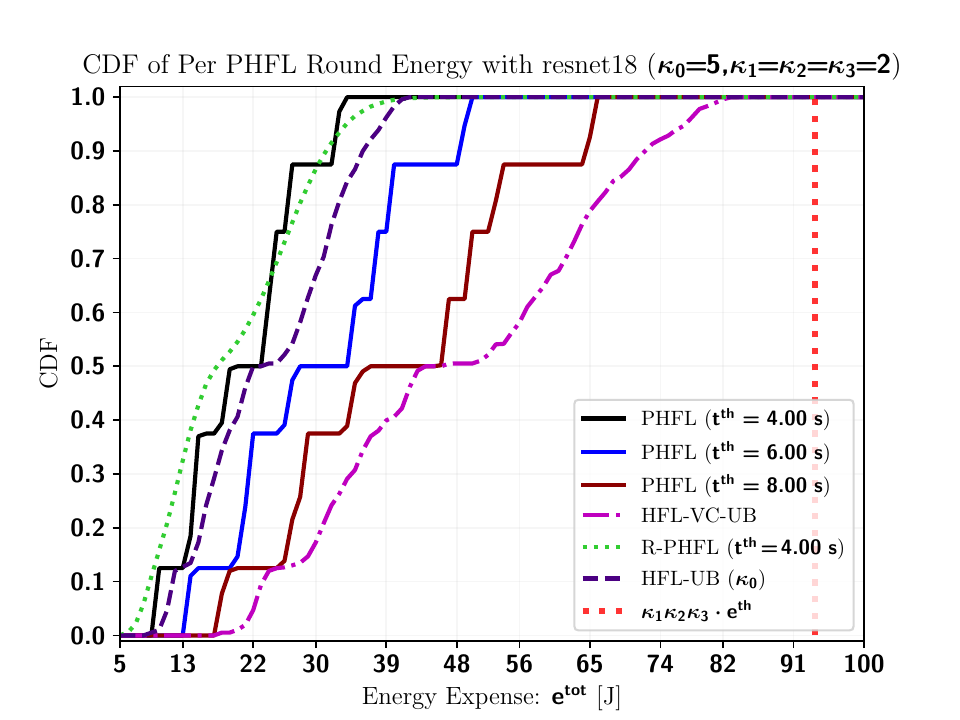}
    \subcaption{Per round $\mathrm{e^{tot}}$ with ResNet-$18$}
    \label{cdfPHFLRoundEnergy_Resnet18}
\end{subfigure} 
\begin{subfigure}{0.325\textwidth}
    \centering
    \includegraphics[trim=10 5 25 25, clip, width=\textwidth, height=0.16\textheight]{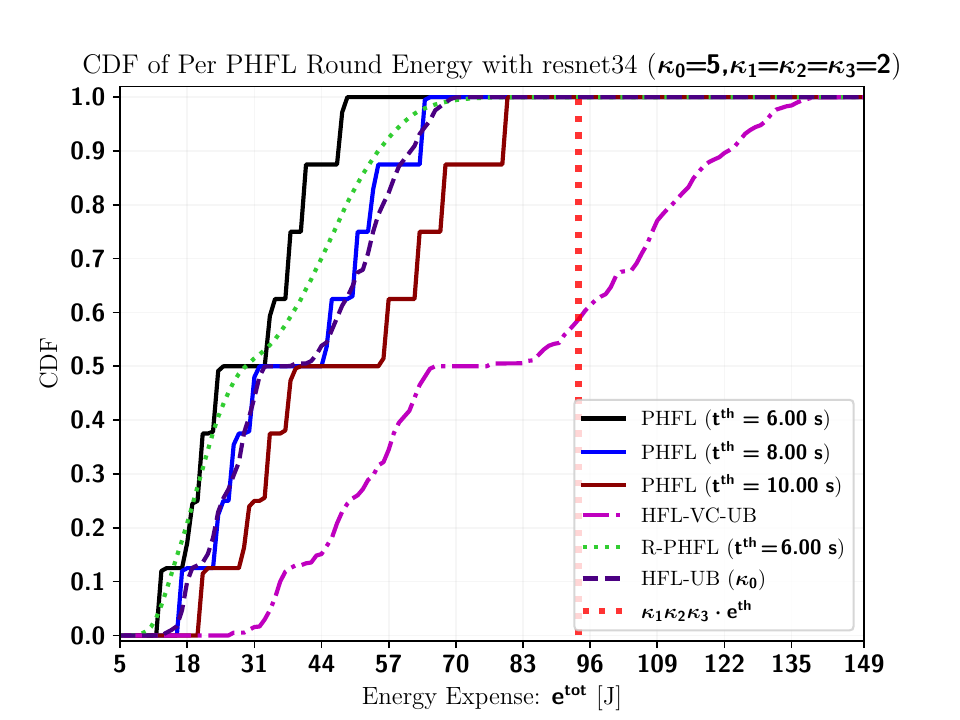} 
    \subcaption{Per round $\mathrm{e^{tot}}$ with ResNet-$34$}
    \label{cdfPHFLRoundEnergy_Resnet34}
\end{subfigure}
\caption{\bblue{CDF of per PHFL round time and energy consumption with different ML models on CIFAR-$10$}}
\end{figure*}
\begin{table*}[!t] 
\centering
\caption{Test Accuracy with Trained $\mathbf{w}^{T}$ on CIFAR-$10$ dataset with $\kappa_0=5, \kappa_1=\kappa_2=\kappa_3=2$ and $T=100$}
\fontsize{6.3}{8}\selectfont

\begin{tabular}{|C {0.88cm} | C{0.4cm} | C {1.44cm} | C{1cm} | C{1.4cm} | C{1.45cm} | C{1cm} | C{1.45cm} | C{1.5cm} | C{1.25cm} | C{1.5cm} |}
\hline 
\multirow{2}{*}{\textbf{Methods}} & \multirow{2}{*}{\rs\rs \textbf{Dir($\bar{\alpha}$)}} & \multicolumn{3}{C{2.65cm}|} {\textbf{With CNN Model}} & \multicolumn{3}{C{2.65cm}|} {\textbf{With ResNet-$18$ Model}} & \multicolumn{3}{C{2.65cm}|} {\textbf{With ResNet-$34$ Model}} \\  \cline{3-11} 
&  & Acc & Req T [s] & Req E [J] & Acc & Req T [s] & Req E [J] & Acc & Req T [s] & Req E [J] \\ \hline
    & $0.5$ & $0.6791 \pm 0.0049$ & \multirow{3}{*}{$1040$} & $48355 \pm 14664$ & $0.7613 \pm 0.0026$ & \multirow{3}{*}{$3200$} & $98122 \pm 15878$ & $0.7677 \pm 0.0017$ & \multirow{3}{*}{$4800$} & $140469 \pm 20993$ \\ \cline{2-3} \cline{5-6} \cline{8-9} \cline{11-11}  
    PHFL & $0.9$ & $0.6930 \pm 0.0049$ & ~ & $48355 \pm 14663$ & $0.7780 \pm 0.0043$ & ~ & $98122 \pm 15878$ & $0.7854 \pm 0.0026$ & ~  & $140469 \pm 20993$ \\ \cline{2-3} \cline{5-6} \cline{8-9} \cline{11-11} 
    (Ours) & $10$ & $0.7091 \pm 0.0022$ & ~ & $48355 \pm 14663$ & $0.7899 \pm 0.0020$ & ~ & $98122 \pm 15878$ & $0.7994 \pm 0.0010$ & ~ & $140469 \pm 20994$ \\ \cline{1-11} 
    HFL-VC & $0.5$ & $0.6971 \pm 0.0037$ & $1293 \pm 110$ & $53041 \pm 12377$ & $0.7689 \pm 0.0054$ & $8985 \pm 182$ & $231634 \pm 32761$ & $0.7789 \pm 0.0030$ & $15370 \pm 269$ & $378436 \pm 51167$ \\ \cline{2-11}
    (UB) & $0.9$ & $0.7066 \pm 0.0017$ & $1293 \pm 110$ & $53041 \pm 12377$ & $0.7807 \pm 0.0019$ & $8985 \pm 182$ & $231634 \pm 32761$ & $0.7942 \pm 0.0033$ & $15370 \pm 269$ & $378436 \pm 51167$ \\ \cline{2-11}
    (Ours) & $10$ & $0.7211 \pm 0.0010$ & $1293 \pm 110$ & $53041 \pm 12377$ & $0.7920 \pm 0.0037$ & $8985 \pm 182$ & $231634 \pm 32761$ & $0.8031 \pm 0.0027$ & $15370 \pm 269$ & $378436 \pm 51167$ \\ \cline{1-11}
    & $0.5$ & $0.5346 \pm 0.0074$ & \multirow{3}{*}{$1040$} & $31483 \pm 9454$ & $0.2447 \pm 0.0161$ & \multirow{3}{*}{$3200$} & $92529 \pm 12400$ & $0.1112 \pm 0.0193$ & \multirow{3}{*}{$4800$} & $158622 \pm 20589$ \\ \cline{2-3} \cline{5-6} \cline{8-9} \cline{11-11}
    R-PHFL & $0.9$ & $0.5471 \pm 0.0089$ & $~$ & $31467 \pm 9427$ & $0.2265 \pm 0.0439$ & $~$ & $92495 \pm 12356$ & $0.1049 \pm 0.0086$ & $~$ & $158827 \pm 20834$ \\ \cline{2-3} \cline{5-6} \cline{8-9} \cline{11-11}
    & $10$ & $0.5847 \pm 0.0164$ & $~$ & $31455 \pm 9406$ & $0.3265 \pm 0.0273$ & $~$ & $92441 \pm 12287$ & $0.1096 \pm 0.0108$ & $~$ & $158651 \pm 20624$  \\ \cline{1-11}
    & $0.5$ & $0.6624 \pm 0.0021$ & $646 \pm 55$ & $26520 \pm 6188$ & $0.7539 \pm 0.0044$ & $4493 \pm 90$ & $115810 \pm 16379$  & $0.7445 \pm 0.0049$ & $7686 \pm 132$ & $189207 \pm 25580$ \\ \cline{2-11}
    HFL & $0.9$ & $0.6752 \pm 0.0018$ & $646 \pm 55$ & $26520 \pm 6188$ & $0.7695 \pm 0.0014$ & $4493 \pm 90$ & $115810 \pm 16379$ & $0.7664 \pm 0.0045$ & $7686 \pm 132$ & $189207 \pm 25580$  \\ \cline{2-11}
    (UB) - $\kappa_0$ & $10$ & $0.6934 \pm 0.0012$ & $646 \pm 55$ & $26520 \pm 6188$ & $0.7833 \pm 0.0032$ & $4493 \pm 90$ & $115810 \pm 16379$ & $0.7844 \pm 0.0013$ & $7686 \pm 132$ & $189207 \pm 25580$ \\ \cline{1-11} \hline
\end{tabular}
\label{performanceComparison_5_2_2_2_cifar10}
\end{table*} 
\begin{table*}[!t]
\caption{Test Accuracy with Trained $\mathbf{w}^{T}$ on CIFAR-$10$ dataset with $\kappa_0=10, \kappa_1=\kappa_2=\kappa_3=2$ and $T=100$}
\centering
\fontsize{6.3}{8}\selectfont

\begin{tabular}{|C {0.85cm} | C{0.4cm} | C {1.44cm} | C{0.99cm} | C{1.43cm} | C{1.44cm} | C{1.1cm} | C{1.45cm} | C{1.5cm} | C{1.1cm} | C{1.45cm} |}
\hline 
\multirow{2}{*}{\textbf{Methods}} & \multirow{2}{*}{\rs\rs \textbf{Dir($\bar{\alpha}$)}} & \multicolumn{3}{C{2.65cm}|} {\textbf{With CNN Model}} & \multicolumn{3}{C{2.65cm}|} {\textbf{With ResNet-$18$ Model}} & \multicolumn{3}{C{2.65cm}|} {\textbf{With ResNet-$34$ Model}} \\  \cline{3-11} 
&  & Acc & Req T [s] & Req E [J] & Acc & Req T [s] & Req E [J] & Acc & Req T [s] & Req E [J] \\ \hline
    & $0.5$ & $0.6859 \pm 0.0037$ & \multirow{3}{*}{$1600$} & $76084 \pm 23665$ & $0.7668 \pm 0.0016$ & \multirow{3}{*}{$3200$} & $103997 \pm 18692$ & $0.7774 \pm 0.0016$ & \multirow{3}{*}{$4800$} & $146238 \pm 23406$ \\ \cline{2-3} \cline{5-6} \cline{8-9} \cline{11-11} 
    PHFL & $0.9$ & $0.6966 \pm 0.0034$ & ~ &$76084 \pm 23665$ & $0.7804 \pm 0.0014$ & ~ & $103998 \pm 18693$ & $0.7948 \pm 0.0023$ & ~ & $146239 \pm 23408$ \\ \cline{2-3} \cline{5-6} \cline{8-9} \cline{11-11}
    (Ours) & $10$ & $0.7093 \pm 0.0024$ & ~ & $76083 \pm 23664$ & $0.7935 \pm 0.0015$ & ~ & $103999 \pm 18695$ & $0.8092 \pm 0.0019$ & ~ & $146238 \pm 23407$ \\ \cline{1-11} 
    HFL-VC & $0.5$ & $0.6932 \pm 0.0021$ & $2417 \pm 221$ & $102117 \pm 24405$ & $0.7704 \pm 0.0042$ & $10052 \pm 267$ & $280710 \pm 43428$ & $0.7825 \pm 0.0033$ & $16416 \pm 346$ & $427512 \pm 61116$ \\ \cline{2-11}
    (UB) & $0.9$ & $0.7070 \pm 0.0024$ & $2417 \pm 221$ & $102117 \pm 24405$ & $0.7866 \pm 0.0016$ & $10052 \pm 267$ & $280710 \pm 43428$ & $0.7975 \pm 0.0025$ & $16416 \pm 346$ & $427512 \pm 61116$ \\ \cline{2-11}
    (Ours) & $10$ & $0.7164 \pm 0.0042$ & $2417 \pm 221$ & $102117 \pm 24405$ & $0.7936 \pm 0.0009$ & $10052 \pm 267$ & $280710 \pm 43428$ & $0.8102 \pm 0.0018$ & $16416 \pm 346$ & $427512 \pm 61116$ \\ \cline{1-11}
    & $0.5$ & $0.6934 \pm 0.0030$ & $1208 \pm 110$ & $51058 \pm 12202$ & $ 0.7634 \pm 0.0017$ & $5026 \pm 132$ & $140349 \pm 21712$ & $0.7726 \pm 0.0014$ & $8209 \pm 171$ & $213745 \pm 30554$ \\ \cline{2-11}
    HFL & $0.9$ & $0.7033 \pm 0.0031$ & $1208 \pm 110$ & $51058 \pm 12202$ & $ 0.7769 \pm 0.0060$ & $5026 \pm 132$ & $140349 \pm 21712$ & $0.7847 \pm 0.0014$ & $8209 \pm 171$ & $213745 \pm 30554$ \\ \cline{2-11}
    (UB)- $\kappa_0$ & $10$ & $0.7208 \pm 0.0018$ & $1208 \pm 110$ & $51058 \pm 12202$ & $0.7895 \pm 0.0023$ & $5026 \pm 132$ & $140349 \pm 21712$ & $0.8004 \pm 0.0016$ & $8209 \pm 171$ & $213745 \pm 30554$ \\ \cline{1-11} \hline
\end{tabular}

\label{performanceComparison_10_2_2_2_cifar10}
\end{table*} 
\begin{table*}
\caption{Test Accuracy with Trained $\mathbf{w}^{T}$ on CIFAR-$10$ dataset with $\kappa_0=5, \kappa_1=4, \kappa_2=\kappa_3=2$ and $T=100$}
\fontsize{6.3}{8}\selectfont
\centering

\begin{tabular}{|C {0.8cm} | C{0.4cm} | C {1.44cm} | C{1cm} | C{1.45cm} | C{1.44cm} | C{1.1cm} | C{1.43cm} | C{1.44cm} | C{1.1cm} | C{1.6cm} |}
\hline 
\multirow{2}{*}{\textbf{Methods}} & \multirow{2}{*}{\rs\rs \textbf{Dir($\bar{\alpha}$)}} & \multicolumn{3}{C{2.65cm}|} {\textbf{With CNN Model}} & \multicolumn{3}{C{2.65cm}|} {\textbf{With ResNet-$18$ Model}} & \multicolumn{3}{C{2.65cm}|} {\textbf{With ResNet-$34$ Model}} \\  \cline{3-11} 
&  & Acc & Req T [s] & Req E [J] & Acc & Req T [s] & Req E [J] & Acc & Req T [s] & Req E [J] \\ \hline
 & $0.5$ & $0.6950 \pm 0.0047$ & \multirow{3}{*}{$2080$} & $96710 \pm 29327$ & $0.7699 \pm 0.0016$ & \multirow{3}{*}{$6400$} & $196244 \pm 31755$ & $0.7826 \pm 0.0023$ & \multirow{3}{*}{$9600$} & $280939 \pm 41988$ \\ \cline{2-3} \cline{5-6} \cline{8-9} \cline{11-11} 
 PHFL & $0.9$ & $0.7060 \pm 0.0026$ & ~ & $96710 \pm 29327$ & $0.7828 \pm 0.0018$ & ~ & $196245 \pm 31757$ & $0.8010 \pm 0.0033$ & ~ & $280938 \pm 41986$ \\ \cline{2-3} \cline{5-6} \cline{8-9} \cline{11-11} 
 (Ours) & $10$ & $0.7136 \pm 0.0043$ & ~ & $96710 \pm 29327$ & $0.7945 \pm 0.0014$ & ~ & $196244 \pm 31756$ & $0.8111 \pm 0.0014$ & ~ & $280938 \pm 41987$  \\ \cline{1-11} 
 HFL-VC & $0.5$ & $0.7052 \pm 0.0038$ & $2587 \pm 220$ & $106082 \pm 24755$ & $0.7726 \pm 0.0057$ & $17969 \pm 364$ & $463257 \pm 65516$ & $0.7881 \pm 0.0028$ & $30737 \pm 534$ & $756853 \pm 102321$  \\ \cline{2-11}
 (UB) & $0.9$ & $0.7122 \pm 0.0016$ & $2587 \pm 220$ & $106082 \pm 24755$ & $0.7874 \pm 0.0015$ & $17969 \pm 364$ & $463257 \pm 65516$ & $0.8012 \pm 0.0009$ & $30737 \pm 534$ & $756853 \pm 102321$ \\ \cline{2-11}
 (Ours) & $10$ & $0.7195 \pm 0.0015$ & $2587 \pm 220$ & $106082 \pm 24755$ & $0.7961 \pm 0.0028$ & $17969 \pm 364$ & $463257 \pm 65516$ & $0.8117 \pm 0.0031$ & $30737 \pm 534$ & $756853 \pm 102321$ \\ \cline{1-11}
 R-PHFL & $0.5$ & $ 0.5813 \pm 0.0146$ & \multirow{3}{*}{$2080$} & $62927 \pm 18914$ & $0.3863 \pm 0.0408$ & \multirow{3}{*}{$6400$} & $184493 \pm 24419$ & $0.1042 \pm 0.0069$ & \multirow{3}{*}{$9600$} & $316985 \pm 41406$ \\ \cline{2-3} \cline{5-6} \cline{8-9} \cline{11-11} 
 & $0.9$ & $0.6013 \pm 0.0164$ & ~ & $62933 \pm 18924$ & $0.4010 \pm 0.0208$ & ~ & $184766 \pm 24766$ & $0.1157 \pm 0.0137$ & ~ & $316870 \pm 41267$ \\ \cline{2-3} \cline{5-6} \cline{8-9} \cline{11-11}
 & $10$ & $0.6360 \pm 0.0082$ & ~ & $62849 \pm 18781$ & $0.4648 \pm 0.0352$ & ~ & $184617 \pm 24577$ & $0.1111 \pm 0.0078$ & ~ & $317122 \pm 41570$ \\ \cline{1-11} \hline
\end{tabular}

\label{performanceComparison_5_4_2_2_cifar10}
\end{table*} 

First, we illustrate the required computation time and the corresponding energy expenses for the baselines and our proposed PHFL solution with different deadline thresholds.
Naturally, if we increase the $\mathrm{t^{th}}$ for each VC aggregation round, our proposed solution will take longer time and energy to perform $T=100$ global rounds.
Moreover, since there are $\kappa_1$ VC aggregation rounds in our proposed system model, the original model training and offloading with HFL-VC-UB is expected to take significantly longer and consume more energy than the HFL-UB baseline.
Therefore, the effectiveness, in terms of time and energy consumption, of PHFL is largely dependent upon the deadline threshold $\mathrm{t^{th}}$.
While R-PHFL requires the same time as our proposed PHFL, the energy requirements for R-PHFL can vary significantly due to the random selection of the $\delta_i^t$'s that lead to different model sizes and different local training episodes in different VCs.

Our results in Figs. \ref{cdfPHFLRoundTime_CNN}-\ref{cdfPHFLRoundTime_Resnet34} and Figs. \ref{cdfPHFLRoundEnergy_CNN}-\ref{cdfPHFLRoundEnergy_Resnet34} clearly show these trade-offs with time and energy consumption, respectively, for the three models.
It is worth pointing out that we adopted the popular lottery ticket hypothesis \cite{franklelottery} for finding the winning ticket that requires performing $\rho$ iterations on the initial model with full parameter space as described in Section \ref{subSecPHFL}.
This incurs additional $\mathrm{t}_i^{\mathrm{cp_d}}$ time and $\mathrm{e}_i^{\mathrm{cp_d}}$ energy overheads, as calculated in (\ref{localCompLottery}) and (\ref{energyConLot}), respectively.
As such, when the $\mathrm{t^{th}}$ is large, if the UEs have sufficient energy budgets, they can prune only a few neurons.
\bblue{Therefore}, the total time and energy consumption for the original model computation for getting the winning ticket, training the pruned model and offloading the trained pruned model parameters can become slightly larger than the HFL-VC-UB.
This is observed in Fig. \ref{cdfPHFLRoundTime_CNN} and Fig. \ref{cdfPHFLRoundEnergy_CNN} for CNN model when $\mathrm{t^{th}}=1.3$s, and also in Fig. \ref{cdfPHFLRoundTime_Resnet18} \bblue{and Fig. \ref{cdfPHFLRoundEnergy_Resnet18}} when \bblue{$\mathrm{t^{th}}=8$s} for the ResNet-$18$ model.
Besides, the dashed vertical lines in Fig. \ref{cdfPHFLRoundEnergy_Resnet18} and Fig. \ref{cdfPHFLRoundEnergy_Resnet34} show the mean energy budget of the clients for each PHFL global round.
While all UEs are able to perform the learning and offloading within this mean energy budget with CNN and ResNet-$18$ models, clearly, when the ResNet-$34$ model is used, more than \bblue{$57\%$} of the clients will fail to perform the HFL-VC-UB due to their limited energy budgets.

\begin{table*}
\caption{Test Accuracy with Trained $\mathbf{w}^{T}$ on CIFAR-$100$ dataset with $\kappa_0=5, \kappa_1=\kappa_2=\kappa_3=2$ and $T=100$}
\fontsize{6.3}{8}\selectfont
\centering

\begin{tabular}{|C {0.88cm} | C{0.4cm} | C {1.44cm} | C{1cm} | C{1.4cm} | C{1.45cm} | C{1cm} | C{1.45cm} | C{1.5cm} | C{1.25cm} | C{1.5cm} |}
\hline 
\multirow{2}{*}{\textbf{Methods}} & \multirow{2}{*}{\rs\rs \textbf{Dir($\bar{\alpha}$)}} & \multicolumn{3}{C{2.65cm}|} {\textbf{With CNN Model}} & \multicolumn{3}{C{2.65cm}|} {\textbf{With ResNet-$18$ Model}} & \multicolumn{3}{C{2.65cm}|} {\textbf{With ResNet-$34$ Model}} \\  \cline{3-11} 
&  & Acc & Req T [s] & Req E [J] & Acc & Req T [s] & Req E [J] & Acc & Req T [s] & Req E [J] \\ \hline
 & $0.5$ & $0.3723 \pm 0.0101$ & \multirow{3}{*}{$1120$} & $51667 \pm 15530$ & $0.4725 \pm 0.0032$ & \multirow{3}{*}{$3200$} & $98074 \pm 15856$ & $0.4770 \pm 0.0032$ & \multirow{3}{*}{$4800$} & $140445 \pm 20984$ \\ \cline{2-3} \cline{5-6} \cline{8-9} \cline{11-11} 
 PHFL & $0.9$ & $0.3786 \pm 0.0096$ & ~ & $51667 \pm 15531$ & $0.4765 \pm 0.0003$ & ~ & $98075 \pm 15857$ & $0.4840 \pm 0.0032$ & ~ & $140444 \pm 20983$ \\ \cline{2-3} \cline{5-6} \cline{8-9} \cline{11-11} 
 (Ours) & $10$ & $0.3795 \pm 0.0084$ & ~ & $51667 \pm 15531$ & $0.4811 \pm 0.0024$ & ~ & $98075 \pm 15857$ & $0.4834 \pm 0.0023$ & ~ & $140444 \pm 20984$ \\ \cline{1-11} 
 HFL-VC & $0.5$ & $0.3962 \pm 0.0030$ & $1319 \pm 109$ & $53645 \pm 12431$ & $0.4805 \pm 0.0022$ & $9038 \pm 183$ & $232839 \pm 32910$ & $0.4909 \pm 0.0017$ & $15422 \pm 270$ & $379641 \pm 51319$ \\ \cline{2-11}
 (UB) & $0.9$ & $0.3956 \pm 0.0030$ & $1319 \pm 109$ & $53645 \pm 12431$ & $0.4832 \pm 0.0034$ & $9038 \pm 183$ & $232839 \pm 32910$ & $0.4922 \pm 0.0019$ & $15422 \pm 270$ & $379641 \pm 51319$ \\ \cline{2-11}
 (Ours) & $10$ & $0.4004 \pm 0.0042$ & $1319 \pm 109$ & $53645 \pm 12431$ & $0.4800 \pm 0.0019$ & $9038 \pm 183$ & $232839 \pm 32910$ & $0.4895 \pm 0.0024$ & $15422 \pm 270$ & $379641 \pm 51319$ \\ \cline{1-11}
 R-PHFL & $0.5$ & $0.1994 \pm 0.0075$ & \multirow{3}{*}{$1120$} & $33751 \pm 10069$ & $0.0587 \pm 0.0212$ & \multirow{3}{*}{$3200$} & $92996 \pm 12374$ & $0.0123 \pm 0.0028$ & \multirow{3}{*}{$4800$} & $159083 \pm 20562$ \\ \cline{2-3} \cline{5-6} \cline{8-9} \cline{11-11} 
 & $0.9$ & $0.2119 \pm 0.0156$ & ~ & $33763 \pm 10089$ & $0.0641 \pm 0.0371$ & ~ & $ 92893 \pm 12245$ & $0.0104 \pm 0.0007$ & ~ & $159122 \pm 20608$ \\ \cline{2-3} \cline{5-6} \cline{8-9} \cline{11-11} 
 & $10$ & $0.2165 \pm 0.0057$ & ~ & $33733 \pm 10040$ & $0.0728 \pm 0.0427$ & ~ & $92893 \pm 12245$ & $0.0109 \pm 0.0011$ & ~ & $159175 \pm 20671$ \\ \cline{1-11} \hline
 HFL & $0.5$ & $0.3509 \pm 0.0023$ & $659 \pm 54$ & $26822 \pm 6215$ & $0.4730 \pm 0.0025$ & $4519 \pm 90$ & $116413 \pm 16453$ & $0.4663 \pm 0.0018$ & $7712 \pm 133$ & $189809 \pm 25656$ \\ \cline{2-11}
 (UB) - $\kappa_0$ & $0.9$ & $0.3524 \pm 0.0029$ & $659 \pm 54$ & $26822 \pm 6215$ & $0.4768 \pm 0.0047$ & $4519 \pm 90$ & $116413 \pm 16453$ & $ 0.4756 \pm 0.0021$ & $7712 \pm 133$ & $189809 \pm 25656$ \\ \cline{2-11} 
 & $10$ & $0.3590 \pm 0.0019$ & $659 \pm 54$ & $26822 \pm 6215$ & $0.4841 \pm 0.0057$ & $4519 \pm 90$ & $116413 \pm 16453$ & $0.4776 \pm 0.0035$ & $7712 \pm 133$ & $189809 \pm 25656$ \\ \cline{1-11} \hline
\end{tabular}
\label{performanceComparison_5_2_2_2_cifar100}
\end{table*} 
\begin{table*}
\caption{Test Accuracy with Trained $\mathbf{w}^{T}$ on CIFAR-$100$ dataset with $\kappa_0=10, \kappa_1 = \kappa_2=\kappa_3=2$ and $T=100$}
\fontsize{6.3}{8}\selectfont
\centering

\begin{tabular}{|C {0.7cm} | C{0.4cm} | C {1.44cm} | C{1cm} | C{1.45cm} | C{1.45cm} | C{1.1cm} | C{1.45cm} | C{1.5cm} | C{1.25cm} | C{1.5cm} |}
\hline 
\multirow{2}{*}{\rs\textbf{Methods}} & \multirow{2}{*}{\rs\rs \textbf{Dir($\bar{\alpha}$)}} & \multicolumn{3}{C{2.65cm}|} {\textbf{With CNN Model}} & \multicolumn{3}{C{2.65cm}|} {\textbf{With ResNet-$18$ Model}} & \multicolumn{3}{C{2.65cm}|} {\textbf{With ResNet-$34$ Model}} \\  \cline{3-11} 
&  & Acc & Req T [s] & Req E [J] & Acc & Req T [s] & Req E [J] & Acc & Req T [s] & Req E [J] \\ \hline
 & $0.5$ & $0.3700 \pm 0.0035$ & \multirow{3}{*}{$2000$} & $94632 \pm 29277$ & $0.4654 \pm 0.0035$  & \multirow{3}{*}{$3200$} & $103923 \pm 18656$ & $0.4801 \pm 0.0009$ & \multirow{3}{*}{$4800$} & $146195 \pm 23388$ \\ \cline{2-3} \cline{5-6} \cline{8-9} \cline{11-11} 
 PHFL & $0.9$ & $0.3648 \pm 0.0044$ & ~ & $94632 \pm 29277$ & $0.4699 \pm 0.0026$ & ~ & $103923 \pm 18655$ & $0.4810 \pm 0.0047$ & ~ & $146196 \pm 23388$ \\ \cline{2-3} \cline{5-6} \cline{8-9} \cline{11-11} 
 (Ours) & $10$ & $0.3634 \pm 0.0039$ & ~ & $94632 \pm 29277$ & $0.4763 \pm 0.0019$ & ~ & $103922 \pm 18655$ & $0.4780 \pm 0.0018 $ & ~ & $146195 \pm 23388$ \\ \cline{1-11} 
 \rs\rs HFL-VC & $0.5$ & $0.3756 \pm 0.0041$ & $2443 \pm 220$ & $102721 \pm 24458$ & $0.4720 \pm 0.0020$ & $10104 \pm 268$ & $281915 \pm 43569$ & $0.4877 \pm 0.0058$ & $16469 \pm 346$ & $428717 \pm 61264$ \\ \cline{2-11}
 (UB) & $0.9$ & $0.3754 \pm 0.0053$ & $2443 \pm 220$ & $102721 \pm 24458$ & $ 0.4701 \pm 0.0043$ & $10104 \pm 268$ & $281915 \pm 43569$ & $0.4895 \pm 0.0015$ & $16469 \pm 346$ & $428717 \pm 61264$ \\ \cline{2-11}
 (Ours) & $10$ & $0.3794 \pm 0.0025$ & $2443 \pm 220$ & $102721 \pm 24458$ & $0.4736 \pm 0.0012$ & $ 10104 \pm 268$ & $281915 \pm 43569$ & $0.4736 \pm 0.0012$ & $10104 \pm 268$ & $281915 \pm 43569$ \\ \cline{1-11}
 HFL & $0.5$ & $0.3877 \pm 0.0031$ & $1221 \pm 110$ & $51360 \pm 12228$ & $0.4713 \pm 0.0052 $ & $5052 \pm 132 $ & $140951 \pm 21783$ & $0.4806 \pm 0.0058$ & $8235 \pm 171$ & $214347 \pm 30628$ \\ \cline{2-11}
 (UB) - & $0.9$ & $0.3877 \pm 0.0024$ & $1221 \pm 110$ & $51360 \pm 12228$ & $0.4731 \pm 0.0020$ & $5052 \pm 132$ & $140951 \pm 21783$ & $0.4855 \pm 0.0011$ & $8235 \pm 171 $ & $214347 \pm 30628 $ \\ \cline{2-11}
 $\kappa_0$ & $10$ & $0.3922 \pm 0.0042$ & $1221 \pm 110$ & $51360 \pm 12228$ & $0.4772 \pm 0.0047$ & $5052 \pm 132$ & $140951 \pm 21783$ & $0.4844 \pm 0.0031$ & $8235 \pm 171$ & $214347 \pm 30628$ \\ \cline{1-11} \hline
\end{tabular}
\label{performanceComparison_10_2_2_2_cifar100}
\end{table*} 
\begin{table*}[!t]
\caption{Test Accuracy with Trained $\mathbf{w}^{T}$ on CIFAR-$100$ dataset with $\kappa_0=5, \kappa_1=4, \kappa_2=\kappa_3=2$ and $T=100$}
\fontsize{6.3}{8}\selectfont
\centering
\begin{tabular}{|C {0.78cm} | C{0.4cm} | C {1.44cm} | C{1cm} | C{1.44cm} | C{1.44cm} | C{1.1cm} | C{1.45cm} | C{1.5cm} | C{1.1cm} | C{1.55cm} |}
\hline 
\multirow{2}{*}{\textbf{Methods}} & \multirow{2}{*}{\rs\rs \textbf{Dir($\bar{\alpha}$)}} & \multicolumn{3}{C{2.65cm}|} {\textbf{With CNN Model}} & \multicolumn{3}{C{2.65cm}|} {\textbf{With ResNet-$18$ Model}} & \multicolumn{3}{C{2.65cm}|} {\textbf{With ResNet-$34$ Model}} \\  \cline{3-11} 
&  & Acc & Req T [s] & Req E [J] & Acc & Req T [s] & Req E [J] & Acc & Req T [s] & Req E [J] \\ \hline
 & $0.5$ & $0.3817 \pm 0.0053$ & \multirow{3}{*}{$2240$} & $103334 \pm 31061$ & $0.4712 \pm 0.0043$ & \multirow{3}{*}{$6400$} & $196151 \pm 31714$ & $0.4932 \pm 0.0039$ & \multirow{3}{*}{$9600$} & $280889 \pm 41968$ \\ \cline{2-3} \cline{5-6} \cline{8-9} \cline{11-11} 
 PHFL & $0.9$ & $0.3782 \pm 0.0022$ & ~ & $103335 \pm 31061$ & $0.4815 \pm 0.0049$ & ~ & $196151 \pm 31714$ & $0.4928 \pm 0.0029$ & ~ & $280889 \pm 41968$ \\ \cline{2-3} \cline{5-6} \cline{8-9} \cline{11-11} 
 (Ours) & $10$ & $0.3853 \pm 0.0035$ & ~ & $103335 \pm 31061$ & $0.4806 \pm 0.0035$ & ~ & $196150 \pm 31714$ & $0.4935 \pm 0.0037$ & ~ & $280889 \pm 41968$ \\ \cline{1-11} 
 \rs HFL-VC & $0.5$ & $0.3924 \pm 0.0055$ & $2638 \pm 219 $ & $107290 \pm 24863$ & $0.4780 \pm 0.0043$ & $18074 \pm 365$ & $465668 \pm 65814$ & $0.4974 \pm 0.0032$ & $30841 \pm 535$ & $759264 \pm 102626$ \\ \cline{2-11}
 (UB) & $0.9$ & $0.3871 \pm 0.0063$ & $2638 \pm 219 $ & $107290 \pm 24863$ & $0.4819 \pm 0.0012$ & $18074 \pm 365$ & $465668 \pm 65814$ & $0.4992 \pm 0.0022$ & $30841 \pm 535$ & $759264 \pm 102626$ \\ \cline{2-11}
 (Ours) & $10$ & $0.3934 \pm 0.0048$ & $2638 \pm 219$ & $107290 \pm 24863$ & $0.4824 \pm 0.0035$ & $18074 \pm 365$ & $465668 \pm 65814$ & $0.4954 \pm 0.0058$ & $30841 \pm 535$ & $759264 \pm 102626$ \\ \cline{1-11}
 R-PHFL & $0.5$ & $0.2338 \pm 0.0047 $ & \multirow{3}{*}{$2240$} & $67480 \pm 20211$ & $0.1562 \pm 0.0181$ & \multirow{3}{*}{$6400$} & $185662 \pm 24670$ & $0.0153 \pm 0.0032$ & \multirow{3}{*}{$9600$} & $318013 \pm 41476$ \\ \cline{2-3} \cline{5-6} \cline{8-9} \cline{11-11} 
 & $0.9$ & $0.2282 \pm 0.0092$ & ~ & $67530 \pm 20296$ & $ 0.1706 \pm 0.0180$ & ~ & $185738 \pm 24765$ & $0.0149 \pm 0.0011$ & ~ & $318004 \pm 41465$ \\ \cline{2-3} \cline{5-6} \cline{8-9} \cline{11-11} 
 & $10$ & $0.2398 \pm 0.0085$ & ~ & $67535 \pm 20304$ & $0.1837 \pm 0.0220$ & ~ & $185639 \pm 24640$ & $0.0144 \pm 0.0036$ & ~ & $317815 \pm 41240$ \\ \cline{1-11} \hline
\end{tabular}
\label{performanceComparison_5_4_2_2_cifar100}
\end{table*}

Finally, we show the impact of $\kappa_0$ and $\kappa_1$ for different dataset heterogeneity levels on CIFAR-$10$ dataset in Table \ref{performanceComparison_5_2_2_2_cifar10} - Table \ref{performanceComparison_5_4_2_2_cifar10}, where $\mathrm{t^{th}}=4$s and $\mathrm{t^{th}}=6$s is used in our proposed PHFL algorithms for ResNet-$18$ and ResNet-$34$ models, respectively. 
Besides, for the CNN model, we used $\mathrm{t^{th}}=1.3$s and $\mathrm{2}$s, respectively, in our PHFL algorithm when $\kappa_0=5$ and $\kappa_0=10$, respectively, due to the facts that pruning exacerbates test performance for a shallow model and computational time dominates the offloading delay.
Similarly, we considered $\mathrm{t^{th}}=1.4$s and $\mathrm{t^{th}}=2.5$s, respectively, for $\kappa_0=5$ and $\kappa_0=10$ on CIFAR-$100$ datasets for the CNN model.
The performance comparisons for different $\bar{\alpha}$'s are shown in Table \ref{performanceComparison_5_2_2_2_cifar100} - Table \ref{performanceComparison_5_4_2_2_cifar100}.
From the tables, it is quite clear that pruning helps with negligible performance deviation from its original non-pruned counterparts.
Besides, for the shallow CNN model, the performance gain, in terms of test accuracy, of our proposed PHFL is insignificant compared to the bulky ResNets.  
Moreover, increasing $\kappa_0$ or $\kappa_1$ generally improves the test accuracy.
However, if the same $\mathrm{t^{th}}$ is to be used, it is beneficial to increase $\kappa_0$ compared to increasing $\kappa_1$.

\section{Conclusion}
\label{section_conclusion}
\noindent
This work proposed a model pruning solution to alleviate bandwidth scarcity and limited computational capacity of wireless clients in heterogeneous networks. 
Using the convergence upper-bound, pruning ratio, computation frequency and transmission power of the clients were jointly optimized to maximize the convergence rate.
The performances were evaluated on two popular datasets using three popular machine learning models of different total training parameter sizes.
The results suggest that pruning can significantly reduce training time, energy expense and bandwidth requirement while incurring negligible test performance.

\bibliography{Ref.bib}
\bibliographystyle{IEEEtran}

\begin{IEEEbiography}[{\includegraphics[width=1in,height=1.25in,clip]{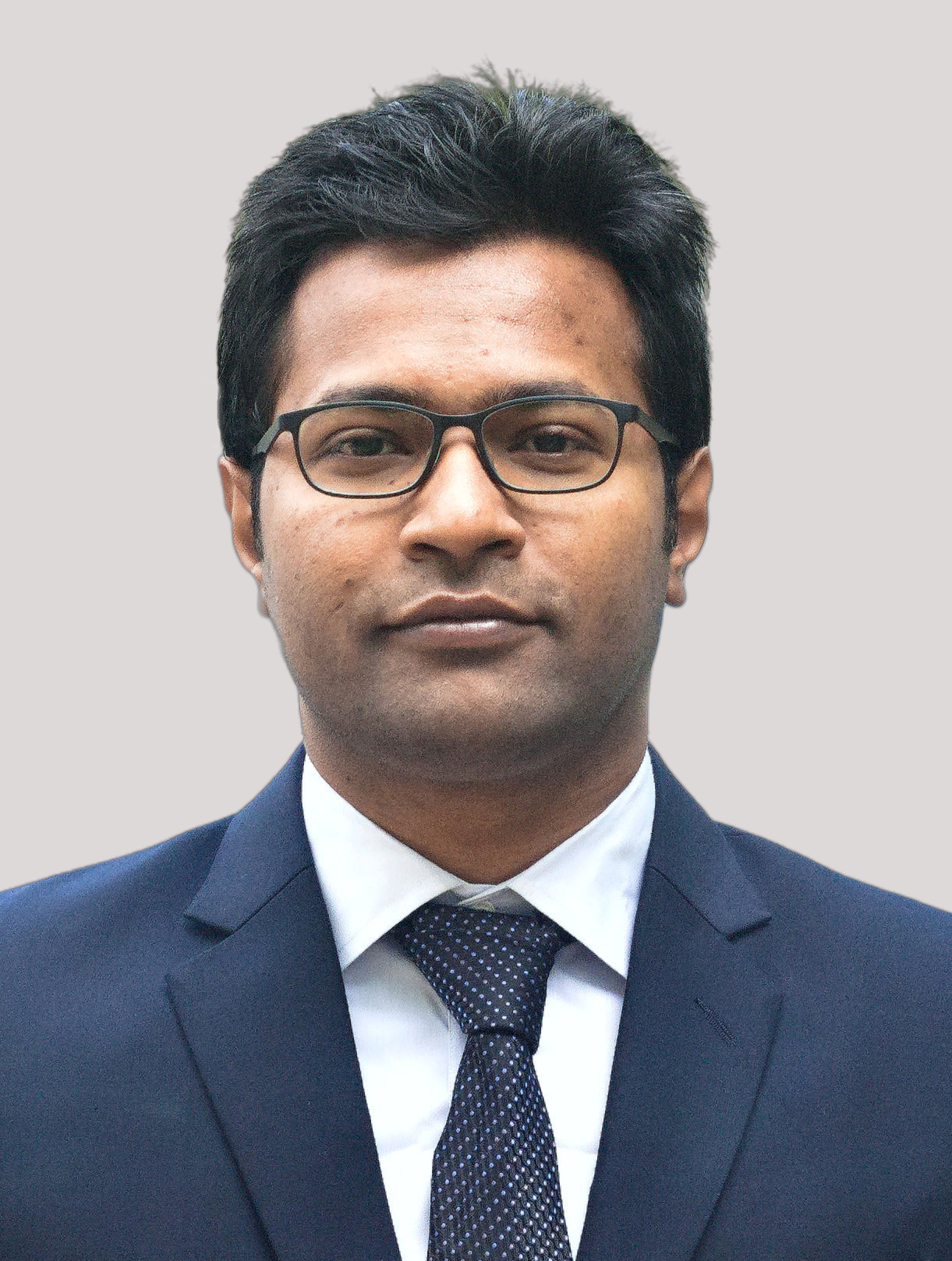}}]{Md Ferdous Pervej (M'23)} 
received the B.Sc. degree in electronics and telecommunication engineering from the Rajshahi University of Engineering and Technology, Rajshahi, Bangladesh, in 2014, and the M.S. and Ph.D. degrees in electrical engineering from Utah State University, Logan, UT, USA, in $2019$ and North Carolina State University, Raleigh, NC, USA, in  $2023$, respectively.  

He was with Mitsubishi Electric Research Laboratories, Cambridge, MA, in summer $2021$, and with Futurewei Wireless Research and Standards, Schaumburg, IL from May to December $2022$.
He is currently a Postdoctoral Scholar -- Research Associate in the Ming Hsieh Department of Electrical and Computer Engineering at the University of Southern California, Los Angeles, CA, USA.
His primary research interests are wireless networks, distributed machine learning, vehicle-to-everything communication, edge caching/computing, and machine learning for wireless networks.
\end{IEEEbiography}

\begin{IEEEbiography}[{\includegraphics[width=1in,height=1.25in,clip]{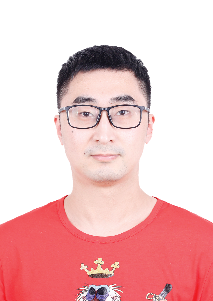}}]{Richeng Jin (M'21)} received the B.S. degree in information and communication engineering from Zhejiang University, Hangzhou, China, in 2015, and the Ph.D. degree in electrical engineering from North Carolina State University, Raleigh, NC, USA, in 2020.

He was a Postdoctoral Researcher in electrical and computer engineering at North Carolina State University, Raleigh, NC, USA, from 2021 to 2022. He is currently a faculty member of the department of information and communication engineering with Zhejiang University, Hangzhou, China. His research interests are in the general area of wireless AI, game theory, and security and privacy in machine learning/artificial intelligence and wireless networks.
\end{IEEEbiography}

\begin{IEEEbiography}[{\includegraphics[width=1.1in,height=1.2in,clip]{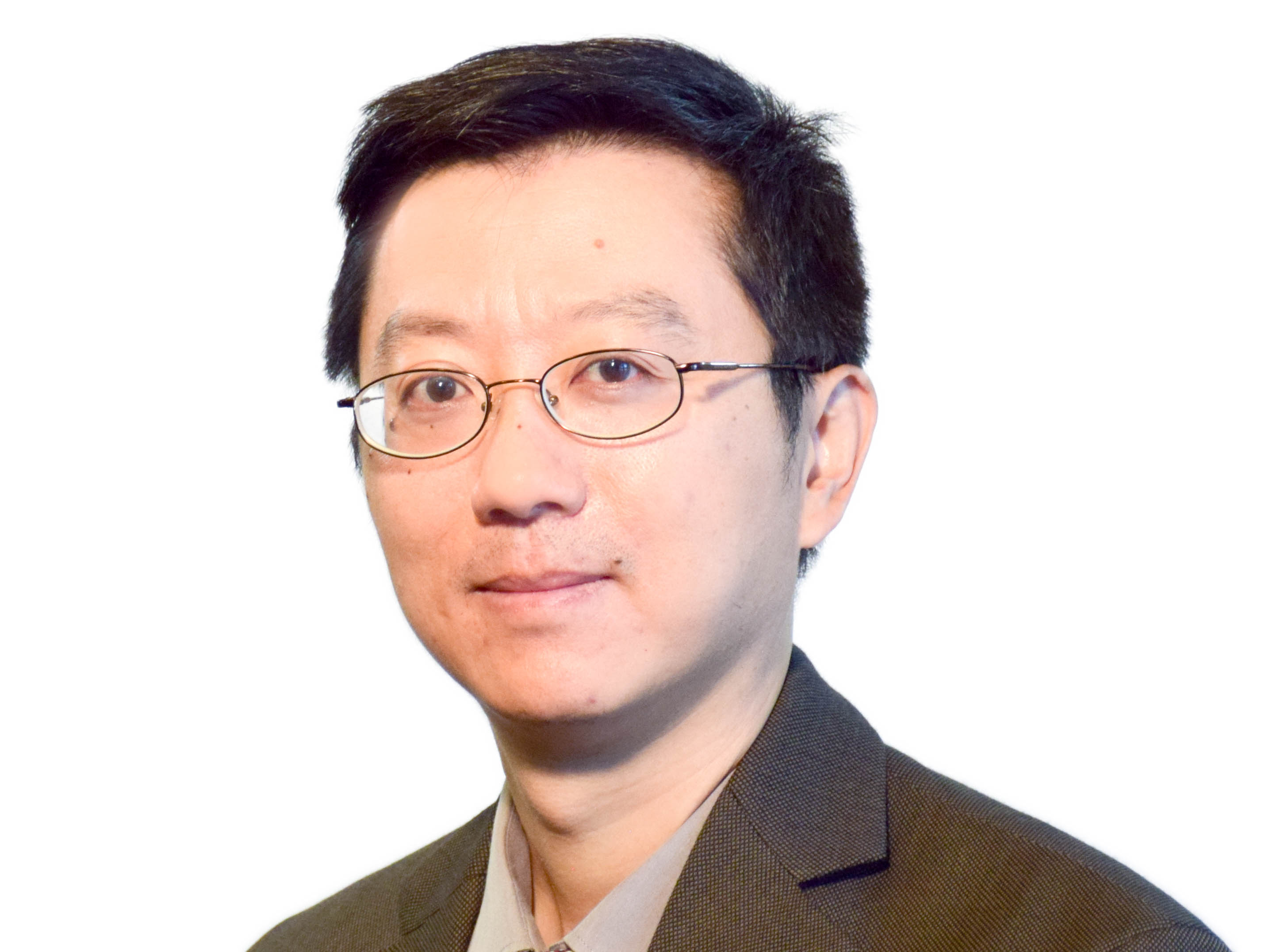}}]{Huaiyu Dai (F'17)} 
received the B.E. and M.S. degrees in electrical engineering from Tsinghua University, Beijing, China, in 1996 and 1998, respectively, and the Ph.D. degree in electrical engineering from Princeton University, Princeton, NJ in 2002. 

He was with Bell Labs, Lucent Technologies, Holmdel, NJ, in summer 2000, and with AT\&T Labs-Research, Middletown, NJ, in summer 2001. He is currently a Professor of Electrical and Computer Engineering with NC State University, Raleigh, holding the title of University Faculty Scholar. His research interests are in the general areas of communications, signal processing, networking, and computing. His current research focuses on machine learning and artificial intelligence for communications and networking, multilayer and interdependent networks, dynamic spectrum access and sharing, as well as security and privacy issues in the above systems.

He has served as an area editor for IEEE Transactions on Communications, a member of the Executive Editorial Committee for IEEE Transactions on Wireless Communications, and an editor for IEEE Transactions on Signal Processing. Currently he serves as an area editor for IEEE Transactions on Wireless Communications. He was a co-recipient of best paper awards at 2010 IEEE International Conference on Mobile Ad-hoc and Sensor Systems (MASS 2010), 2016 IEEE INFOCOM BIGSECURITY Workshop, and 2017 IEEE International Conference on Communications (ICC 2017). He received Qualcomm Faculty Award in 2019.
\end{IEEEbiography}

\newpage
\clearpage
\appendices 
\onecolumn

\section*{Supplementary Materials}
\noindent
\textbf{Additional Notations:} Analogous to our previous notations, we express the pruned virtual models at different hierarchy levels as $\tilde{\bar{\mathbf{w}}}_j^{t,0} \coloneqq \sum_{i=1}^{U_{j,k,l}} \alpha_i \tilde{\mathbf{w}}_i^{t,0}$, $\tilde{\bar{\mathbf{w}}}_k^{t,0} \coloneqq \sum_{j=1}^{V_{k,l}} \alpha_j \sum_{i=1}^{U_{j,k,l}} \alpha_i \tilde{\mathbf{w}}_i^{t,0} = \sum_{j=1}^{V_{k,l}} \alpha_j \tilde{\bar{\mathbf{w}}}_j^{t,0}$, $\tilde{\bar{\mathbf{w}}}_{l}^{t,0} \coloneqq \sum_{k=1}^{B_l} \alpha_k \sum_{j=1}^{V_{k,l}} \alpha_j \sum_{i=1}^{U_{j,k,l}} \alpha_i \tilde{\mathbf{w}}_i^{t,0} = \sum_{k=1}^{B_l} \alpha_k \tilde{\bar{\mathbf{w}}}_k^{t,0}$ and
$\tilde{\bar{\mathbf{w}}}^{t,0} \coloneqq \sum_{l=1}^{L} \alpha_l \tilde{\bar{\mathbf{w}}}_{l}^{t,0} = \sum_{u=1}^{\mathrm{U}} \alpha_u \tilde{\mathbf{w}}_u^{t,0}$.
Moreover, $\Tilde{g} (\tilde{\mathbf{w}}_i^{t,0}) \coloneqq g (\tilde{\mathbf{w}}_i^{t,0}) \odot \mathbf{m}_i^t$, $\nabla \tilde{f}_i(\tilde{\bar{\mathbf{w}}}^{t,0}) \coloneqq \nabla f_i(\tilde{\bar{\mathbf{w}}}^{t,0}) \odot \mathbf{m}_i^t$ and $ \nabla \tilde{f}(\tilde{\bar{\mathbf{w}}}^{t,0})) = \sum_{u=1}^{\mathrm{U}} \alpha_u \nabla \tilde{f}_u(\tilde{\bar{\mathbf{w}}}^{t,0}))$. 

\noindent
\textbf{Additional Assumptions:} Since the element-wise mask operation is equivalent to replacing a subset of the actual entries with zero, similar to \cite{jiang2022model, liu2021adaptive}, we also assume that the assumptions made in Section III-A hold if we apply the binary masks to all gradients.

\section{Proof of Theorem \ref{theorem_1}}
\label{proof_Theorem1}

\setcounter{Theorem}{0}
\begin{Theorem}
\label{theorem_1}
When the assumptions in Section III-A hold and $\eta \leq 1/\beta$, we have  
\begin{align}
\label{theorem_1_eqn_Sup}
    \rs \rs \theta_\mathrm{{PHFL}} \leq
    & \mathcal{O} \big( [f(\tilde{\bar{\mathbf{w}}}^0) - f_{\mathrm{inf}} ]/[\eta T] \big)  + \mathcal{O} ( \beta \eta \sigma^2 / \mathrm{U} ) + \underbrace{\mathcal{O} \big(\delta^{\mathrm{th}} \beta^2 D^2 \big)}_{\mathrm{pruning ~error}} + \underbrace{\mathcal{O} \big( \beta \eta G^2 \cdot \varphi_\mathrm{w,0}(\pmb{\delta},\pmb{\mathrm{f}}, \pmb{\mathrm{P}}) \big)}_{\mathrm{wireless ~ factor}} +  \mathcal{O} \big( \beta^2 \big[\mathrm{L}_1 + \mathrm{L}_2 + \mathrm{L}_3 + \mathrm{L}_4 \big] \big),
\end{align}
where $\pmb{\delta}=\{ \delta_1^t, \dots, \delta_U^t \}_{t=0}^{T-1}$, $\pmb{\mathrm{f}}=\{\mathrm{f}_1^t, \dots, \mathrm{f}_U^t\}_{t=0}^{T-1}$, $\pmb{\mathrm{P}}=\{P_1^t,\dots, P_U^t\}_{t=0}^{T-1}$ and $\mathrm{f}_i^t$ is the $i^{\mathrm{th}}$ client's CPU frequency in the wireless factors.
Besides, the terms $\mathrm{L}_1$, $\mathrm{L}_2$, $\mathrm{L}_3$ and $\mathrm{L}_4$ are  
\begin{align}
    \varphi_\mathrm{w,0}(\pmb{\delta},\pmb{\mathrm{f}}, \pmb{\mathrm{P}}) &= [1/T] \sum\nolimits_{t=0}^{T-1} \sum\nolimits_{l=1}^{L} \alpha_{l}^2 \sum\nolimits_{k=1}^{B_l} \alpha_k^2 \sum\nolimits_{j=1}^{V_{k,l}} \alpha_j^2 \sum\nolimits_{i=1}^{U_{j,k,l}} \alpha_i^2 \big[1/p_i^t - 1 \big], \\
    \mathrm{L}_1 &= [1/T] \sum\nolimits_{t=0}^{T-1}\sum\nolimits_{l=1}^{L} \alpha_l \sum\nolimits_{k=1}^{B_l} \alpha_k \sum\nolimits_{j=1}^{V_{k,l}} \alpha_j \sum\nolimits_{i=1}^{U_{j,k,l}} \alpha_i\mathbb{E} \Vert \bar{\mathbf{w}}_j^t - \tilde{\mathbf{w}}_i^{t} \Vert^2, \label{HL1}\\
    \mathrm{L}_2 &= [1/T] \sum\nolimits_{t=0}^{T-1} \sum\nolimits_{l=1}^{L} \alpha_l \sum\nolimits_{k=1}^{B_l} \alpha_k \sum\nolimits_{j=1}^{V_{k,l}} \alpha_j \mathbb{E} \Vert \bar{\mathbf{w}}_k^t - \bar{\mathbf{w}}_j^t \Vert^2, \label{HL2}\\
    \mathrm{L}_3 &= [1/T] \sum\nolimits_{t=0}^{T-1} \sum\nolimits_{l=1}^{L} \alpha_l \sum\nolimits_{k=1}^{B_l} \alpha_k \mathbb{E} \Vert \bar{\mathbf{w}}_{l}^t - \bar{\mathbf{w}}_k^t \Vert^2, \label{HL3}\\
    \mathrm{L}_4 &= [1/T] \sum\nolimits_{t=0}^{T-1} \sum\nolimits_{l=1}^{L} \alpha_l \mathbb{E} \Vert \bar{\mathbf{w}}^t - \bar{\mathbf{w}}_{l}^t \Vert^2.\label{HL4}
\end{align}
\end{Theorem}

\begin{proof}
Since there are $U$ clients and we consider the virtual models at different hierarchy levels, we start the convergence proof assuming the global model is the weighted combination of all these $U$ clients. 
After that, we break down our derivations for different hierarchy levels based on our proposed PHFL algorithm.
Note that this is also a standard practice in the literature \cite{feng2022Mobility, xu2021adaptive}.

First, let us write the update rule for the global model as
\begin{align}
\label{globalPrunedAgg}
    \bar{\mathbf{w}}^{t+1} 
    &= \sum_{u=1}^{\mathrm{U}} \alpha_u \tilde{\mathbf{w}}_u^{t+1} = \sum_{u=1}^{\mathrm{U}} \alpha_u \left( \tilde{\mathbf{w}}_u^{t,0} - \eta \tilde{g}(\tilde{\mathbf{w}}_u^{t,0}) \frac{\pmb{1}_u^t}{p_u^t} \right) 
    = \tilde{\bar{\mathbf{w}}}^{t,0} - \eta \sum_{u=1}^{\mathrm{U}} \alpha_u \tilde{g}(\tilde{\mathbf{w}}_u^{t,0}) \frac{\pmb{1}_u^t}{p_u^t},  
\end{align}
where $\tilde{\bar{\mathbf{w}}}^{t,0} \coloneqq \sum_{u=1}^{\mathrm{U}} \alpha_u \tilde{\mathbf{w}}_u^{t,0}$. 
As such, we write 
\begin{align}
\label{mainEq}
    f(\bar{\mathbf{w}}^{t+1}) 
    & = f( \tilde{\bar{\mathbf{w}}}^{t,0} - \eta \sum_{u=1}^{\mathrm{U}} \alpha_u \tilde{g}(\tilde{\mathbf{w}}_u^{t,0})\frac{\pmb{1}_u^t}{p_u^t}),
    \overset{(a)}{\leq} f(\tilde{\bar{\mathbf{w}}}^{t,0}) + \left<\nabla f(\tilde{\bar{\mathbf{w}}}^{t,0}), - \eta \sum_{u=1}^{\mathrm{U}} \alpha_u \tilde{g}(\tilde{\mathbf{w}}_u^{t,0}) \frac{\pmb{1}_u^t}{p_u^t} \right> + \frac{\beta \eta^2}{2} \left\Vert \sum_{u=1}^{\mathrm{U}} \alpha_u \tilde{g}(\tilde{\mathbf{w}}_u^{t,0}) \frac{\pmb{1}_u^t}{p_u^t} \right\Vert^2,
\end{align}
where $(a)$ stems from $\beta$-smoothness assumption.
For the third term in (\ref{mainEq}), we can write
\begin{align}
\label{3rdTerm}
    &\left<\nabla f(\tilde{\bar{\mathbf{w}}}^{t,0}), - \eta \sum_{u=1}^{\mathrm{U}} \alpha_u \tilde{g}(\tilde{\mathbf{w}}_u^{t,0}) \frac{\pmb{1}_u^t}{p_u^t} \right> 
    = -\eta \left< \sum_{u=1}^{\mathrm{U}} \alpha_u \nabla f_u (\tilde{\bar{\mathbf{w}}}^{t,0}), \sum_{u=1}^{\mathrm{U}} \alpha_u \tilde{g}(\tilde{\mathbf{w}}_u^{t,0}) \frac{\pmb{1}_u^t}{p_u^t} \right> \nonumber \\
    &= -\eta \sum_{u=1}^{\mathrm{U}} \alpha_u  \bigg< \nabla f_u (\tilde{\bar{\mathbf{w}}}^{t,0}) \odot \mathbf{m}_u^t,  \big(g(\tilde{\mathbf{w}}_u^{t,0})\odot \mathbf{m}_u^t\big) \frac{\pmb{1}_u^{t}}{p_u^t} \bigg> 
    = -\eta \sum_{u=1}^{\mathrm{U}} \alpha_u  \bigg<\nabla \tilde{f}_u (\tilde{\bar{\mathbf{w}}}^{t,0}),  \tilde{g} (\tilde{\mathbf{w}}_u^{t,0}) \frac{\pmb{1}_u^t}{p_u^t} \bigg>, 
\end{align}
where $\nabla \tilde{f}_u (\tilde{\bar{\mathbf{w}}}^{t,0}) \coloneqq \nabla f_u (\tilde{\bar{\mathbf{w}}}^{t,0}) \odot \mathbf{m}_u^t$.

Plugging (\ref{3rdTerm}) into (\ref{mainEq}) and taking expectation over both sides, we get
\begin{align}
\label{mainEq1}
    \mathbb{E} \left[ f(\bar{\mathbf{w}}^{t+1}) \right] 
    &\leq f(\tilde{\bar{\mathbf{w}}}^{t,0}) - \eta \sum_{u=1}^{\mathrm{U}} \alpha_u \mathbb{E} \left[ \left<\nabla \tilde{f}_u (\tilde{\bar{\mathbf{w}}}^{t,0}),  \tilde{g} (\tilde{\mathbf{w}}_u^{t,0}) \frac{\pmb{1}_u^t}{p_u^t} \right>\right] + \frac{\beta \eta^2}{2} \mathbb{E} \left[ \left\Vert \sum_{u=1}^{\mathrm{U}} \alpha_u \tilde{g}(\tilde{\mathbf{w}}_u^{t,0}) \frac{\pmb{1}_u^t}{p_u^t} \right\Vert^2\right],
\end{align}
We simply the third term in (\ref{mainEq1}) as follows:
\begin{subequations}
\label{3rdTerm1}
\begin{align}
    &- \eta \sum_{u=1}^{\mathrm{U}} \alpha_u \mathbb{E} \left[ \left< \nabla \tilde{f}_u (\tilde{\bar{\mathbf{w}}}^{t,0}),  \tilde{g} (\tilde{\mathbf{w}}_u^{t,0}) \frac{\pmb{1}_u^t}{p_u^t} \right>\right] 
    \overset{(a)}{=} - \eta \sum_{u=1}^{\mathrm{U}} \alpha_u \left< \nabla \tilde{f}_u (\tilde{\bar{\mathbf{w}}}^{t,0}), \mathbb{E} \left[ \tilde{g} (\tilde{\mathbf{w}}_u^{t,0}) \right] \mathbb{E} \bigg[\frac{\pmb{1}_u^t}{p_u^t}\bigg] \right>,\nonumber\\
    & \qquad = - \eta \sum_{u=1}^{\mathrm{U}} \alpha_u \left< \nabla \tilde{f}_u (\tilde{\bar{\mathbf{w}}}^{t,0}),  \nabla \tilde{f}_u (\tilde{\mathbf{w}}_u^{t,0}) \right>, \nonumber\\
    &\qquad= - \frac{\eta}{2} \sum_{u=1}^{\mathrm{U}} \alpha_u \Big\{\Vert \nabla \tilde{f}_u (\tilde{\bar{\mathbf{w}}}^{t,0}) \Vert^2 + \Vert \nabla \tilde{f}_u (\tilde{\mathbf{w}}_u^{t,0}) \Vert^2 - \Vert \nabla \tilde{f}_u (\tilde{\bar{\mathbf{w}}}^{t,0}) -  \nabla \tilde{f}_u (\tilde{\mathbf{w}}_u^{t,0}) \Vert^2 \Big\}, \nonumber\\
    &\qquad= \frac{\eta}{2} \sum_{u=1}^{\mathrm{U}} \alpha_u \Big\{ \Vert \nabla \tilde{f}_u (\tilde{\bar{\mathbf{w}}}^{t,0}) -  \nabla \tilde{f}_u (\tilde{\mathbf{w}}_u^{t,0}) \Vert^2 - \Vert \nabla \tilde{f}_u (\tilde{\bar{\mathbf{w}}}^{t,0}) \Vert^2 - \Vert \nabla \tilde{f}_u (\tilde{\mathbf{w}}_u^{t,0}) \Vert^2 \Big\}, \tag{\ref{3rdTerm1}} \\
    &\qquad= \frac{\eta}{2} \sum_{u=1}^{\mathrm{U}} \alpha_u \Big\{ \Vert \nabla \tilde{f}_u (\tilde{\bar{\mathbf{w}}}^{t,0}) -  \nabla \tilde{f}_u (\tilde{\mathbf{w}}_u^{t,0}) \Vert^2  -  \Vert \nabla \tilde{f}_u (\tilde{\mathbf{w}}_u^{t,0}) \Vert^2  \Big\} - \frac{\eta}{2} \sum_{u=1}^{\mathrm{U}} \alpha_u \Vert \nabla \tilde{f}_u (\tilde{\bar{\mathbf{w}}}^{t,0}) \Vert^2,\nonumber\\
    &\qquad \leq \frac{\eta}{2} \sum_{u=1}^{\mathrm{U}} \alpha_u \Big\{ \Vert \nabla \tilde{f}_u (\tilde{\bar{\mathbf{w}}}^{t,0}) -  \nabla \tilde{f}_u (\tilde{\mathbf{w}}_u^{t,0}) \Vert^2  -  \Vert \nabla \tilde{f}_u (\tilde{\mathbf{w}}_u^{t,0}) \Vert^2  \Big\} - \frac{\eta}{2} \left\Vert \sum_{u=1}^{\mathrm{U}} \alpha_u\nabla \tilde{f}_u (\tilde{\bar{\mathbf{w}}}^{t,0}) \right\Vert^2,\nonumber\\
    &\qquad\overset{(b)}{=} \frac{\eta}{2} \sum_{u=1}^{\mathrm{U}} \alpha_u \Big\{ \Vert \nabla \tilde{f}_u (\tilde{\bar{\mathbf{w}}}^{t,0}) -  \nabla \tilde{f}_u (\tilde{\mathbf{w}}_u^{t,0}) \Vert^2  - \Vert \nabla \tilde{f}_u (\tilde{\mathbf{w}}_u^{t,0}) \Vert^2 \Big\} - \frac{\eta}{2} \Vert \nabla \tilde{f} (\tilde{\bar{\mathbf{w}}}^{t,0}) \Vert^2, \nonumber
\end{align}    
\end{subequations}
where $(a)$ follows from the independence of client selection and SGD. 
Besides, we define $\Vert \nabla \tilde{f} (\tilde{\bar{\mathbf{w}}}^{t,0}) \Vert^2 \coloneqq  \Vert \sum_{u=1}^{\mathrm{U}} \alpha_u \nabla \tilde{f}_u (\tilde{\bar{\mathbf{w}}}^{t,0}) \Vert^2 = \Vert \sum_{u=1}^{\mathrm{U}} \alpha_u \nabla f_u (\tilde{\bar{\mathbf{w}}}^{t,0}) \odot \mathbf{m}_u^t \Vert^2$ in $(b)$.

The second term in (\ref{mainEq1}) can be simplified as follows:
\begin{align}
\label{2ndTerm1}
    & \mathbb{E} \left[ \left\Vert \sum_{u=1}^{\mathrm{U}} \alpha_u \tilde{g}\left(\tilde{\mathbf{w}}_u^{t,0}\right) \frac{\pmb{1}_u^t}{p_u^t} \right\Vert^2\right] \nonumber\\
    &\overset{(a)}{=} \mathbb{E} \left \Vert \sum_{u=1}^{\mathrm{U}} \alpha_u \tilde{g}\left(\tilde{\mathbf{w}}_u^{t,0}\right) \frac{\pmb{1}_u^t}{p_u^t} - \mathbb{E} \left[ \sum_{u=1}^{\mathrm{U}} \alpha_u \tilde{g}\left(\tilde{\mathbf{w}}_u^{t,0}\right) \frac{\pmb{1}_u^t}{p_u^t} \right] \right\Vert^2 + \left(\mathbb{E} \left[ \sum_{u=1}^{\mathrm{U}} \alpha_u \tilde{g}\left(\tilde{\mathbf{w}}_u^{t,0}\right) \frac{\pmb{1}_u^t}{p_u^t} \right]\right)^2, \nonumber\\
    &=\mathbb{E} \left \Vert \sum_{u=1}^{\mathrm{U}} \alpha_u \tilde{g}\left(\tilde{\mathbf{w}}_u^{t,0}\right) \frac{\pmb{1}_u^t}{p_u^t} \pm \sum_{u=1}^{\mathrm{U}} \alpha_u \tilde{g}\left(\tilde{\mathbf{w}}_u^{t,0}\right) - \sum_{u=1}^{\mathrm{U}} \alpha_u \nabla \tilde{f}_u \left(\tilde{\mathbf{w}}_u^{t,0} \right) \right\Vert^2 + \left\Vert \sum_{u=1}^{\mathrm{U}} \alpha_u \nabla \tilde{f}_u \left(\tilde{\mathbf{w}}_u^{t,0} \right) \right\Vert^2, \nonumber\\
    &\overset{(b)}{\leq} \mathbb{E} \left \Vert \sum_{u=1}^{\mathrm{U}} \alpha_u \left(\frac{\pmb{1}_u^t}{p_u^t} -1\right)\tilde{g}\left(\tilde{\mathbf{w}}_u^{t,0}\right) + \sum_{u=1}^{\mathrm{U}} \alpha_u \left( \tilde{g}\left(\tilde{\mathbf{w}}_u^{t,0}\right)- \nabla \tilde{f}_u \left(\tilde{\mathbf{w}}_u^{t,0} \right) \right) \right\Vert^2 + \sum_{u=1}^{\mathrm{U}} \alpha_u \left\Vert \nabla \tilde{f}_u \left(\tilde{\mathbf{w}}_u^{t,0} \right) \right\Vert^2,\nonumber\\
    &\overset{(c)}{\leq} 2\mathbb{E} \left \Vert \sum_{u=1}^{\mathrm{U}} \alpha_u \left(\frac{\pmb{1}_u^t}{p_u^t} -1\right)\tilde{g}\left(\tilde{\mathbf{w}}_u^{t,0}\right) \right\Vert^2 + 2\mathbb{E} \left \Vert\sum_{u=1}^{\mathrm{U}} \alpha_u \left( \tilde{g}\left(\tilde{\mathbf{w}}_u^{t,0}\right)- \nabla \tilde{f}_u \left(\tilde{\mathbf{w}}_u^{t,0} \right) \right) \right\Vert^2 + \sum_{u=1}^{\mathrm{U}} \alpha_u \left\Vert \nabla \tilde{f}_u \left(\tilde{\mathbf{w}}_u^{t,0} \right) \right\Vert^2, \nonumber\\
    &=2\mathbb{E} \Big[\sum_{u=1}^{\mathrm{U}} \alpha_u^2 \left \Vert \left(\frac{\pmb{1}_u^t}{p_u^t} -1\right)\tilde{g}\left(\tilde{\mathbf{w}}_u^{t,0}\right) \right\Vert^2 + \sum_{u=1}^{\mathrm{U}} \alpha_u \sum_{u'=1, u'\neq u}^U \alpha_{u'}  \left(\frac{\pmb{1}_u^t}{p_u^t} -1\right)\tilde{g}\left(\tilde{\mathbf{w}}_u^{t,0}\right) \left(\frac{\pmb{1}_{u'}^t}{p_{u'}^t} - 1 \right)\tilde{g}(\tilde{\mathbf{w}}_{u'}^{t,0}) \Big] + \nonumber\\
    &\Squad 2\mathbb{E} \Big[\sum_{u=1}^{\mathrm{U}} \alpha_u^2 \left \Vert \left( \tilde{g}\left(\tilde{\mathbf{w}}_u^{t,0}\right)- \nabla \tilde{f}_u \left(\tilde{\mathbf{w}}_u^{t,0} \right) \right) \right\Vert^2 + \nonumber\\
    &\qquad \sum_{u=1}^{\mathrm{U}} \alpha_u \sum_{u'=1, u \neq u'}^U \alpha_{u'} \left( \tilde{g}\left(\tilde{\mathbf{w}}_u^{t,0}\right) - \nabla \tilde{f}_u \left(\tilde{\mathbf{w}}_u^{t,0} \right) \right) \left( \tilde{g}(\tilde{\mathbf{w}}_{u'}^{t,0})- \nabla \tilde{f}_{u'}(\tilde{\mathbf{w}}_{u'}^{t,0}) \right) \Big] + \sum_{u=1}^{\mathrm{U}} \alpha_u \left\Vert \nabla \tilde{f}_u \left(\tilde{\mathbf{w}}_u^{t,0} \right) \right\Vert^2,\nonumber\\
    &\overset{(d)}{\leq} 2\mathbb{E} \bigg[\sum_{u=1}^{\mathrm{U}} \alpha_u^2 \left \Vert \left(\frac{\pmb{1}_u^t}{p_u^t} -1\right)\tilde{g}\left(\tilde{\mathbf{w}}_u^{t,0}\right) \right\Vert^2 \bigg] + 2 \sigma^2 \sum_{u=1}^{\mathrm{U}} \alpha_u^2 + \sum_{u=1}^{\mathrm{U}} \alpha_u \left\Vert \nabla \tilde{f}_u \left(\tilde{\mathbf{w}}_u^{t,0} \right) \right\Vert^2,\nonumber\\ 
    &= 2\sum_{u=1}^{\mathrm{U}} \alpha_u^2  \left(\frac{1-p_u^t}{p_u^t}\right) \mathbb{E} \left \Vert \tilde{g}\left(\tilde{\mathbf{w}}_u^{t,0}\right) \right\Vert^2 + 2 \sigma^2 \sum_{u=1}^{\mathrm{U}} \alpha_u^2 + \sum_{u=1}^{\mathrm{U}} \alpha_u \left\Vert \nabla \tilde{f}_u \left(\tilde{\mathbf{w}}_u^{t,0} \right) \right\Vert^2,
\end{align}
where $(a)$ comes from the definition of variance, $(b)$ follows from Jensen inequality, i.e., $\left\Vert \sum_{u=1}^{\mathrm{U}} \alpha_u \nabla \tilde{f}_u (\tilde{\mathbf{w}}_u^{t,0}) \right\Vert^2 \leq \sum_{u=1}^{\mathrm{U}} \alpha_u  \left\Vert \nabla \tilde{f}_u (\tilde{\mathbf{w}}_u^{t,0}) \right\Vert^2$, $(c)$ stems from the fact that $\Vert \sum_{i=1}^n \mathbf{a}_i \Vert^2 \leq n \sum_{i=1}^n \Vert \mathbf{a}_i \Vert^2$ and $(d)$ is due to the bounded stochastic gradient assumption.

Plugging (\ref{3rdTerm1}) and (\ref{2ndTerm1}) in (\ref{mainEq1}), we get
\begin{subequations}
\label{mainEq2}
\begin{align}
    &\mathbb{E} \left[ f(\bar{\mathbf{w}}^{t+1}) \right] 
    \leq  f(\tilde{\bar{\mathbf{w}}}^{t,0}) + \frac{\eta}{2} \sum_{u=1}^{\mathrm{U}} \alpha_u \Big\{ \Vert \nabla \tilde{f}_u (\tilde{\bar{\mathbf{w}}}^{t,0}) -  \nabla \tilde{f}_u (\tilde{\mathbf{w}}_u^{t,0}) \Vert^2  - \Vert \nabla \tilde{f}_u (\tilde{\mathbf{w}}_u^{t,0}) \Vert^2 \Big\} - \nonumber\\
    &\qquad \frac{\eta}{2} \Vert \nabla \tilde{f} (\tilde{\bar{\mathbf{w}}}^{t,0}) \Vert^2 + \eta^2 \sum_{u=1}^{\mathrm{U}} \alpha_u^2  \left(\frac{1-p_u^t}{p_u^t}\right) \mathbb{E} \left \Vert \tilde{g}\left(\tilde{\mathbf{w}}_u^{t,0}\right) \right\Vert^2 + \beta \eta^2 \sigma^2 \sum_{u=1}^{\mathrm{U}} \alpha_u^2 + \frac{\beta \eta^2}{2} \sum_{u=1}^{\mathrm{U}} \alpha_u  \left\Vert \nabla \tilde{f}_u (\tilde{\mathbf{w}}_u^{t,0}) \right\Vert^2, \nonumber\\
    &= f(\tilde{\bar{\mathbf{w}}}^{t,0}) + \frac{\eta}{2} \sum_{u=1}^{\mathrm{U}} \alpha_u \bigg \Vert \nabla \tilde{f}_u (\tilde{\bar{\mathbf{w}}}^{t,0}) -  \nabla \tilde{f}_u (\tilde{\mathbf{w}}_u^{t,0}) \bigg\Vert^2  - \frac{\eta}{2} \left \Vert \nabla \tilde{f} (\tilde{\bar{\mathbf{w}}}^{t,0}) \right\Vert^2 - \frac{\eta}{2} \left(1 - \beta \eta\right) \sum_{u=1}^{\mathrm{U}} \Vert \nabla \tilde{f}_u (\tilde{\mathbf{w}}_u^{t,0}) \Vert^2 \nonumber\\
    & \Squad + \beta \eta^2 \sum_{u=1}^{\mathrm{U}} \alpha_u^2  \left(\frac{1-p_u^t}{p_u^t}\right) \mathbb{E} \left \Vert \tilde{g}\left(\tilde{\mathbf{w}}_u^{t,0}\right) \right\Vert^2 + \beta \eta^2 \sigma^2 \sum_{u=1}^{\mathrm{U}} \alpha_u^2. \tag{\ref{mainEq2}}
\end{align}
\end{subequations}
Note that, when $\eta \leq \frac{1}{\beta}$, $\left(1 - \beta \eta\right) \geq 0$. 
Therefore, we can drop the fourth term in (\ref{mainEq2}).

To that end, rearranging the terms in (\ref{mainEq2}), then dividing both sides by $\frac{\eta}{2}$, taking expectations on both sides and averaging over time, we get the following
\begin{align}
\label{mainEq3}
    \frac{1}{T} \sum_{t=0}^{T-1} \mathbb{E} \left[\left\Vert \nabla \tilde{f}\left(\tilde{\bar{\mathbf{w}}}^{t,0} \right) \right\Vert^2\right] 
    &\leq \frac{2\left(f(\tilde{\bar{\mathbf{w}}}^{0}) - \mathbb{E} [f(\bar{\mathbf{w}}^{T})] \right)}{\eta T} + 2\beta \eta \sigma^2 \rs \sum_{u=1}^{\mathrm{U}} \alpha_u^2 + \frac{2 \beta \eta G^2}{T} \sum_{t=0}^{T-1} \sum_{u=1}^{\mathrm{U}} \alpha_u^2 \left(\frac{1-p_u^t}{p_u^t}\right) + \nonumber \\ 
    &\Squad  \frac{1}{T} \sum_{t=0}^{T-1} \sum_{u=1}^{\mathrm{U}} \alpha_u \mathbb{E} \left[ \Vert \nabla \tilde{f}_u (\tilde{\bar{\mathbf{w}}}^{t,0}) -  \nabla \tilde{f}_u (\tilde{\mathbf{w}}_u^{t,0}) \Vert^2 \right], 
\end{align} 
where we use the notation $\tilde{\bar{\mathbf{w}}}^{0}$ to represent $\tilde{\bar{\mathbf{w}}}^{0,0} $ for simplicity.

The last term in (\ref{mainEq3}) can be expanded as follows:
\begin{subequations}
\label{lastTerm}
\begin{align}
    &\frac{1}{T} \sum_{t=0}^{T-1}  \sum_{u=1}^{\mathrm{U}} \alpha_u \mathbb{E} \left[ \Vert \nabla \tilde{f}_u (\tilde{\bar{\mathbf{w}}}^{t,0}) -  \nabla \tilde{f}_u (\tilde{\mathbf{w}}_u^{t,0}) \Vert^2 \right] 
    = \sum_{u=1}^{\mathrm{U}} \alpha_u \mathbb{E} \bigg[ \Big\Vert \nabla \tilde{f}_u (\tilde{\bar{\mathbf{w}}}^{t,0}) \pm \nabla \tilde{f}_u (\tilde{\bar{\mathbf{w}}}_{l,(u)}^{t,0}) \pm \nabla \tilde{f}_u (\tilde{\bar{\mathbf{w}}}_{k,(u)}^{t,0}) \pm \nabla \tilde{f}_u (\tilde{\bar{\mathbf{w}}}_{j,(u)}^{t,0}) -  \nabla \tilde{f}_u (\tilde{\mathbf{w}}_u^{t,0}) \Big\Vert^2 \bigg], \nonumber \\
    &\overset{(a)}{\leq} \frac{4}{T} \sum_{t=0}^{T-1} \sum_{u=1}^{\mathrm{U}} \alpha_u \mathbb{E} \Big[ \Vert \nabla \tilde{f}_u (\tilde{\bar{\mathbf{w}}}^{t,0}) - \nabla \tilde{f}_u (\tilde{\bar{\mathbf{w}}}_{l,(u)}^{t,0}) \Vert^2 + \Vert \nabla \tilde{f}_u (\tilde{\bar{\mathbf{w}}}_{l,(u)}^{t,0}) - \nabla \tilde{f}_u (\tilde{\bar{\mathbf{w}}}_{k,(u)}^{t,0})\Vert^2 + \Vert \nabla \tilde{f}_u (\tilde{\bar{\mathbf{w}}}_{k,(u)}^{t,0}) - \nabla \tilde{f}_u (\tilde{\bar{\mathbf{w}}}_{j,(u)}^{t,0}) \Vert^2 + \nonumber\\
    &\Bquad \Bquad \Vert \nabla \tilde{f}_u (\tilde{\bar{\mathbf{w}}}_{j,(u)}^{t,0}) - \nabla \tilde{f}_u (\tilde{\mathbf{w}}_u^{t,0}) \Vert^2 \Big], \nonumber \\
    &\overset{(b)}{\leq} \frac{4\beta^2}{T} \sum_{t=0}^{T-1} \sum_{u=1}^{\mathrm{U}} \alpha_u \Big\{  \mathbb{E} \left[\Vert \tilde{\bar{\mathbf{w}}}^{t,0} - \tilde{\bar{\mathbf{w}}}_{l,(u)}^{t,0} \Vert^2 \right] + \mathbb{E} \left[\Vert \tilde{\bar{\mathbf{w}}}_{l,(u)}^{t,0} - \tilde{\bar{\mathbf{w}}}_{k,(u)}^{t,0} \Vert^2 \right] + \mathbb{E} \left[\Vert \tilde{\bar{\mathbf{w}}}_{k,(u)}^{t,0} - \tilde{\bar{\mathbf{w}}}_{j,(u)}^{t,0} \Vert^2 \right] + \mathbb{E} \left[\Vert \tilde{\bar{\mathbf{w}}}_{j,(u)}^{t,0} - \tilde{\mathbf{w}}_u^{t,0} \Vert^2 \right]\Big\},\nonumber\\
    &= \frac{4\beta^2}{T} \sum_{t=0}^{T-1} \sum_{l=1}^{L} \alpha_l \sum_{k=1}^{B_l} \alpha_k \sum_{j=1}^{V_{k,l}} \alpha_j \sum_{i=1}^{U_{j,k,l}} \alpha_i \Big\{ \mathbb{E} \left[\Vert \tilde{\bar{\mathbf{w}}}^{t,0} - \tilde{\bar{\mathbf{w}}}_{l}^{t,0} \Vert^2 \right] + \mathbb{E} \left[\Vert \tilde{\bar{\mathbf{w}}}_{l}^{t,0} - \tilde{\bar{\mathbf{w}}}_{k}^{t,0} \Vert^2 \right] + \mathbb{E} \left[\Vert \tilde{\bar{\mathbf{w}}}_{k}^{t,0} - \tilde{\bar{\mathbf{w}}}_j^{t,0} \Vert^2 \right] +  \mathbb{E} \left[\Vert \tilde{\bar{\mathbf{w}}}_j^{t,0} - \tilde{\mathbf{w}}_i^{t,0} \Vert^2 \right]\Big\},\nonumber\\
    &= \frac{4\beta^2}{T} \sum_{t=0}^{T-1} \Bigg\{ \sum_{l=1}^{L} \alpha_l \mathbb{E} \left[\Vert \tilde{\bar{\mathbf{w}}}^{t,0} - \tilde{\bar{\mathbf{w}}}_{l}^{t,0} \Vert^2 \right] + \sum_{l=1}^{L} \alpha_l \sum_{k=1}^{B_l} \alpha_k \mathbb{E} \left[\Vert \tilde{\bar{\mathbf{w}}}_l^{t,0} - \tilde{\bar{\mathbf{w}}}_k^{t,0} \Vert^2 \right] + \sum_{l=1}^{L} \alpha_l \sum_{k=1}^{B_l} \alpha_k \sum_{j=1}^{V_{k,l}} \alpha_j \mathbb{E} \left[\Vert \tilde{\bar{\mathbf{w}}}_{k}^{t,0} - \tilde{\bar{\mathbf{w}}}_j^{t,0} \Vert^2 \right] + \nonumber\\
    &\Bquad \Mquad \sum_{l=1}^{L} \alpha_l \sum_{k=1}^{B_l} \alpha_k \sum_{j=1}^{V_{k,l}} \alpha_j \sum_{i=1}^{U_{j,k,l}} \alpha_i \mathbb{E} \left[\Vert \tilde{\bar{\mathbf{w}}}_j^{t,0} - \tilde{\mathbf{w}}_i^{t,0} \Vert^2 \right]\Bigg\}, \nonumber\\
    &= \frac{4\beta^2}{T} \sum_{t=0}^{T-1} \Big\{ \sum\nolimits_{l=1}^{L} \alpha_l  \mathbb{E} \left[\Vert \tilde{\bar{\mathbf{w}}}^{t,0}  \pm \bar{\mathbf{w}}^t \pm \bar{\mathbf{w}}_{l}^t -  \tilde{\mathbf{w}}_{l}^t \Vert^2 \right] + \sum_{l=1}^{L} \alpha_l \sum_{k=1}^{B_l} \alpha_k \mathbb{E} \left[\Vert \tilde{\bar{\mathbf{w}}}_{l}^{t,0} \pm \bar{\mathbf{w}}_{l}^t \pm \bar{\mathbf{w}}_k^t - \tilde{\bar{\mathbf{w}}}_k^{t,0} \Vert^2 \right] + \nonumber \\
    & \Squad \sum_{l=1}^{L} \alpha_l \sum_{k=1}^{B_l} \alpha_k \sum_{j=1}^{V_{k,l}} \alpha_j \mathbb{E} \left[\Vert \tilde{\bar{\mathbf{w}}}_k^{t,0} \pm \bar{\mathbf{w}}_k^t \pm \bar{\mathbf{w}}_j^t - \tilde{\bar{\mathbf{w}}}_j^{t,0} \Vert^2 \right] + \sum_{l=1}^{L} \alpha_l \sum_{k=1}^{B_l} \alpha_k \sum_{j=1}^{V_{k,l}} \alpha_j \sum_{i=1}^{U_{j,k,l}} \alpha_i \mathbb{E} \left[\Vert \tilde{\bar{\mathbf{w}}}_j^{t,0} \pm \bar{\mathbf{w}}_j^t \pm \Tilde{\mathbf{w}}_i^t - \tilde{\mathbf{w}}_i^{t,0} \Vert^2 \right]\Big\},\nonumber\\
    &\overset{(c)}{\leq} \frac{12 \beta^2}{T} \sum_{t=0}^{T-1} \bigg\{ \sum_{l=1}^{L} \alpha_l \mathbb{E} \left[ \Vert \tilde{\bar{\mathbf{w}}}^{t,0}  - \bar{\mathbf{w}}^t \Vert^2 +  \Vert \bar{\mathbf{w}}_{l}^t - \tilde{\bar{\mathbf{w}}}_{l}^{t,0} \Vert^2 + \Vert \bar{\mathbf{w}}^t - \bar{\mathbf{w}}_{l}^t \Vert^2 \right] + \sum_{l=1}^{L} \alpha_l \sum_{k=1}^{B_l} \alpha_k \mathbb{E} \Big[ \Vert \tilde{\bar{\mathbf{w}}}_{l}^{t,0} - \bar{\mathbf{w}}_{l}^t \Vert^2 + \Vert \bar{\mathbf{w}}_k^t - \tilde{\bar{\mathbf{w}}}_k^{t,0} \Vert^2 + \Vert \bar{\mathbf{w}}_{l}^t - \bar{\mathbf{w}}_k^t \Vert^2 \Big] + \nonumber \\
    &\Squad \sum_{l=1}^{L} \alpha_l \sum_{k=1}^{B_l} \alpha_k \sum_{j=1}^{V_{k,l}} \alpha_j \mathbb{E} \Big[ \Vert \tilde{\bar{\mathbf{w}}}_k^{t,0} - \bar{\mathbf{w}}_k^t \Vert^2 + \Vert \bar{\mathbf{w}}_j^t - \tilde{\bar{\mathbf{w}}}_j^{t,0} \Vert^2 + \Vert \bar{\mathbf{w}}_k^t - \bar{\mathbf{w}}_j^t \Vert^2  \Big]  + \nonumber \\
    &\Squad \sum_{l=1}^{L} \alpha_l \sum_{k=1}^{B_l} \alpha_k \sum_{j=1}^{V_{k,l}} \alpha_j \sum_{i=1}^{U_{j,k,l}} \alpha_i \Big\{ \mathbb{E} \left[\Vert \tilde{\bar{\mathbf{w}}}_j^{t,0} - \bar{\mathbf{w}}_j^t\Vert^2 + \Vert \bar{\mathbf{w}}_j^t - \tilde{\mathbf{w}}_i^t \Vert^2 + \Vert \tilde{\mathbf{w}}_i^t - \tilde{\mathbf{w}}_i^{t,0} \Vert^2 \right]\bigg\}, \nonumber\\
    &= \frac{12 \beta^2}{T} \underbrace{\sum_{t=0}^{T-1} \sum_{l=1}^{L} \alpha_l  \mathbb{E} \Vert \bar{\mathbf{w}}^t - \bar{\mathbf{w}}_{l}^t \Vert^2}_{\mathrm{L}_4} +  \frac{12 \beta^2}{T} \underbrace{\sum_{t=0}^{T-1} \sum_{l=1}^{L} \alpha_l \sum_{k=1}^{B_l} \alpha_k \mathbb{E} \Vert \bar{\mathbf{w}}_{l}^t - \bar{\mathbf{w}}_k^t \Vert^2}_{\mathrm{L}_3} + \frac{12 \beta^2}{T} \underbrace{\sum_{t=0}^{T-1} \sum_{l=1}^{L} \alpha_l \sum_{k=1}^{B_l} \alpha_k \sum_{j=1}^{V_{k,l}} \alpha_j\mathbb{E} \Vert \bar{\mathbf{w}}_k^t - \bar{\mathbf{w}}_j^t \Vert^2}_{\mathrm{L}_2} + \nonumber\\
    &\Squad\frac{12 \beta^2}{T} \underbrace{\sum_{t=0}^{T-1} \sum_{l=1}^{L} \alpha_l \sum_{k=1}^{B_l} \alpha_k \sum_{j=1}^{V_{k,l}} \alpha_j \sum_{i=1}^{U_{j,k,l}} \alpha_i \mathbb{E} \Vert \bar{\mathbf{w}}_j^t - \tilde{\mathbf{w}}_i^t \Vert^2}_{\mathrm{L}_1} + \underbrace{ \frac{12 \beta^2}{T} \sum_{t=0}^{T-1} \mathbb{E} \Vert \bar{\mathbf{w}}^t - \tilde{\bar{\mathbf{w}}}^{t,0} \Vert^2}_{\mathrm{T}_5} + \underbrace{\frac{24 \beta^2}{T} \sum_{t=0}^{T-1} \sum_{l=1}^{L} \alpha_l  \mathbb{E} \Vert \bar{\mathbf{w}}_{l}^t - \tilde{\bar{\mathbf{w}}}_{l}^{t,0} \Vert^2}_{\mathrm{T}_4} + \nonumber\\
    &\Squad\underbrace{\frac{24 \beta^2}{T} \sum_{t=0}^{T-1} \sum_{l=1}^{L} \alpha_l \sum_{k=1}^{B_l} \alpha_k  \mathbb{E} \Vert \bar{\mathbf{w}}_k^t - \tilde{\bar{\mathbf{w}}}_k^{t,0} \Vert^2}_{\mathrm{T}_3} + \underbrace{\frac{24 \beta^2}{T} \sum_{t=0}^{T-1} \sum_{l=1}^{L} \alpha_l \sum_{k=1}^{B_l} \alpha_k \sum_{j=1}^{V_{k,l}} \alpha_j \mathbb{E} \Vert \bar{\mathbf{w}}_j^t - \tilde{\bar{\mathbf{w}}}_j^{t,0} \Vert^2}_{\mathrm{T}_2} + \nonumber\\
    &\Squad \underbrace{\frac{12 \beta^2}{T} \sum_{t=0}^{T-1} \sum_{l=1}^{L} \alpha_l \sum_{k=1}^{B_l} \alpha_k \sum_{j=1}^{V_{k,l}} \alpha_j \sum_{i=1}^{U_{j,k,l}} \mathbb{E} \Vert \tilde{\mathbf{w}}_i^t - \tilde{\mathbf{w}}_i^{t,0} \Vert^2}_{\mathrm{T}_1}, \tag{\ref{lastTerm}}
\end{align}
\end{subequations} 
where we used the notation $\bar{\mathbf{w}}_{j,(u)}^t, \bar{\mathbf{w}}_{k,(u)}^t$ and $\bar{\mathbf{w}}_{l,(u)}^t$ to represent the model of the VC, sBS and mBS, respectively, that UE $u$ is connected to.
Furthermore, $(a)$ and $(c)$ stem from the fact that $\Vert \sum_{i=1}^n \mathbf{a}_i \Vert^2 \leq n \sum_{i=1}^n \Vert \mathbf{a}_i \Vert^2$.
Moreover, $(b)$ arises from the $\beta$-smoothness and divergence between the loss functions assumptions.

To this end, we bound the terms $\mathrm{T}_1$ to $\mathrm{T}_5$ using our aggregation rules and definitions.
First, let us focus on the term $\mathrm{T}_1$
\begin{align}
    \mathrm{T}_1 
    &\overset{(a)}{=} \frac{12 \beta^2}{T} \sum_{l=1}^{L} \alpha_l \sum_{k=1}^{B_l} \alpha_k \sum_{j=1}^{V_{k,l}} \alpha_j \sum_{i=1}^{U_{j,k,l}} \alpha_i \sum_{t=0}^{T-1} \mathbb{E} \bigg\Vert \mathbf{w}_i^{t} -  \tilde{\mathbf{w}}_i^{t,0}\bigg\Vert^2 \nonumber \\ 
    &\overset{(b)}{\bblue{\leq}} \frac{12 \beta^2}{T} \sum_{l=1}^{L} \alpha_l \sum_{k=1}^{B_l} \alpha_k \sum_{j=1}^{V_{k,l}} \alpha_j \sum_{i=1}^{U_{j,k,l}} \alpha_i \sum_{t=0}^{T-1} \delta_i^t \mathbb{E} \Vert \mathbf{w}_i^t \Vert^2, \label{localPrune}
\end{align}
where in $(a)$, we trace back to the nearest synchronization iteration where $\mathbf{w}_i^t \gets \bar{\mathbf{w}}_j^t$ and $\tilde{\mathbf{w}}_i^{t} = \mathbf{w}_i^t$.
Besides, we use the definition of pruning ratio in $(b)$.

Now, we calculate the bound for the term $\mathrm{T}_2$ as follows
\begin{align}
   \mathrm{T}_2 
   &= \frac{24 \beta^2}{T} \sum_{t=0}^{T-1} \sum_{l=1}^{L} \alpha_l \sum_{k=1}^{B_l} \alpha_k \sum_{j=1}^{V_{k,l}} \alpha_j \mathbb{E} \bigg \Vert \sum_{i=1}^{U_{j,k,l}} \alpha_i \big( \tilde{\mathbf{w}}_i^t - \tilde{\mathbf{w}}_i^{t,0} \big) \bigg \Vert^2 \nonumber\\
   &\overset{(a)}{\bblue{\leq}} \frac{24 \beta^2}{T} \sum_{t=0}^{T-1} \sum_{l=1}^{L} \alpha_l \sum_{k=1}^{B_l} \alpha_k \sum_{j=1}^{V_{k,l}} \alpha_j \sum_{i=1}^{U_{j,k,l}} \bblue{\alpha_i} \mathbb{E} \bigg \Vert  \tilde{\mathbf{w}}_i^t - \tilde{\mathbf{w}}_i^{t,0} \bigg \Vert^2 \nonumber\\
   &\overset{(b)}{\bblue{\leq}} \frac{24 \beta^2}{T} \sum_{t=0}^{T-1} \sum_{l=1}^{L} \alpha_l \sum_{k=1}^{B_l} \alpha_k \sum_{j=1}^{V_{k,l}} \alpha_j \sum_{i=1}^{U_{j,k,l}} \bblue{\alpha_i} \delta_i^t \mathbb{E} \Vert \mathbf{w}_i^t \Vert^2,  \label{vCPrune}
\end{align}
where ($a$) arises from \bblue{Jensen inequality} and $(b)$ appears from the same reasoning as in $\mathrm{T}_1$.
Using similar steps, we write the following:
\begin{align}
    \mathrm{T_3} 
    &\leq \frac{24 \beta^2} {T} \sum_{l=1}^{L} \alpha_l \sum_{k=1}^{B_l} \alpha_k \sum_{j=1}^{V_{k,l}} \bblue{\alpha_j} \sum_{i=1}^{U_{j,k,l}} \bblue{\alpha_i} \delta_i^t \mathbb{E} \Vert \mathbf{w}_i^t \Vert^2, \label{sBSPrune} \\
    \mathrm{T}_4  
    &\leq \frac{24 \beta^2} {T} \sum_{l=1}^{L} \alpha_l \sum_{k=1}^{B_l} \bblue{\alpha_k} \sum_{j=1}^{V_{k,l}} \bblue{\alpha_j} \sum_{i=1}^{U_{j,k,l}} \bblue{\alpha_i} \delta_i^t \mathbb{E} \Vert \mathbf{w}_i^t \Vert^2, \label{mBSPrune} \\
    \mathrm{T}_5 
    & \leq \frac{24 \beta^2} {T} \sum_{l=1}^{L} \bblue{\alpha_l} \sum_{k=1}^{B_l} \bblue{\alpha_k} \sum_{j=1}^{V_{k,l}} \bblue{\alpha_j} \sum_{i=1}^{U_{j,k,l}} \bblue{\alpha_i} \delta_i^t \mathbb{E} \Vert \mathbf{w}_i^t \Vert^2. \label{globalPrune}
\end{align}

Plugging the above values into (\ref{mainEq3}), we get
\begin{subequations}
\label{mainEq4}
\begin{align}
    &\frac{1}{T}\rs \sum_{t=0}^{T-1} \rs \mathbb{E} \left[\left\Vert \nabla \tilde{f}\left(\tilde{\bar{\mathbf{w}}}^{t,0} \right) \right\Vert^2\right] 
    \leq \frac{2\left(f(\tilde{\bar{\mathbf{w}}}^{0}) - \mathbb{E} [f(\bar{\mathbf{w}}^{T})] \right)}{\eta T} + 2 \beta \eta \sigma^2 \sum_{l=1}^{L} \alpha_{l}^2 \sum_{k=1}^{B_l} \alpha_k^2 \sum_{j=1}^{V_{k,l}} \alpha_j^2  \sum_{i=1}^{U_{j,k,l}} \alpha_i^2 + \nonumber\\
    &\Mquad 2\beta \eta G^2 \cdot \underbrace{\frac{1}{T} \sum_{t=0}^{T-1} \sum_{l=1}^{L} \alpha_{l}^2 \sum_{k=1}^{B_l} \alpha_k^2 \sum_{j=1}^{V_{k,l}} \alpha_j^2  \sum_{i=1}^{U_{j,k,l}} \alpha_i^2 \bigg(\frac{1 - p_i^t}{p_i^t}\bigg)}_{\varphi_\mathrm{w,0} (\pmb{\delta},\pmb{\mathrm{f}}, \pmb{\mathrm{P}})} + 12 \beta^2 \cdot \big( \mathrm{L}_1 + 2 \big\{\mathrm{L}_2 + \mathrm{L}_3 + \mathrm{L}_4 \big\} \big) + \nonumber\\
    &\Mquad \bblue{96} \beta^2 \cdot \underbrace{\frac{1}{T} \sum_{t=0}^{T-1}  \sum_{l=1}^{L} \alpha_l \sum_{k=1}^{B_l} \alpha_k \sum_{j=1}^{V_{k,l}} \alpha_j \sum_{i=1}^{U_{j,k,l}} \alpha_i \delta_i^t \Vert \mathbf{w}_i^t \Vert^2}_{\mathrm{e}_0 (\pmb{\delta})} . \tag{\ref{mainEq4}}
\end{align}
\end{subequations}

When $\alpha_i = \frac{1}{U_{j,k,l}}$, $\alpha_j = \frac{1}{V_{k,l}}$, $\alpha_k = \frac{1}{B_{l}}$ and $\alpha_l = \frac{1}{L}$, we have $\sum_{l=1}^{L} \alpha_{l}^2 \sum_{k=1}^{B_l} \alpha_k^2 \sum_{j=1}^{V_{k,l}} \alpha_j^2 \sum_{i=1}^{U_{j,k,l}} \alpha_i^2 = \frac{1}{\mathrm{U}}$.
As such, using the fact that $f(\bar{\mathbf{w}}^{T}) \geq f_{\mathrm{inf}}$ and definition of $\theta_{\mathrm{PHFL}}$, we get
\begin{align}
\label{theorem_1_eqn_FinalTerm}
\theta_\mathrm{{PHFL}} \leq
    & \mathcal{O} \bigg( \frac{f(\tilde{\bar{\mathbf{w}}}^0) - f_{\mathrm{inf}} } {\eta T} \bigg)  + \mathcal{O} \bigg( \frac{\beta \eta \sigma^2}{\mathrm{U}} \bigg) + \mathcal{O} \big(\delta^{\mathrm{th}} \beta^2 D^2 \big) + \mathcal{O} \big( \beta \eta G^2 \cdot \varphi_\mathrm{w,0}(\pmb{\delta},\pmb{\mathrm{f}}, \pmb{\mathrm{P}}) \big) + \mathcal{O} \big( \beta^2 \big[\mathrm{L}_1 + \mathrm{L}_2 + \mathrm{L}_3 + \mathrm{L}_4 \big] \big).
\end{align}

\end{proof}

\section{Proof of Lemma \ref{Lemma2}}
\label{proofLemma2}

\setcounter{Lemma}{0}
\begin{Lemma}
\label{Lemma2}
When $\eta \leq 1/[2 \sqrt{10} \kappa_0 \beta]$, the average difference between the VC and local model parameters, i.e., the $\mathrm{L}_1$ term of (\ref{theorem_1_eqn_Sup}), is upper bounded as 
\begin{align}
\label{lem2_Sup}
    &[\beta^2/T] \sum\nolimits_{t=0}^T \sum\nolimits_{l=1}^{L} \alpha_l \sum\nolimits_{k=1}^{B_l} \alpha_k \sum\nolimits_{j=1}^{V_{k,l}} \alpha_j \sum\nolimits_{i=1}^{U_{j,k,l}} \alpha_i \mathbb{E} \left\Vert \bar{\mathbf{w}}_j^t - \tilde{\mathbf{w}}_i^t \right\Vert^2 \nonumber\\
    &\qquad \leq \mathcal{O} \big( \kappa_0 \eta^2 \beta^2 \sigma^2 \big) + \mathcal{O} \big( \kappa_0^2 \eta^2 \beta^2 \epsilon_{\mathrm{vc}}^2 \big) + \mathcal{O} \big( \kappa_0 \eta^2 \beta^2 G^2 \cdot \varphi_{\mathrm{w, L}_1} (\pmb{\delta},\pmb{\mathrm{f}}, \pmb{\mathrm{P}}) \big) + \mathcal{O} \big(\delta^{\mathrm{th}}  \beta^2 D^2 \big),
\end{align}
where $\varphi_{\mathrm{w, L}_1} (\pmb{\delta},\pmb{\mathrm{f}}, \pmb{\mathrm{P}}) = [1/T] \sum_{l=1}^{L} \alpha_l \sum_{k=1}^{B_l} \alpha_k \sum_{j=1}^{V_{k,l}} \alpha_j \sum_{i=1}^{U_{j,k,l}} \alpha_i \sum_{t=0}^{T-1} \left(1/p_i^t - 1\right)$.
\end{Lemma}

\begin{proof}
\begin{align}
\label{mainLem2}
    &\frac{1}{T} \sum_{t=0}^{T-1} \sum_{l=1}^{L} \alpha_l \sum_{k=1}^{B_l} \alpha_k \sum_{j=1}^{V_{k,l}} \alpha_j \sum_{i=1}^{U_{j,k,l}} \alpha_i \mathbb{E} \left\Vert \bar{\mathbf{w}}_j^t - \tilde{\mathbf{w}}_i^{t} \right\Vert^2 \nonumber\\
    &= \frac{1}{T} \sum_{t=0}^{T-1} \sum_{l=1}^{L} \alpha_l \sum_{k=1}^{B_l} \alpha_k \sum_{j=1}^{V_{k,l}} \alpha_j \sum_{i=1}^{U_{j,k,l}} \alpha_i \mathbb{E} \Bigg\Vert \tilde{\mathbf{w}}_j^{\Bar{t}_0, 0} - \eta \sum_{i'=1}^{U_{j,k,l}} \alpha_{i'} \sum_{\tau=\Bar{t}_0}^{t-1} \tilde{g} \left(\tilde{\mathbf{w}}_{i'}^{\tau,0} \right) \frac{\pmb{1}_{i'}^\tau}{p_{i'}^\tau} -  \Big(\tilde{\mathbf{w}}_i^{\Bar{t}_0,0} - \eta \sum_{\tau= \Bar{t}_0}^{t-1} \tilde{g} \big(\tilde{\mathbf{w}}_i^{\tau,0} \big) \frac{\pmb{1}_i^\tau}{p_i^\tau}  \Big) \Bigg\Vert^2, \nonumber\\
    &\overset{(a)}{\leq} \frac{2}{T} \sum_{t=0}^{T-1} \sum_{l=1}^L \alpha_l \sum_{k=1}^{B_l} \alpha_k \sum_{j=1}^{V_{k,l}} \alpha_j \sum_{i=1}^{U_{j,k,l}} \alpha_i \mathbb{E} \Bigg\Vert  \tilde{\mathbf{w}}_j^{\Bar{t}_0, 0} - \tilde{\mathbf{w}}_i^{\Bar{t}_0,0} \Bigg\Vert^2  + \nonumber \\
    &\qquad \frac{2\eta^2}{T} \sum_{t=0}^{T-1}  \sum_{l=1}^{L} \alpha_l \sum_{k=1}^{B_l}  \alpha_k  \sum_{j=1}^{V_{k,l}} \alpha_j  \sum_{i=1}^{U_{j,k,l}} \alpha_i \mathbb{E} \Bigg\Vert \sum_{\tau= \Bar{t}_0}^{t-1} \tilde{g} \big(\tilde{\mathbf{w}}_i^{\tau,0} \big) \frac{\pmb{1}_i^\tau}{p_i^\tau} -  \sum_{i'=1}^{U_{j,k,l}} \alpha_{i'} \sum_{\tau=\Bar{t}_0}^{t-1} \tilde{g}\Big(\tilde{\mathbf{w}}_{i'}^{\tau,0}\Big) \frac{\pmb{1}_{i'}^\tau}{p_{i'}^\tau} \Bigg\Vert^2, \rs
\end{align}
where $\Bar{t}_0 = [\{(m \kappa_3 + t_3)\kappa_2 + t_2\} \kappa_1 + t_1] \kappa_0 $ and $(a)$ comes from $\Vert \sum_{i=1}^n \mathbf{a}_i \Vert^2 \leq n \sum_{i=1}^n \Vert \mathbf{a}_i \Vert^2$.

Note that the first term in (\ref{mainLem2}) comes from the VC receiving a weighted combination of the pruned models of its associated UEs.
For this term, we have
\begin{align}
    & 2 \sum_{l=1}^L \alpha_l \sum_{k=1}^{B_l} \alpha_k \sum_{j=1}^{V_{k,l}} \alpha_j \sum_{i=1}^{U_{j,k,l}} \alpha_i \mathbb{E} \bigg\Vert  \tilde{\mathbf{w}}_j^{\Bar{t}_0, 0} - \tilde{\mathbf{w}}_i^{\Bar{t}_0,0} \bigg\Vert^2, \nonumber \\
    &\overset{(a)}{=} 2 \sum_{l=1}^L \alpha_l \sum_{k=1}^{B_l} \alpha_k \sum_{j=1}^{V_{k,l}} \alpha_j \sum_{i=1}^{U_{j,k,l}} \alpha_i \mathbb{E} \bigg\Vert \left( \tilde{\mathbf{w}}_j^{\Bar{t}_0, 0} - \mathbf{w}_j^{\Bar{t}_0} \right)  + \left( \mathbf{w}_{i}^{\Bar{t}_0} - \tilde{\mathbf{w}}_i^{\Bar{t}_0,0} \right) \bigg\Vert^2,\nonumber \\
    &\overset{(b)}{\leq} 4 \sum_{l=1}^L \alpha_l \sum_{k=1}^{B_l} \alpha_k \sum_{j=1}^{V_{k,l}} \alpha_j \sum_{i=1}^{U_{j,k,l}} \alpha_i \mathbb{E} \bigg\Vert \tilde{\mathbf{w}}_j^{\Bar{t}_0, 0} - \mathbf{w}_j^{\Bar{t}_0} \bigg\Vert^2 + 4 \sum_{l=1}^L \alpha_l \sum_{k=1}^{B_l} \alpha_k \sum_{j=1}^{V_{k,l}} \alpha_j \sum_{i=1}^{U_{j,k,l}} \alpha_i \mathbb{E} \bigg\Vert \mathbf{w}_{i}^{\Bar{t}_0} - \tilde{\mathbf{w}}_i^{\Bar{t}_0,0} \bigg\Vert^2, \nonumber\\
    &\leq 4 \sum_{l=1}^L \alpha_l \sum_{k=1}^{B_l} \alpha_k \sum_{j=1}^{V_{k,l}} \alpha_j \bigg\{ \mathbb{E} \bigg\Vert \tilde{\mathbf{w}}_j^{\Bar{t}_0, 0} - \mathbf{w}_j^{\Bar{t}_0} \bigg\Vert^2  +  \sum_{i=1}^{U_{j,k,l}} \alpha_i \delta_i^{\Bar{t}_0} \mathbb{E} \left\Vert \mathbf{w}_{i}^{\Bar{t}_0} \right\Vert^2 \bigg\},\nonumber\\
    &\leq \bblue{8} \sum_{l=1}^L \alpha_l \sum_{k=1}^{B_l} \alpha_k \sum_{j=1}^{V_{k,l}} \alpha_j \sum_{i=1}^{U_{j,k,l}} \alpha_i\delta_i^{\Bar{t}_0} \mathbb{E} \left\Vert \mathbf{w}_{i}^{\Bar{t}_0} \right\Vert^2, \label{firstTermLemma2_0}  
\end{align}    
where $(a)$ is true since $\mathbf{w}_j^{\Bar{t}_0} = \mathbf{w}_{i}^{\Bar{t}_0}$ and $(b)$ comes from $\Vert \sum_{i=1}^n \mathbf{a}_i \Vert^2 \leq n \sum_{i=1}^n \Vert \mathbf{a}_i \Vert^2$.

As such, the first term is bounded as 
\begin{align}
    &\frac{2}{T} \sum_{t=0}^{T-1} \sum_{l=1}^L \alpha_l \sum_{k=1}^{B_l} \alpha_k \sum_{j=1}^{V_{k,l}} \alpha_j \sum_{i=1}^{U_{j,k,l}} \alpha_i \mathbb{E} \bigg\Vert  \tilde{\mathbf{w}}_j^{\Bar{t}_0, 0} - \tilde{\mathbf{w}}_i^{\Bar{t}_0,0} \bigg\Vert^2 
    \leq \frac{\bblue{8}}{T} \sum_{t=0}^{T-1} \sum_{l=1}^L \alpha_l \sum_{k=1}^{B_l} \alpha_k \sum_{j=1}^{V_{k,l}} \alpha_j \sum_{i=1}^{U_{j,k,l}} \alpha_i \delta_i^{\lfloor t/\kappa_0 \rfloor} \mathbb{E} \left\Vert \mathbf{w}_{i}^{\lfloor t/\kappa_0 \rfloor} \right\Vert^2. \label{firstTermLemma2}  
\end{align}

For the second term of (\ref{mainLem2}), we have 
\begin{align}
\label{mainLemSecondTerm}
    &\frac{2\eta^2}{T} \sum_{t=0}^{T-1}  \sum_{l=1}^{L} \alpha_l \sum_{k=1}^{B_l}  \alpha_k  \sum_{j=1}^{V_{k,l}} \alpha_j  \sum_{i=1}^{U_{j,k,l}} \alpha_i \mathbb{E} \Bigg\Vert \sum_{\tau= \Bar{t}_0}^{t-1} \bigg[ \tilde{g} \big(\tilde{\mathbf{w}}_i^{\tau,0} \big) \frac{\pmb{1}_i^\tau}{p_i^\tau} -  \sum_{i'=1}^{U_{j,k,l}} \alpha_{i'}  \tilde{g}\Big(\tilde{\mathbf{w}}_{i'}^{\tau,0}\Big) \frac{\pmb{1}_{i'}^\tau}{p_{i'}^\tau} \bigg] \Bigg\Vert^2 \nonumber\\
    &= \frac{2\eta^2}{T} \sum_{t=0}^{T-1}  \sum_{l=1}^{L} \alpha_l \sum_{k=1}^{B_l}  \alpha_k  \sum_{j=1}^{V_{k,l}} \alpha_j  \sum_{i=1}^{U_{j,k,l}} \alpha_i \mathbb{E} \Bigg\Vert \sum_{\tau= \Bar{t}_0}^{t-1} \bigg[ \tilde{g} \big(\tilde{\mathbf{w}}_i^{\tau,0} \big) \frac{\pmb{1}_i^\tau}{p_i^\tau} \pm \nabla \tilde{f}_i \big(\tilde{\mathbf{w}}_i^{\tau,0} \big) \pm \sum_{i'=1}^{U_{j,k,l}} \alpha_{i'} \nabla \tilde{f}_i\Big(\tilde{\mathbf{w}}_{i'}^{\tau,0}\Big) - \sum_{i'=1}^{U_{j,k,l}} \alpha_{i'}  \tilde{g}\Big(\tilde{\mathbf{w}}_{i'}^{\tau,0}\Big) \frac{\pmb{1}_{i'}^\tau}{p_{i'}^\tau}  \bigg] \Bigg\Vert^2 \nonumber\\
    &\leq \frac{4\eta^2}{T} \sum_{t=0}^{T-1}  \sum_{l=1}^{L} \alpha_l \sum_{k=1}^{B_l}  \alpha_k  \sum_{j=1}^{V_{k,l}} \alpha_j  \sum_{i=1}^{U_{j,k,l}} \alpha_i \mathbb{E} \Bigg\Vert \sum_{\tau= \Bar{t}_0}^{t-1} \bigg[ \bigg( \tilde{g} \big(\tilde{\mathbf{w}}_i^{\tau,0} \big) \frac{\pmb{1}_i^\tau}{p_i^\tau} - \nabla \tilde{f}_i \big(\tilde{\mathbf{w}}_i^{\tau,0} \big) \bigg)  - \sum_{i'=1}^{U_{j,k,l}} \alpha_{i'}  \bigg( \tilde{g}\Big(\tilde{\mathbf{w}}_{i'}^{\tau,0}\Big) \frac{\pmb{1}_{i'}^\tau}{p_{i'}^\tau} - \nabla \tilde{f}_i\Big(\tilde{\mathbf{w}}_{i'}^{\tau,0}\Big) \bigg) \bigg] \Bigg\Vert^2 + \nonumber\\
    &\Squad \frac{4\eta^2}{T} \sum_{t=0}^{T-1}  \sum_{l=1}^{L} \alpha_l \sum_{k=1}^{B_l}  \alpha_k  \sum_{j=1}^{V_{k,l}} \alpha_j  \sum_{i=1}^{U_{j,k,l}} \alpha_i \mathbb{E} \Bigg\Vert \sum_{\tau= \Bar{t}_0}^{t-1} \bigg[ \nabla \tilde{f}_i \big(\tilde{\mathbf{w}}_i^{\tau,0} \big) - \sum_{i'=1}^{U_{j,k,l}} \alpha_{i'}  \nabla \tilde{f}_i\Big(\tilde{\mathbf{w}}_{i'}^{\tau,0}\Big) \bigg] \Bigg\Vert^2.
\end{align}

We calculate the first and the second terms of (\ref{mainLemSecondTerm}) in Lemma {\ref{lemma20_1}} and Lemma \ref{lemma20_2}, respectively.

\setcounter{Lemma}{4}
\begin{Lemma}
\label{lemma20_1}
\begin{align}
    & \frac{4\eta^2}{T} \sum_{t=0}^{T-1}  \sum_{l=1}^{L} \alpha_l \sum_{k=1}^{B_l}  \alpha_k  \sum_{j=1}^{V_{k,l}} \alpha_j  \sum_{i=1}^{U_{j,k,l}} \alpha_i \mathbb{E} \Bigg\Vert \sum_{\tau= \Bar{t}_0}^{t-1} \bigg[ \bigg( \tilde{g} \big(\tilde{\mathbf{w}}_i^{\tau,0} \big) \frac{\pmb{1}_i^\tau}{p_i^\tau} - \nabla \tilde{f}_i \big(\tilde{\mathbf{w}}_i^{\tau,0} \big) \bigg) - \sum_{i'=1}^{U_{j,k,l}} \alpha_{i'}  \bigg( \tilde{g}\Big(\tilde{\mathbf{w}}_{i'}^{\tau,0}\Big) \frac{\pmb{1}_{i'}^\tau}{p_{i'}^\tau} -  \nabla \tilde{f}_i\Big(\tilde{\mathbf{w}}_{i'}^{\tau,0}\Big) \bigg) \bigg] \Bigg\Vert^2 \nonumber\\
    &\leq 8 \kappa_0 \eta^2 \sigma^2 + 8 \kappa_0 \eta^2 G^2 \cdot \varphi_{\mathrm{w, L}_1},
\end{align}
where $\varphi_{\mathrm{w, L}_1} = \frac{1}{T} \sum_{l=1}^{L} \alpha_l \sum_{k=1}^{B_l} \alpha_k \sum_{j=1}^{V_{k,l}} \alpha_j \sum_{i=1}^{U_{j,k,l}} \alpha_i \sum_{t=0}^{T-1} \left(\frac{1-p_i^t}{p_i^t}\right)$.
\end{Lemma}

\begin{Lemma}
\label{lemma20_2}
\begin{align}
    & \frac{4\eta^2}{T} \sum_{t=0}^{T-1}  \sum_{l=1}^{L} \alpha_l \sum_{k=1}^{B_l}  \alpha_k  \sum_{j=1}^{V_{k,l}} \alpha_j  \sum_{i=1}^{U_{j,k,l}} \alpha_i \mathbb{E} \Bigg\Vert \sum_{\tau= \Bar{t}_0}^{t-1} \bigg[ \nabla \tilde{f}_i \big(\tilde{\mathbf{w}}_i^{\tau,0} \big) - \sum_{i'=1}^{U_{j,k,l}} \alpha_{i'}  \nabla \tilde{f}_{i'} \Big(\tilde{\mathbf{w}}_{i'}^{\tau,0}\Big) \bigg] \Bigg\Vert^2 \nonumber\\
    &\leq 20\kappa_0^2 \eta^2 \epsilon_{\mathrm{vc}}^2 + 40 \kappa_0^2 \eta^2 \beta^2 \cdot \bar{\mathrm{e}}_{\pmb{\delta}} + \frac{40 \kappa_0^2 \eta^2 \beta^2}{T} \sum_{t=0}^{T-1} \sum_{l=1}^{L} \alpha_l \sum_{k=1}^{B_l}  \alpha_k  \sum_{j=1}^{V_{k,l}} \alpha_j  \sum_{i=1}^{U_{j,k,l}} \alpha_i \mathbb{E} \left\Vert \bar{\mathbf{w}}_j^t - \tilde{\mathbf{w}}_i^{t} \right\Vert^2,
\end{align}
where $\bar{\mathrm{e}}_{\pmb{\delta}} = \frac{1}{T} \sum_{t=0}^{T-1} \sum_{l=1}^{L} \alpha_l \sum_{k=1}^{B_l}  \alpha_k  \sum_{j=1}^{V_{k,l}} \alpha_j  \sum_{i=1}^{U_{j,k,l}} \alpha_i \delta_i^t \mathbb{E} \Vert \mathbf{w}_i^t \Vert^2$.
\end{Lemma}

Using Lemma \ref{lemma20_1} and Lemma \ref{lemma20_2}, we write 
\begin{align}
\label{mainLem2_2}
    &\frac{1}{T} \sum_{t=0}^T \sum_{l=1}^{L} \alpha_l \sum_{k=1}^{B_l} \alpha_k \sum_{j=1}^{V_{k,l}} \alpha_j \sum_{i=1}^{U_{j,k,l}} \alpha_i \mathbb{E} \left\Vert \bar{\mathbf{w}}_j^t - \tilde{\mathbf{w}}_i^t \right\Vert^2 
    \leq \frac{8 \kappa_0 \eta^2 \sigma^2 + 20\kappa_0^2 \eta^2 \epsilon_{\mathrm{vc}}^2 + 8 \kappa_0 \eta^2 G^2 \cdot \varphi_{\mathrm{w, L}_1} + 40 \kappa_0^2 \eta^2 \beta^2 \cdot \bar{\mathrm{e}}_{\pmb{\delta}} + \bblue{8} \cdot \mathrm{e}_{\mathrm{p,L}_1}  } { 1 - 40\kappa_0^2 \eta^2 \beta^2}, 
\end{align}
where $\mathrm{e}_{\mathrm{p,L}_1} = \frac{1}{T} \sum_{t=0}^{T-1}\sum_{l=1}^L \alpha_l \sum_{k=1}^{B_l} \alpha_k \sum_{j=1}^{V_{k,l}} \alpha_j \sum_{i=1}^{U_{j,k,l}} \alpha_i \delta_i^{\lfloor t/\kappa_0 \rfloor} \mathbb{E} \left\Vert \mathbf{w}_{i}^{\lfloor t/\kappa_0 \rfloor} \right\Vert^2$.

Now multiplying both sides of (\ref{mainLem2_2}) by $12\beta^2$, we get
\begin{align}
    &\frac{12 \beta^2}{T} \sum_{t=0}^T \sum_{l=1}^{L} \alpha_l \sum_{k=1}^{B_l} \alpha_k \sum_{j=1}^{V_{k,l}} \alpha_j \sum_{i=1}^{U_{j,k,l}} \alpha_i \mathbb{E} \left\Vert \bar{\mathbf{w}}_j^t - \tilde{\mathbf{w}}_i^t \right\Vert^2 \nonumber\\
    &\leq \frac{96 \kappa_0 \eta^2 \beta^2 \sigma^2 + 240 \kappa_0^2 \eta^2 \beta^2 \epsilon_{\mathrm{vc}}^2 + 96 \kappa_0 \eta^2 \beta^2 G^2 \cdot \varphi_{\mathrm{w, L}_1} + 480 \kappa_0^2 \eta^2 \beta^4 \cdot \bar{\mathrm{e}}_{\pmb{\delta}} + \bblue{96} \beta^2 \cdot \mathrm{e}_{\mathrm{p,L}_1}  } { 1 - 40\kappa_0^2 \eta^2 \beta^2}.
\end{align}

When $\eta \leq \frac{1}{2\sqrt{10} \kappa_0 \beta}$, we have $1 - 40\kappa_0^2 \eta^2 \beta^2 \geq 1$, and the previous assumption of $\eta \leq \frac{1}{\beta}$ is automatically satisfied.
As such, we write 
\begin{align}
    &\frac{12 \beta^2}{T} \sum_{t=0}^T \sum_{l=1}^{L} \alpha_l \sum_{k=1}^{B_l} \alpha_k \sum_{j=1}^{V_{k,l}} \alpha_j \sum_{i=1}^{U_{j,k,l}} \alpha_i \mathbb{E} \left\Vert \bar{\mathbf{w}}_j^t - \tilde{\mathbf{w}}_i^t \right\Vert^2 \nonumber\\
    &\leq 96 \kappa_0 \eta^2 \beta^2 \sigma^2 + 240 \kappa_0^2 \eta^2 \beta^2 \epsilon_{\mathrm{vc}}^2 + 96 \kappa_0 \eta^2 \beta^2 G^2 \cdot \varphi_{\mathrm{w, L}_1} + 480 \kappa_0^2 \eta^2 \beta^4 \cdot \bar{\mathrm{e}}_{\pmb{\delta}} + \bblue{96} \beta^2 \cdot \mathrm{e}_{\mathrm{p,L}_1} \nonumber \\
    &\approx \mathcal{O} \big( \kappa_0 \eta^2 \beta^2 \sigma^2 \big) + \mathcal{O} \big( \kappa_0^2 \eta^2 \beta^2 \epsilon_{\mathrm{vc}}^2 \big) + \mathcal{O} \big( \kappa_0 \eta^2 \beta^2 G^2 \cdot \varphi_{\mathrm{w, L}_1} \big) + \mathcal{O} \big(\delta^{\mathrm{th}}  \beta^2 D^2 \big).
\end{align}

\subsection{Missing Proof of Lemma \ref{lemma20_1}}
\begin{align}
\label{lemma2secondTerm0_SecondTerm}
    &4 \eta^2 \sum_{l=1}^{L} \alpha_l \sum_{k=1}^{B_l} \alpha_k \sum_{j=1}^{V_{k,l}} \alpha_j \sum_{i=1}^{U_{j,k,l}} \alpha_i \mathbb{E}  \Bigg\Vert \sum_{\tau= \Bar{t}_0}^{t-1} \Bigg\{ \left(\tilde{g} \big(\tilde{\mathbf{w}}_i^{\tau,0} \big)\frac{\pmb{1}_i^\tau}{p_i^\tau} - \nabla \tilde{f}_i (\tilde{\mathbf{w}}_i^{\tau,0})\right) - \sum_{i'=1}^{U_{j,k,l}}  \alpha_{i'} \left( \tilde{g} \left(\tilde{\mathbf{w}}_{i'}^{\tau,0} \right) \frac{\pmb{1}_{i'}^\tau}{p_{i'}^\tau} - \nabla \tilde{f}_{i'} (\tilde{\mathbf{w}}_{i'}^{\tau,0} ) \right)  \Bigg\} \Bigg\Vert^2, \nonumber\\
    &\overset{(a)}{=} 4 \eta^2 \sum_{l=1}^{L} \alpha_l \sum_{k=1}^{B_l} \alpha_k \sum_{j=1}^{V_{k,l}} \alpha_j \Bigg\{ \sum_{i=1}^{U_{j,k,l}} \alpha_i \mathbb{E}  \Bigg\Vert \sum_{\tau= \Bar{t}_0}^{t-1} \left(\tilde{g} \big(\tilde{\mathbf{w}}_i^{\tau,0} \big) \frac{\pmb{1}_i^\tau}{p_i^\tau} - \nabla \tilde{f}_i (\tilde{\mathbf{w}}_i^{\tau,0})\right) \Bigg\Vert^2 - \mathbb{E}  \Bigg\Vert \sum_{\tau = \Bar{t}_0}^{t-1} \sum_{i'=1}^{U_{j,k,l}} \alpha_{i'} \left( \tilde{g} \left(\tilde{\mathbf{w}}_{i'}^{\tau,0} \right) \frac{\pmb{1}_{i'}^\tau}{p_{i'}^\tau} - \nabla \tilde{f}_{i'} (\tilde{\mathbf{w}}_{i'}^{\tau})\right) \Bigg\Vert^2 \Bigg\}, \nonumber\\
    &\overset{(b)}{=} 4 \eta^2 \sum_{l=1}^{L} \alpha_l \sum_{k=1}^{B_l} \alpha_k \sum_{j=1}^{V_{k,l}} \alpha_j \Bigg\{ \sum_{i=1}^{U_{j,k,l}} \alpha_i \sum_{\tau= \Bar{t}_0}^{t-1} \mathbb{E}  \Bigg\Vert \tilde{g} \big(\tilde{\mathbf{w}}_i^{\tau,0} \big)\frac{\pmb{1}_i^\tau}{p_i^\tau} \pm \Tilde{g} \big( \Tilde{\mathbf{w}}_i^{\tau,0} \big) - \nabla \tilde{f}_i (\tilde{\mathbf{w}}_i^{\tau,0}) \Bigg\Vert^2 -  \nonumber\\
    & \Mquad \qquad \sum_{\tau= \Bar{t}_0}^{t-1} \mathbb{E}  \Bigg\Vert \sum_{i'=1}^{U_{j,k,l}}  \alpha_{i'} \left( \tilde{g} \left(\tilde{\mathbf{w}}_{i'}^{\tau,0} \right) \frac{\pmb{1}_{i'}^\tau}{p_{i'}^\tau} \pm \tilde{g} \left(\tilde{\mathbf{w}}_{i'}^{\tau,0} \right) - \nabla \tilde{f}_{i'} (\tilde{\mathbf{w}}_{i'}^{\tau,0} ) \right) \Bigg\Vert^2 \Bigg\}, \nonumber\\
    &\overset{(c)}{=} 8 \eta^2 \sum_{l=1}^{L} \alpha_l \sum_{k=1}^{B_l} \alpha_k \sum_{j=1}^{V_{k,l}} \alpha_j \sum_{i=1}^{U_{j,k,l}} \alpha_i \sum_{\tau= \Bar{t}_0}^{t-1} \mathbb{E} \Bigg\{ \Bigg\Vert \left(\frac{\pmb{1}_i^\tau}{p_i^\tau} - 1\right)\tilde{g} \big(\tilde{\mathbf{w}}_i^{\tau,0} \big)\Bigg\Vert^2 + \Bigg\Vert \Tilde{g} \big( \Tilde{\mathbf{w}}_i^{\tau,0} \big) - \nabla \tilde{f}_i (\tilde{\mathbf{w}}_i^{\tau,0}) \Bigg\Vert^2 \Bigg\} - \nonumber\\
    &8 \eta^2 \sum_{l=1}^{L} \alpha_l \sum_{k=1}^{B_l} \alpha_k \sum_{j=1}^{V_{k,l}} \alpha_j \sum_{i=1}^{U_{j,k,l}} \alpha_i \sum_{\tau= \Bar{t}_0}^{t-1} \sum_{i'=1}^{U_{j,k,l}} \rs \rs \rs \alpha_{i'}^2 \mathbb{E} \Bigg\{ \Bigg\Vert \left( \frac{\pmb{1}_{i'}^\tau}{p_{i'}^\tau} - 1 \right) \tilde{g} \left(\tilde{\mathbf{w}}_{i'}^{\tau,0} \right)\Bigg\Vert^2 + \Bigg\Vert \tilde{g} \left(\tilde{\mathbf{w}}_{i'}^{\tau,0} \right) - \nabla \tilde{f}_{i'} (\tilde{\mathbf{w}}_{i'}^{\tau,0} ) \Bigg\Vert^2 \Bigg\}, \nonumber\\
    & \overset{(d)}{\leq} 8 \eta^2 \sum_{l=1}^{L} \alpha_l \sum_{k=1}^{B_l} \alpha_k \sum_{j=1}^{V_{k,l}} \alpha_j \sum_{i=1}^{U_{j,k,l}} \alpha_i \sum_{\tau= \Bar{t}_0}^{t-1} \Bigg\{ \left(\frac{1-p_i^\tau}{p_i^\tau}\right) \mathbb{E} \left\Vert \tilde{g} (\tilde{\mathbf{w}}_i^{\tau,0}) \right\Vert^2 + \sigma^2 \Bigg\} - \nonumber\\
    &\Squad 8 \eta^2 \sum_{l=1}^{L} \alpha_l \sum_{k=1}^{B_l} \alpha_k \sum_{j=1}^{V_{k,l}} \alpha_j \sum_{\tau= \Bar{t}_0}^{t-1} \sum_{i=1}^{U_{j,k,l}} \alpha_i^2 \Bigg\{\left( \frac{1 - p_i^\tau}{p_i^\tau} \right) \mathbb{E} \left \Vert \tilde{g} (\tilde{\mathbf{w}}_i^{\tau,0} ) \right\Vert^2 + \sigma^2 \Bigg\}, \nonumber\\ 
    & \overset{(e)}{\leq} 8 \kappa_0 \eta^2 \sigma^2 - 8 \eta^2 \sum\nolimits_{l=1}^{L} \alpha_l \sum\nolimits_{k=1}^{B_l} \alpha_k \sum\nolimits_{j=1}^{V_{k,l}} \alpha_j \sum\nolimits_{\tau= \Bar{t}_0}^{t-1} \sum\nolimits_{i=1}^{U_{j,k,l}} \alpha_i^2 \sigma^2 + \nonumber\\
    &\Squad 8 \eta^2 \sum\nolimits_{l=1}^{L} \alpha_l \sum\nolimits_{k=1}^{B_l} \alpha_k \sum\nolimits_{j=1}^{V_{k,l}} \alpha_j \sum\nolimits_{\tau= \Bar{t}_0}^{t-1} \sum\nolimits_{i=1}^{U_{j,k,l}} \alpha_i \left(1 - \alpha_i \right) \left([1-p_i^\tau]/p_i^\tau\right) \mathbb{E} \left\Vert \tilde{g} (\tilde{\mathbf{w}}_i^{\tau,0}) \right\Vert^2, \nonumber \\
    &\leq 8 \kappa_0 \eta^2 \sigma^2 + 8 \eta^2 \sum_{l=1}^{L} \alpha_l \sum_{k=1}^{B_l} \alpha_k \sum_{j=1}^{V_{k,l}} \alpha_j \sum_{\tau= \Bar{t}_0}^{t-1} \sum_{i=1}^{U_{j,k,l}} \alpha_i \left(\frac{1-p_i^\tau}{p_i^\tau}\right) \mathbb{E} \left\Vert \tilde{g} (\tilde{\mathbf{w}}_i^{\tau,0}) \right\Vert^2, 
\end{align}
where $(a)$ stems from the fact that $\sum_{i=1}^n p_i \Vert \mathbf{x}_i - \Bar{\mathbf{x}} \Vert^2 = \sum_{i=1}^n p_i \Vert \mathbf{x}_i \Vert^2 - \Vert \Bar{\mathbf{x}} \Vert^2$, where $\bar{\mathbf{x}} = \sum_{i=1}^n p_i \mathbf{x}_i$ for any $0 \leq p_i\leq 1$ and $\sum_{i=1}^n p_i = 1$.
Besides, $(b)$ is true due to the time independence of the SGD assumption.
Furthermore, in $(c)$ and $(d)$, we use the bounded divergence of the mini-batch gradient assumption and client independence property.

\subsection{Missing Proof of Lemma \ref{lemma20_2}}
\begin{align}
    & \frac{4\eta^2}{T} \sum_{t=0}^{T-1}  \sum_{l=1}^{L} \alpha_l \sum_{k=1}^{B_l}  \alpha_k  \sum_{j=1}^{V_{k,l}} \alpha_j  \sum_{i=1}^{U_{j,k,l}} \alpha_i \mathbb{E} \Bigg\Vert \sum_{\tau= \Bar{t}_0}^{t-1} \bigg[ \nabla \tilde{f}_i \big(\tilde{\mathbf{w}}_i^{\tau,0} \big) - \sum_{i'=1}^{U_{j,k,l}} \alpha_{i'} \nabla \tilde{f}_{i'} \Big(\tilde{\mathbf{w}}_{i'}^{\tau,0}\Big) \bigg] \Bigg\Vert^2 \nonumber\\
    &= \frac{4\kappa_0 \eta^2}{T} \sum_{t=0}^{T-1}  \sum_{l=1}^{L} \alpha_l \sum_{k=1}^{B_l}  \alpha_k  \sum_{j=1}^{V_{k,l}} \alpha_j  \sum_{i=1}^{U_{j,k,l}} \alpha_i \sum_{\tau= \Bar{t}_0}^{t-1} \mathbb{E} \Bigg\Vert \bigg( \nabla \tilde{f}_i \big(\tilde{\mathbf{w}}_i^{\tau,0} \big) - \nabla \tilde{f}_i \big(\tilde{\mathbf{w}}_i^{\tau} \big) \bigg) + \bigg(\nabla \tilde{f}_i \big(\tilde{\mathbf{w}}_i^{\tau} \big) - \nabla \tilde{f}_i \big(\bar{\mathbf{w}}_j^{\tau} \big) \bigg) + \nonumber \\
    &\Squad \bigg(\nabla \tilde{f}_i \big(\bar{\mathbf{w}}_j^{\tau} \big) - \sum_{i'=1}^{U_{j,k,l}} \alpha_{i'} \nabla \tilde{f}_{i'} \big(\bar{\mathbf{w}}_{j}^{\tau}\big) \bigg) + \bigg(\sum_{i'=1}^{U_{j,k,l}} \alpha_{i'} \nabla \tilde{f}_{i'} \big(\bar{\mathbf{w}}_{j}^{\tau}\big) - \sum_{i'=1}^{U_{j,k,l}} \alpha_{i'} \nabla \tilde{f}_{i'} \big(\tilde{\mathbf{w}}_{i}^{\tau}\big) \bigg) + \nonumber\\ 
    &\Squad \bigg( \sum_{i'=1}^{U_{j,k,l}} \alpha_{i'} \nabla \tilde{f}_{i'} \big(\tilde{\mathbf{w}}_{i}^{\tau}\big) - \sum_{i'=1}^{U_{j,k,l}} \alpha_{i'} \nabla \tilde{f}_{i'}\big(\tilde{\mathbf{w}}_{i'}^{\tau,0}\big) \bigg) \Bigg\Vert^2  \nonumber \\
    &\leq \frac{20\kappa_0^2 \eta^2}{T} \sum_{t=0}^{T-1}  \sum_{l=1}^{L} \alpha_l \sum_{k=1}^{B_l}  \alpha_k  \sum_{j=1}^{V_{k,l}} \alpha_j  \sum_{i=1}^{U_{j,k,l}} \alpha_i \bigg[ 2\beta^2 \mathbb{E} \Vert \tilde{\mathbf{w}}_i^{t} - \tilde{\mathbf{w}}_i^{t,0} \Vert^2 + \epsilon_{\mathrm{vc}}^2 + 2\beta^2 \mathbb{E} \Vert \bar{\mathbf{w}}_j^{t} - \tilde{\mathbf{w}}_i^{t} \Vert^2   \bigg] \nonumber \\
    &\leq 20\kappa_0^2 \eta^2 \epsilon_{\mathrm{vc}}^2 + 40 \kappa_0^2 \eta^2 \beta^2 \cdot \bar{\mathrm{e}}_{\pmb{\delta}} + \frac{40 \kappa_0^2 \eta^2 \beta^2}{T} \sum_{t=0}^{T-1} \sum_{l=1}^{L} \alpha_l \sum_{k=1}^{B_l}  \alpha_k  \sum_{j=1}^{V_{k,l}} \alpha_j  \sum_{i=1}^{U_{j,k,l}} \alpha_i \mathbb{E} \left\Vert \bar{\mathbf{w}}_j^t - \tilde{\mathbf{w}}_i^{t} \right\Vert^2,
\end{align}
where $\bar{\mathrm{e}}_{\pmb{\delta}} = \frac{1}{T} \sum_{t=0}^{T-1} \sum_{l=1}^{L} \alpha_l \sum_{k=1}^{B_l}  \alpha_k  \sum_{j=1}^{V_{k,l}} \alpha_j  \sum_{i=1}^{U_{j,k,l}} \alpha_i \delta_i^t \mathbb{E} \Vert \mathbf{w}_i^t \Vert^2$.
\end{proof}

\section{Proof of Lemma \ref{Lemma3}}
\setcounter{Lemma}{1}
\begin{Lemma}
\label{Lemma3}
When $\eta \leq 1/[2 \sqrt{10} \kappa_0 \kappa_1 \beta]$, the difference between the sBS model parameters and VC model parameters, i.e., the $\mathrm{L}_2$ term of (\ref{theorem_1_eqn_Sup}), is upper bounded as 
\begin{align}
\label{lem3_Sup}
    & \frac{\beta^2}{T} \sum_{t=0}^{T-1} \sum_{l=1}^{L} \alpha_l \sum_{k=1}^{B_l} \alpha_k \sum_{j=1}^{V_{k,l}} \alpha_j \mathbb{E} \left\Vert \bar{\mathbf{w}}_k^t - \bar{\mathbf{w}}_j^t \right\Vert^2 \nonumber\\ 
    &\leq \mathcal{O} \big(\beta^4 \kappa_0^4 \kappa_1^2 \eta^4 \epsilon_{\mathrm{vc}}^2 \big) + \mathcal{O} \big(\kappa_0^2 \kappa_1^2 \eta^2 \beta^2 \epsilon_{\mathrm{sbs}}^2 \big) + \mathcal{O} \big( \kappa_0 \kappa_1 \eta^2 \sigma^2 \beta^2 \big) + \mathcal{O } \big( \delta^{th} \beta^2 D^2 \big) + \nonumber\\
    &\Mquad \mathcal{O} \big( \kappa_0^3 \kappa_1^2 \beta^4 \eta^4 G^2 \cdot \varphi_{\mathrm{w, L}_1} (\pmb{\delta},\pmb{\mathrm{f}}, \pmb{\mathrm{P}}) \big) + \mathcal{O} \big( \kappa_0 \kappa_1 \beta^2 \eta^2 \cdot \varphi_{\mathrm{w, L}_2} (\pmb{\delta},\pmb{\mathrm{f}}, \pmb{\mathrm{P}}) \big),
\end{align}
where $\varphi_{\mathrm{w, L}_2} (\pmb{\delta},\pmb{\mathrm{f}}, \pmb{\mathrm{P}}) = [1/T] \sum_{t=0}^{T-1} \sum_{l=1}^{L} \alpha_l \sum_{k=1}^{B_l} \alpha_k \sum_{j=1}^{V_{k,l}} \alpha_j \sum_{i=1}^{U_{j,k,l}} \alpha_i^2 ( 1/p_i^t - 1 )$.
\end{Lemma}

\begin{proof}
\label{proofLemma3}
\begin{align}
\label{lemma3Main0}
    &\frac{1}{T} \sum_{t=0}^{T-1} \sum_{l=1}^{L} \alpha_l \sum_{k=1}^{B_l} \alpha_k \sum_{j=1}^{V_{k,l}} \alpha_j \mathbb{E} \left\Vert \bar{\mathbf{w}}_k^t - \bar{\mathbf{w}}_j^t \right\Vert^2 \nonumber\\ 
    &= \frac{1}{T} \sum_{t=0}^{T-1} \sum_{l=1}^{L} \alpha_l \sum_{k=1}^{B_l} \alpha_k \sum_{j=1}^{V_{k,l}} \alpha_j \mathbb{E} \Big\Vert \tilde{\mathbf{w}}_k^{\Bar{t}_1,0} - \eta  \sum_{\tau=\Bar{t}_1}^{t-1} \sum_{j'=1}^{V_{k,l}} \alpha_{j'} \sum_{i'=1}^{U_{j',k,l}} \alpha_{i'} \Tilde{g} \left( \tilde{\mathbf{w}}_{i'}^{\tau,0} \right) \frac{\pmb{1}_{i'}^\tau}{p_{i'}^\tau} -  \tilde{\mathbf{w}}_j^{\Bar{t}_1,0} + \eta \sum_{\tau=\Bar{t}_1}^{t-1} \sum_{i=1}^{U_{j,k,l}} \alpha_i \tilde{g} \big(\tilde{\mathbf{w}}_i^{\tau,0} \big) \frac{\pmb{1}_i^\tau}{p_i^\tau} \Big\Vert^2, \nonumber \\
    &\leq \frac{2}{T} \sum_{t=0}^{T-1} \sum_{l=1}^{L} \alpha_l \sum_{k=1}^{B_l} \alpha_k \sum_{j=1}^{V_{k,l}} \alpha_j \mathbb{E} \left\Vert \tilde{\mathbf{w}}_k^{\Bar{t}_1,0} - \tilde{\mathbf{w}}_j^{\Bar{t}_1,0} \right \Vert^2 + \\
    &\qquad \frac{2\eta^2} {T} \sum_{t=0}^{T-1} \sum_{l=1}^{L} \alpha_l \sum_{k=1}^{B_l} \alpha_k \sum_{j=1}^{V_{k,l}} \alpha_j \mathbb{E} \left \Vert \sum_{\tau=\Bar{t}_1}^{t-1} \bigg[ \sum_{i=1}^{U_{j,k,l}} \alpha_i \tilde{g} \big(\tilde{\mathbf{w}}_i^{\tau,0} \big) \frac{\pmb{1}_i^\tau}{p_i^\tau} -  \sum_{j'=1}^{V_{k,l}} \alpha_{j'} \sum_{i'=1}^{U_{j',k,l}} \alpha_{i'} \Tilde{g} \left( \tilde{\mathbf{w}}_{i'}^{\tau,0} \right) \frac{\pmb{1}_{i'}^\tau}{p_{i'}^\tau} \bigg] \right\Vert^2, \nonumber
\end{align}
where $\Bar{t}_1 = \{(m \kappa_3 + t_3)\kappa_2 + t_2\} \kappa_1 \kappa_0 $ and the inequalities in the last term arise from Jensen inequality. 

For the first term in (\ref{lemma3Main0}), we have
\begin{align}
\label{Lemma2Main0_First_Term}
    & 2 \sum_{l=1}^{L} \alpha_l \sum_{k=1}^{B_l} \alpha_k \sum_{j=1}^{V_{k,l}} \alpha_j \mathbb{E} \left\Vert \tilde{\mathbf{w}}_k^{\Bar{t}_1,0} - \tilde{\mathbf{w}}_j^{\Bar{t}_1,0} \right \Vert^2 \nonumber \\
    &\overset{(a)}{=} 2 \sum_{l=1}^{L} \alpha_l \sum_{k=1}^{B_l} \alpha_k \sum_{j=1}^{V_{k,l}} \alpha_j \mathbb{E} \left\Vert \tilde{\mathbf{w}}_k^{\Bar{t}_1,0} - \mathbf{w}_k^{\Bar{t}_1} + \mathbf{w}_j^{\Bar{t}_1} - \tilde{\mathbf{w}}_j^{\Bar{t}_1,0} \right \Vert^2, \nonumber \\
    & \leq 4 \sum_{l=1}^{L} \alpha_l \sum_{k=1}^{B_l} \alpha_k \sum_{j=1}^{V_{k,l}} \alpha_j \mathbb{E} \left\Vert \mathbf{w}_k^{\Bar{t}_1} - \tilde{\mathbf{w}}_k^{\Bar{t}_1,0} \right\Vert^2 + 4 \sum_{l=1}^{L} \alpha_l \sum_{k=1}^{B_l} \alpha_k \sum_{j=1}^{V_{k,l}} \alpha_j \mathbb{E} \left\Vert\mathbf{w}_j^{\Bar{t}_1} - \tilde{\mathbf{w}}_j^{\Bar{t}_1,0} \right \Vert^2, \nonumber \\
    &\overset{(b)}{\leq} 4 \sum_{l=1}^{L} \alpha_l \sum_{k=1}^{B_l} \alpha_k \sum_{j=1}^{V_{k,l}} \bblue{\alpha_j} \sum_{i=1}^{U_{j,k,l}} \bblue{\alpha_i} \delta_i^{\Bar{t}_1} \mathbb{E} \left\Vert \mathbf{w}_{i}^{\Bar{t}_1} \right\Vert^2 + 4 \sum_{l=1}^{L} \alpha_l \sum_{k=1}^{B_l} \alpha_k \sum_{j=1}^{V_{k,l}} \alpha_j \sum_{i=1}^{U_{j,k,l}} \bblue{\alpha_i} \delta_{i}^{\Bar{t}_1} \mathbb{E} \left\Vert \mathbf{w}_{i}^{\Bar{t}_1} \right\Vert^2, \nonumber\\
    &= \bblue{8} \sum_{l=1}^{L} \alpha_l \sum_{k=1}^{B_l} \alpha_k \sum_{j=1}^{V_{k,l}} \alpha_j \sum_{i=1}^{U_{j,k,l}} \bblue{\alpha_i} \delta_i^{\Bar{t}_1} \mathbb{E} \left\Vert \mathbf{w}_{i}^{\Bar{t}_1} \right\Vert^2,
\end{align}
where in $(a)$ we use the fact that $\mathbf{w}_k^{\Bar{t}_1} = \mathbf{w}_j^{\Bar{t}_1}$ and $(b)$ stems from following similar steps as in (\ref{sBSPrune}) and (\ref{vCPrune}).

As such we derive the upper bound of the first term of (\ref{lemma3Main0}) as
\begin{align}
\label{Lemma2Main0_First_TermBound}
    & 2 \sum_{l=1}^{L} \alpha_l \sum_{k=1}^{B_l} \alpha_k \sum_{j=1}^{V_{k,l}} \alpha_j \mathbb{E} \left\Vert \tilde{\mathbf{w}}_k^{\Bar{t}_1,0} - \tilde{\mathbf{w}}_j^{\Bar{t}_1,0} \right \Vert^2 \nonumber \\
    &\leq \frac{\bblue{8}}{T} \sum_{t=0}^{T-1} \sum_{l=1}^{L} \alpha_l \sum_{k=1}^{B_l} \alpha_k \sum_{j=1}^{V_{k,l}} \alpha_j \sum_{i=1}^{U_{j,k,l}} \bblue{\alpha_i} \delta_i^{\left\lfloor t/(\kappa_0\kappa_1)\right\rfloor} \mathbb{E} \left\Vert \mathbf{w}_{i}^{\left\lfloor t/(\kappa_0\kappa_1)\right\rfloor} \right\Vert^2 
    \approx \mathcal{O} \big( \delta^{\mathrm{th}} D^2 \big).
\end{align}

For the second term in (\ref{lemma3Main0}), we have
\begin{align}
\label{Lemma3Main0_Second_Term}
    & \frac{2\eta^2}{T} \sum_{t=0}^{T-1} \sum_{l=1}^{L} \alpha_l \sum_{k=1}^{B_l} \alpha_k \sum_{j=1}^{V_{k,l}} \alpha_j \mathbb{E} \left \Vert \sum_{\tau=\Bar{t}_1}^{t-1} \Big( \sum_{i=1}^{U_{j,k,l}} \alpha_i \tilde{g} \big(\tilde{\mathbf{w}}_i^{\tau,0} \big) \frac{\pmb{1}_i^\tau}{p_i^\tau} -  \sum_{j'=1}^{V_{k,l}} \alpha_{j'} \sum_{i'=1}^{U_{j',k,l}} \alpha_{i'} \Tilde{g} \left( \tilde{\mathbf{w}}_{i'}^{\tau,0} \right) \frac{\pmb{1}_{i'}^\tau}{p_{i'}^\tau} \Big) \right\Vert^2 \nonumber\\
    &= \frac{2\eta^2}{T} \sum_{t=0}^{T-1} \sum_{l=1}^{L} \alpha_l \sum_{k=1}^{B_l} \alpha_k \sum_{j=1}^{V_{k,l}} \alpha_j \mathbb{E} \bigg \Vert \sum_{\tau=\Bar{t}_1}^{t-1} \bigg[  \bigg\{ \sum_{i=1}^{U_{j,k,l}} \alpha_i \tilde{g} \big(\tilde{\mathbf{w}}_i^{\tau,0} \big) \frac{\pmb{1}_i^\tau}{p_i^\tau} - \sum_{i=1}^{U_{j,k,l}} \alpha_i \nabla \tilde{f}_i \big(\tilde{\mathbf{w}}_i^{\tau,0} \big) \bigg\} +  \nonumber\\
    &\Squad \bigg\{ \sum_{j'=1}^{V_{k,l}} \alpha_{j'} \sum_{i'=1}^{U_{j',k,l}} \alpha_{i'} \nabla \Tilde{f}_{i'} \left( \tilde{\mathbf{w}}_{i'}^{\tau,0} \right) - \sum_{j'=1}^{V_{k,l}} \alpha_{j'} \sum_{i'=1}^{U_{j',k,l}} \alpha_{i'} \Tilde{g} \left( \tilde{\mathbf{w}}_{i'}^{\tau,0} \right) \frac{\pmb{1}_{i'}^\tau}{p_{i'}^\tau} \bigg\} + \nonumber\\
    &\Squad \bigg\{\sum_{i=1}^{U_{j,k,l}} \alpha_i \nabla \tilde{f}_i \big(\tilde{\mathbf{w}}_i^{\tau,0} \big) - \sum_{j'=1}^{V_{k,l}} \alpha_{j'} \sum_{i'=1}^{U_{j',k,l}} \alpha_{i'} \nabla \Tilde{f}_{i'} \left( \tilde{\mathbf{w}}_{i'}^{\tau,0} \right) \bigg\} \bigg] \bigg\Vert^2 , \nonumber\\
    &\leq \frac{4\eta^2}{T} \sum_{t=0}^{T-1} \sum_{l=1}^{L} \alpha_l \sum_{k=1}^{B_l} \alpha_k \sum_{j=1}^{V_{k,l}} \alpha_j \mathbb{E} \bigg \Vert \sum_{\tau=\Bar{t}_1}^{t-1} \bigg[  \bigg\{ \sum_{i=1}^{U_{j,k,l}} \alpha_i \tilde{g} \big(\tilde{\mathbf{w}}_i^{\tau,0} \big) \frac{\pmb{1}_i^\tau}{p_i^\tau} - \sum_{i=1}^{U_{j,k,l}} \alpha_i \nabla \tilde{f}_i \big(\tilde{\mathbf{w}}_i^{\tau,0} \big) \bigg\} +  \nonumber\\
    &\Mquad \bigg\{ \sum_{j'=1}^{V_{k,l}} \alpha_{j'} \sum_{i'=1}^{U_{j',k,l}} \alpha_{i'} \nabla \Tilde{f}_{i'} \left( \tilde{\mathbf{w}}_{i'}^{\tau,0} \right) - \sum_{j'=1}^{V_{k,l}} \alpha_{j'} \sum_{i'=1}^{U_{j',k,l}} \alpha_{i'} \Tilde{g} \left( \tilde{\mathbf{w}}_{i'}^{\tau,0} \right) \frac{\pmb{1}_{i'}^\tau}{p_{i'}^\tau} \bigg\} \bigg] \bigg\Vert^2 + \nonumber\\
    &~ \frac{4\eta^2}{T} \sum_{t=0}^{T-1} \sum_{l=1}^{L} \alpha_l \sum_{k=1}^{B_l} \alpha_k \sum_{j=1}^{V_{k,l}} \alpha_j \mathbb{E} \bigg\Vert \sum_{\tau=\Bar{t}_1}^{t-1} \bigg[ \sum_{i=1}^{U_{j,k,l}} \alpha_i \nabla \tilde{f}_i \big(\tilde{\mathbf{w}}_i^{\tau,0} \big) - \sum_{j'=1}^{V_{k,l}} \alpha_{j'} \sum_{i'=1}^{U_{j',k,l}} \alpha_{i'} \nabla \Tilde{f}_{i'} \left( \tilde{\mathbf{w}}_{i'}^{\tau,0} \right) \bigg] \bigg\Vert^2,
\end{align}

\setcounter{Lemma}{6}
\begin{Lemma}
\label{Lemma_30}
\begin{align}
    &\frac{4\eta^2}{T} \sum_{t=0}^{T-1} \sum_{l=1}^{L} \alpha_l \sum_{k=1}^{B_l} \alpha_k \sum_{j=1}^{V_{k,l}} \alpha_j \mathbb{E} \bigg \Vert \sum_{\tau=\Bar{t}_1}^{t-1} \bigg[  \bigg\{ \sum_{i=1}^{U_{j,k,l}} \alpha_i \tilde{g} \big(\tilde{\mathbf{w}}_i^{\tau,0} \big) \frac{\pmb{1}_i^\tau}{p_i^\tau} - \sum_{i=1}^{U_{j,k,l}} \alpha_i \nabla \tilde{f}_i \big(\tilde{\mathbf{w}}_i^{\tau,0} \big) \bigg\} +  \nonumber\\
    &\Mquad \bigg\{ \sum_{j'=1}^{V_{k,l}} \alpha_{j'} \sum_{i'=1}^{U_{j',k,l}} \alpha_{i'} \nabla \Tilde{f}_{i'} \left( \tilde{\mathbf{w}}_{i'}^{\tau,0} \right) - \sum_{j'=1}^{V_{k,l}} \alpha_{j'} \sum_{i'=1}^{U_{j',k,l}} \alpha_{i'} \Tilde{g} \left( \tilde{\mathbf{w}}_{i'}^{\tau,0} \right) \frac{\pmb{1}_{i'}^\tau}{p_{i'}^\tau} \bigg\} \bigg] \bigg\Vert^2 \nonumber\\
    & \leq 8 \kappa_0 \kappa_1 \eta^2 \sigma^2 \sum_{l=1}^{L} \alpha_l \sum_{k=1}^{B_l} \alpha_k \sum_{j=1}^{V_{k,l}} \alpha_j \sum_{i=1}^{U_{j,k,l}} \alpha_i^2 + 8 \kappa_0 \kappa_1 \eta^2 \cdot \varphi_{\mathrm{w, L}_2} \\
    & \approx \mathcal{O} \big( \kappa_0 \kappa_1 \eta^2 \sigma^2 \big) + \mathcal{O} \big( \kappa_0 \kappa_1 \eta^2 \cdot \varphi_{\mathrm{w, L}_2} \big), \label{Lemma_30_Eqn}
\end{align}
where $\varphi_{\mathrm{w, L}_2} = \frac{1}{T} \sum_{t=0}^{T-1} \sum_{l=1}^{L} \alpha_l \sum_{k=1}^{B_l} \alpha_k \sum_{j=1}^{V_{k,l}} \alpha_j \sum_{i=1}^{U_{j,k,l}} \alpha_i^2 \big(\frac{1}{p_i^t} - 1 \big)$.
\end{Lemma}

\begin{Lemma}
\label{Lemma_30_1}
\begin{align}
    &\frac{4\eta^2}{T} \sum_{t=0}^{T-1} \sum_{l=1}^{L} \alpha_l \sum_{k=1}^{B_l} \alpha_k \sum_{j=1}^{V_{k,l}} \alpha_j \mathbb{E} \bigg\Vert \sum_{\tau=\Bar{t}_1}^{t-1} \bigg[ \sum_{i=1}^{U_{j,k,l}} \alpha_i \nabla \tilde{f}_i \big(\tilde{\mathbf{w}}_i^{\tau,0} \big) - \sum_{j'=1}^{V_{k,l}} \alpha_{j'} \sum_{i'=1}^{U_{j',k,l}} \alpha_{i'} \nabla \Tilde{f}_{i'} \left( \tilde{\mathbf{w}}_{i'}^{\tau,0} \right) \bigg] \bigg\Vert^2 \nonumber\\ 
    &\leq  \mathcal{O} \big(\beta^2 \kappa_0^4 \kappa_1^2 \eta^4 \epsilon_{\mathrm{vc}}^2 \big) + \mathcal{O} \big(\kappa_0^2 \kappa_1^2 \eta^2 \epsilon_{\mathrm{sbs}}^2 \big) + \mathcal{O} \big(\beta^2 \sigma^2 \kappa_0^3 \kappa_1^2 \eta^4 \big) + \mathcal{O} \big( \delta^{\mathrm{th}} D^2 \beta^2 \kappa_0^2 \kappa_1^2 \eta^2 \big) + \nonumber\\
    &\qquad \qquad \mathcal{O} \big( \beta^2 \kappa_0^3 \kappa_1^2 \eta^4 G^2 \cdot \varphi_{\mathrm{w, L}_1} \big) + \frac{40 \beta^2 \kappa_0^2 \kappa_1^2 \eta^2}{T} \sum_{t=0}^{T-1} \sum_{l=1}^{L} \alpha_l \sum_{k=1}^{B_l} \alpha_k \sum_{j=1}^{V_{k,l}} \alpha_j \mathbb{E} \big\Vert \Bar{\mathbf{w}}_k^t - \Bar{\mathbf{w}}_j^t \big\Vert^2. 
\end{align}
\end{Lemma}

Now, using Lemma \ref{Lemma_30}, Lemma \ref{Lemma_30_1} and assuming $\eta \leq \frac{1}{2\sqrt{10} \kappa_0 \kappa_1 \beta }$, we have the following
\begin{align}
    & \frac{\beta^2}{T} \sum_{t=0}^{T-1} \sum_{l=1}^{L} \alpha_l \sum_{k=1}^{B_l} \alpha_k \sum_{j=1}^{V_{k,l}} \alpha_j \mathbb{E} \left\Vert \bar{\mathbf{w}}_k^t - \bar{\mathbf{w}}_j^t \right\Vert^2 \nonumber\\ 
    &\leq \mathcal{O } \big( \delta^{th} \beta^2 D^2 \big) + \mathcal{O} \big( \kappa_0 \kappa_1 \eta^2 \sigma^2 \beta^2 \big)  + \mathcal{O} \big( \kappa_0 \kappa_1 \beta^2 \eta^2 \cdot \varphi_{\mathrm{w, L}_2}) + \mathcal{O} \big(\beta^4 \kappa_0^4 \kappa_1^2 \eta^4 \epsilon_{\mathrm{vc}}^2 \big) + \mathcal{O} \big(\kappa_0^2 \kappa_1^2 \eta^2 \beta^2 \epsilon_{\mathrm{sbs}}^2 \big) + \nonumber\\
    &\Squad \mathcal{O} \big( \sigma^2 \kappa_0^3 \kappa_1^2 \beta^4 \eta^4 \big) + \mathcal{O} \big(\delta^{\mathrm{th}} D^2 \kappa_0^2 \kappa_1^2 \eta^2 \beta^4 \big) + \mathcal{O} \big( \kappa_0^3 \kappa_1^2 \beta^4 \eta^4 G^2 \cdot \varphi_{\mathrm{w, L}_1} \big) \nonumber\\
    &\approx \mathcal{O} \big(\beta^4 \kappa_0^4 \kappa_1^2 \eta^4 \epsilon_{\mathrm{vc}}^2 \big) + \mathcal{O} \big(\kappa_0^2 \kappa_1^2 \eta^2 \beta^2 \epsilon_{\mathrm{sbs}}^2 \big) + \mathcal{O} \big( \kappa_0 \kappa_1 \eta^2 \sigma^2 \beta^2 \big) + \mathcal{O } \big( \delta^{th} \beta^2 D^2 \big) + \nonumber\\
    &\Squad \mathcal{O} \big( \kappa_0^3 \kappa_1^2 \beta^4 \eta^4 G^2 \cdot \varphi_{\mathrm{w, L}_1} \big) + \mathcal{O} \big( \kappa_0 \kappa_1 \beta^2 \eta^2 \cdot \varphi_{\mathrm{w, L}_2}) .
\end{align}

\subsection{Missing Proof of Lemma \ref{Lemma_30}}
\begin{align}
\label{Lemma_30_main_1}
    &\frac{4\eta^2}{T} \sum_{t=0}^{T-1} \sum_{l=1}^{L} \alpha_l \sum_{k=1}^{B_l} \alpha_k \sum_{j=1}^{V_{k,l}} \alpha_j \mathbb{E} \bigg \Vert \sum_{\tau=\Bar{t}_1}^{t-1} \bigg[  \bigg\{ \sum_{i=1}^{U_{j,k,l}} \alpha_i \tilde{g} \big(\tilde{\mathbf{w}}_i^{\tau,0} \big) \frac{\pmb{1}_i^\tau}{p_i^\tau} - \sum_{i=1}^{U_{j,k,l}} \alpha_i \nabla \tilde{f}_i \big(\tilde{\mathbf{w}}_i^{\tau,0} \big) \bigg\} +  \nonumber\\
    &\Mquad \bigg\{ \sum_{j'=1}^{V_{k,l}} \alpha_{j'} \sum_{i'=1}^{U_{j',k,l}} \alpha_{i'} \nabla \Tilde{f}_{i'} \left( \tilde{\mathbf{w}}_{i'}^{\tau,0} \right) - \sum_{j'=1}^{V_{k,l}} \alpha_{j'} \sum_{i'=1}^{U_{j',k,l}} \alpha_{i'} \Tilde{g} \left( \tilde{\mathbf{w}}_{i'}^{\tau,0} \right) \frac{\pmb{1}_{i'}^\tau}{p_{i'}^\tau} \bigg\} \bigg] \bigg\Vert^2 \nonumber\\
    &\overset{(a)}{=} \frac{4\eta^2}{T} \sum_{t=0}^{T-1} \sum_{l=1}^{L} \alpha_l \sum_{k=1}^{B_l} \alpha_k \sum_{j=1}^{V_{k,l}} \alpha_j \mathbb{E} \bigg \Vert \sum_{\tau=\Bar{t}_1}^{t-1} \bigg[  \sum_{i=1}^{U_{j,k,l}} \alpha_i \tilde{g} \big(\tilde{\mathbf{w}}_i^{\tau,0} \big) \frac{\pmb{1}_i^\tau}{p_i^\tau} - \sum_{i=1}^{U_{j,k,l}} \alpha_i \nabla \tilde{f}_i \big(\tilde{\mathbf{w}}_i^{\tau,0} \big)  \bigg] \bigg\Vert^2 -  \nonumber\\
    &\Squad \frac{4\eta^2}{T} \sum_{t=0}^{T-1} \sum_{l=1}^{L} \alpha_l \sum_{k=1}^{B_l} \alpha_k  \mathbb{E} \bigg \Vert \sum_{\tau=\Bar{t}_1}^{t-1}  \sum_{j'=1}^{V_{k,l}} \alpha_{j'} \sum_{i'=1}^{U_{j',k,l}} \alpha_{i'} \bigg[ \Tilde{g} \left( \tilde{\mathbf{w}}_{i'}^{\tau,0} \right) \frac{\pmb{1}_{i'}^\tau}{p_{i'}^\tau} - \nabla \Tilde{f}_{i'} \left( \tilde{\mathbf{w}}_{i'}^{\tau,0} \right) \bigg]  \bigg\Vert^2 \nonumber \\
    &\overset{(b)}{=} \frac{4\eta^2}{T} \sum_{t=0}^{T-1} \sum_{l=1}^{L} \alpha_l \sum_{k=1}^{B_l} \alpha_k \sum_{j=1}^{V_{k,l}} \alpha_j \sum_{\tau=\Bar{t}_1}^{t-1} \mathbb{E} \bigg \Vert  \sum_{i=1}^{U_{j,k,l}} \alpha_i \bigg[ \tilde{g} \big(\tilde{\mathbf{w}}_i^{\tau,0} \big) \frac{\pmb{1}_i^\tau}{p_i^\tau} -  \nabla \tilde{f}_i \big(\tilde{\mathbf{w}}_i^{\tau,0} \big) \bigg] \bigg\Vert^2 -  \nonumber\\
    &\Squad \frac{4\eta^2}{T} \sum_{t=0}^{T-1} \sum_{l=1}^{L} \alpha_l \sum_{k=1}^{B_l} \alpha_k  \sum_{\tau=\Bar{t}_1}^{t-1} \mathbb{E} \bigg \Vert  \sum_{j'=1}^{V_{k,l}} \alpha_{j'} \sum_{i'=1}^{U_{j',k,l}} \alpha_{i'} \bigg[ \Tilde{g} \left( \tilde{\mathbf{w}}_{i'}^{\tau,0} \right) \frac{\pmb{1}_{i'}^\tau}{p_{i'}^\tau} - \nabla \Tilde{f}_{i'} \left( \tilde{\mathbf{w}}_{i'}^{\tau,0} \right) \bigg]  \bigg\Vert^2 \nonumber \\
    &\overset{(c)}{\leq} \frac{4 \kappa_0 \kappa_1 \eta^2}{T} \sum_{t=0}^{T-1} \sum_{l=1}^{L} \alpha_l \sum_{k=1}^{B_l} \alpha_k \sum_{j=1}^{V_{k,l}} \alpha_j  \mathbb{E} \bigg \Vert  \sum_{i=1}^{U_{j,k,l}} \alpha_i \bigg[ \tilde{g} \big(\tilde{\mathbf{w}}_i^{t,0} \big) \frac{\pmb{1}_i^t}{p_i^t} -  \nabla \tilde{f}_i \big(\tilde{\mathbf{w}}_i^{t,0} \big) \bigg] \bigg\Vert^2 -  \nonumber\\
    &\Squad \frac{4\kappa_0 \kappa_1\eta^2}{T} \sum_{t=0}^{T-1} \sum_{l=1}^{L} \alpha_l \sum_{k=1}^{B_l} \alpha_k  \mathbb{E} \bigg \Vert  \sum_{j'=1}^{V_{k,l}} \alpha_{j'} \sum_{i'=1}^{U_{j',k,l}} \alpha_{i'} \bigg[ \Tilde{g} \left( \tilde{\mathbf{w}}_{i'}^{t,0} \right) \frac{\pmb{1}_{i'}^t}{p_{i'}^t} - \nabla \Tilde{f}_{i'} \left( \tilde{\mathbf{w}}_{i'}^{t,0} \right) \bigg]  \bigg\Vert^2 \nonumber \\
    &= \frac{4 \kappa_0 \kappa_1 \eta^2}{T} \sum_{t=0}^{T-1} \sum_{l=1}^{L} \alpha_l \sum_{k=1}^{B_l} \alpha_k \sum_{j=1}^{V_{k,l}} \alpha_j \sum_{i=1}^{U_{j,k,l}} \alpha_i^2 \mathbb{E} \bigg \Vert  \tilde{g} \big(\tilde{\mathbf{w}}_i^{t,0} \big) \frac{\pmb{1}_i^t}{p_i^t} \pm \Tilde{g} \left( \tilde{\mathbf{w}}_{i}^{t,0} \right) - \nabla \tilde{f}_i \big( \tilde{\mathbf{w}}_i^{t,0} \big) \bigg\Vert^2 -  \nonumber\\
    &\Squad \frac{4\kappa_0 \kappa_1\eta^2}{T} \sum_{t=0}^{T-1} \sum_{l=1}^{L} \alpha_l \sum_{k=1}^{B_l} \alpha_k    \sum_{j=1}^{V_{k,l}} \alpha_{j}^2 \sum_{i=1}^{U_{j,k,l}} \alpha_{i}^2 \mathbb{E} \bigg \Vert \Tilde{g} \left( \tilde{\mathbf{w}}_{i}^{t,0} \right) \frac{\pmb{1}_{i}^t}{p_{i}^t} \pm \Tilde{g} \left( \tilde{\mathbf{w}}_{i}^{t,0} \right) - \nabla \Tilde{f}_{i} \left( \tilde{\mathbf{w}}_{i}^{t,0} \right)  \bigg\Vert^2 \nonumber \\
    &\leq \frac{8 \kappa_0 \kappa_1 \eta^2}{T} \sum_{t=0}^{T-1} \sum_{l=1}^{L} \alpha_l \sum_{k=1}^{B_l} \alpha_k \sum_{j=1}^{V_{k,l}} \alpha_j \sum_{i=1}^{U_{j,k,l}} \alpha_i^2 \bigg[ \mathbb{E} \bigg \Vert  \bigg(\frac{\pmb{1}_i^t}{p_i^t} - 1 \bigg) \tilde{g} \big(\tilde{\mathbf{w}}_i^{t,0} \big) \bigg\Vert^2 + \mathbb{E} \bigg \Vert \Tilde{g} \left( \tilde{\mathbf{w}}_{i}^{t,0} \right) - \nabla \tilde{f}_i \big( \tilde{\mathbf{w}}_i^{t,0} \big) \bigg\Vert^2 \bigg] -  \nonumber\\
    &\quad \frac{8 \kappa_0 \kappa_1 \eta^2}{T} \sum_{t=0}^{T-1} \sum_{l=1}^{L} \alpha_l \sum_{k=1}^{B_l} \alpha_k \sum_{j=1}^{V_{k,l}} \alpha_j^2 \sum_{i=1}^{U_{j,k,l}} \alpha_i^2 \bigg[ \mathbb{E} \bigg \Vert  \bigg(\frac{\pmb{1}_i^t}{p_i^t} - 1 \bigg) \tilde{g} \big(\tilde{\mathbf{w}}_i^{t,0} \big) \bigg\Vert^2 + \mathbb{E} \bigg \Vert \Tilde{g} \left( \tilde{\mathbf{w}}_{i}^{t,0} \right) - \nabla \tilde{f}_i \big( \tilde{\mathbf{w}}_i^{t,0} \big) \bigg\Vert^2 \bigg] \nonumber \\
    &\leq \frac{8 \kappa_0 \kappa_1 \eta^2 G^2}{T} \sum_{t=0}^{T-1} \sum_{l=1}^{L} \alpha_l \sum_{k=1}^{B_l} \alpha_k \sum_{j=1}^{V_{k,l}} \alpha_j \sum_{i=1}^{U_{j,k,l}} \alpha_i^2 \bigg(\frac{1}{p_i^t} - 1 \bigg) + 8 \kappa_0 \kappa_1 \eta^2 \sigma^2 G^2 \sum_{l=1}^{L} \alpha_l \sum_{k=1}^{B_l} \alpha_k \sum_{j=1}^{V_{k,l}} \alpha_j \sum_{i=1}^{U_{j,k,l}} \alpha_i^2 - \nonumber \\
    & \frac{8 \kappa_0 \kappa_1 \eta^2 G^2}{T} \sum_{t=0}^{T-1} \sum_{l=1}^{L} \alpha_l \sum_{k=1}^{B_l} \alpha_k \sum_{j=1}^{V_{k,l}} \alpha_j^2 \sum_{i=1}^{U_{j,k,l}} \alpha_i^2 \bigg(\frac{1}{p_i^t} - 1 \bigg) - 8 \kappa_0 \kappa_1 \eta^2 \sigma^2 G^2 \sum_{l=1}^{L} \alpha_l \sum_{k=1}^{B_l} \alpha_k \sum_{j=1}^{V_{k,l}} \alpha_j^2 \sum_{i=1}^{U_{j,k,l}} \alpha_i^2 \nonumber \\
    & \leq 8 \kappa_0 \kappa_1 \eta^2 \sigma^2 \sum_{l=1}^{L} \alpha_l \sum_{k=1}^{B_l} \alpha_k \sum_{j=1}^{V_{k,l}} \alpha_j \sum_{i=1}^{U_{j,k,l}} \alpha_i^2 + 8 \kappa_0 \kappa_1 \eta^2 \cdot \varphi_{\mathrm{w, L}_2},  
\end{align}
where $\varphi_{\mathrm{w, L}_2} = \frac{1}{T} \sum_{t=0}^{T-1} \sum_{l=1}^{L} \alpha_l \sum_{k=1}^{B_l} \alpha_k \sum_{j=1}^{V_{k,l}} \alpha_j \sum_{i=1}^{U_{j,k,l}} \alpha_i^2 \bigg(\frac{1}{p_i^t} - 1 \bigg)$.

\subsection{Missing Proof of Lemma \ref{Lemma_30_1}}
\begin{align}
    &\frac{4\eta^2}{T} \sum_{t=0}^{T-1} \sum_{l=1}^{L} \alpha_l \sum_{k=1}^{B_l} \alpha_k \sum_{j=1}^{V_{k,l}} \alpha_j \mathbb{E} \bigg\Vert \sum_{\tau=\Bar{t}_1}^{t-1} \bigg[ \sum_{i=1}^{U_{j,k,l}} \alpha_i \nabla \tilde{f}_i \big(\tilde{\mathbf{w}}_i^{\tau,0} \big) - \sum_{j'=1}^{V_{k,l}} \alpha_{j'} \sum_{i'=1}^{U_{j',k,l}} \alpha_{i'} \nabla \Tilde{f}_{i'} \left( \tilde{\mathbf{w}}_{i'}^{\tau,0} \right) \bigg] \bigg\Vert^2 \nonumber\\ 
    &= \frac{4\kappa_0 \kappa_1 \eta^2}{T} \sum_{t=0}^{T-1} \sum_{l=1}^{L} \alpha_l \sum_{k=1}^{B_l} \alpha_k \sum_{j=1}^{V_{k,l}} \alpha_j \sum_{\tau=\Bar{t}_1}^{t-1} \mathbb{E} \bigg\Vert \bigg( \sum_{i=1}^{U_{j,k,l}} \alpha_i \big[ \nabla \tilde{f}_i \big(\tilde{\mathbf{w}}_i^{\tau,0} \big) - \nabla \tilde{f}_i \big(\bar{\mathbf{w}}_j^{\tau} \big)  \big] \bigg) + \nonumber\\
    & \bigg( \sum_{i=1}^{U_{j,k,l}} \alpha_i \big[ \nabla \tilde{f}_i \big(\bar{\mathbf{w}}_j^{\tau} \big) - \nabla \tilde{f}_i \big(\bar{\mathbf{w}}_k^{\tau} \big)  \big] \bigg) + \bigg( \sum_{i=1}^{U_{j,k,l}} \alpha_i \nabla \tilde{f}_i \big(\bar{\mathbf{w}}_k^{\tau} \big) - \sum_{j'=1}^{V_{k,l}} \alpha_{j'} \sum_{i'=1}^{U_{j',k,l}} \alpha_{i'} \nabla \Tilde{f}_{i'} \big( \bar{\mathbf{w}}_{k}^{\tau} \big) \bigg) + \nonumber \\
    & \bigg( \sum_{j'=1}^{V_{k,l}} \alpha_{j'} \sum_{i'=1}^{U_{j',k,l}} \alpha_{i'} \big[ \nabla \Tilde{f}_{i'} \big( \bar{\mathbf{w}}_{k}^{\tau} \big) - \nabla \Tilde{f}_{i'} \big( \bar{\mathbf{w}}_{j}^{\tau} \big) \big] \bigg) + \bigg( \sum_{j'=1}^{V_{k,l}} \alpha_{j'} \sum_{i'=1}^{U_{j',k,l}} \alpha_{i'} \big[ \nabla \Tilde{f}_{i'} \big( \bar{\mathbf{w}}_{j}^{\tau} \big) - \nabla \Tilde{f}_{i'} \big( \tilde{\mathbf{w}}_{i'}^{\tau,0} \big) \big] \bigg) \bigg\Vert^2 \nonumber\\  
    & \leq 20 \kappa_0^2 \kappa_1^2 \eta^2 \epsilon_{\mathrm{sbs}}^2 + \frac{40 \beta^2 \kappa_0^2 \kappa_1^2 \eta^2}{T} \sum_{t=0}^{T-1} \sum_{l=1}^{L} \alpha_l \sum_{k=1}^{B_l} \alpha_k \sum_{j=1}^{V_{k,l}} \alpha_j \mathbb{E} \big\Vert \Bar{\mathbf{w}}_k^t - \Bar{\mathbf{w}}_j^t \big\Vert^2 +  \frac{80 \beta^2 \kappa_0^2 \kappa_1^2 \eta^2}{T} \sum_{t=0}^{T-1} \sum_{l=1}^{L} \alpha_l \sum_{k=1}^{B_l} \alpha_k \sum_{j=1}^{V_{k,l}} \alpha_j \sum_{i=1}^{U_{j,k,l}} \alpha_i \mathbb{E} \big\Vert \Bar{\mathbf{w}}_j^t - \tilde{\mathbf{w}}_i^{t} \big\Vert^2 + \nonumber\\
    &\Squad \frac{80 \beta^2 \kappa_0^2 \kappa_1^2 \eta^2}{T} \sum_{t=0}^{T-1} \sum_{l=1}^{L} \alpha_l \sum_{k=1}^{B_l} \alpha_k \sum_{j=1}^{V_{k,l}} \alpha_j \sum_{i=1}^{U_{j,k,l}} \alpha_i \mathbb{E} \big\Vert \tilde{\mathbf{w}}_i^{t} - \tilde{\mathbf{w}}_i^{t,0} \big\Vert^2 \nonumber\\
    &\approx \mathcal{O} \big(\beta^2 \kappa_0^4 \kappa_1^2 \eta^4 \epsilon_{\mathrm{vc}}^2 \big) + \mathcal{O} \big(\kappa_0^2 \kappa_1^2 \eta^2 \epsilon_{\mathrm{sbs}}^2 \big) + \mathcal{O} \big(\beta^2 \sigma^2 \kappa_0^3 \kappa_1^2 \eta^4 \big) + \mathcal{O} \big( \delta^{\mathrm{th}} D^2 \beta^2 \kappa_0^2 \kappa_1^2 \eta^2 \big) + \mathcal{O} \big( \beta^2 \kappa_0^3 \kappa_1^2 \eta^4 G^2 \cdot \varphi_{\mathrm{w, L}_1} \big) + \nonumber\\
    &\Squad \frac{40 \beta^2 \kappa_0^2 \kappa_1^2 \eta^2}{T} \sum_{t=0}^{T-1} \sum_{l=1}^{L} \alpha_l \sum_{k=1}^{B_l} \alpha_k \sum_{j=1}^{V_{k,l}} \alpha_j \mathbb{E} \big\Vert \Bar{\mathbf{w}}_k^t - \Bar{\mathbf{w}}_j^t \big\Vert^2. 
\end{align}
which concludes the proof of Lemma \ref{Lemma_30_1}.
\end{proof}

\section{Proof of Lemma \ref{Lemma4}}
\setcounter{Lemma}{2}
\begin{Lemma}
\label{Lemma4}
When $\eta \leq 1/[2\sqrt{14} \kappa_0 \kappa_1 \kappa_2 \beta]$, the average difference between the sBS and mBS model parameters, i.e., the $\mathrm{L}_3$ term of (\ref{theorem_1_eqn_Sup}), is upper bounded as
\begin{align}
\label{Lem4_Sup}
    &[\beta^2/T] \sum\nolimits_{t=0}^{T-1} \sum\nolimits_{l=1}^{L} \alpha_l \sum\nolimits_{k=1}^{B_l} \alpha_k \mathbb{E} \left\Vert \bar{\mathbf{w}}_l^t - \bar{\mathbf{w}}_k^t \right\Vert^2 \nonumber \\
    &\leq \mathcal{O} \big( \kappa_0^3 \kappa_1^2 \kappa_2^2 \eta^4 \beta^4 \epsilon_{\mathrm{vc}}^2 \big) + \mathcal{O} \big(\kappa_0^4 \kappa_1^4 \kappa_2^2 \eta^4 \beta^4 \epsilon_{\mathrm{sbs}}^2 \big) + \mathcal{O} \big( \kappa_0^2 \kappa_1^2 \kappa_2^2 \eta^2 \beta^4 \epsilon_{\mathrm{mbs}}^2 \big) + \mathcal{O} \big( \kappa_0\kappa_1\kappa_2 \eta^2 \beta^2 \sigma^2 \big) + \mathcal{O} \big( \delta^{\mathrm{th}} \beta^2 D^2 \big) + \nonumber\\
    &\qquad \mathcal{O} \big( \kappa_0^3 \kappa_1^2 \kappa_2^2 \eta^4 \beta^4 G^2 \cdot \varphi_{\mathrm{w, L}_1} (\pmb{\delta},\pmb{\mathrm{f}}, \pmb{\mathrm{P}}) \big) + \mathcal{O} \big( \kappa_0^2 \kappa_1^2 \kappa_2^2 \beta^4 \eta^4 \cdot \varphi_{\mathrm{w, L}_2} (\pmb{\delta},\pmb{\mathrm{f}}, \pmb{\mathrm{P}}) ) + \mathcal{O} \big( \kappa_0\kappa_1\kappa_2 \eta^2 \beta^2 G^2 \cdot \varphi_{\mathrm{w,L}_3} (\pmb{\delta},\pmb{\mathrm{f}}, \pmb{\mathrm{P}}) \big).
\end{align}
where $\varphi_{\mathrm{w,L}_3} (\pmb{\delta},\pmb{\mathrm{f}}, \pmb{\mathrm{P}}) = [1/T] \sum_{l=1}^{L} \alpha_l \sum_{k=1}^{B_l} \alpha_k \sum_{j=1}^{V_{k,l}} \alpha_j^2 \sum_{i=1}^{U_{j,k,l}} \alpha_i^2 \sum_{t=0}^{T-1} \left(1/p_i^t - 1 \right)$. 
\end{Lemma}

\begin{proof}
\label{proofLemma4}
\begin{align}
\label{Lemma4_Main_0}
    &\frac{1}{T} \sum_{t=0}^{T-1} \sum_{l=1}^{L} \alpha_l \sum_{k=1}^{B_l} \alpha_k \mathbb{E} \left\Vert \bar{\mathbf{w}}_l^t - \bar{\mathbf{w}}_k^t \right\Vert^2 \nonumber\\ 
    &= \sum_{l=1}^{L} \alpha_l \sum_{k=1}^{B_l} \alpha_k \mathbb{E} \Bigg\Vert \tilde{ \mathbf{w} }_l^{\Bar{t}_2, 0} - \tilde{ \mathbf{w} }_k^{\Bar{t}_2, 0} - \eta  \sum_{\tau=\Bar{t}_2}^{t-1} \sum_{k'=1}^{B_l} \alpha_{k'} \sum_{j'=1}^{V_{k',l}} \alpha_{j'}  \sum_{i'=1}^{U_{j',k',l}} \alpha_{i'} \Tilde{g} \big(\tilde{\mathbf{w}}_{i'}^{\tau,0}\big) \frac{\pmb{1}_{i'}^\tau}{p_{i'}^\tau} + \eta \sum_{\tau=\Bar{t}_2}^{t-1} \sum_{j=1}^{V_{k,l}} \alpha_j \sum_{i=1}^{U_{j,k,l} } \alpha_i \tilde{g} \big(\tilde{\mathbf{w}}_i^{\tau,0} \big) \frac{\pmb{1}_{i}^\tau}{p_{i}^\tau} \Bigg\Vert^2, \nonumber \\
    &\leq \frac{2}{T} \sum_{t=0}^{T-1} \sum_{l=1}^{L} \rs \alpha_l \rs \sum_{k=1}^{B_l} \rs \alpha_k \mathbb{E} \left\Vert \tilde{ \mathbf{w} }_l^{\Bar{t}_2, 0} - \tilde{ \mathbf{w} }_k^{\Bar{t}_2, 0} \right \Vert^2 + \frac{2\eta^2}{T} \sum_{t=0}^{T-1} \rs \sum_{l=1}^{L} \rs \alpha_l \rs \sum_{k=1}^{B_l} \rs \alpha_k \mathbb{E} \bigg \Vert \sum_{\tau=\Bar{t}_2}^{t-1} \bigg[ \sum_{j=1}^{V_{k,l}} \rs \alpha_j \rs \sum_{i=1}^{U_{j,k,l}} \rs \alpha_i \tilde{g} \big(\tilde{\mathbf{w}}_i^{\tau,0} \big) \frac{\pmb{1}_{i}^\tau}{p_{i}^\tau} - \sum_{k'=1}^{B_l} \rs \alpha_{k'} \rs \sum_{j'=1}^{V_{k',l}} \rs \alpha_{j'} \rs \sum_{i'=1}^{U_{j',k',l}} \rs \alpha_{i'} \Tilde{g} \big(\tilde{\mathbf{w}}_{i'}^{\tau,0}\big) \frac{\pmb{1}_{i'}^\tau}{p_{i'}^\tau} \bigg] \bigg\Vert^2 \rs, 
\end{align}
where $\Bar{t}_2 = (m \kappa_3 + t_3)\kappa_2\kappa_1 \kappa_0 $ and the inequalities in the last term arise from Jensen inequality. 

For the first term in (\ref{Lemma4_Main_0}), we have
\begin{align}
\label{Lemma4Main0_First_Term}
    & 2 \sum_{l=1}^{L} \alpha_l \sum_{k=1}^{B_l} \alpha_k \mathbb{E} \left\Vert \tilde{ \mathbf{w} }_l^{\Bar{t}_2, 0} - \tilde{ \mathbf{w} }_k^{\Bar{t}_2, 0} \right \Vert^2 \nonumber \\ 
    & \overset{(a)}{=} 2 \sum_{l=1}^{L} \alpha_l \sum_{k=1}^{B_l} \alpha_k \mathbb{E} \left\Vert \tilde{ \mathbf{w} }_l^{\Bar{t}_2, 0} - \mathbf{w}_l^{\Bar{t}_2} + \mathbf{w}_k^{\Bar{t}_2} - \tilde{ \mathbf{w} }_k^{\Bar{t}_2, 0} \right \Vert^2, \nonumber \\
    &\leq 4 \sum_{l=1}^{L} \alpha_l \sum_{k=1}^{B_l} \alpha_k \mathbb{E} \left\Vert \mathbf{w}_l^{\Bar{t}_2} - \tilde{ \mathbf{w} }_l^{\Bar{t}_2, 0} \right\Vert^2 + 4 \sum_{l=1}^{L} \alpha_l \sum_{k=1}^{B_l} \alpha_k \mathbb{E} \left\Vert \mathbf{w}_k^{\Bar{t}_2} - \tilde{ \mathbf{w} }_k^{\Bar{t}_2, 0} \right \Vert^2, \nonumber \\
    &\overset{(b)}{\leq} 4 \sum_{l=1}^{L} \alpha_l \sum_{k=1}^{B_l} \bblue{\alpha_k} \sum_{j=1}^{V_{k,l}} \bblue{\alpha_j} \sum_{i=1}^{U_{j,k,l}} \bblue{\alpha_i} \delta_i^{\Bar{t}_2} \mathbb{E} \left\Vert \mathbf{w}_{i}^{\Bar{t}_2} \right\Vert^2 + 4 \sum_{l=1}^{L} \alpha_l \sum_{k=1}^{B_l} \alpha_k \sum_{j=1}^{V_{k,l}} \bblue{\alpha_j} \sum_{i=1}^{U_{j,k,l}} \bblue{\alpha_i} \delta_i^{\Bar{t}_2} \mathbb{E} \left\Vert \mathbf{w}_{i}^{\Bar{t}_2} \right \Vert^2, \nonumber \\
    &= \bblue{8} \sum_{l=1}^{L} \alpha_l \sum_{k=1}^{B_l} \alpha_k \sum_{j=1}^{V_{k,l}} \bblue{\alpha_j} \sum_{i=1}^{U_{j,k,l}} \bblue{\alpha_i} \delta_i^{\Bar{t}_2} \mathbb{E} \left\Vert \mathbf{w}_{i}^{\Bar{t}_2} \right\Vert^2,
\end{align}
where in $(a)$ we use the fact that $\mathbf{w}_l^{\Bar{t}_2} = \mathbf{w}_k^{\Bar{t}_2}$.
Moreover, $(b)$ appears following similar steps as in (\ref{mBSPrune}) and (\ref{sBSPrune}).

As such, we have 
\begin{align}
    & \frac{2}{T} \sum_{t=0}^{T-1} \sum_{l=1}^{L} \rs \alpha_l \rs \sum_{k=1}^{B_l} \rs \alpha_k \mathbb{E} \left\Vert \tilde{ \mathbf{w} }_l^{\Bar{t}_2, 0} - \tilde{ \mathbf{w} }_k^{\Bar{t}_2, 0} \right \Vert^2 \nonumber\\
    & \leq \frac{\bblue{8}}{T} \sum_{t=0}^{T-1} \sum_{l=1}^{L} \alpha_l \sum_{k=1}^{B_l} \alpha_k \sum_{j=1}^{V_{k,l}} \bblue{\alpha_j} \sum_{i=1}^{U_{j,k,l}} \bblue{\alpha_i} \delta_i^{\big\lfloor \frac{t}{\kappa_0\kappa_1\kappa_2}\big\rfloor} \mathbb{E} \bigg\Vert \mathbf{w}_{i}^{\big\lfloor \frac{t}{\kappa_0\kappa_1\kappa_2}\big\rfloor} \bigg\Vert^2 \nonumber\\
    &\approx \mathcal{O} \big( \delta^{\mathrm{th}} D^2 \big).
\end{align}

For the second term in (\ref{Lemma4_Main_0}), we have
\begin{align}
\label{Lemma4Main0_Second_Term}
    & \frac{2\eta^2}{T} \sum_{t=0}^{T-1} \sum_{l=1}^{L} \alpha_l \sum_{k=1}^{B_l} \alpha_k \mathbb{E} \left \Vert \sum_{\tau=\Bar{t}_2}^{t-1} \Big( \sum_{j=1}^{V_{k,l}} \alpha_j \sum_{i=1}^{U_{j,k,l} } \alpha_i \tilde{g} \big(\tilde{\mathbf{w}}_i^{\tau,0} \big) \frac{\pmb{1}_{i}^\tau}{p_{i}^\tau} -  \sum_{k'=1}^{B_l} \alpha_{k'} \sum_{j'=1}^{V_{k',l}} \alpha_{j'}  \sum_{i'=1}^{U_{j',k',l}} \alpha_{i'} \Tilde{g} \big(\tilde{\mathbf{w}}_{i'}^{\tau,0}\big) \frac{\pmb{1}_{i'}^\tau}{p_{i'}^\tau} \Big) \right\Vert^2 \nonumber\\
    &= \frac{2\eta^2}{T} \sum_{t=0}^{T-1} \sum_{l=1}^{L} \alpha_l \sum_{k=1}^{B_l} \alpha_k \mathbb{E} \Bigg \Vert \sum_{\tau=\Bar{t}_2}^{t-1} \rs \bigg[ \sum_{j=1}^{V_{k,l}} \alpha_j \sum_{i=1}^{U_{j,k,l}} \alpha_i \bigg( \tilde{g} \big(\tilde{\mathbf{w}}_i^{\tau,0} \big) \frac{\pmb{1}_{i}^\tau}{p_{i}^\tau} - \nabla \tilde{f}_i \big(\tilde{\mathbf{w}}_i^{\tau,0} \big) \bigg) - \sum_{k'=1}^{B_l} \alpha_{k'} \sum_{j'=1}^{V_{k',l}}  \alpha_{j'} \sum_{i'=1}^{U_{j',k',l} } \alpha_{i'} \bigg( \Tilde{g} \big(\tilde{\mathbf{w}}_{i'}^{\tau,0}\big) \frac{\pmb{1}_{i'}^\tau}{p_{i'}^\tau} - \nabla \tilde{f}_{i'} \big(\tilde{\mathbf{w}}_{i'}^{\tau,0} \big) \bigg) + \nonumber \\
    &\Mquad \sum_{k'=1}^{B_l} \alpha_{k'} \sum_{j'=1}^{V_{k',l}}  \alpha_{j'} \sum_{i'=1}^{U_{j',k',l} } \alpha_{i'} \nabla \tilde{f}_{i'} \big(\tilde{\mathbf{w}}_{i'}^{\tau,0} \big) - \sum_{j=1}^{V_{k,l}} \alpha_j \sum_{i=1}^{U_{j,k,l}} \alpha_i \nabla \tilde{f}_i \big(\tilde{\mathbf{w}}_i^{\tau,0} \big) \bigg] \Bigg\Vert^2, \nonumber\\
    &\leq \frac{4\eta^2}{T} \sum_{t=0}^{T-1} \sum_{l=1}^{L} \alpha_l \sum_{k=1}^{B_l} \alpha_k \mathbb{E} \Bigg \Vert \sum_{\tau=\Bar{t}_2}^{t-1} \rs \bigg[ \sum_{j=1}^{V_{k,l}} \alpha_j \sum_{i=1}^{U_{j,k,l}} \alpha_i \bigg( \tilde{g} \big(\tilde{\mathbf{w}}_i^{\tau,0} \big) \frac{\pmb{1}_{i}^\tau}{p_{i}^\tau} - \nabla \tilde{f}_i \big(\tilde{\mathbf{w}}_i^{\tau,0} \big) \bigg) - \sum_{k'=1}^{B_l} \alpha_{k'} \sum_{j'=1}^{V_{k',l}}  \alpha_{j'} \sum_{i'=1}^{U_{j',k',l} } \alpha_{i'} \bigg( \Tilde{g} \big(\tilde{\mathbf{w}}_{i'}^{\tau,0}\big) \frac{\pmb{1}_{i'}^\tau}{p_{i'}^\tau} - \nabla \tilde{f}_{i'} \big(\tilde{\mathbf{w}}_{i'}^{\tau,0} \big) \bigg) \bigg] \Bigg \Vert^2 + \nonumber \\
    &\Mquad \frac{4\eta^2}{T} \sum_{t=0}^{T-1} \sum_{l=1}^{L} \alpha_l \sum_{k=1}^{B_l} \alpha_k \mathbb{E} \Bigg \Vert \sum_{\tau=\Bar{t}_2}^{t-1} \rs \bigg[\sum_{k'=1}^{B_l} \alpha_{k'} \sum_{j'=1}^{V_{k',l}}  \alpha_{j'} \sum_{i'=1}^{U_{j',k',l} } \alpha_{i'} \nabla \tilde{f}_{i'} \big(\tilde{\mathbf{w}}_{i'}^{\tau,0} \big) - \sum_{j=1}^{V_{k,l}} \alpha_j \sum_{i=1}^{U_{j,k,l}} \alpha_i \nabla \tilde{f}_i \big(\tilde{\mathbf{w}}_i^{\tau,0} \big) \bigg] \Bigg\Vert^2, 
\end{align}

\setcounter{Lemma}{8}
\begin{Lemma}
\label{Lemma_40}
\begin{align}
    &\frac{4\eta^2}{T} \sum_{t=0}^{T-1} \sum_{l=1}^{L} \alpha_l \sum_{k=1}^{B_l} \alpha_k \mathbb{E} \Bigg \Vert \sum_{\tau=\Bar{t}_2}^{t-1} \rs \bigg[ \sum_{j=1}^{V_{k,l}} \alpha_j \sum_{i=1}^{U_{j,k,l}} \alpha_i \bigg( \tilde{g} \big(\tilde{\mathbf{w}}_i^{\tau,0} \big) \frac{\pmb{1}_{i}^\tau}{p_{i}^\tau} - \nabla \tilde{f}_i \big(\tilde{\mathbf{w}}_i^{\tau,0} \big) \bigg) - \sum_{k'=1}^{B_l} \alpha_{k'} \sum_{j'=1}^{V_{k',l}}  \alpha_{j'} \sum_{i'=1}^{U_{j',k',l} } \alpha_{i'} \bigg( \Tilde{g} \big(\tilde{\mathbf{w}}_{i'}^{\tau,0}\big) \frac{\pmb{1}_{i'}^\tau}{p_{i'}^\tau} - \nabla \tilde{f}_{i'} \big(\tilde{\mathbf{w}}_{i'}^{\tau,0} \big) \bigg) \bigg] \Bigg \Vert^2 \nonumber\\
    &\leq 8 \kappa_0\kappa_1\kappa_2 \eta^2 \sigma^2 \sum_{l=1}^{L} \alpha_l \sum_{k=1}^{B_l} \alpha_k \sum_{j=1}^{V_{k,l}} \alpha_j^2 \sum_{i=1}^{U_{j,k,l}} \alpha_i^2 + 8 \kappa_0\kappa_1\kappa_2 \eta^2 G^2 \cdot \varphi_{\mathrm{w,L}_3} \nonumber\\
    & \approx \mathcal{O} \big( \kappa_0\kappa_1\kappa_2 \eta^2 \sigma^2 \big) + \mathcal{O} \big( \kappa_0\kappa_1\kappa_2 \eta^2 G^2 \cdot \varphi_{\mathrm{w,L}_3} \big),
\end{align}
where $\varphi_{\mathrm{w,L}_3} = \frac{1}{T} \sum_{l=1}^{L} \alpha_l \sum_{k=1}^{B_l} \alpha_k \sum_{j=1}^{V_{k,l}} \alpha_j^2 \sum_{i=1}^{U_{j,k,l}} \alpha_i^2 \sum_{t=0}^{T-1} \left(\frac{1}{p_i^t} - 1 \right)$.
\end{Lemma}

\begin{Lemma}
\label{Lemma_40_1}
\begin{align}
    & \frac{4\eta^2}{T} \sum_{t=0}^{T-1} \sum_{l=1}^{L} \alpha_l \sum_{k=1}^{B_l} \alpha_k \mathbb{E} \Bigg \Vert \sum_{\tau=\Bar{t}_2}^{t-1} \rs \bigg[\sum_{k'=1}^{B_l} \alpha_{k'} \sum_{j'=1}^{V_{k',l}}  \alpha_{j'} \sum_{i'=1}^{U_{j',k',l} } \alpha_{i'} \nabla \tilde{f}_{i'} \big(\tilde{\mathbf{w}}_{i'}^{\tau,0} \big) - \sum_{j=1}^{V_{k,l}} \alpha_j \sum_{i=1}^{U_{j,k,l}} \alpha_i \nabla \tilde{f}_i \big(\tilde{\mathbf{w}}_i^{\tau,0} \big) \bigg] \Bigg\Vert^2  \nonumber\\
    &\leq \mathcal{O} \big( \kappa_0^3 \kappa_1^2 \kappa_2^2 \eta^4 \beta^2 \sigma^2 \big) + \mathcal{O} \big( \kappa_0^3 \kappa_1^2 \kappa_2^2 \eta^4 \beta^2 \epsilon_{\mathrm{vc}}^2 \big) + \mathcal{O} \big(\kappa_0^4 \kappa_1^4 \kappa_2^2 \eta^4 \beta^2 \epsilon_{\mathrm{sbs}}^2 \big) + \mathcal{O} \big( \kappa_0^2 \kappa_1^2 \kappa_2^2 \eta^2 \beta^2 \epsilon_{\mathrm{mbs}}^2 \big) +  \mathcal{O } \big( \delta^{th} \kappa_0^2 \kappa_1^2 \kappa_2^2 \eta^2 \beta^2 D^2 \big) + \nonumber\\
    &\Squad \mathcal{O} \big( \kappa_0^3 \kappa_1^2 \kappa_2^2 \eta^4 \beta^2 G^2 \cdot \varphi_{\mathrm{w, L}_1} \big) + \mathcal{O} \big( \kappa_0^2 \kappa_1^2 \kappa_2^2 \beta^2 \eta^4 \cdot \varphi_{\mathrm{w, L}_2}) + \frac{56 \kappa_0^2 \kappa_1^2 \kappa_2^2 \eta^2 \beta^2}{T} \sum_{t=0}^{T-1} \sum_{l=1}^{L} \alpha_l \sum_{k=1}^{B_l} \alpha_k \mathbb{E} \Vert \bar{\mathbf{w}}_l^t - \bar{\mathbf{w}}_k^t \Vert^2.
\end{align}
\end{Lemma}

Using Lemma \ref{Lemma_40} and Lemma \ref{Lemma_40_1}, when $\eta \leq \frac{1}{2\sqrt{14} \kappa_0 \kappa_1 \kappa_2 \beta}$, we have
\begin{align}
    &\frac{\beta^2}{T} \sum_{t=0}^{T-1} \sum_{l=1}^{L} \alpha_l \sum_{k=1}^{B_l} \alpha_k \mathbb{E} \left\Vert \bar{\mathbf{w}}_l^t - \bar{\mathbf{w}}_k^t \right\Vert^2  \nonumber\\
    &\leq \mathcal{O} \big( \delta^{\mathrm{th}} \beta^2 D^2 \big) + \mathcal{O} \big( \kappa_0\kappa_1\kappa_2 \eta^2 \beta^2 \sigma^2 \big) + \mathcal{O} \big( \kappa_0\kappa_1\kappa_2 \eta^2 \beta^2 G^2 \cdot \varphi_{\mathrm{w,L}_3} \big) + \mathcal{O} \big( \kappa_0^3 \kappa_1^2 \kappa_2^2 \eta^4 \beta^4 \sigma^2 \big) + \mathcal{O} \big( \kappa_0^3 \kappa_1^2 \kappa_2^2 \eta^4 \beta^4 \epsilon_{\mathrm{vc}}^2 \big) + \nonumber\\
    &\mathcal{O} \big(\kappa_0^4 \kappa_1^4 \kappa_2^2 \eta^4 \beta^4 \epsilon_{\mathrm{sbs}}^2 \big) + \mathcal{O} \big( \kappa_0^2 \kappa_1^2 \kappa_2^2 \eta^2 \beta^4 \epsilon_{\mathrm{mbs}}^2 \big) + \mathcal{O} \big( \delta^{th} \kappa_0^2 \kappa_1^2 \kappa_2^2 \eta^2 \beta^4 D^2 \big) + \mathcal{O} \big( \kappa_0^3 \kappa_1^2 \kappa_2^2 \eta^4 \beta^4 G^2 \cdot \varphi_{\mathrm{w, L}_1} \big) + \mathcal{O} \big( \kappa_0^2 \kappa_1^2 \kappa_2^2 \beta^4 \eta^4 \cdot \varphi_{\mathrm{w, L}_2}) \nonumber\\
    &\approx \mathcal{O} \big( \kappa_0^3 \kappa_1^2 \kappa_2^2 \eta^4 \beta^4 \epsilon_{\mathrm{vc}}^2 \big) + \mathcal{O} \big(\kappa_0^4 \kappa_1^4 \kappa_2^2 \eta^4 \beta^4 \epsilon_{\mathrm{sbs}}^2 \big) + \mathcal{O} \big( \kappa_0^2 \kappa_1^2 \kappa_2^2 \eta^2 \beta^4 \epsilon_{\mathrm{mbs}}^2 \big) + \mathcal{O} \big( \kappa_0\kappa_1\kappa_2 \eta^2 \beta^2 \sigma^2 \big) + \mathcal{O} \big( \delta^{\mathrm{th}} \beta^2 D^2 \big) + \nonumber\\
    &\mathcal{O} \big( \kappa_0^3 \kappa_1^2 \kappa_2^2 \eta^4 \beta^4 G^2 \cdot \varphi_{\mathrm{w, L}_1} \big) + \mathcal{O} \big( \kappa_0^2 \kappa_1^2 \kappa_2^2 \beta^4 \eta^4 \cdot \varphi_{\mathrm{w, L}_2}) +  \mathcal{O} \big( \kappa_0\kappa_1\kappa_2 \eta^2 \beta^2 G^2 \cdot \varphi_{\mathrm{w,L}_3} \big).
\end{align}

\subsection{Missing Proof of Lemma \ref{Lemma_40}}
\begin{align}
\label{Lemma_40_main_1}
    &\frac{4\eta^2}{T} \sum_{t=0}^{T-1} \sum_{l=1}^{L} \alpha_l \sum_{k=1}^{B_l} \alpha_k \mathbb{E} \Bigg \Vert \sum_{\tau=\Bar{t}_2}^{t-1} \rs \bigg[ \sum_{j=1}^{V_{k,l}} \alpha_j \sum_{i=1}^{U_{j,k,l}} \alpha_i \bigg( \tilde{g} \big(\tilde{\mathbf{w}}_i^{\tau,0} \big) \frac{\pmb{1}_{i}^\tau}{p_{i}^\tau} - \nabla \tilde{f}_i \big(\tilde{\mathbf{w}}_i^{\tau,0} \big) \bigg) - \sum_{k'=1}^{B_l} \alpha_{k'} \sum_{j'=1}^{V_{k',l}}  \alpha_{j'} \sum_{i'=1}^{U_{j',k',l} } \alpha_{i'} \bigg( \Tilde{g} \big(\tilde{\mathbf{w}}_{i'}^{\tau,0}\big) \frac{\pmb{1}_{i'}^\tau}{p_{i'}^\tau} - \nabla \tilde{f}_{i'} \big(\tilde{\mathbf{w}}_{i'}^{\tau,0} \big) \bigg) \bigg] \Bigg \Vert^2 \nonumber\\
    &\overset{(a)}{=} \frac{4\eta^2}{T} \sum_{t=0}^{T-1} \sum_{l=1}^{L} \alpha_l \sum_{k=1}^{B_l} \alpha_k \mathbb{E} \Bigg \Vert \sum_{\tau=\Bar{t}_2}^{t-1} \rs \bigg[ \sum_{j=1}^{V_{k,l}} \alpha_j \sum_{i=1}^{U_{j,k,l}} \alpha_i \bigg( \tilde{g} \big(\tilde{\mathbf{w}}_i^{\tau,0} \big) \frac{\pmb{1}_{i}^\tau}{p_{i}^\tau} - \nabla \tilde{f}_i \big(\tilde{\mathbf{w}}_i^{\tau,0} \big) \bigg) \bigg] \Bigg \Vert^2 - \nonumber\\
    &\Squad \frac{4\eta^2}{T} \sum_{t=0}^{T-1} \sum_{l=1}^{L} \alpha_l \mathbb{E} \Bigg \Vert \sum_{\tau=\Bar{t}_2}^{t-1} \rs \bigg[\sum_{k'=1}^{B_l} \alpha_{k'} \sum_{j'=1}^{V_{k',l}}  \alpha_{j'} \sum_{i'=1}^{U_{j',k',l} } \alpha_{i'} \bigg( \Tilde{g} \big(\tilde{\mathbf{w}}_{i'}^{\tau,0}\big) \frac{\pmb{1}_{i'}^\tau}{p_{i'}^\tau} - \nabla \tilde{f}_{i'} \big(\tilde{\mathbf{w}}_{i'}^{\tau,0} \big) \bigg) \bigg] \Bigg \Vert^2 \nonumber\\
    &\overset{(b)}{=} \frac{4\eta^2}{T} \sum_{t=0}^{T-1} \sum_{l=1}^{L} \alpha_l \sum_{k=1}^{B_l} \alpha_k \sum_{\tau=\Bar{t}_2}^{t-1} \mathbb{E} \Bigg \Vert  \sum_{j=1}^{V_{k,l}} \alpha_j \sum_{i=1}^{U_{j,k,l}} \alpha_i \bigg( \tilde{g} \big(\tilde{\mathbf{w}}_i^{\tau,0} \big) \frac{\pmb{1}_{i}^\tau}{p_{i}^\tau} - \nabla \tilde{f}_i \big(\tilde{\mathbf{w}}_i^{\tau,0} \big) \bigg) \Bigg \Vert^2 - \nonumber\\
    &\Squad \frac{4\eta^2}{T} \sum_{t=0}^{T-1} \sum_{l=1}^{L} \alpha_l \sum_{\tau=\Bar{t}_2}^{t-1} \mathbb{E} \Bigg \Vert \sum_{k'=1}^{B_l} \alpha_{k'} \sum_{j'=1}^{V_{k',l}}  \alpha_{j'} \sum_{i'=1}^{U_{j',k',l} } \alpha_{i'} \bigg( \Tilde{g} \big(\tilde{\mathbf{w}}_{i'}^{\tau,0}\big) \frac{\pmb{1}_{i'}^\tau}{p_{i'}^\tau} - \nabla \tilde{f}_{i'} \big(\tilde{\mathbf{w}}_{i'}^{\tau,0} \big) \bigg) \Bigg \Vert^2 \nonumber\\
    &\overset{(c)}{\leq} \frac{4 \kappa_0\kappa_1\kappa_2 \eta^2}{T} \sum_{t=0}^{T-1} \sum_{l=1}^{L} \alpha_l \sum_{k=1}^{B_l} \alpha_k \mathbb{E} \Bigg \Vert \sum_{j=1}^{V_{k,l}} \alpha_j \sum_{i=1}^{U_{j,k,l}} \alpha_i \bigg( \tilde{g} \big(\tilde{\mathbf{w}}_i^{t,0} \big) \frac{\pmb{1}_{i}^t}{p_{i}^t} - \nabla \tilde{f}_i \big(\tilde{\mathbf{w}}_i^{t,0} \big) \bigg) \Bigg \Vert^2 - \nonumber\\
    &\Squad \frac{4 \kappa_0\kappa_1\kappa_2 \eta^2}{T} \sum_{t=0}^{T-1} \sum_{l=1}^{L} \alpha_l \mathbb{E} \Bigg \Vert \sum_{k=1}^{B_l} \alpha_{k} \sum_{j=1}^{V_{k,l}}  \alpha_{j} \sum_{i=1}^{U_{j,k,l} } \alpha_{i} \bigg( \Tilde{g} \big(\tilde{\mathbf{w}}_{i}^{t,0}\big) \frac{\pmb{1}_{i}^t}{p_{i}^t} - \nabla \tilde{f}_{i} \big(\tilde{\mathbf{w}}_{i}^{t,0} \big) \bigg) \Bigg \Vert^2 \nonumber\\
    &\overset{(d)}{=} \frac{4 \kappa_0\kappa_1\kappa_2 \eta^2}{T} \sum_{t=0}^{T-1} \sum_{l=1}^{L} \alpha_l \sum_{k=1}^{B_l} \alpha_k \sum_{j=1}^{V_{k,l}} \alpha_j^2 \sum_{i=1}^{U_{j,k,l}} \alpha_i^2 \mathbb{E} \Bigg \Vert \tilde{g} \big(\tilde{\mathbf{w}}_i^{t,0} \big) \frac{\pmb{1}_{i}^t}{p_{i}^t} \pm \tilde{g} \big(\tilde{\mathbf{w}}_i^{t,0} \big) - \nabla \tilde{f}_i \big(\tilde{\mathbf{w}}_i^{t,0} \big) \Bigg \Vert^2 - \nonumber\\
    &\Squad \frac{4 \kappa_0\kappa_1\kappa_2 \eta^2}{T} \sum_{t=0}^{T-1} \sum_{l=1}^{L} \alpha_l \sum_{k=1}^{B_l} \alpha_{k}^2 \sum_{j=1}^{V_{k,l}}  \alpha_{j}^2 \sum_{i=1}^{U_{j,k,l} } \alpha_{i}^2 \mathbb{E} \Bigg \Vert \Tilde{g} \big(\tilde{\mathbf{w}}_{i}^{t,0}\big) \frac{\pmb{1}_{i}^t}{p_{i}^t} \pm \tilde{g} \big(\tilde{\mathbf{w}}_i^{t,0} \big) - \nabla \tilde{f}_{i} \big(\tilde{\mathbf{w}}_{i}^{t,0} \big) \Bigg \Vert^2 \nonumber \\
    &\leq 8 \kappa_0\kappa_1\kappa_2 \eta^2 \sigma^2 \sum_{l=1}^{L} \alpha_l \sum_{k=1}^{B_l} \alpha_k \sum_{j=1}^{V_{k,l}} \alpha_j^2 \sum_{i=1}^{U_{j,k,l}} \alpha_i^2 + 8 \kappa_0\kappa_1\kappa_2 \eta^2 G^2 \cdot \varphi_{\mathrm{w,L}_3} \nonumber\\
    & \approx \mathcal{O} \big( \kappa_0\kappa_1\kappa_2 \eta^2 \sigma^2 \big) + \mathcal{O} \big( \kappa_0\kappa_1\kappa_2 \eta^2 G^2 \cdot \varphi_{\mathrm{w,L}_3} \big),
\end{align}
where $\varphi_{\mathrm{w,L}_3} = \frac{1}{T} \sum_{l=1}^{L} \alpha_l \sum_{k=1}^{B_l} \alpha_k \sum_{j=1}^{V_{k,l}} \alpha_j^2 \sum_{i=1}^{U_{j,k,l}} \alpha_i^2 \sum_{t=0}^{T-1} \left(\frac{1}{p_i^t} - 1 \right)$.

\subsection{Missing Proof of Lemma \ref{Lemma_40_1}}
\begin{align}
\label{Lemma_40_1_Main}
    &\frac{4\eta^2}{T} \sum_{t=0}^{T-1} \sum_{l=1}^{L} \alpha_l \sum_{k=1}^{B_l} \alpha_k \mathbb{E} \Bigg \Vert \sum_{\tau=\Bar{t}_2}^{t-1} \rs \bigg[\sum_{k'=1}^{B_l} \alpha_{k'} \sum_{j'=1}^{V_{k',l}}  \alpha_{j'} \sum_{i'=1}^{U_{j',k',l} } \alpha_{i'} \nabla \tilde{f}_{i'} \big(\tilde{\mathbf{w}}_{i'}^{\tau,0} \big) - \sum_{j=1}^{V_{k,l}} \alpha_j \sum_{i=1}^{U_{j,k,l}} \alpha_i \nabla \tilde{f}_i \big(\tilde{\mathbf{w}}_i^{\tau,0} \big) \bigg] \Bigg\Vert^2  \nonumber\\
    &=\frac{4\kappa_0\kappa_1\kappa_2\eta^2}{T} \sum_{t=0}^{T-1} \sum_{l=1}^{L} \alpha_l \sum_{k=1}^{B_l} \alpha_k \sum_{\tau=\Bar{t}_2}^{t-1} \mathbb{E} \Bigg \Vert \bigg( \sum_{j=1}^{V_{k,l}} \alpha_j \sum_{i=1}^{U_{j,k,l}} \alpha_i \big[ \nabla \tilde{f}_i \big(\tilde{\mathbf{w}}_i^{\tau,0} \big) - \nabla \tilde{f}_i \big(\bar{\mathbf{w}}_j^{\tau} \big)\big]  \bigg) + \nonumber\\
    &\Squad \bigg( \sum_{j=1}^{V_{k,l}} \alpha_j \sum_{i=1}^{U_{j,k,l}} \alpha_i \big[ \nabla \tilde{f}_i \big(\bar{\mathbf{w}}_j^{\tau} \big) - \nabla \tilde{f}_i \big(\bar{\mathbf{w}}_k^{\tau} \big) \big]  \bigg) + \bigg(\sum_{j=1}^{V_{k,l}} \alpha_j \sum_{i=1}^{U_{j,k,l}} \alpha_i \big[ \nabla \tilde{f}_i \big(\bar{\mathbf{w}}_k^{\tau} \big) - \nabla \tilde{f}_i \big(\bar{\mathbf{w}}_l^{\tau} \big) \big] \bigg) + \nonumber \\
    &\Squad \bigg( \sum_{j=1}^{V_{k,l}} \alpha_j \sum_{i=1}^{U_{j,k,l}} \alpha_i \nabla \tilde{f}_i \big(\bar{\mathbf{w}}_l^{\tau} \big) - \sum_{k'=1}^{B_l} \alpha_{k'} \sum_{j'=1}^{V_{k',l}}  \alpha_{j'} \sum_{i'=1}^{U_{j',k',l} } \alpha_{i'} \nabla \tilde{f}_{i'} \big(\bar{\mathbf{w}}_{l}^{\tau} \big) \bigg) +  \nonumber \\
    &\Squad \bigg( \sum_{k'=1}^{B_l} \alpha_{k'} \sum_{j'=1}^{V_{k',l}}  \alpha_{j'} \sum_{i'=1}^{U_{j',k',l} } \alpha_{i'} \big[ \nabla \tilde{f}_{i'} \big(\bar{\mathbf{w}}_{l}^{\tau} \big) - \nabla \tilde{f}_{i'} \big(\bar{\mathbf{w}}_{k}^{\tau} \big) \big] \bigg) + \nonumber \\
    &\Squad \bigg( \sum_{k'=1}^{B_l} \alpha_{k'} \sum_{j'=1}^{V_{k',l}}  \alpha_{j'} \sum_{i'=1}^{U_{j',k',l} } \alpha_{i'} \big[ \nabla \tilde{f}_{i'} \big(\bar{\mathbf{w}}_{k}^{\tau} \big) - \nabla \tilde{f}_{i'} \big(\bar{\mathbf{w}}_{j}^{\tau} \big) \big] \bigg) + \nonumber \\
    &\Squad \bigg( \sum_{k'=1}^{B_l} \alpha_{k'} \sum_{j'=1}^{V_{k',l}}  \alpha_{j'} \sum_{i'=1}^{U_{j',k',l} } \alpha_{i'} \big[ \nabla \tilde{f}_{i'} \big(\bar{\mathbf{w}}_{j}^{\tau} \big) - \nabla \tilde{f}_{i'} \big(\tilde{\mathbf{w}}_{i'}^{\tau,0} \big) \big] \bigg) \Bigg\Vert^2 \nonumber \\
    &\leq 28 \kappa_0^2 \kappa_1^2 \kappa_2^2 \eta^2 \beta^2 \epsilon_{\mathrm{mbs}}^2 + \frac{128 \kappa_0^2 \kappa_1^2 \kappa_2^2 \eta^2 \beta^2}{T} \sum_{t=0}^{T-1} \sum_{l=1}^{L} \alpha_l \sum_{k=1}^{B_l} \alpha_k \sum_{j=1}^{V_{k,l}} \alpha_{j} \sum_{i=1}^{U_{j,k,l} } \alpha_{i} \mathbb{E} \Vert \tilde{\mathbf{w}}_i^t - \tilde{\mathbf{w}}_i^{t,0} \Vert^2 + \nonumber \\
    &\Squad \frac{128 \kappa_0^2 \kappa_1^2 \kappa_2^2 \eta^2 \beta^2}{T} \sum_{t=0}^{T-1} \sum_{l=1}^{L} \alpha_l \sum_{k=1}^{B_l} \alpha_k \sum_{j=1}^{V_{k,l}} \alpha_{j} \sum_{i=1}^{U_{j,k,l} } \alpha_{i} \mathbb{E} \Vert \bar{\mathbf{w}}_j^t - \tilde{\mathbf{w}}_i^{t} \Vert^2  + \nonumber \\ 
    &\Squad \frac{56 \kappa_0^2 \kappa_1^2 \kappa_2^2 \eta^2 \beta^2}{T} \sum_{t=0}^{T-1} \sum_{l=1}^{L} \alpha_l \sum_{k=1}^{B_l} \alpha_k \sum_{j=1}^{V_{k,l}} \alpha_{j} \mathbb{E} \Vert \bar{\mathbf{w}}_k^t - \bar{\mathbf{w}}_j^t \Vert^2 + \nonumber\\
    &\Squad \frac{56 \kappa_0^2 \kappa_1^2 \kappa_2^2 \eta^2 \beta^2}{T} \sum_{t=0}^{T-1} \sum_{l=1}^{L} \alpha_l \sum_{k=1}^{B_l} \alpha_k \mathbb{E} \Vert \bar{\mathbf{w}}_l^t - \bar{\mathbf{w}}_k^t \Vert^2 \nonumber \\
    &\leq \mathcal{O} \big( \kappa_0^2 \kappa_1^2 \kappa_2^2 \eta^2 \beta^2 \epsilon_{\mathrm{mbs}}^2 \big) + \mathcal{O} \big( \delta^{\mathrm{th}} \kappa_0^2 \kappa_1^2 \kappa_3^2 \eta^2 \beta^2 D^2  \big) + \mathcal{O} \big( \kappa_0^3 \kappa_1^2 \kappa_2^2 \eta^4 \beta^2 \sigma^2 \big) + \mathcal{O} \big( \kappa_0^3 \kappa_1^2 \kappa_2^2 \eta^4 \beta^2 \epsilon_{\mathrm{vc}}^2 \big) + \nonumber\\
    & \mathcal{O} \big( \kappa_0^3 \kappa_1^2 \kappa_2^2 \eta^4 \beta^2 G^2 \cdot \varphi_{\mathrm{w, L}_1} \big) + \mathcal{O} \big(\delta^{\mathrm{th}} \kappa_0^2 \kappa_1^2 \kappa_2^2 \eta^2 \beta^2 D^2 \big) + \mathcal{O} \big(\beta^4 \kappa_0^6 \kappa_1^4 \kappa_2^2 \eta^6 \epsilon_{\mathrm{vc}}^2 \big) + \mathcal{O} \big(\kappa_0^4 \kappa_1^4 \kappa_2^2 \eta^4 \beta^2 \epsilon_{\mathrm{sbs}}^2 \big) + \nonumber\\
    & \mathcal{O} \big( \kappa_0^3 \kappa_1^3 \kappa_2^2 \eta^4 \sigma^2 \beta^2 \big) + \mathcal{O } \big( \delta^{th} \kappa_0^2 \kappa_1^2 \kappa_2^2 \eta^2 \beta^2 D^2 \big) + \mathcal{O} \big( \kappa_0^5 \kappa_1^4 \kappa_2^2 \beta^4 \eta^6 G^2 \cdot \varphi_{\mathrm{w, L}_1} \big) +  \nonumber\\
    & \mathcal{O} \big( \kappa_0^2 \kappa_1^2 \kappa_2^2 \beta^2 \eta^4 \cdot \varphi_{\mathrm{w, L}_2}) + \frac{56 \kappa_0^2 \kappa_1^2 \kappa_2^2 \eta^2 \beta^2}{T} \sum_{t=0}^{T-1} \sum_{l=1}^{L} \alpha_l \sum_{k=1}^{B_l} \alpha_k \mathbb{E} \Vert \bar{\mathbf{w}}_l^t - \bar{\mathbf{w}}_k^t \Vert^2 \nonumber \\
    & \approx \mathcal{O} \big( \kappa_0^3 \kappa_1^2 \kappa_2^2 \eta^4 \beta^2 \sigma^2 \big) + \mathcal{O} \big( \kappa_0^3 \kappa_1^2 \kappa_2^2 \eta^4 \beta^2 \epsilon_{\mathrm{vc}}^2 \big) + \mathcal{O} \big(\kappa_0^4 \kappa_1^4 \kappa_2^2 \eta^4 \beta^2 \epsilon_{\mathrm{sbs}}^2 \big) + \mathcal{O} \big( \kappa_0^2 \kappa_1^2 \kappa_2^2 \eta^2 \beta^2 \epsilon_{\mathrm{mbs}}^2 \big) +  \nonumber\\
    &\Squad \mathcal{O } \big( \delta^{th} \kappa_0^2 \kappa_1^2 \kappa_2^2 \eta^2 \beta^2 D^2 \big) +  \mathcal{O} \big( \kappa_0^3 \kappa_1^2 \kappa_2^2 \eta^4 \beta^2 G^2 \cdot \varphi_{\mathrm{w, L}_1} \big)  + \nonumber\\
    &\Squad \mathcal{O} \big( \kappa_0^2 \kappa_1^2 \kappa_2^2 \beta^2 \eta^4 \cdot \varphi_{\mathrm{w, L}_2}) + \frac{56 \kappa_0^2 \kappa_1^2 \kappa_2^2 \eta^2 \beta^2}{T} \sum_{t=0}^{T-1} \sum_{l=1}^{L} \alpha_l \sum_{k=1}^{B_l} \alpha_k \mathbb{E} \Vert \bar{\mathbf{w}}_l^t - \bar{\mathbf{w}}_k^t \Vert^2. 
\end{align}

\end{proof}

\section{Proof of Lemma \ref{Lemma5}}
\setcounter{Lemma}{3}
\begin{Lemma}
\label{Lemma5}
When $\eta \leq 1/[6\sqrt{2} \kappa_0 \kappa_1 \kappa_2 \kappa_3 \beta]$, the average difference between the global and the mBS models, i.e., the $\mathrm{L}_4$ term, is bounded as follows:
\begin{align}
\label{Lem5_Sup}
    &[\beta^2/T] \sum\nolimits_{t=0}^{T-1} \sum\nolimits_{l=1}^{L} \alpha_l \mathbb{E} \left\Vert \bar{\mathbf{w}}^t - \bar{\mathbf{w}}_l^t \right\Vert^2 
    \leq \mathcal{O} \big( \kappa_0^4 \kappa_1^2 \kappa_2^2 \kappa_3^2 \eta^4 \beta^4 \epsilon_{\mathrm{vc}}^2 \big) + \mathcal{O} \big(\kappa_0^4 \kappa_1^4 \kappa_2^2 \kappa_3^2 \eta^4 \beta^4 \epsilon_{\mathrm{sbs}}^2 \big) + \nonumber\\
    &\qquad \mathcal{O} \big( \kappa_0^4 \kappa_1^4 \kappa_2^4 \kappa_3^2 \eta^4 \beta^6 \epsilon_{\mathrm{mbs}}^2 \big) + \mathcal{O} \big(\kappa_0^2 \kappa_1^2 \kappa_2^2 \kappa_3^2 \beta^4 \eta^2 \epsilon^2 \big) + \mathcal{O} \big( \kappa_0\kappa_1\kappa_2\kappa_3 \beta^2 \eta^2 \sigma^2 \big) + \mathcal{O} \big( \delta^{\mathrm{th}} \beta^2 D^2 \big) + \nonumber\\
    &\qquad \mathcal{O} \big( \kappa_0^3 \kappa_1^2 \kappa_2^2 \kappa_3^2 \eta^4 \beta^4 G^2 \cdot \varphi_{\mathrm{w, L}_1} (\pmb{\delta},\pmb{\mathrm{f}}, \pmb{\mathrm{P}}) \big) + \mathcal{O} \big( \kappa_0^3 \kappa_1^3 \kappa_2^2 \kappa_3^2 \beta^4 \eta^4 \cdot \varphi_{\mathrm{w, L}_2} (\pmb{\delta},\pmb{\mathrm{f}}, \pmb{\mathrm{P}}) \big) + \nonumber\\
    &\qquad \mathcal{O} \big( \kappa_0^3 \kappa_1^3 \kappa_2^3 \kappa_3^2 \eta^4 \beta^4 G^2 \cdot \varphi_{\mathrm{w,L}_3} (\pmb{\delta},\pmb{\mathrm{f}}, \pmb{\mathrm{P}}) \big) + \mathcal{O} \big(\kappa_0\kappa_1\kappa_2\kappa_3 \beta^2 \eta^2 G^2 \cdot \varphi_{\mathrm{w,L}_4} (\pmb{\delta},\pmb{\mathrm{f}}, \pmb{\mathrm{P}}) \big), 
\end{align}
where $\varphi_{\mathrm{w,L}_4} (\pmb{\delta},\pmb{\mathrm{f}}, \pmb{\mathrm{P}}) = [1/T] \sum_{l=1}^{L} \alpha_l \sum_{k=1}^{B_l} \alpha_k^2 \sum_{j=1}^{V_{k,l}} \alpha_j^2 \sum_{i=1}^{U_{j,k,l}} \alpha_i^2  \sum_{t = 0}^{T-1} \left(1/p_i^t - 1 \right)$.
\end{Lemma}

\begin{proof}
\label{proofLemma5}

\begin{align}
\label{Lemma5_Main_0}
    &\frac{1}{T} \sum_{t=0}^{T-1} \sum_{l=1}^{L} \alpha_l \mathbb{E} \left\Vert \bar{\mathbf{w}}^t - \bar{\mathbf{w}}_{l}^t \right\Vert^2 \nonumber\\ 
    &= \frac{1}{T} \sum_{t=0}^{T-1}  \sum_{l=1}^{L} \alpha_l \mathbb{E} \bigg\Vert \tilde{ \mathbf{w} }^{m\prod_{z=0}^3\kappa_z, 0} - \eta \sum_{\tau=m\prod_{z=0}^3\kappa_z}^{t-1}\sum_{l'=1}^L \alpha_{l'} \sum_{k'=1}^{B_{l'}} \alpha_{k'}  \sum_{j'=1}^{V_{k',l'}} \alpha_{j'} \sum_{i'=1}^{U_{j',k',l'}} \alpha_{i'} \Tilde{g} \big(\tilde{\mathbf{w}}_{i'}^{\tau,0}\big) \frac{\pmb{1}_{i'}^\tau}{p_{i}^\tau} - \nonumber\\
    &\Squad \Big( \tilde{ \mathbf{w} }_l^{m\prod_{z=0}^3\kappa_z, 0} - \eta \sum_{\tau = m\prod_{z=0}^3\kappa_z} \sum_{k=1}^{B_l} \alpha_k \sum_{j=1}^{V_{k,l}} \alpha_j \sum_{i=1}^{U_{j,k,l} } \alpha_i \Tilde{g} \big( \tilde{\mathbf{w}}_i^{\tau,0} \big) \frac{\pmb{1}_i^\tau}{p_i^\tau} \Big) \bigg\Vert^2,  \nonumber \\
    &\leq \frac{2}{T} \sum_{t=0}^{T-1}  \sum_{l=1}^{L} \alpha_l \mathbb{E} \left\Vert \tilde{ \mathbf{w} }^{m\prod_{z=0}^3\kappa_z, 0} - \tilde{ \mathbf{w} }_l^{m\prod_{z=0}^3\kappa_z, 0} \right\Vert^2 + \nonumber \\
    &\Squad \frac{2 \eta^2}{T} \sum_{t=0}^{T-1}  \sum_{l=1}^{L} \alpha_l \mathbb{E} \bigg\Vert \sum_{\tau = m\prod_{z=0}^3\kappa_z}^{t-1} \bigg[ \sum_{k=1}^{B_l} \alpha_k \sum_{j=1}^{V_{k,l}} \alpha_j \sum_{i=1}^{U_{j,k,l} } \alpha_i \Tilde{g} \big( \tilde{\mathbf{w}}_i^{\tau,0} \big) \frac{\pmb{1}_i^\tau}{p_i^\tau} -  \sum_{l'=1}^L \alpha_{l'} \sum_{k'=1}^{B_{l'}} \alpha_{k'}  \sum_{j'=1}^{V_{k',l'}} \alpha_{j'} \sum_{i'=1}^{U_{j',k',l'}} \alpha_{i'} \Tilde{g} \big(\tilde{\mathbf{w}}_{i'}^{\tau,0}\big) \frac{\pmb{1}_{i'}^\tau}{p_{i'}^\tau} \bigg] \bigg\Vert^2, 
\end{align}
where the last inequality follows from Jensen inequality. 

For the first term of (\ref{Lemma5_Main_0}), we have
\begin{align}
\label{Lemma5_Main_0_FirstTerm}
    & 2 \sum_{l=1}^{L} \alpha_l \mathbb{E} \left\Vert \tilde{ \mathbf{w} }^{m\prod_{z=0}^3\kappa_z, 0} - \tilde{ \mathbf{w} }_l^{m\prod_{z=0}^3\kappa_z, 0} \right\Vert^2, \nonumber\\ &\overset{(a)}{=} 2 \sum_{l=1}^{L} \alpha_l \mathbb{E} \left\Vert \tilde{ \mathbf{w} }^{m\prod_{z=0}^3\kappa_z, 0} - \mathbf{w}^{m\prod_{z=0}^3\kappa_z} + \mathbf{w}_l^{m\prod_{z=0}^3\kappa_z} - \tilde{ \mathbf{w} }_l^{m\prod_{z=0}^3\kappa_z, 0} \right\Vert^2, \nonumber \\
    &\overset{(b)}{\leq} 4 \sum_{l=1}^{L} \alpha_l \mathbb{E} \left\Vert \tilde{ \mathbf{w} }^{m\prod_{z=0}^3\kappa_z, 0} - \mathbf{w}^{m\prod_{z=0}^3\kappa_z} \right\Vert^2 +  4 \sum_{l=1}^{L} \alpha_l \mathbb{E} \left\Vert \mathbf{w}_l^{m\prod_{z=0}^3\kappa_z} - \tilde{ \mathbf{w} }_l^{m\prod_{z=0}^3\kappa_z, 0} \right\Vert^2, \nonumber \\
    &\leq \bblue{8} \sum_{l=1}^{L} \alpha_l \sum_{k=1}^{B_l} \bblue{\alpha_k} \sum_{j=1}^{V_{k,l}} \bblue{\alpha_j} \sum_{i=1}^{U_{j,k,l}} \bblue{\alpha_i} \delta_i^{m\prod_{z=0}^3\kappa_z} \mathbb{E} \left\Vert \mathbf{w}_i^{m\prod_{z=0}^3\kappa_z} \right\Vert^2,
\end{align}
where in $(a)$, we use the fact that $\mathbf{w}^{m\prod_{z=0}^3\kappa_z} = \mathbf{w}_l^{m\prod_{z=0}^3\kappa_z}$ and $(b)$ stems from $\Vert \sum_{i=1}^n \mathbf{a}_i \Vert^2 \leq n \sum_{i=1}^n \Vert \mathbf{a}_i \Vert^2$.
Moreover, the last inequality appears following steps as in (\ref{globalPrune}) and (\ref{mBSPrune}).

As such, we simplify the first term as
\begin{align}
    &\frac{2}{T} \sum_{t=0}^{T-1} \sum_{l=1}^{L} \alpha_l \mathbb{E} \left\Vert \tilde{ \mathbf{w} }^{m\prod_{z=0}^3\kappa_z, 0} - \tilde{ \mathbf{w} }_l^{m\prod_{z=0}^3\kappa_z, 0} \right\Vert^2 
    \leq \frac{\bblue{8}}{T} \sum_{t=0}^{T-1} \sum_{l=1}^{L} \alpha_l \sum_{k=1}^{B_l} \bblue{\alpha_k} \sum_{j=1}^{V_{k,l}} \bblue{\alpha_j} \sum_{i=1}^{U_{j,k,l}} \bblue{\alpha_i} \delta_i^{\left\lfloor t/[\prod_{z=0}^3\kappa_z] \right\rfloor} \mathbb{E} \left\Vert \mathbf{w}_i^{\left\lfloor t/[\prod_{z=0}^3\kappa_z] \right\rfloor} \right\Vert^2 \nonumber\\
    &\approx \mathcal{O} \big( \delta^{\mathrm{th}} D^2 \big).
\end{align}

The second term of (\ref{Lemma5_Main_0}) is further simplified as follows:
\begin{align}
\label{Lemma5_Main_0_SecondTerm}
    & \frac{2 \eta^2}{T} \sum_{t=0}^{T-1} \sum_{l=1}^{L} \alpha_l \mathbb{E} \bigg\Vert \sum_{\tau = m\prod_{z=0}^3\kappa_z}^{t-1} \bigg[ \sum_{k=1}^{B_l} \alpha_k \sum_{j=1}^{V_{k,l}} \alpha_j \sum_{i=1}^{U_{j,k,l} } \alpha_i \Tilde{g} \big( \tilde{\mathbf{w}}_i^{\tau,0} \big) \frac{\pmb{1}_i^\tau}{p_i^\tau} -  \sum_{l'=1}^L \alpha_{l'} \sum_{k'=1}^{B_{l'}} \alpha_{k'}  \sum_{j'=1}^{V_{k',l'}} \alpha_{j'} \sum_{i'=1}^{U_{j',k',l'}} \alpha_{i'} \Tilde{g} \big(\tilde{\mathbf{w}}_{i'}^{\tau,0}\big) \frac{\pmb{1}_{i'}^\tau}{p_{i'}^\tau} \bigg] \bigg\Vert^2 \nonumber\\
    &= \frac{2 \eta^2}{T} \sum_{t=0}^{T-1} \sum_{l=1}^{L} \alpha_l \mathbb{E} \bigg\Vert \sum_{\tau = m\prod_{z=0}^3\kappa_z}^{t-1} \bigg[ \sum_{k=1}^{B_l} \alpha_k \sum_{j=1}^{V_{k,l}} \alpha_j \sum_{i=1}^{U_{j,k,l} } \alpha_i \bigg( \Tilde{g} \big( \tilde{\mathbf{w}}_i^{\tau,0} \big) \frac{\pmb{1}_i^\tau}{p_i^\tau} - \nabla \tilde{f}_i \big( \tilde{\mathbf{w}}_i^{\tau,0} \big) \bigg) - \nonumber\\
    &\Squad \sum_{l'=1}^L \alpha_{l'} \sum_{k'=1}^{B_{l'}} \alpha_{k'}  \sum_{j'=1}^{V_{k',l'}} \alpha_{j'} \sum_{i'=1}^{U_{j',k',l'}} \alpha_{i'} \bigg( \Tilde{g} \big(\tilde{\mathbf{w}}_{i'}^{\tau,0}\big) \frac{\pmb{1}_{i'}^\tau}{p_{i'}^\tau} -  \Tilde{f}_{i'} \big(\tilde{\mathbf{w}}_{i'}^{\tau,0}\big) \bigg) + \nonumber\\
    &\Squad \sum_{k=1}^{B_l} \alpha_k \sum_{j=1}^{V_{k,l}} \alpha_j \sum_{i=1}^{U_{j,k,l} } \alpha_i \nabla \tilde{f}_i \big( \tilde{\mathbf{w}}_i^{\tau,0} \big) - \sum_{l'=1}^L \alpha_{l'} \sum_{k'=1}^{B_{l'}} \alpha_{k'}  \sum_{j'=1}^{V_{k',l'}} \alpha_{j'} \sum_{i'=1}^{U_{j',k',l'}} \alpha_{i'} \Tilde{f}_{i'} \big(\tilde{\mathbf{w}}_{i'}^{\tau,0}\big) \bigg] \bigg\Vert^2 \nonumber\\
    &\leq \frac{4 \eta^2}{T} \sum_{t=0}^{T-1} \sum_{l=1}^{L} \alpha_l \mathbb{E} \bigg\Vert \sum_{\tau = m\prod_{z=0}^3\kappa_z}^{t-1} \bigg[ \sum_{k=1}^{B_l} \alpha_k \sum_{j=1}^{V_{k,l}} \alpha_j \sum_{i=1}^{U_{j,k,l} } \alpha_i \bigg( \Tilde{g} \big( \tilde{\mathbf{w}}_i^{\tau,0} \big) \frac{\pmb{1}_i^\tau}{p_i^\tau} - \nabla \tilde{f}_i \big( \tilde{\mathbf{w}}_i^{\tau,0} \big) \bigg) - \nonumber\\
    &\Squad \sum_{l'=1}^L \alpha_{l'} \sum_{k'=1}^{B_{l'}} \alpha_{k'}  \sum_{j'=1}^{V_{k',l'}} \alpha_{j'} \sum_{i'=1}^{U_{j',k',l'}} \alpha_{i'} \bigg( \Tilde{g} \big(\tilde{\mathbf{w}}_{i'}^{\tau,0}\big) \frac{\pmb{1}_{i'}^\tau}{p_{i'}^\tau} -  \Tilde{f}_{i'} \big(\tilde{\mathbf{w}}_{i'}^{\tau,0}\big) \bigg) \bigg] \bigg\Vert^2 + \nonumber\\
    &\Squad\frac{4 \eta^2}{T} \sum_{t=0}^{T-1} \sum_{l=1}^{L} \alpha_l \mathbb{E} \bigg\Vert \sum_{\tau = m\prod_{z=0}^3\kappa_z}^{t-1} \bigg[ \sum_{k=1}^{B_l} \alpha_k \sum_{j=1}^{V_{k,l}} \alpha_j \sum_{i=1}^{U_{j,k,l} } \alpha_i \nabla \tilde{f}_i \big( \tilde{\mathbf{w}}_i^{\tau,0} \big) -  \sum_{l'=1}^L \alpha_{l'} \sum_{k'=1}^{B_{l'}} \alpha_{k'}  \sum_{j'=1}^{V_{k',l'}} \alpha_{j'} \sum_{i'=1}^{U_{j',k',l'}} \alpha_{i'} \Tilde{f}_{i'} \big(\tilde{\mathbf{w}}_{i'}^{\tau,0}\big) \bigg] \bigg\Vert^2.
\end{align}

\setcounter{Lemma}{10}
\begin{Lemma}
\label{Lemma_50}
The first term of (\ref{Lemma5_Main_0_SecondTerm}) is bounded as follow:
\begin{align}
\label{Lemma_50_Main_0}
    & \frac{4 \eta^2}{T} \sum_{t=0}^{T-1} \sum_{l=1}^{L} \alpha_l \mathbb{E} \bigg\Vert \sum_{\tau = m\prod_{z=0}^3\kappa_z}^{t-1} \bigg[ \sum_{k=1}^{B_l} \alpha_k \sum_{j=1}^{V_{k,l}} \alpha_j \sum_{i=1}^{U_{j,k,l} } \alpha_i \bigg( \Tilde{g} \big( \tilde{\mathbf{w}}_i^{\tau,0} \big) \frac{\pmb{1}_i^\tau}{p_i^\tau} - \nabla \tilde{f}_i \big( \tilde{\mathbf{w}}_i^{\tau,0} \big) \bigg) - \nonumber\\
    &\Squad \sum_{l'=1}^L \alpha_{l'} \sum_{k'=1}^{B_{l'}} \alpha_{k'}  \sum_{j'=1}^{V_{k',l'}} \alpha_{j'} \sum_{i'=1}^{U_{j',k',l'}} \alpha_{i'} \bigg( \Tilde{g} \big(\tilde{\mathbf{w}}_{i'}^{\tau,0}\big) \frac{\pmb{1}_{i'}^\tau}{p_{i'}^\tau} -  \Tilde{f}_{i'} \big(\tilde{\mathbf{w}}_{i'}^{\tau,0}\big) \bigg) \bigg] \bigg\Vert^2 \nonumber \\
    & \leq 8 \Big(\prod_{z=0}^3\kappa_z \Big) \eta^2 \sigma^2 \sum_{l=1}^{L} \alpha_l \sum_{k=1}^{B_l} \alpha_k^2 \sum_{j=1}^{V_{k,l}} \alpha_j^2  \sum_{i=1}^{U_{j,k,l}} \alpha_i^2 +  \frac{8 \Big(\prod_{z=0}^3\kappa_z \Big) \eta^2 G^2}{T} \sum_{l=1}^{L} \alpha_l \sum_{k=1}^{B_l} \alpha_k^2 \sum_{j=1}^{V_{k,l}} \alpha_j^2 \sum_{i=1}^{U_{j,k,l}} \alpha_i^2  \sum_{t = 0}^{T-1} \bigg(\frac{1 - p_i^t}{p_i^t}\bigg) \nonumber\\
    & \approx \mathcal{O} \big( \kappa_0\kappa_1\kappa_2\kappa_3 \eta^2 \sigma^2 \big) + \mathcal{O} \big(\kappa_0\kappa_1\kappa_2\kappa_3 \eta^2 G^2 \cdot \varphi_{\mathrm{w,L}_4}\big),  
\end{align}
where $\varphi_{\mathrm{w,L}_4} = \frac{1}{T} \sum_{l=1}^{L} \alpha_l \sum_{k=1}^{B_l} \alpha_k^2 \sum_{j=1}^{V_{k,l}} \alpha_j^2 \sum_{i=1}^{U_{j,k,l}} \alpha_i^2  \sum_{t = 0}^{T-1} \big(\frac{1 - p_i^t}{p_i^t}\big)$.
\end{Lemma}

\begin{Lemma}
\label{Lemma_50_1}
The second term of (\ref{Lemma5_Main_0_SecondTerm}) is bounded as follow:
\begin{align}
    & \frac{4 \eta^2}{T} \sum_{t=0}^{T-1} \sum_{l=1}^{L} \alpha_l \mathbb{E} \bigg\Vert \sum_{\tau = m\prod_{z=0}^3\kappa_z}^{t-1} \bigg[ \sum_{k=1}^{B_l} \alpha_k \sum_{j=1}^{V_{k,l}} \alpha_j \sum_{i=1}^{U_{j,k,l} } \alpha_i \nabla \tilde{f}_i \big( \tilde{\mathbf{w}}_i^{\tau,0} \big) - \sum_{l'=1}^L \alpha_{l'} \sum_{k'=1}^{B_{l'}} \alpha_{k'}  \sum_{j'=1}^{V_{k',l'}} \alpha_{j'} \sum_{i'=1}^{U_{j',k',l'}} \alpha_{i'} \Tilde{f}_{i'} \big(\tilde{\mathbf{w}}_{i'}^{\tau,0}\big) \bigg] \bigg\Vert^2 \nonumber\\
    &\leq \mathcal{O} \big( \kappa_0^4 \kappa_1^2 \kappa_2^2 \kappa_3^2 \eta^4 \beta^2 \epsilon_{\mathrm{vc}}^2 \big) + \mathcal{O} \big(\kappa_0^4 \kappa_1^4 \kappa_2^2 \kappa_3^2 \eta^4 \beta^2 \epsilon_{\mathrm{sbs}}^2 \big) + \mathcal{O} \big( \kappa_0^4 \kappa_1^4 \kappa_2^4 \kappa_3^2 \eta^4 \beta^4 \epsilon_{\mathrm{mbs}}^2 \big) + \nonumber\\
    &\Squad \mathcal{O} \big(\kappa_0^2 \kappa_1^2 \kappa_2^2 \kappa_3^2 \beta^2 \eta^2 \epsilon^2 \big) + \mathcal{O} \big( \kappa_0^3 \kappa_1^2 \kappa_2^2 \kappa_3^2 \eta^4 \beta^2 \sigma^2 \big)  + \mathcal{O} \big( \delta^{\mathrm{th}} \kappa_0^2 \kappa_1^2 \kappa_2^2 \kappa_3^2 \eta^2 \beta^2 D^2 \big) + \nonumber \\
    &\Squad \mathcal{O} \big( \kappa_0^3 \kappa_1^2 \kappa_2^2 \kappa_3^2 \eta^4 \beta^2 G^2 \cdot \varphi_{\mathrm{w, L}_1} \big) + \mathcal{O} \big( \kappa_0^3 \kappa_1^3 \kappa_2^2 \kappa_3^2 \beta^2 \eta^4 \cdot \varphi_{\mathrm{w, L}_2}) + \nonumber\\
    &\Squad \rs \rs \mathcal{O} \big( \kappa_0^3 \kappa_1^3 \kappa_2^3 \kappa_3^2 \eta^4 \beta^2 G^2 \cdot \varphi_{\mathrm{w,L}_3} \big) + \frac{72 \big(\beta \eta \kappa_0 \kappa_1 \kappa_2 \kappa_3 \big)^2}{T} \sum_{t=0}^{T-1} \sum_{l=1}^{L} \alpha_l  \mathbb{E} \bigg\Vert \Bar{\mathbf{w}}^t - \Bar{\mathbf{w}}_l^t \bigg\Vert^2.
\end{align}
\end{Lemma}

Using Lemma \ref{Lemma_50} and Lemma \ref{Lemma_50_1}, and assuming $\eta \leq \frac{1}{6\sqrt{2} \kappa_0 \kappa_1 \kappa_2 \kappa_3 \beta}$, we get
\begin{align}
    &\frac{\beta^2}{T} \sum_{t=0}^{T-1} \sum_{l=1}^{L} \alpha_l \mathbb{E} \left\Vert \bar{\mathbf{w}}^t - \bar{\mathbf{w}}_{l}^t \right\Vert^2 \nonumber\\ 
    &\leq \mathcal{O} \big( \delta^{\mathrm{th}} \beta^2 D^2 \big) + \mathcal{O} \big( \kappa_0\kappa_1\kappa_2\kappa_3 \beta^2 \eta^2 \sigma^2 \big) + \mathcal{O} \big(\kappa_0\kappa_1\kappa_2\kappa_3 \beta^2 \eta^2 G^2 \cdot \varphi_{\mathrm{w,L}_4}\big) + \nonumber\\
    &\Squad \mathcal{O} \big( \kappa_0^4 \kappa_1^2 \kappa_2^2 \kappa_3^2 \eta^4 \beta^4 \epsilon_{\mathrm{vc}}^2 \big) + \mathcal{O} \big(\kappa_0^4 \kappa_1^4 \kappa_2^2 \kappa_3^2 \eta^4 \beta^4 \epsilon_{\mathrm{sbs}}^2 \big) + \mathcal{O} \big( \kappa_0^4 \kappa_1^4 \kappa_2^4 \kappa_3^2 \eta^4 \beta^6 \epsilon_{\mathrm{mbs}}^2 \big) + \nonumber\\
    &\Squad \mathcal{O} \big(\kappa_0^2 \kappa_1^2 \kappa_2^2 \kappa_3^2 \beta^4 \eta^2 \epsilon^2 \big) + \mathcal{O} \big( \kappa_0^3 \kappa_1^2 \kappa_2^2 \kappa_3^2 \eta^4 \beta^4 \sigma^2 \big)  + \mathcal{O} \big( \delta^{\mathrm{th}} \kappa_0^2 \kappa_1^2 \kappa_2^2 \kappa_3^2 \eta^2 \beta^4 D^2 \big) + \nonumber \\
    &\Squad \mathcal{O} \big( \kappa_0^3 \kappa_1^2 \kappa_2^2 \kappa_3^2 \eta^4 \beta^4 G^2 \cdot \varphi_{\mathrm{w, L}_1} \big) + \mathcal{O} \big( \kappa_0^3 \kappa_1^3 \kappa_2^2 \kappa_3^2 \beta^4 \eta^4 \cdot \varphi_{\mathrm{w, L}_2}) + \nonumber\\
    &\Squad \rs \rs \mathcal{O} \big( \kappa_0^3 \kappa_1^3 \kappa_2^3 \kappa_3^2 \eta^4 \beta^4 G^2 \cdot \varphi_{\mathrm{w,L}_3} \big) \nonumber\\
    &\approx \mathcal{O} \big( \kappa_0^4 \kappa_1^2 \kappa_2^2 \kappa_3^2 \eta^4 \beta^4 \epsilon_{\mathrm{vc}}^2 \big) + \mathcal{O} \big(\kappa_0^4 \kappa_1^4 \kappa_2^2 \kappa_3^2 \eta^4 \beta^4 \epsilon_{\mathrm{sbs}}^2 \big) + \mathcal{O} \big( \kappa_0^4 \kappa_1^4 \kappa_2^4 \kappa_3^2 \eta^4 \beta^6 \epsilon_{\mathrm{mbs}}^2 \big) + \nonumber\\
    &\Squad \mathcal{O} \big(\kappa_0^2 \kappa_1^2 \kappa_2^2 \kappa_3^2 \beta^4 \eta^2 \epsilon^2 \big) + \mathcal{O} \big( \kappa_0\kappa_1\kappa_2\kappa_3 \beta^2 \eta^2 \sigma^2 \big) + \mathcal{O} \big( \delta^{\mathrm{th}} \beta^2 D^2 \big) + \nonumber\\
    &\Squad \mathcal{O} \big( \kappa_0^3 \kappa_1^2 \kappa_2^2 \kappa_3^2 \eta^4 \beta^4 G^2 \cdot \varphi_{\mathrm{w, L}_1} \big) + \mathcal{O} \big( \kappa_0^3 \kappa_1^3 \kappa_2^2 \kappa_3^2 \beta^4 \eta^4 \cdot \varphi_{\mathrm{w, L}_2}) + \nonumber\\
    &\Squad \mathcal{O} \big( \kappa_0^3 \kappa_1^3 \kappa_2^3 \kappa_3^2 \eta^4 \beta^4 G^2 \cdot \varphi_{\mathrm{w,L}_3} \big) + \mathcal{O} \big(\kappa_0\kappa_1\kappa_2\kappa_3 \beta^2 \eta^2 G^2 \cdot \varphi_{\mathrm{w,L}_4}\big).
\end{align}

\subsection{Missing Proof of Lemma \ref{Lemma_50}}
\begin{align}
    & \frac{4 \eta^2}{T} \sum_{t=0}^{T-1} \sum_{l=1}^{L} \alpha_l \mathbb{E} \bigg\Vert \sum_{\tau = m\prod_{z=0}^3\kappa_z}^{t-1} \bigg[ \sum_{k=1}^{B_l} \alpha_k \sum_{j=1}^{V_{k,l}} \alpha_j \sum_{i=1}^{U_{j,k,l} } \alpha_i \bigg( \Tilde{g} \big( \tilde{\mathbf{w}}_i^{\tau,0} \big) \frac{\pmb{1}_i^\tau}{p_i^\tau} - \nabla \tilde{f}_i \big( \tilde{\mathbf{w}}_i^{\tau,0} \big) \bigg) - \nonumber\\
    &\Squad \sum_{l'=1}^L \alpha_{l'} \sum_{k'=1}^{B_{l'}} \alpha_{k'}  \sum_{j'=1}^{V_{k',l'}} \alpha_{j'} \sum_{i'=1}^{U_{j',k',l'}} \alpha_{i'} \bigg( \Tilde{g} \big(\tilde{\mathbf{w}}_{i'}^{\tau,0}\big) \frac{\pmb{1}_{i'}^\tau}{p_{i'}^\tau} -  \Tilde{f}_{i'} \big(\tilde{\mathbf{w}}_{i'}^{\tau,0}\big) \bigg) \bigg] \bigg\Vert^2 \nonumber \\
    &\overset{(a)}{=} \frac{4 \eta^2}{T} \sum_{t=0}^{T-1} \sum_{l=1}^{L} \alpha_l \mathbb{E} \bigg\Vert \sum_{\tau = m\prod_{z=0}^3\kappa_z}^{t-1} \bigg[ \sum_{k=1}^{B_l} \alpha_k \sum_{j=1}^{V_{k,l}} \alpha_j \sum_{i=1}^{U_{j,k,l} } \alpha_i \bigg( \Tilde{g} \big( \tilde{\mathbf{w}}_i^{\tau,0} \big) \frac{\pmb{1}_i^\tau}{p_i^\tau} - \nabla \tilde{f}_i \big( \tilde{\mathbf{w}}_i^{\tau,0} \big) \bigg) \bigg] \bigg\Vert^2 - \nonumber\\
    &\Squad \frac{4 \eta^2}{T} \sum_{t=0}^{T-1} \mathbb{E} \bigg\Vert \sum_{\tau = m\prod_{z=0}^3\kappa_z}^{t-1} \bigg[\sum_{l=1}^L \alpha_{l} \sum_{k=1}^{B_{l}} \alpha_{k}  \sum_{j=1}^{V_{k,l}} \alpha_{j} \sum_{i=1}^{U_{j,k,l}} \alpha_{i} \bigg( \Tilde{g} \big(\tilde{\mathbf{w}}_{i}^{\tau,0}\big) \frac{\pmb{1}_{i}^\tau}{p_{i}^\tau} -  \Tilde{f}_{i} \big(\tilde{\mathbf{w}}_{i}^{\tau,0}\big) \bigg) \bigg] \bigg\Vert^2 \nonumber \\ 
    &\overset{(b)}{=} \frac{4 \eta^2}{T} \sum_{t=0}^{T-1} \sum_{l=1}^{L} \alpha_l \sum_{\tau = m\prod_{z=0}^3\kappa_z}^{t-1} \mathbb{E} \bigg\Vert \sum_{k=1}^{B_l} \alpha_k \sum_{j=1}^{V_{k,l}} \alpha_j \sum_{i=1}^{U_{j,k,l} } \alpha_i \bigg( \Tilde{g} \big( \tilde{\mathbf{w}}_i^{\tau,0} \big) \frac{\pmb{1}_i^\tau}{p_i^\tau} - \nabla \tilde{f}_i \big( \tilde{\mathbf{w}}_i^{\tau,0} \big) \bigg) \bigg\Vert^2 - \nonumber\\
    &\Squad \frac{4 \eta^2}{T} \sum_{t=0}^{T-1} \sum_{\tau = m\prod_{z=0}^3\kappa_z}^{t-1} \mathbb{E} \bigg\Vert \sum_{l=1}^L \alpha_{l} \sum_{k=1}^{B_{l}} \alpha_{k}  \sum_{j=1}^{V_{k,l}} \alpha_{j} \sum_{i=1}^{U_{j,k,l}} \alpha_{i} \bigg( \Tilde{g} \big(\tilde{\mathbf{w}}_{i}^{\tau,0}\big) \frac{\pmb{1}_{i}^\tau}{p_{i}^\tau} -  \Tilde{f}_{i} \big(\tilde{\mathbf{w}}_{i}^{\tau,0}\big) \bigg) \bigg\Vert^2 \nonumber \\ 
    &\overset{(c)}{\leq} \frac{4\kappa_0 \kappa_1 \kappa_2 \kappa_3 \eta^2}{T} \sum_{t=0}^{T-1} \sum_{l=1}^{L} \alpha_l \mathbb{E} \bigg\Vert \sum_{k=1}^{B_l} \alpha_k \sum_{j=1}^{V_{k,l}} \alpha_j \sum_{i=1}^{U_{j,k,l} } \alpha_i \bigg( \Tilde{g} \big( \tilde{\mathbf{w}}_i^{t,0} \big) \frac{\pmb{1}_i^t} {p_i^t} \pm \Tilde{g} \big(\tilde{\mathbf{w}}_{i}^{t,0}\big) - \nabla \tilde{f}_i \big( \tilde{\mathbf{w}}_i^{t,0} \big) \bigg) \bigg\Vert^2 - \nonumber\\
    &\qquad \frac{4 \kappa_0 \kappa_1 \kappa_2 \kappa_3 \eta^2}{T} \sum_{t=0}^{T-1} \mathbb{E} \bigg\Vert \sum_{l=1}^L \alpha_{l} \sum_{k=1}^{B_{l}} \alpha_{k}  \sum_{j=1}^{V_{k,l}} \alpha_{j} \sum_{i=1}^{U_{j,k,l}} \alpha_{i} \bigg( \Tilde{g} \big(\tilde{\mathbf{w}}_{i}^{t,0}\big) \frac{\pmb{1}_{i}^t}{p_{i}^t} \pm \Tilde{g} \big(\tilde{\mathbf{w}}_{i}^{t,0}\big) -  \Tilde{f}_{i} \big(\tilde{\mathbf{w}}_{i}^{t,0}\big) \bigg) \bigg\Vert^2 \nonumber \\ 
    & \leq 8 \kappa_0 \kappa_1 \kappa_2 \kappa_3 \eta^2 \sigma^2 \sum_{l=1}^{L} \alpha_l \sum_{k=1}^{B_l} \alpha_k^2 \sum_{j=1}^{V_{k,l}} \alpha_j^2  \sum_{i=1}^{U_{j,k,l}} \alpha_i^2 +  \frac{8 \kappa_0 \kappa_1 \kappa_2 \kappa_3 \eta^2 G^2}{T} \sum_{l=1}^{L} \alpha_l \sum_{k=1}^{B_l} \alpha_k^2 \sum_{j=1}^{V_{k,l}} \alpha_j^2 \sum_{i=1}^{U_{j,k,l}} \alpha_i^2  \sum_{t = 0}^{T-1} \bigg(\frac{1 - p_i^t}{p_i^t}\bigg) \nonumber\\
    & \approx \mathcal{O} \big( \kappa_0\kappa_1\kappa_2\kappa_3 \eta^2 \sigma^2 \big) + \mathcal{O} \big(\kappa_0\kappa_1\kappa_2\kappa_3 \eta^2 G^2 \cdot \varphi_{\mathrm{w,L}_4}\big),  
\end{align}
where $\varphi_{\mathrm{w,L}_4} = \frac{1}{T} \sum_{l=1}^{L} \alpha_l \sum_{k=1}^{B_l} \alpha_k^2 \sum_{j=1}^{V_{k,l}} \alpha_j^2 \sum_{i=1}^{U_{j,k,l}} \alpha_i^2  \sum_{t = 0}^{T-1} \big(\frac{1 - p_i^t}{p_i^t}\big)$.

\subsection{Missing Proof of Lemma \ref{Lemma_50_1}}
\begin{align}
    & \frac{4 \eta^2}{T} \sum_{t=0}^{T-1} \sum_{l=1}^{L} \alpha_l \mathbb{E} \bigg\Vert \sum_{\tau = m\prod_{z=0}^3\kappa_z}^{t-1} \bigg[ \sum_{k=1}^{B_l} \alpha_k \sum_{j=1}^{V_{k,l}} \alpha_j \sum_{i=1}^{U_{j,k,l} } \alpha_i \nabla \tilde{f}_i \big( \tilde{\mathbf{w}}_i^{\tau,0} \big) -  \sum_{l'=1}^L \alpha_{l'} \sum_{k'=1}^{B_{l'}} \alpha_{k'}  \sum_{j'=1}^{V_{k',l'}} \alpha_{j'} \sum_{i'=1}^{U_{j',k',l'}} \alpha_{i'} \nabla \Tilde{f}_{i'} \big(\tilde{\mathbf{w}}_{i'}^{\tau,0}\big) \bigg] \bigg\Vert^2 \nonumber\\
    &= \frac{4\kappa_0 \kappa_1 \kappa_2 \kappa_3 \eta^2}{T} \sum_{t=0}^{T-1} \sum_{l=1}^{L} \alpha_l \mathbb{E} \bigg\Vert \bigg( \sum_{k=1}^{B_l} \alpha_k \sum_{j=1}^{V_{k,l}} \alpha_j \sum_{i=1}^{U_{j,k,l} } \alpha_i \big[\nabla \tilde{f}_i \big( \tilde{\mathbf{w}}_i^{t,0} \big) - \nabla \tilde{f}_i \big( \bar{\mathbf{w}}_j^t \big) \big] \bigg) + \nonumber \\
    &\qquad \bigg( \sum_{k=1}^{B_l} \alpha_k \sum_{j=1}^{V_{k,l}} \alpha_j \sum_{i=1}^{U_{j,k,l} } \alpha_i \big[\nabla \tilde{f}_i \big( \bar{\mathbf{w}}_j^t \big) - \nabla \tilde{f}_i \big( \bar{\mathbf{w}}_k^t \big) \big] \bigg) + \bigg( \sum_{k=1}^{B_l} \alpha_k \sum_{j=1}^{V_{k,l}} \alpha_j \sum_{i=1}^{U_{j,k,l} } \alpha_i \big[\nabla \tilde{f}_i \big( \bar{\mathbf{w}}_k^t \big) - \nabla \tilde{f}_i \big( \bar{\mathbf{w}}_l^t \big) \big] \bigg) + \nonumber\\
    &\qquad \bigg( \sum_{k=1}^{B_l} \alpha_k \sum_{j=1}^{V_{k,l}} \alpha_j \sum_{i=1}^{U_{j,k,l} } \alpha_i \big[\nabla \tilde{f}_i \big( \bar{\mathbf{w}}_l^t \big) - \nabla \tilde{f}_i \big( \bar{\mathbf{w}}^t \big) \big] \bigg) + \bigg( \sum_{k=1}^{B_l} \alpha_k \sum_{j=1}^{V_{k,l}} \alpha_j \sum_{i=1}^{U_{j,k,l} } \alpha_i \nabla \tilde{f}_i \big( \bar{\mathbf{w}}^t \big) - \nonumber\\
    &\qquad \sum_{l'=1}^L \alpha_{l'} \sum_{k'=1}^{B_{l'}} \alpha_{k'} \sum_{j'=1}^{V_{k',l'}} \alpha_{j'} \sum_{i'=1}^{U_{j',k',l'}} \alpha_{i'} \Tilde{f}_{i'} \big(\bar{\mathbf{w}}^t\big) \bigg) + \bigg( \sum_{l'=1}^L \alpha_{l'} \sum_{k'=1}^{B_{l'}} \alpha_{k'} \sum_{j'=1}^{V_{k',l'}} \alpha_{j'} \sum_{i'=1}^{U_{j',k',l'}} \alpha_{i'} \big[ \nabla \Tilde{f}_{i'} \big(\bar{\mathbf{w}}^t\big) - \nabla \Tilde{f}_{i'} \big(\bar{\mathbf{w}}_l^t\big) \big] \bigg) + \nonumber \\
    &\qquad \bigg( \sum_{l'=1}^L \alpha_{l'} \sum_{k'=1}^{B_{l'}} \alpha_{k'} \sum_{j'=1}^{V_{k',l'}} \alpha_{j'} \sum_{i'=1}^{U_{j',k',l'}} \alpha_{i'} \big[ \nabla \Tilde{f}_{i'} \big(\bar{\mathbf{w}}_l^t\big) - \nabla \Tilde{f}_{i'} \big(\bar{\mathbf{w}}_k^t\big) \big] \bigg) + \nonumber \\
    &\qquad \bigg( \sum_{l'=1}^L \alpha_{l'} \sum_{k'=1}^{B_{l'}} \alpha_{k'} \sum_{j'=1}^{V_{k',l'}} \alpha_{j'} \sum_{i'=1}^{U_{j',k',l'}} \alpha_{i'} \big[ \nabla \Tilde{f}_{i'} \big(\bar{\mathbf{w}}_k^t\big) - \nabla \Tilde{f}_{i'} \big(\bar{\mathbf{w}}_j^t\big) \big] \bigg) + \nonumber\\
    &\qquad \bigg( \sum_{l'=1}^L \alpha_{l'} \sum_{k'=1}^{B_{l'}} \alpha_{k'} \sum_{j'=1}^{V_{k',l'}} \alpha_{j'} \sum_{i'=1}^{U_{j',k',l'}} \alpha_{i'} \big[ \nabla \Tilde{f}_{i'} \big(\bar{\mathbf{w}}_l^t\big) -  \nabla \Tilde{f}_{i'} \big(\tilde{\mathbf{w}}_{i'}^{t,0}\big) \big] \bigg) ~ \bigg\Vert^2  \\
    &\leq 36 \big(\beta \epsilon \eta \kappa_0 \kappa_1 \kappa_2 \kappa_3 \big)^2 + \frac{144 \big(\beta \eta \kappa_0 \kappa_1 \kappa_2 \kappa_3 \big)^2}{T} \sum_{t=0}^{T-1} \sum_{l=1}^{L} \alpha_l \sum_{k=1}^{B_l} \alpha_k \sum_{j=1}^{V_{k,l}} \alpha_j \sum_{i=1}^{U_{j,k,l} } \alpha_i \mathbb{E} \bigg\Vert \tilde{\mathbf{w}}_i^t - \Tilde{\mathbf{w}}_i^{t,0} \bigg\Vert^2 + \nonumber\\
    &\Squad \frac{144 \big(\beta \eta \kappa_0 \kappa_1 \kappa_2 \kappa_3 \big)^2}{T} \sum_{t=0}^{T-1} \sum_{l=1}^{L} \alpha_l \sum_{k=1}^{B_l} \alpha_k \sum_{j=1}^{V_{k,l}} \alpha_j \sum_{i=1}^{U_{j,k,l} } \alpha_i \mathbb{E} \bigg\Vert \Bar{\mathbf{w}}_j^t - \Tilde{\mathbf{w}}_i^t \bigg\Vert^2 + \nonumber\\
    &\Squad \frac{72 \big(\beta \eta \kappa_0 \kappa_1 \kappa_2 \kappa_3 \big)^2}{T} \sum_{t=0}^{T-1} \sum_{l=1}^{L} \alpha_l \sum_{k=1}^{B_l} \alpha_k \sum_{j=1}^{V_{k,l}} \alpha_j \mathbb{E} \bigg\Vert \Bar{\mathbf{w}}_k^t - \Bar{\mathbf{w}}_j^t \bigg\Vert^2 + \nonumber\\
    &\Squad \frac{72 \big(\beta \eta \kappa_0 \kappa_1 \kappa_2 \kappa_3 \big)^2}{T} \sum_{t=0}^{T-1} \sum_{l=1}^{L} \alpha_l \sum_{k=1}^{B_l} \alpha_k \mathbb{E} \bigg\Vert \Bar{\mathbf{w}}_l^t - \Bar{\mathbf{w}}_k^t \bigg\Vert^2 + \nonumber\\
    &\Squad \frac{72 \big(\beta \eta \kappa_0 \kappa_1 \kappa_2 \kappa_3 \big)^2}{T} \sum_{t=0}^{T-1} \sum_{l=1}^{L} \alpha_l  \mathbb{E} \bigg\Vert \Bar{\mathbf{w}}^t - \Bar{\mathbf{w}}_l^t \bigg\Vert^2 \nonumber \\
    &\leq \mathcal{O} \big(\kappa_0^2 \kappa_1^2 \kappa_2^2 \kappa_3^2 \beta^2 \eta^2 \epsilon^2 \big) + \mathcal{O} \big( \delta^{\mathrm{th}} \kappa_0^2 \kappa_1^2 \kappa_2^2 \kappa_3^2 \beta^2 \eta^2 D^2 \big) + \mathcal{O} \big( \kappa_0^3 \kappa_1^2 \kappa_2^2 \kappa_3^2 \eta^4 \beta^2 \sigma^2 \big) + \nonumber\\
    &\qquad \mathcal{O} \big( \kappa_0^4 \kappa_1^2 \kappa_2^2 \kappa_3^2 \eta^4 \beta^2 \epsilon_{\mathrm{vc}}^2 \big) + \mathcal{O} \big( \kappa_0^3 \kappa_1^2 \kappa_2^2 \kappa_3^2 \eta^4 \beta^2 G^2 \cdot \varphi_{\mathrm{w, L}_1} \big) + \mathcal{O} \big(\delta^{\mathrm{th}} \kappa_0^2 \kappa_1^2 \kappa_2^2 \kappa_3^2 \eta^2 \beta^2 D^2 \big) + \nonumber\\
    &\qquad \mathcal{O} \big(\kappa_0^6 \kappa_1^4 \kappa_2^2 \kappa_3^2 \eta^6 \beta^4 \epsilon_{\mathrm{vc}}^2 \big) + \mathcal{O} \big(\kappa_0^4 \kappa_1^4 \kappa_2^2 \kappa_3^2 \eta^4 \beta^2 \epsilon_{\mathrm{sbs}}^2 \big) + \mathcal{O} \big( \kappa_0^3 \kappa_1^3 \kappa_2^2 \kappa_3^2 \eta^4 \sigma^2 \beta^2 \big) + \nonumber\\
    &\qquad \mathcal{O } \big( \delta^{th} \kappa_0^2 \kappa_1^2 \kappa_2^2 \kappa_3^2 \beta^2 \eta^2 D^2 \big) + \mathcal{O} \big( \kappa_0^5 \kappa_1^4 \kappa_2^2 \kappa_3^2 \beta^4 \eta^6 G^2 \cdot \varphi_{\mathrm{w, L}_1} \big) + \mathcal{O} \big( \kappa_0^3 \kappa_1^3 \kappa_2^2 \kappa_3^2 \beta^2 \eta^4 \cdot \varphi_{\mathrm{w, L}_2}) + \nonumber\\
    &\qquad \mathcal{O} \big( \kappa_0^5 \kappa_1^4 \kappa_2^4 \kappa_3^2 \eta^6 \beta^4 \epsilon_{\mathrm{vc}}^2 \big) + \mathcal{O} \big(\kappa_0^6 \kappa_1^6 \kappa_2^4 \kappa_3^2 \eta^6 \beta^4 \epsilon_{\mathrm{sbs}}^2 \big) + \mathcal{O} \big( \kappa_0^4 \kappa_1^4 \kappa_2^4 \kappa_3^2 \eta^4 \beta^4 \epsilon_{\mathrm{mbs}}^2 \big) + \nonumber\\
    &\qquad \mathcal{O} \big( \kappa_0^3 \kappa_1^3 \kappa_2^3 \kappa_3^2 \eta^4 \beta^2 \sigma^2 \big) + \mathcal{O} \big( \delta^{\mathrm{th}} \kappa_0^2 \kappa_1^2 \kappa_2^2 \kappa_3^2 \eta^2 \beta^2 D^2 \big) + \mathcal{O} \big( \kappa_0^5 \kappa_1^4 \kappa_2^4 \kappa_3^2 \eta^6 \beta^4 G^2 \cdot \varphi_{\mathrm{w, L}_1} \big) + \nonumber\\
    &\qquad \mathcal{O} \big( \kappa_0^4 \kappa_1^4 \kappa_2^4 \kappa_3^2 \beta^4 \eta^6 \cdot \varphi_{\mathrm{w, L}_2}) + \mathcal{O} \big( \kappa_0^3 \kappa_1^3 \kappa_2^3 \kappa_3^2 \eta^4 \beta^2 G^2 \cdot \varphi_{\mathrm{w,L}_3} \big) +\nonumber\\
    &\qquad \frac{72 \big(\beta \eta \kappa_0 \kappa_1 \kappa_2 \kappa_3 \big)^2}{T} \sum_{t=0}^{T-1} \sum_{l=1}^{L} \alpha_l  \mathbb{E} \bigg\Vert \Bar{\mathbf{w}}^t - \Bar{\mathbf{w}}_l^t \bigg\Vert^2 \nonumber \\
    & \approx \mathcal{O} \big( \kappa_0^4 \kappa_1^2 \kappa_2^2 \kappa_3^2 \eta^4 \beta^2 \epsilon_{\mathrm{vc}}^2 \big) + \mathcal{O} \big(\kappa_0^4 \kappa_1^4 \kappa_2^2 \kappa_3^2 \eta^4 \beta^2 \epsilon_{\mathrm{sbs}}^2 \big) + \mathcal{O} \big( \kappa_0^4 \kappa_1^4 \kappa_2^4 \kappa_3^2 \eta^4 \beta^4 \epsilon_{\mathrm{mbs}}^2 \big) + \nonumber\\
    &\Squad \mathcal{O} \big(\kappa_0^2 \kappa_1^2 \kappa_2^2 \kappa_3^2 \beta^2 \eta^2 \epsilon^2 \big) + \mathcal{O} \big( \kappa_0^3 \kappa_1^2 \kappa_2^2 \kappa_3^2 \eta^4 \beta^2 \sigma^2 \big)  + \mathcal{O} \big( \delta^{\mathrm{th}} \kappa_0^2 \kappa_1^2 \kappa_2^2 \kappa_3^2 \eta^2 \beta^2 D^2 \big) + \nonumber \\
    &\Squad \mathcal{O} \big( \kappa_0^3 \kappa_1^2 \kappa_2^2 \kappa_3^2 \eta^4 \beta^2 G^2 \cdot \varphi_{\mathrm{w, L}_1} \big) + \mathcal{O} \big( \kappa_0^3 \kappa_1^3 \kappa_2^2 \kappa_3^2 \beta^2 \eta^4 \cdot \varphi_{\mathrm{w, L}_2}) + \nonumber\\
    &\Squad \rs \mathcal{O} \big( \kappa_0^3 \kappa_1^3 \kappa_2^3 \kappa_3^2 \eta^4 \beta^2 G^2 \cdot \varphi_{\mathrm{w,L}_3} \big) + \frac{72 \big(\beta \eta \kappa_0 \kappa_1 \kappa_2 \kappa_3 \big)^2}{T} \sum_{t=0}^{T-1} \sum_{l=1}^{L} \alpha_l  \mathbb{E} \bigg\Vert \Bar{\mathbf{w}}^t - \Bar{\mathbf{w}}_l^t \bigg\Vert^2\rs. \rs 
\end{align}
\end{proof}

\end{document}